\documentclass[
 reprint,
 superscriptaddress,
 amsmath,amssymb,
 aps,physrev
]{revtex4-2}
\usepackage[dvipsnames]{xcolor}
\usepackage[
    colorlinks=true,
    linkcolor=NavyBlue,
    citecolor=NavyBlue,
    urlcolor=NavyBlue
]{hyperref}
\usepackage{graphicx}
\usepackage{dcolumn}
\usepackage{bm}
\usepackage{braket}
\usepackage{amsthm}
\usepackage{amsmath}
\newtheorem*{theorem*}{Theorem} 
\newtheorem{fact}{Fact}
\newtheorem{theorem}{Theorem}

\newtheorem{lemma}{Lemma}
\newtheorem{proposition}{Proposition}
\newtheorem{corollary}{Corollary}
\theoremstyle{definition}
\newtheorem{definition}{Definition}
\usepackage{mathtools}
\theoremstyle{remark}
\newtheorem{remark}{Remark}
\usepackage{kotex}
\usepackage{booktabs}
\usepackage[colorinlistoftodos]{todonotes}

\allowdisplaybreaks

\begin{document}

\title{Quantum state learning beyond approximate unitary designs}

\author{Gyungmin Cho}
 \email{km950501@snu.ac.kr}
 \affiliation{NextQuantum Center, Department of Physics and Astronomy, and Institute of Applied Physics, Seoul National University, Seoul 08826, Korea}

\author{Changhun Oh}
\email{changhun0218@gmail.com}
\affiliation{Department of Physics, Korea Advanced Institute of Science and Technology, Daejeon 34141, Korea}

\author{Dohun Kim}
\email{dohunkim@snu.ac.kr}
\affiliation{NextQuantum Center, Department of Physics and Astronomy, and Institute of Applied Physics, Seoul National University, Seoul 08826, Korea}
 
             
\begin{abstract}
    Approximate unitary designs reproduce the statistics of Haar-random unitaries to a given order and accuracy. Recent constructions realize such designs with logarithmic-depth circuits, enabling shallow measurement protocols for various quantum state-learning tasks while preserving performance. These results raise the question of whether approximate designs can replace exact designs more generally. We show that even exponentially small design error need not preserve the learning guarantees of exact designs. This motivates deriving learning guarantees directly from the measurement circuit structure. For observable estimation with classical shadows, we prove that logarithmic-depth two-layer Clifford circuits yield an unbiased estimator matching the global Clifford variance scaling for every state and Hermitian observable. Beyond observable estimation, this bound allows shallow measurements to retain global Clifford guarantees for other tasks, including setting-efficient tomography, mixed-state metrology, and stabilizer structure learning. We complement these statistical guarantees with a compact, exact tensor-network representation of the inverse shadow channel. When measurement bases are reused, approximate unitary designs of arbitrarily high order need not uniformly match the Haar variance in the many-shot limit. For the all-to-all random two-local circuits considered, doing so requires nearly linear depth. Together, our results reveal the capabilities and limitations of shallow quantum state learning, highlighting the distinction between approximating Haar randomness and reproducing its learning guarantees.
\end{abstract}

\maketitle

\section{introduction}

\begin{figure*}[t]
    \centering
    \includegraphics[width=1\linewidth, trim=2.5cm 19.7cm 2cm 1cm, clip]{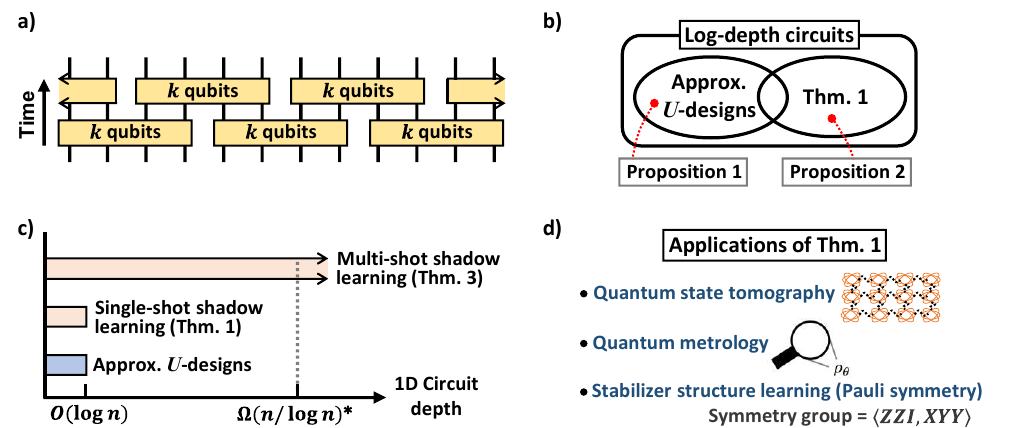}
    \caption{
    \textbf{Overview of the two-layer architecture and main results.}
    (a) Two-layer block unitary ensemble consisting of two staggered layers of non-overlapping $k$-qubit blocks with periodic boundary conditions.
    (b) The approximate unitary design condition is neither sufficient nor necessary for the variance guarantee of Theorem~\ref{main:thm:1}, as shown by Proposition~\ref{main:prop:not_sufficient} and~\ref{main:prop:not_necessary}, respectively.
    (c) Circuit-depth requirements in an ancilla-free architecture. Approximate unitary designs in relative error and the variance scaling of Theorem~\ref{main:thm:1} are attainable at depth $O(\log n)$ in a 1D ancilla-free architecture. For multi-shot shadow learning with circuits of independent Haar-random two-qubit gates, matching the global exact unitary $4$-design variance floor within a constant factor uniformly over all states and observables requires depth $\Omega(n/\log n)$. The asterisk indicates that this lower bound remains valid even under all-to-all connectivity.
    (d) Applications of Theorem~\ref{main:thm:1}, where global Clifford measurements can be replaced by shallow two-layer Clifford measurements.
    }
    \label{main:fig:1}
\end{figure*}

Quantum state learning aims to infer properties of an unknown state $\rho$ from measurement data. Examples include estimating observable expectation values~\cite{huang2020predicting} and entanglement entropies~\cite{satzinger2021realizing,brydges2019probing}, reconstructing $\rho$ through quantum state tomography~\cite{odonnell2016tomography, haah2017tomography}, verifying or certifying its closeness to a target state~\cite{du2025certifying, huang2025certifying}, and estimating parameters $\theta$ encoded in $\rho_\theta$ through quantum metrology~\cite{zhou2026randomized}. Quantum state learning thus arises in a wide range of problems across quantum information and many-body physics.

A single measurement basis generally provides only partial information about $\rho$. For example, computational basis measurements reveal only its diagonal elements. Recovering additional information therefore requires measurements in different bases. Randomized measurements~(RMs) do so by sampling a unitary $U$ from an ensemble before each measurement~\cite{elben2023randomized}. With suitable ensembles and estimators, RMs can achieve optimal sample complexity for predicting many observable expectation values~\cite{huang2020predicting} and can also estimate nonlinear properties~\cite{satzinger2021realizing,brydges2019probing}, with only modest hardware overhead.

Each measurement produces a bit string $b\in\{0,1\}^n$, and the pair $(U,b)$ forms a measurement record. An RM learning protocol applies a map $\mathcal{A}$ to these data to estimate a property $y(\rho)$. Schematically,
\begin{equation}\label{main:eq:algo}
    \rho
    \xrightarrow{\mathrm{RM}}
    (U,b)
    \xrightarrow{\mathcal{A}}
    \hat{y}(\rho).
\end{equation}
Here, $\hat{y}(\rho)$ denotes an estimate of the target quantity $y(\rho)$, which depends on the learning task. 
In classical shadows, the post-processing is constructed from the shadow channel associated with the measurement ensemble $\mathcal E$~\cite{huang2020predicting}, defined by
\begin{equation}
    \mathcal{M}_{\mathcal{E}}(\rho)
    =
    \mathbb{E}_{U\sim\mathcal{E}}
    \sum_b
    \langle b|U\rho U^\dagger|b\rangle
    U^\dagger\ket{b}\!\bra{b}U.
\end{equation}
If $\mathcal{M}_{\mathcal{E}}$ is invertible, each record $(U,b)$ gives the shadow estimator
\begin{equation}
    \hat{\rho}
    =
    \mathcal{M}_{\mathcal{E}}^{-1}
    \left(
        U^\dagger\ket{b}\!\bra{b}U
    \right).
\end{equation}
By construction, $\mathbb{E}_{U\sim \mathcal{E},b}[\hat{\rho}]=\rho$, so $\hat{\rho}$ is unbiased. Its statistical properties depend on the unitary ensemble $\mathcal{E}$. 

Haar-random unitaries provide a natural reference ensemble for RMs, but generally require exponentially many gates to implement. The corresponding shadow channel is $\mathcal M_{\mathrm H}(A)=(A+\operatorname{Tr}(A)I)/(2^n+1)$, giving the estimator $\hat{\rho}_{\mathrm H}=(2^n+1)U^\dagger\ket b\!\bra bU-I$. Nonetheless, many learning guarantees depend only on moments up to a fixed order $t$. 
For these guarantees, Haar randomness can be replaced by a unitary $t$-design~\cite{gross2007evenly}, which reproduces the Haar moments up to the order $t$. Formally, for an ensemble $\mathcal{E}$, its $t$-th moment channel is $\Phi_{\mathcal{E}}^{(t)}(X)=\mathbb{E}_{U\sim\mathcal{E}}[U^{\otimes t}X(U^\dagger)^{\otimes t}]$. The ensemble is a unitary $t$-design if $\Phi_{\mathcal{E}}^{(t)}(X)=\Phi_{\mathrm{H}}^{(t)}(X)$ for every operator $X$, where $\Phi_{\mathrm{H}}^{(t)}$ is the Haar moment channel~\cite{collins2006integration}.

For example, the $n$-qubit Clifford group $\mathrm{Cl}(n)$ forms an exact unitary $3$-design~\cite{webb2015clifford}, whereas no general efficient construction is known for $t\geq 4$. 
To address this limitation, approximate unitary designs relax exact moment matching while retaining controllable errors, as expressed by
\begin{equation}\label{main:eq:cp}
    (1-\epsilon_{\mathrm{des}})
    \Phi_{\mathrm{H}}^{(t)}
    \preceq
    \Phi_{\mathcal{E}}^{(t)}
    \preceq
    (1+\epsilon_{\mathrm{des}})
    \Phi_{\mathrm{H}}^{(t)},
\end{equation}
where $\Phi_1 \preceq \Phi_2$ denotes the completely positive (CP) order, i.e., $\Phi_2-\Phi_1$ is completely positive. Recently, it has been proven that such $\epsilon_{\mathrm{des}}$-approximate unitary $t$-designs can be realized by 1D circuits without ancillary qubits with depth $O\!\left(t\,\mathrm{poly}(\log t)\,\log(n/\epsilon_{\mathrm{des}})\right)$~\cite{schuster2025random, laracuente2026approximate}.

Building on these developments, the CP-order condition in Eq.~\eqref{main:eq:cp} has been used to extend learning guarantees from exact to approximate designs for state overlap estimation~\cite{cioli2025approximate,schuster2025random}, nonlinear property estimation~\cite{zheng2025distributed,li2026quantum}, and quantum metrology~\cite{du2026complexity}. This motivates a broader question: does the approximate-design condition in Eq.~\eqref{main:eq:cp} suffice to reproduce the learning guarantees of exact designs?

We show that it does not. CP-order bounds do not directly control post-processing maps that are generally not completely positive, such as the inverse shadow channel. 
In existing analyses~\cite{schuster2025random, li2026quantum}, the condition in Eq.~\eqref{main:eq:cp} is used to bound the bias and variance of the Haar-inverse estimator $\hat{\rho}_{\mathrm H}$.
However, when the measurements are drawn from an approximate design $\mathcal E$, this estimator is generally biased: $\mathbb E_{U\sim\mathcal E,b}[\hat{\rho}_{\mathrm H}]\neq\rho$. Averaging more samples does not reduce this bias, which can limit the extension of these guarantees beyond the settings covered above.

More fundamentally, this limitation is not specific to the Haar-inverse estimator: the approximate-design condition alone does not ensure the learning guarantees of exact designs.
For example, exact unitary $3$-design measurements allow quantum state tomography with near-optimal sample complexity~\cite{guctua2020fast, Lowe2025lower}. In contrast, we construct approximate unitary $3$-designs with exponentially small design error for which tomography at constant trace-norm accuracy requires arbitrarily many independent randomized measurements, regardless of the estimator.

One way to overcome this limitation is to impose a stronger condition, provided by approximate strong unitary designs in relative error, which control mixed moments of $U$ and $U^*$~\cite{schuster2025strong}.
This notion of approximation allows us to recover the relevant Haar variance bounds for the learning tasks considered here (see Appendix~\ref{app:sec:strong_design_learning}). However, circuits composed of independent Haar-random two-qubit gates require depth $\Omega(n)$ to satisfy this condition, even with all-to-all connectivity~\cite{schuster2025strong}. 
These limitations motivate a direct analysis of the measurement circuit. We focus on the following two-layer block ensemble.

\begin{definition}[Two-layer block unitary ensemble~\cite{schuster2025random}]
    The two-layer block unitary ensemble $\mathcal{E}_{2\mathrm{L}}$ consists of two staggered layers of $k$-qubit unitaries. Each layer is a tensor product of unitaries acting on non-overlapping contiguous blocks, with periodic boundary conditions~[Fig.~\ref{main:fig:1}a]. The block unitaries are sampled independently from the ensembles specified below.
\end{definition}

We denote the Haar-random unitary ensemble on $n$ qubits as $\mathrm{Haar}(n)$. 
The same two-layer architecture also appears in shallow constructions of approximate unitary designs~\cite{schuster2025random, laracuente2026approximate}. 
Instead of inferring learning guarantees from the approximate-design condition in Eq.~\eqref{main:eq:cp}, we derive them directly from the staggered circuit structure.
Specifically, the moments governing the learning statistics can be expressed as transfer matrix contractions, providing a direct route from the circuit structure to variance bounds.

Using this approach, we show that logarithmic-depth two-layer Clifford circuits yield an unbiased single-shot shadow estimator that matches the $\mathrm{Cl}(n)$ variance scaling for every state and Hermitian observable. The same bound transfers global-Clifford guarantees to several learning tasks beyond observable estimation, including QST, quantum metrology, and stabilizer learning. Moreover, the inverse shadow channel admits an exact matrix product operator~(MPO) representation with bond dimension $n/k$, which enables exact and efficient evaluation of shadow estimates for observables that also admit MPO representations with polynomial bond dimension.
We then examine whether this performance extends to multi-shot shadow estimation, where each measurement basis is reused and the variance also depends on fourth-order moments.
Even when every $k$-qubit block is Haar random, the two-layer architecture cannot reproduce the worst-case variance scaling of a global exact unitary $4$-design unless $k=\Omega(n)$. 
Together, these results show that shallow circuits can retain strong learning guarantees beyond those established by the approximate-design condition alone, while the attainable performance depends on the learning task and the statistical quantity being estimated.

Throughout, we set $d=2^n$ and use $P=\bigotimes_{i=1}^n P_i\in\{I,X,Y,Z\}^{\otimes n}$ to denote an $n$-qubit Pauli operator. Its weight is $|P|=|\{i\in[n]:P_i\neq I\}|$, and we write $\mathcal Z=\{I,Z\}^{\otimes n}$. We write $\|A\|_p=(\operatorname{Tr}|A|^p)^{1/p}$ for the Schatten $p$-norm. We use $N_U$ for the number of independently sampled measurement unitaries and $N_S$ for the number of measurement shots performed with each unitary, with the total number of samples $T=N_U N_S$. For an observable $O$, we denote its traceless part by $O_0=O-\operatorname{Tr}(O)I/d$.


\section{Results}
\subsection{Statistical guarantees and limitations for learning with shallow measurements}

We begin with a single-shot shadow estimation guarantee for the two-layer Clifford ensemble.

\begin{theorem}[Single-shot shadow estimation, Informal]\label{main:thm:1}
    For the two-layer Clifford ensemble with block
    size $k$ satisfying $k2^{k/2}=\Omega(n)$, and for every state $\rho$ and
    Hermitian observable $O$, the unbiased estimator
    $\hat{o}=\operatorname{Tr}(O\hat{\rho})$ satisfies
    \begin{equation}
        \operatorname{Var}[\hat{o}]
        =
        O\!\left(\lVert O_0\rVert_2^2\right).
    \end{equation}
\end{theorem}
The estimator uses the exact inverse shadow channel of the two-layer ensemble, ensuring unbiasedness. Establishing the variance bound requires controlling both the inverse Pauli visibilities arising from this channel and the joint probabilities that pairs of Pauli operators are mapped to $Z$-type operators. We exploit the staggered block structure to express these quantities as transfer matrix contractions and bound their combined contribution. 

Theorem~\ref{main:thm:1} resolves an open question posed in~\cite{bertoni2024shallow} and partially addressed for low-rank observables in~\cite{schuster2025random}: whether shallow shadows can achieve the same Frobenius-norm variance scaling as global Clifford measurements. 
This bound has applications beyond observable estimation: the performance guarantees of several quantum state learning protocols based on $\mathrm{Cl}(n)$ can be derived from the same variance bound.

For a $\gamma$-local observable $O=\sum_{P:|P|\leq\gamma}\operatorname{Tr}(OP)P/d$ with $\gamma=O(1)$ and bounded $\|O_0\|_\infty$, shallow Clifford measurements can offer an advantage over global Clifford measurements. Accounting for the locality of $O$ gives a sharper bound than the general variance guarantee in Theorem~\ref{main:thm:1}: with $k=O(\log n)$ satisfying the same block-size condition, the two-layer Clifford ensemble satisfies
\begin{equation}
    \operatorname{Var}[\hat o]
    \leq
    2^{O(k\gamma)}\|O_0\|_\infty^2
    =
    \operatorname{poly}(n),
\end{equation}
whereas the global Clifford bound can grow exponentially with $n$.


We focus on a 1D ancilla-free implementation to keep the hardware assumptions minimal. In this setting, each $k$-qubit Clifford block can be implemented in depth $O(k)$, so choosing $k=O(\log n)$ to satisfy Theorem~\ref{main:thm:1} yields an overall depth of $O(\log n)$. If all-to-all connectivity and $\widetilde O(n)$ ancillary qubits are available, the same ensemble can instead be implemented in depth $O(\log\log n)$~\cite{jiang2020optimal}. Here, $\widetilde O(\cdot)$ suppresses polylogarithmic factors in $n$.

For comparison, Ref.~\cite{schuster2025random} uses the Haar inverse shadow channel rather than the exact inverse channel of the measurement ensemble. The resulting estimator $\hat{o}_{\mathrm H}=\operatorname{Tr}(O\hat{\rho}_{\mathrm H})$ is generally biased, but its bias and variance can be bounded using the approximate-design condition in Eq.~\eqref{main:eq:cp}. For an $\epsilon_{\mathrm{des}}$-approximate unitary $3$-design,
\begin{align}\label{main:eq:prev_est}
    \left|
        \mathbb E[\hat{o}_{\mathrm H}]
        -\operatorname{Tr}(O\rho)
    \right|
    &=
    O\!\left(
        \epsilon_{\mathrm{des}}\|O_0\|_1
    \right),
    \nonumber\\
    \operatorname{Var}[\hat{o}_{\mathrm H}]
    &=
    O\!\left(
        \|O_0\|_2^2
        +\epsilon_{\mathrm{des}}\|O_0\|_1^2
    \right).
\end{align}
While averaging more samples reduces the statistical fluctuations, it does not reduce the bias.
The above bound guarantees a bias of at most $\epsilon$ when, for $O_0\neq0$,
\begin{equation}\label{main:eq:err_des}
    \epsilon_{\mathrm{des}}
    =
    O\!\left(
        \frac{\epsilon}{\|O_0\|_1}
    \right).
\end{equation}
Since the circuit-depth bound scales as $O(\log(n/\epsilon_{\mathrm{des}}))$, the depth required by this construction depends on the design accuracy needed for the learning task.
For learning protocols based on $\hat{\rho}_{\mathrm H}$, such as quantum state tomography or mixed-state quantum metrology, certifying even constant accuracy in the task's natural error measure can require exponentially small $\epsilon_{\mathrm{des}}$. In such cases, the resulting circuit-depth bound no longer guarantees logarithmic depth.

In contrast, Theorem~\ref{main:thm:1} gives an unbiased estimator with the Frobenius-norm variance bound. The measurement depth remains $O(\log n)$ independently of the target error, so higher precision can be achieved by collecting more samples at the same circuit depth. This raises the question of whether the same guarantee follows from the approximate design condition in Eq.~\eqref{main:eq:cp}, without any additional information about the circuit structure. The following proposition shows that this implication does not hold.

\begin{proposition}[Approximate designs are not sufficient (informal)]\label{main:prop:not_sufficient}
    For any $\epsilon_{\mathrm{des}}>14/2^n$, $\alpha>0$, and $0\leq\beta<1$, there exists an $\epsilon_{\mathrm{des}}$-approximate unitary $3$-design $\mathcal E$ satisfying Eq.~\eqref{main:eq:cp} such that, 
    for every estimator $\hat f_O(U,b)$ with relative bias at most $\beta$, there exist a state $\rho$ and a Hermitian observable $O$ with $O_0\neq0$ for which
    \begin{equation}
        \operatorname{Var}_{U,b}[\hat f_O(U,b)]
        \geq
        \alpha\lVert O_0\rVert_2^2\label{main:eq:var_lower}.
    \end{equation}
\end{proposition}

Some control of estimation accuracy is essential for a variance lower bound: an estimator that always outputs a constant has zero variance, regardless of the measurement ensemble. Here, we impose the uniform relative bias condition $|\mathbb E_{U,b}[\hat f_O(U,b)]-\operatorname{Tr}(\rho O)|\leq\beta|\operatorname{Tr}(\rho O)|$ for every state $\rho$ and Hermitian observable $O$. The case $\beta=0$ includes every unbiased estimator.

Since $\alpha$ can be arbitrarily large in Eq.~\eqref{main:eq:var_lower}, the approximate design condition in Eq.~\eqref{main:eq:cp} alone cannot guarantee the Frobenius-norm variance bound of global Clifford measurements. The same construction also rules out uniform sample complexity guarantees for certain quantum state learning problems; see Appendix~\ref{app:subsec:approximate_designs_not_sufficient}. We next ask whether the approximate design condition is necessary for the variance bound.

\begin{proposition}[Approximate designs are not necessary]\label{main:prop:not_necessary}
    Under the same block size condition $k2^{k/2}=\Omega(n)$, the variance bound in Theorem~\ref{main:thm:1} remains valid for a two-layer Clifford circuit with $\ell=O(1)$ disconnected components. For $\ell\ge2$, the resulting tensor product ensemble does not satisfy the approximate unitary 3-design condition in Eq.~\eqref{main:eq:cp} for any $\epsilon_{\mathrm{des}}<1$.
\end{proposition}

Together with Proposition~\ref{main:prop:not_sufficient}, this shows that the approximate unitary design condition in Eq.~\eqref{main:eq:cp} is neither necessary nor sufficient for the guarantee in Theorem~\ref{main:thm:1} [Fig.~\ref{main:fig:1}b]. 

We next ask whether the two-layer architecture can also reproduce exact-design performance in multi-shot observable estimation~\cite{zhou2023performance,helsen2023thrifty}, where each measurement unitary is reused for multiple shots and the variance depends on fourth-order moments.
Given $\{U_i\}_{i=1}^{N_U}$ independently sampled unitaries, we use the multi-shot shadow estimator
\begin{equation}
    \hat{\rho}_{\mathrm{mul}}
    =
    \frac{1}{N_U N_S}
    \sum_{i=1}^{N_U}\sum_{j=1}^{N_S}
    \mathcal{M}^{-1}
    \left(
    U_i^{\dagger}\ket{b_{i,j}}\!\bra{b_{i,j}}U_i
    \right)\label{main:eq:est_mul}.
\end{equation}
Here, $\{b_{i,j}\}_{j=1}^{N_S}$ are measurement outcomes obtained from $U_i\rho U_i^\dagger$, and the observable $O$ is estimated as $\hat{o}=\operatorname{Tr}(O\hat{\rho}_{\mathrm{mul}})$. Reusing the same measurement unitary introduces an additional variance term governed by fourth-order moments. 
At fixed $N_U$, we refer to the variance that remains as $N_S\to\infty$ as the multi-shot \textit{variance floor}.
For an exact unitary $4$-design, this variance floor is $O(2^{-n}\|O_0\|_2^2/N_U)$.
For the two-layer architecture, however, this worst-case scaling cannot be recovered with a sublinear block size.

\begin{theorem}[Informal]\label{main:thm:2}
    Even when every block unitary in the two-layer circuit is sampled independently from the Haar measure, recovering the worst-case multi-shot variance-floor scaling of a global exact unitary $4$-design requires a block size $k=\Omega(n)$.
\end{theorem}
In contrast to the single-shot result, Theorem~\ref{main:thm:2} shows that the unbiased shadow estimator in Eq.~\eqref{main:eq:est_mul} can have a variance floor exponentially larger than the global exact unitary $4$-design value for the same state and observable. This separation can occur even when the ensemble is an approximate unitary $4$-design with exponentially small error $\epsilon_{\mathrm{des}}=2^{-\Omega(n)}$~\cite{schuster2025random}. Exact $4$-designs also satisfy the same CP-order condition, so this condition alone establishes neither uniform matching of the variance floor scaling nor its failure. The failure for the two-layer ensemble is established through our analysis of the circuit structure and its exact inverse shadow channel.



Beyond the two-layer architecture, we next consider an all-to-all random two-qubit circuit, in which each layer applies independent Haar-random two-qubit gates to arbitrary, disjoint pairs of qubits. Despite its all-to-all connectivity, we show that this model cannot match the multi-shot variance floor of a global exact unitary $4$-design within a constant factor uniformly over states and observables at shallow depth.

\begin{theorem}[Informal]
\label{main:thm:3}
    For the ancilla-free all-to-all random two-qubit circuit, matching the multi-shot variance floor of a global exact unitary $4$-design within a constant factor for every state--observable pair requires depth $\Omega(n/\log n)$.
\end{theorem}

These two no-go results are specific to multi-shot shadow estimation and do not preclude the above ensembles from achieving favorable guarantees for other learning tasks. They nevertheless show that closeness to Haar randomness alone does not guarantee Haar-level learning performance, motivating direct analysis for each task and target quantity.

\subsection{Exact and efficient classical post-processing}

Practical use of the estimator in Theorem~\ref{main:thm:1} also depends on the efficiency of classical post-processing. We now show that the exact inverse shadow channel of the two-layer Clifford ensemble admits an efficient classical representation.
For a Clifford ensemble $\mathcal{E}\subset\mathrm{Cl}(n)$, the shadow channel is diagonal in the Pauli basis, $\mathcal{M}(P)=m_P P$, where $m_P=\mathbb{E}_{U\sim\mathcal{E}}[\mathbf{1}\{UPU^{\dagger}\in\pm\mathcal{Z}\}]$ denotes the visibility of $P$~\cite{cho2025shallow}. The estimator $\hat{o}=\operatorname{Tr}(O\hat{\rho})$, with $\hat{\rho}=\mathcal{M}^{-1}(U^{\dagger}\ket{b}\!\bra{b}U)$, can therefore be written as
\begin{equation}
    \hat{o}
    =
    \frac{1}{2^n}\sum_P
    m_P^{-1}
    \operatorname{Tr}(OP)
    \operatorname{Tr}(PU^{\dagger}\ket{b}\!\bra{b}U).
\end{equation}

If the observable $O$ has a Pauli decomposition containing $\operatorname{poly}(n)$ terms, one can process each Pauli operator $P$ separately by computing its visibility $m_P$. Since the function $P\mapsto m_P$ admits a matrix product state (MPS) representation with bond dimension $\chi=\operatorname{poly}(n)$, each $m_P$ can be evaluated by a polynomial-time tensor-network contraction~\cite{bertoni2024shallow, akhtar2023scalable}.

For an observable given as a matrix product operator~(MPO), however, its Pauli decomposition generally contains exponentially many terms, making term-by-term evaluation impractical. Two alternative approaches are commonly used. The first numerically approximates the reciprocal visibility function $P\mapsto m_P^{-1}$ as an MPS~\cite{bertoni2024shallow, akhtar2023scalable}. This approach generally lacks rigorous error control and requires the approximation to be recomputed when the system size or circuit depth changes. The second replaces the exact inverse shadow channel with the Haar inverse shadow channel~\cite{schuster2025random, cioli2025approximate}. This replacement can provide controlled bias in certain settings, but no such guarantee holds in general. In particular, previous work has shown that the bias can remain large at shallow depth for tasks such as single-shot purity estimation~\cite{cioli2025approximate} and quantum state tomography~\cite{cho2025sample}.
Consequently, approximate classical post-processing can compromise the sample-complexity advantage. 

In contrast, the inverse shadow channel of the two-layer Clifford ensemble admits an exact and efficient tensor network representation.
\begin{theorem}[Exact inverse shadow channel]
    For the two-layer block Clifford ensemble with block size $k$, the inverse shadow channel $\mathcal{M}_{\mathcal{E}_{2\mathrm{L}}}^{-1}$ admits an exact MPO representation with bond dimension $\chi=n/k$, where each site corresponds to a block of $k$ qubits.
\end{theorem}

Although the visibility function $P\mapsto m_P$ admits an MPS representation with polynomial bond dimension, its entrywise reciprocal $P\mapsto m_P^{-1}$ can require an exponentially large bond dimension. The two-layer ensemble avoids this growth: the reset structure of its visibility yields an exact MPS representation of $m_P^{-1}$ with bond dimension $\chi=n/k$. This provides an explicit example related to a question left open in~\cite{akhtar2023scalable}: whether exact reconstruction coefficients of shallow shadows can be represented efficiently as an MPS.

Consequently, when $O$ is given as an MPO with bond dimension $\operatorname{poly}(n)$, the estimator $\hat{o}=\operatorname{Tr}(O\hat{\rho})$ can be evaluated by a polynomial-time tensor network contraction for the block size $k=O(\log n)$ considered here [see Appendix~\ref{app:subsec:eff_mpo}]. 

\subsection{Applications}
The fourth-order no-go result above identifies a limitation of the two-layer architecture in the multi-shot setting. Nevertheless, Theorem~\ref{main:thm:1} can be applied to several state-learning problems for which global Clifford measurements were previously used, allowing them to be replaced by shallow Clifford measurements.

Theorem~\ref{main:thm:1} is formulated as an observable estimation result: for an unbiased estimator $\hat{\rho}$, it controls the variance of $\operatorname{Tr}(O\hat{\rho})$ for a Hermitian observable $O$. Many state-learning problems, however, are not directly formulated as observable estimation. We show that the statistical quantities governing their performance can nevertheless be reduced to observable estimation for suitably chosen observables $O$ and states $\rho$, allowing Theorem~\ref{main:thm:1} to replace the global Clifford ensemble by the two-layer Clifford ensemble.

The choices of $O$ and $\rho$ depend on the learning problem and, in some cases, the state $\rho$ appearing in this reduction differs from the state being learned. This flexibility allows us to apply Theorem~\ref{main:thm:1} to a broad range of learning tasks. In the following, we describe the resulting guarantees. The explicit choices of $O$ and $\rho$ and the corresponding proofs are given in Appendix~\ref{app:sec:state_learning_applications}.

\subsubsection{Setting-efficient quantum state tomography}

\begin{table}[t]
    \centering
    \begin{tabular}{lcccc}
        \toprule
        
        & Measurement
        & $N_U$
        & $N_S$
        & $T$
        \\
        \midrule

        Voroninski~\cite{voroninski2013quantum}
        & Haar($n$)
        & $r$
        & ---
        & ---
        \\

        Gu\c{t}\u{a} et al.~\cite{guctua2020fast}
        & $2$-design
        & $d$
        & $r^2\epsilon^{-2}$
        & $dr^2\epsilon^{-2}$
        \\

        Fran\c{c}a et al.~\cite{brandao2020fast}
        & Cl($n$)
        & $r^3\epsilon^{-2}$
        & $dr\epsilon^{-2}$
        & $dr^4\epsilon^{-4}$
        \\

        Fran\c{c}a et al.~\cite{brandao2020fast}
        & $4$-design
        & $r\epsilon^{-2}$
        & $dr\epsilon^{-2}$
        & $dr^2\epsilon^{-4}$
        \\

        Cho et al.~\cite{cho2025sample}
        & $\mathcal{E}_{2\mathrm{L}}$
        & $dr^2\epsilon^{-2}$
        & $1$
        & $dr^2\epsilon^{-2}$
        \\

        This work
        & $\mathcal{E}_{2\mathrm{L}}$
        & $r^2\epsilon^{-2}$
        & $d$
        & $dr^2\epsilon^{-2}$
        \\

        \bottomrule
    \end{tabular}
    \caption{
        Comparison of single-copy quantum state tomography protocols for an unknown rank-$r$ state in dimension $d=2^n$, with trace norm error $\epsilon$. Here, $N_U$ denotes the number of measurement settings, $N_S$ denotes the number of shots per basis, and $T=N_U N_S$ is the total number of samples. All entries suppress polylogarithmic factors in $d$, $r$, and $1/\epsilon$.
    }
    \label{main:tab:qst}
\end{table}

We first consider quantum state tomography (QST). Given an unknown state $\rho$ of rank at most $r$ in dimension $d=2^n$, the goal is to construct an estimator $\hat{\rho}$ satisfying $\|\hat{\rho}-\rho\|_1\leq\epsilon$. For a fixed failure probability, any non-adaptive single-copy protocol requires $T=\Omega(dr^2/\epsilon^2)$ samples in the worst case~\cite{Lowe2025lower}. Ref.~\cite{cho2025sample} achieves this scaling up to a logarithmic factor using two-layer Clifford measurements of depth $O(\log n)$, with $T=O(dr^2\log d/\epsilon^2)$. However, it samples a fresh measurement unitary for every shot, so $N_U=T$ and $N_S=1$. Thus, its near-optimal total sample complexity comes with a large number of measurement settings.

Reusing each measurement unitary for multiple shots can reduce the number of settings. Using global Clifford measurements or unitary $4$-design measurements, Ref.~\cite{brandao2020fast} obtained protocols whose setting complexity depends only polylogarithmically on $d$. Their total sample complexities, however, have a worse dependence on the target accuracy $\epsilon$ than the near-optimal $dr^2/\epsilon^2$ scaling [Table~\ref{main:tab:qst}]. 

This raises the question of whether $N_U$ can be reduced without sacrificing the total sample complexity. We first show that global Clifford measurements admit an estimator satisfying
\begin{equation}
    N_U=\widetilde{O}(r^2/\epsilon^2),
    \qquad
    N_S=O(d),
\end{equation}
and hence $T=\widetilde{O}(dr^2/\epsilon^2)$. Here, $\widetilde O$ suppresses polylogarithmic factors in $d$, $r$, and $1/\epsilon$. We then use Theorem~\ref{main:thm:1} to replace the global Clifford measurements by the two-layer Clifford ensemble without changing these asymptotic scalings. 

Our estimator first averages the exact-inverse snapshots within each measurement basis, applies spectral clipping to each resulting matrix, and then averages over bases before taking a rank-$r$ approximation.
Although the tomography protocol reuses each sampled unitary for multiple shots, its concentration analysis reduces to the variance bound of single-shot observable estimation in Theorem~\ref{main:thm:1}. Compared with the single-shot two-layer Clifford protocol of Ref.~\cite{cho2025sample}, this reduces the number of measurement settings by a factor of $d$, up to polylogarithmic factors, without changing the asymptotic total sample complexity. In contrast, applying the bounds in Eq.~\eqref{main:eq:prev_est} to QST based on $\hat\rho_{\mathrm H}$ from approximate-design measurements can require an exponentially small design error to certify even constant trace-norm accuracy. The corresponding circuit construction then no longer guarantees logarithmic depth.

\subsubsection{Multiparameter quantum metrology}
\begin{table}[t]
    \centering
    \begin{tabular}{lccc}
        \toprule
        & Measurement
        & Pure
        & Mixed
        \\
        \midrule

        Zhou et al.~\cite{zhou2026randomized}
        & $\mathrm{Cl}(n)$
        & $\checkmark$
        & $\checkmark$
        \\

        Du et al.~\cite{du2026complexity}
        & Approx.\ $3$-design
        & $\checkmark$
        & ---
        \\

        Mao et al.~\cite{mao2026near}
        & $\mathrm{Cl}(1)^{\otimes n}$
        & $\checkmark\!(\textit{Avg.})$
        & ---
        \\

        This work
        & $\mathcal{E}_{2\mathrm{L}}$
        & $\checkmark$
        & $\checkmark$
        \\

        \bottomrule
    \end{tabular}
    \caption{
        Comparison of measurement protocols for multiparameter quantum metrology.
        A check mark ($\checkmark$) indicates an established \textit{near-optimality} or \textit{weak near-optimality} guarantee in the corresponding state regime.
        \textit{Avg.} indicates an average-case guarantee.
    }
    \label{main:tab:metrology}
\end{table}

Quantum metrology concerns the estimation of parameters $\theta$ encoded in a quantum state $\rho_{\theta}$. 
For a positive operator-valued measure (POVM) $\mathbb{M}=\{M_x\}_x$, we denote the resulting classical Fisher information matrix (CFIM) by $I(\mathbb{M},\rho_{\theta})$ and the quantum Fisher information matrix (QFIM) by $J(\rho_{\theta})$. 
We call a measurement $\mathbb{M}$ \textit{near-optimal} if, for some constant $c\geq 1$,
\begin{equation}\label{main:eq:near_optiaml}
    I(\mathbb{M},\rho_{\theta})\succeq \frac{1}{c}J(\rho_{\theta}).
\end{equation}
Together with the quantum Cramér--Rao bound $I(\mathbb{M},\rho_{\theta})\preceq J(\rho_{\theta})$~\cite{BraunsteinCaves1994}, Eq.~\eqref{main:eq:near_optiaml} ensures that the POVM $\mathbb{M}$ retains a constant fraction of the available quantum Fisher information in every parameter direction.

In the multiparameter setting, identifying POVMs $\mathbb{M}$ that satisfy Eq.~\eqref{main:eq:near_optiaml} can be computationally demanding~\cite{hayashi2023tight}. 
Randomized measurements based on global Clifford unitaries provide a state-independent alternative~\cite{zhou2026randomized}: they achieve \textit{near-optimality} for pure states and well-conditioned, approximately low-rank states, and \textit{weak near-optimality}, in the sense of matching the Gill--Massar (GM) bound~\cite{gill2000state} up to a constant factor, for well-conditioned full-parameter rank-$r$ states.

More recent works have reduced the measurement complexity for pure-state metrology~\cite{du2026complexity, mao2026near}. Approximate unitary $3$-designs implemented by shallow-depth circuits retain a constant fraction of the QFIM~\cite{du2026complexity}, and randomized Pauli measurements achieve a constant fraction for Haar-random pure states~\cite{mao2026near}. Both results are restricted to pure-state encodings [Table~\ref{main:tab:metrology}].

Using Theorem~\ref{main:thm:1}, we extend shallow randomized measurements to the mixed-state regimes previously accessible with global Clifford measurements. 
For well-conditioned approximately low-rank states, the two-layer Clifford ensemble achieves near-optimality, while for well-conditioned full-parameter rank-$r$ states it achieves the same asymptotic GM bound as global Clifford measurements. 
Thus, the mixed-state guarantees of Ref.~\cite{zhou2026randomized} can be retained using shallow measurements, resolving an open question left by Refs.~\cite{du2026complexity,mao2026near,zhou2026randomized}.

Moreover, unlike the guarantee of Ref.~\cite{du2026complexity}, our result does not require the measurement ensemble to form an approximate unitary $3$-design. As discussed in Proposition~\ref{main:prop:not_necessary}, the same metrological guarantees can hold even when the approximate-design condition fails.

\subsubsection{Stabilizer structure learning}

We next consider learning the Pauli symmetries of a quantum state. Ignoring global phases, we define the unsigned Pauli symmetry group $\mathrm{Weyl}(\rho):=\{P:\operatorname{Tr}(P\rho)^2=1\}$, and refer to $\dim\mathrm{Weyl}(\rho)$ as the stabilizer dimension of $\rho$. A pure stabilizer state has stabilizer dimension $n$, whereas states with $\dim\mathrm{Weyl}(\rho)=n-t$ retain $n-t$ independent Pauli symmetries, where $t$ is the stabilizer nullity. Such states include classes generated by Clifford circuits with a small number of non-Clifford gates~\cite{leone2024learning,chia2024efficient,grewal2025efficient}. We focus on $t=O(\log n)$. This regime is natural in the single-copy setting, since any potentially adaptive single-copy protocol for stabilizer group learning requires $\Omega(2^t)$ samples in the worst case~\cite{cho2026single}.

Previous work considered shallow measurements based on the block Clifford ensemble $\mathrm{Cl}(k)^{\otimes n/k}$ with $k=O(\log n)$~\cite{cho2026single}. These measurements efficiently recover the stabilizer structure for all but an exponentially small fraction of states with stabilizer dimension $n-t$, but do not provide a worst-case guarantee. Here, we show that the two-layer Clifford ensemble provides a worst-case guarantee.

Existing analyses based on global Clifford measurements use the transitive action of $\mathrm{Cl}(n)$ on non-identity Pauli operators~\cite{chia2024efficient,grewal2025efficient}. For the two-layer ensemble, however, different non-identity Pauli operators can have different probabilities of being mapped into $\mathcal{Z}$ and hence different visibilities under computational basis measurements. Therefore, the argument based on global Clifford transitivity does not apply directly. 

Let $S=\mathrm{Weyl}(\rho)$. For every codimension-one subspace $T\subset S$, we show that
\begin{align}\label{main:eq:stab}
    \Pr_{C}\!\left[
    C(S\setminus T)\cap\mathcal{Z}\neq\emptyset
    \right]
    =\Omega(2^{-t}),
\end{align}
where $C$ is sampled from the two-layer Clifford ensemble. Let $S_{\mathrm{vis}}\subsetneq S$ denote the subspace covered so far. Choosing a codimension-one subspace $T\subset S$ with $S_{\mathrm{vis}}\subseteq T$, Eq.~\eqref{main:eq:stab} implies that a new generator independent of $S_{\mathrm{vis}}$ becomes visible with probability $\Omega(2^{-t})$. Since $\dim S=n-t$, $O(n2^t)$ measurement bases suffice to cover $S$, matching the asymptotic scaling of global Clifford measurements.

The stabilizer learning protocol is not formulated as an observable estimation problem and does not apply the inverse shadow channel. Nevertheless, its visibility guarantee can be analyzed using an auxiliary shadow estimator. For a suitable observable, Theorem~\ref{main:thm:1} controls the second moment needed to establish Eq.~\eqref{main:eq:stab} [Appendix~\ref{app:sec:stabilizer_structure_learning}].

Stabilizer group learning also serves as a subroutine for QST of a state $\rho$ with $\dim \mathrm{Weyl}(\rho)=n-t$. Once the stabilizer structure is identified, the remaining tomography problem can be reduced to a system of $t$ qubits and can be performed by measuring Pauli operators on the original state. Consequently, for $t=O(\log n)$, the above result gives a shallow-depth single-copy protocol with $\operatorname{poly}(n,1/\epsilon)$ sample and time complexity for learning $\rho$ to trace norm error $\epsilon$~\cite{cho2026single}, without requiring Bell measurements~\cite{leone2024learning,hangleiter2024bell,grewal2025efficient} or global Clifford measurements~\cite{chia2024efficient,grewal2025efficient}.

\section{Discussion}
Our results show that the CP-order condition in Eq.~\eqref{main:eq:cp} alone does not determine learning performance. Specifically, the approximate unitary $3$-design condition is neither necessary nor sufficient for the Frobenius-norm variance guarantee of Theorem~\ref{main:thm:1}. Approximate strong unitary designs provide sufficient conditions. However, for circuits composed of independent local Haar-random gates, satisfying the relative error strong design condition requires depth $\Omega(n)$ even with all-to-all connectivity~\cite{schuster2025strong}.

By directly analyzing the two-layer Clifford ensemble, we establish the Frobenius-norm variance bound of Theorem~\ref{main:thm:1} at depth $O(\log n)$ in a 1D architecture without ancillary qubits.
This variance bound transfers global Clifford guarantees to setting-efficient quantum state tomography, metrology for mixed states, and stabilizer structure learning. The inverse shadow channel also admits an exact MPO representation with polynomial bond dimension, enabling post-processing without approximating the inverse channel.

For multi-shot observable estimation, matching the variance floor of a global exact unitary $4$-design within a constant factor uniformly over all states and observables imposes stronger requirements. The two-layer Haar ensemble requires block size $k=\Omega(n)$, while circuits composed of independent Haar-random two-qubit gates require depth $\Omega(n/\log n)$ even with all-to-all connectivity. These limitations show that the learning guarantees attainable at shallow depth depend on the task and the statistical quantity being estimated.

Several questions remain open. 
The optimal tradeoff among connectivity, ancilla count, and circuit depth for achieving the variance bound in Theorem~\ref{main:thm:1} remains open.
In the multi-shot setting, it remains open whether other circuit architectures can attain the exact unitary $4$-design scaling at sublinear depth while retaining efficient classical post-processing. 
Finally, our analysis focuses primarily on non-adaptive measurements, and extending our results to adaptive measurement strategies is a natural direction for future work.

\begin{acknowledgments}
This work was supported by a National Research Foundation of Korea (NRF) grant funded by the Korean Government (Ministry of Science and ICT (MSIT)) (RS-2023-NR057243, RS-2023-00283291, RS-2024-00413957, SRC Center for Quantum Coherence in Condensed Matter RS-2023-00207732, RS-2023-NR077112 and Quantum Technology R\&D Leading Program (Quantum Computing) RS-2024-00442994) and a core center program grant funded by the Ministry of Education (No. 2021R1A6C101B418). 
C.O. was supported by the National Research Foundation of Korea Grants (No. RS-2024-00431768 and No. RS-2025-00515456) funded by the Korean government (Ministry of Science and ICT (MSIT)) and the Institute of Information \& Communications Technology Planning \& Evaluation (IITP) Grants funded by the Korean government (MSIT) (No. RS-2024-00437284, No. IITP-2025-RS-2025-02283189 and No. IITP-2025-RS-2025-02263264) by Global Partnership Program of Leading Universities in Quantum Science and Technology (RS-2025-08542968) through the National Research Foundation of Korea~(NRF) funded by the Korean government (Ministry of Science and ICT(MSIT)).
ChatGPT 5.6 was used to assist in refining proof ideas, checking errors, and improving the presentation of the manuscript. All proofs and results were independently written and verified by the authors, who take full responsibility for their correctness.
\end{acknowledgments}

\bibliographystyle{apsrev4-2}
\bibliography{main}

\clearpage
\onecolumngrid
\appendix

\section{RELATED WORK}\label{appx:related_works}

\paragraph{\textbf{Approximate unitary $t$-designs.}}
Haar-random unitaries provide a standard model of quantum randomness but generally require exponentially large circuit complexity. Many applications require agreement with the Haar measure only up to a fixed moment order $t$. Unitary $t$-designs provide this restricted form of randomness by reproducing the corresponding Haar moments. The Clifford group is an exact unitary $3$-design but not a unitary $4$-design, and no general efficient construction of exact unitary $t$-designs is known for $t\geq4$. Approximate unitary designs relax exact moment matching by allowing a controlled error. Recent work has shown that designs satisfying the relative error condition in Eq.~\eqref{main:eq:cp} can be implemented without ancillary qubits in a 1D architecture at depth $O(t\,\operatorname{poly}\log t\,\log(n/\epsilon_{\mathrm{des}}))$~\cite{schuster2025random, laracuente2026approximate}. This depth is logarithmic in $n$ for fixed $t$ and $\epsilon_{\mathrm{des}}$.

\paragraph{\textbf{Approximate strong unitary $t$-designs.}} 
The standard approximate-design condition does not generally control mixed moments involving $U$, $U^\dagger$, $U^T$, and $U^*$, which arise in out-of-time-order correlators~\cite{roberts2017chaos} and the Hayden--Preskill protocol~\cite{yoshida2017efficient}. Approximate strong unitary designs were introduced to control such moments~\cite{schuster2025strong}, allowing the corresponding Haar guarantees to be recovered for many state-learning tasks. For circuits composed of independent local Haar-random gates, however, satisfying the relative-error condition with constant error $\epsilon_{\mathrm{des}}$ requires depth $\Omega(n)$ already at order two, regardless of geometry. Thus, compared to ordinary approximate designs, which admit logarithmic-depth constructions, approximate strong designs lose the advantage of shallow measurements in this circuit model.

\paragraph{\textbf{Shallow shadows.}}
Shadow tomography and classical shadows estimate selected properties of a quantum state without reconstructing the entire state~\cite{aaronson2018shadow,huang2020predicting}. In classical shadows, randomized single-copy measurement data can be reused to estimate many observables that can be chosen later. The measurement ensemble determines the statistical and experimental cost. Product single-qubit Clifford $\mathrm{Cl}(1)^{\otimes n}$ measurements are effective for local observables, whereas global Clifford measurements $\mathrm{Cl}(n)$ provide favorable Frobenius-norm guarantees for general observables but require linear-depth circuits in one dimension without ancillary qubits. Shallow shadows interpolate between these regimes using low-depth brickwork circuits~\cite{akhtar2023scalable,bertoni2024shallow} or block Clifford ensembles $\mathrm{Cl}(k)^{\otimes n/k}$ with $k=O(\log n)$~\cite{cho2025shallow}. Approximate unitary designs provide another logarithmic-depth approach with state uniform guarantees, although the corresponding Haar-inverse estimator is generally biased and does not ensure Frobenius-norm variance scaling~\cite{schuster2025random}.

\paragraph{\textbf{Setting-efficient quantum state tomography.}}

Quantum state tomography reconstructs an unknown state $\rho$. It is widely used for device characterization and as a subroutine in various quantum learning algorithms~\cite{qst_ion,qst_super,qst_nv,huang2022provably}. Previous works have mainly focused on minimizing the total sample complexity, for which tight lower bounds and near-optimal protocols are known~\cite{Lowe2025lower,guctua2020fast,cho2025sample}. In experiments, the number of distinct measurement bases $N_U$ is also important because changing a basis can require recompiling the measurement circuit. Writing the total number of samples as $T=N_U N_S$, one can reduce this overhead by increasing the number of shots $N_S$ taken in each basis while keeping the total number of samples $T$ as small as possible. Several protocols have studied the resulting tradeoff between measurement settings and total samples~\cite{voroninski2013quantum,guctua2020fast, brandao2020fast,cho2025sample}. The present work reuses each shallow Clifford basis for multiple shots, reducing the number of settings $N_U$, while retaining near-optimal total sample complexity.

\paragraph{\textbf{Multiparameter quantum metrology.}}
Quantum metrology studies the estimation of parameters encoded in a quantum state $\rho_\theta$. The information obtained from a measurement is described by the Classical Fisher Information Matrix, whereas the Quantum Fisher Information Matrix characterizes the information available in the state. In multiparameter problems, optimal measurements can be state dependent and incompatible across different parameter directions. Randomized measurements provide a practical state-independent alternative. Unitary $3$-design measurements were shown to retain a constant fraction of the quantum Fisher information for broad classes of pure and low-rank states~\cite{zhou2026randomized}. For pure states, related guarantees have subsequently been obtained using shallow approximate designs and few-qubit randomized measurements~\cite{du2026complexity,mao2026near}. For the mixed-state regimes considered here, however, it had remained open whether shallow-depth measurements could replace global Clifford measurements.

\paragraph{\textbf{Stabilizer structure learning.}}
Stabilizer structure learning identifies the Pauli symmetries of a quantum state, described by its stabilizer group $\mathrm{Weyl}(\rho)$. For states with stabilizer dimension $n-t$, where $t$ is the stabilizer nullity, learning $\mathrm{Weyl}(\rho)$ reduces tomography to the remaining non-stabilizer degrees of freedom~\cite{grewal2025efficient,chia2024efficient}. Existing efficient protocols based on Bell sampling require joint measurements on multiple copies~\cite{grewal2025efficient,leone2024learning, hangleiter2024bell,chen2025stabilizer}, while single-copy protocols use global Clifford measurements. More recently, block Clifford ensembles $\mathrm{Cl}(k)^{\otimes n/k}$ were shown to learn the stabilizer structure efficiently for all but an exponentially small fraction of states when $k=O(\log n)$ and $t=O(\log n)$, while a worst-case guarantee remained open~\cite{cho2026single}.

\section{Preliminaries}
In this section, we introduce the notation and preliminary results used
in the subsequent analysis.

\subsection{Shadow channel}

We first review the shadow channel and its form for Clifford measurement ensembles. 
Let $\mathcal{E}$ be an ensemble of $n$-qubit unitaries. 
For a state $\rho$, the corresponding shadow channel is
\begin{equation}\label{app:eq:shadow_channel}
    \mathcal{M}_{\mathcal{E}}(\rho)
    =
    \mathbb{E}_{U\sim\mathcal{E}}
    \sum_{b\in\{0,1\}^n}
    \operatorname{Tr}
    \left(
        U\rho U^{\dagger}\ket{b}\!\bra{b}
    \right)
    U^{\dagger}\ket{b}\!\bra{b}U .
\end{equation}
Here, $U$ specifies the measurement basis and $b\in\{0,1\}^n$ is the measurement outcome. 
If $\mathcal{M}_{\mathcal{E}}$ is invertible, each measurement record $(U,b)$ gives the unbiased estimator $\hat{\rho}=\mathcal{M}_{\mathcal{E}}^{-1}(U^{\dagger}\ket{b}\!\bra{b}U)$. 
By the definition of the shadow channel, $\mathbb{E}_{U,b}[\hat{\rho}]=\rho$. Hence, for any Hermitian observable $O$, $\hat{o}=\operatorname{Tr}(O\hat{\rho})$ is an unbiased estimator of $\operatorname{Tr}(O\rho)$. The statistical error of this estimator is determined by its second moment. In the single-shot setting,
\begin{align}\label{app:eq:general_variance}
    \operatorname{Var}[\hat{o}]
    &=
    \mathbb{E}_{U\sim\mathcal{E}}
    \sum_{b\in\{0,1\}^n}
    \operatorname{Tr}
    \left(
        O\,
        \mathcal{M}_{\mathcal{E}}^{-1}
        \left(
            U^{\dagger}\ket{b}\!\bra{b}U
        \right)
    \right)^2
    \operatorname{Tr}
    \left(
        U\rho U^{\dagger}\ket{b}\!\bra{b}
    \right)
    -
    \operatorname{Tr}(O\rho)^2 .
\end{align}

We now restrict to Clifford ensembles, $\mathcal E\subset\operatorname{Cl}(n)$. Let $\mathcal P_n=\{I,X,Y,Z\}^{\otimes n}$ denote the $n$-qubit Pauli operators and let $\mathcal Z=\{I,Z\}^{\otimes n}$. 
For any Clifford ensemble, the shadow channel is diagonal in the Pauli basis~\cite{cho2025shallow}:
\begin{equation}
    \mathcal M_{\mathcal E}(P)=m_P P
    \label{app:eq:visibility_channel}
\end{equation}
where
\begin{equation}
    m_P
    :=
    \mathbb E_{U\sim\mathcal E}
    \left[
        \mathbf 1
        \left\{
            UPU^\dagger\in\pm\mathcal Z
        \right\}
    \right]
    =
    \Pr_{U\sim\mathcal E}
    \left[
        UPU^\dagger\in\pm\mathcal Z
    \right].
    \label{app:eq:visibility}
\end{equation}
We refer to $m_P$ as the visibility of $P$. Since the Pauli operators form an operator basis, Eq.~\eqref{app:eq:visibility_channel} also gives a simple condition for informational completeness. 
The shadow channel is invertible if and only if $m_P>0$ for every $P\in\mathcal{P}_n$. 
In this case,
\begin{equation}
    \mathcal{M}_{\mathcal{E}}^{-1}(P)
    =
    m_P^{-1}P.
\end{equation}
\begin{fact}[Self-adjointness of the inverse shadow channel~\cite{huang2020predicting}]
\label{app:fact:dual}
    If $\mathcal M_{\mathcal E}$ is invertible, then
    $\mathcal M_{\mathcal E}^{-1}$ is self-adjoint with respect to the
    Hilbert--Schmidt inner product. Thus, for all Hermitian $A$ and $B$,
    \begin{equation}
        \operatorname{Tr}\left(
            A\,\mathcal M_{\mathcal E}^{-1}(B)
        \right)
        =
        \operatorname{Tr}\left(
            \mathcal M_{\mathcal E}^{-1}(A)\,B
        \right).
    \end{equation}
\end{fact}
Expanding $O=\frac{1}{d}\sum_P\operatorname{Tr}(OP)P$, where $d=2^n$, the single-shot observable estimator can therefore be written as
\begin{equation}\label{app:eq:pauli_estimator}
    \hat{o}
    =
    \frac{1}{d}
    \sum_{P\in\mathcal{P}_n}
    \operatorname{Tr}(OP)\,
    m_P^{-1}
    \operatorname{Tr}
    \left(
        UPU^{\dagger}\ket{b}\!\bra{b}
    \right).
\end{equation}
To describe the second moment of Eq.~\eqref{app:eq:pauli_estimator}, we define the correlated visibility of two Pauli operators~\cite{bertoni2024shallow, cho2025sample},
\begin{equation}\label{app:eq:correlated_visibility}
    \tau(P,Q)
    :=
    \mathbb{E}_{U\sim\mathcal{E}}
    \left[
        \mathbf{1}
        \{UPU^{\dagger}\in\pm\mathcal{Z}\}
        \mathbf{1}
        \{UQU^{\dagger}\in\pm\mathcal{Z}\}
    \right].
\end{equation}
Equivalently, $\tau(P,Q)$ is the probability that $P$ and $Q$ are simultaneously visible in the same measurement basis. 
In particular, $\tau(P,Q)=0$ whenever $P$ and $Q$ anticommute, since two anticommuting Pauli operators cannot both be mapped into the commuting set $\pm\mathcal{Z}$. Using Eqs.~\eqref{app:eq:visibility} and~\eqref{app:eq:correlated_visibility}, the variance in Eq.~\eqref{app:eq:general_variance} becomes
\begin{equation}\label{app:eq:cliff_Ord_variance}
    \operatorname{Var}[\hat{o}]
    =
    \sum_{P,Q\in\mathcal{P}_n}
    \frac{
        \operatorname{Tr}(OP)
        \operatorname{Tr}(OQ)
    }{d^2}
    \frac{\tau(P,Q)}{m_Pm_Q}
    \operatorname{Tr}(\rho PQ)
    -
    \operatorname{Tr}(O\rho)^2 .
\end{equation}
Thus, the dependence of the variance on the Clifford ensemble is entirely captured by $m_P$ and $\tau(P,Q)$. 

Recall that the two-layer Clifford ensemble $\mathcal E_{2\mathrm L}$ consists of two staggered layers of $k$-qubit unitaries. Each layer is a tensor product of unitaries acting on non-overlapping contiguous blocks, with periodic boundary conditions~[Fig.~\ref{main:fig:1}a]. The block unitaries are sampled independently and uniformly from $\mathrm{Cl}(k)$. A unitary $U\sim\mathcal E_{2\mathrm L}$ can be written as $U=U_2U_1$, where $U_1$ and $U_2$ are the products of the block unitaries in the first and second layers, respectively. For $\ell\in\{1,2\}$, let $w_{k,\ell}(P)$ denote the number of blocks in layer $\ell$ on which $P$ acts nontrivially.

For the uniform $n$-qubit Clifford ensemble $\mathcal E=\operatorname{Cl}(n)$, the visibility $m_P$ is
\begin{equation}
    m_{P}
    =
    \begin{cases}
        1 & P=I,\\[2pt]
        \dfrac{1}{2^n+1} & P\neq I.
    \end{cases}
\end{equation}
Similarly, the correlated visibility $\tau(P,Q)$ is
\begin{equation}\label{app:eq:global_cliff_Ord_correlated_visibility}
    \tau(P,Q;n)
    =
    \begin{cases}
        1,
        & P=Q=I, \\[2pt]
        \dfrac{1}{2^n+1},
        & P=I,\ Q\neq I, \\[6pt]
        \dfrac{1}{2^n+1},
        & P\neq I,\ Q=I, \\[6pt]
        \dfrac{1}{2^n+1},
        & P=Q\neq I, \\[6pt]
        \dfrac{2}{(2^n+1)(2^n+2)},
        & P,Q\neq I,\ P\neq Q,\ [P,Q]=0, \\[6pt]
        0,
        & \{P,Q\}=0 .
    \end{cases}
\end{equation}

For the block Clifford ensemble $\operatorname{Cl}(k)^{\otimes n/k}$, the independence of the blocks gives a product form for both visibility $m_P$ and $\tau(P,Q)$~\cite{cho2025sample}. Let $P_{i}$ and $Q_{i}$ denote the restrictions of $P$ and $Q$ to the $i$th block. Then
\begin{align}
    m_P
    &=(2^k+1)^{-w_{k}(P)},
    \label{app:eq:block_product_visibility}\\
    \tau(P,Q)
    &=
    \prod_{i=1}^{n/k}
    \tau(P_{i},Q_{i};k),
    \label{app:eq:block_product_correlated_visibility}
\end{align}
where $w_k(P)$ is the number of blocks for which $P_i\neq I$, and the local factor $\tau(P_i,Q_i;k)$ is given by Eq.~\eqref{app:eq:global_cliff_Ord_correlated_visibility} with $n$ replaced by $k$. The following bound will be used for the product terms appearing below.

\begin{lemma}
    For $a,b >0$, the following holds
    \begin{equation}
        \left(1+\frac{b}{2^{ak}}\right)^{n/k}
        =
        \left(1+\frac{b}{2^{ak}}\right)^{(2^{ak}/b)(bn/(k2^{ak}))}
        \le e^{bn/(k2^{ak})}
        =
        O(1),
    \end{equation}
    when $k2^{ak}=\Omega(n)$. 
\end{lemma}

Equation~\eqref{app:eq:visibility} is specific to Clifford unitary ensembles and does not hold for general unitary ensembles. For Pauli-invariant unitary ensembles, the following relation holds.

\begin{fact}[\cite{bu2024classical}]\label{app:fact:m_P(U)}
    Let $\mathcal E$ be a Pauli-invariant unitary ensemble. For every
    $P\in\mathcal P_n$,
    $\mathcal M_{\mathcal E}(P)=m_PP$, where
    \begin{align}
        m_P
        &=
        \frac{1}{d}
        \mathbb E_{U\sim\mathcal E}
        \sum_b
        \bra{b}UPU^\dagger\ket{b}^2
        =
        \mathbb E_{U\sim\mathcal E}[m_P(U)]\label{app:eq:pauli_invariant_m_P},
        \\
        m_P(U)
        &=
        \frac{1}{d}
        \sum_b
        \bra{b}UPU^\dagger\ket{b}^2.
        \label{app:eq:m_P_U}
    \end{align}
\end{fact}

\begin{lemma}\label{app:lem:m_P_L_bound}
    For any unitaries $U,V$ and $P\in\mathcal P_n$,
    \begin{equation}
        |m_P(U)-m_P(V)|
        \leq
        4\|U-V\|_\infty.
    \end{equation}
\end{lemma}

\begin{proof}
    By the Cauchy--Schwarz inequality and $\sum_b \bra{b}\!A\!\ket{b}^2\le \|A\|_2^2$ for a hermitian matrix $A$,
    \begin{align}
        |m_P(U)-m_P(V)|
        &=
        \frac{1}{d}
        \left|
        \sum_b 
        \bra{b}UPU^\dagger\ket{b}^2
        -
        \bra{b}VPV^\dagger\ket{b}^2
        \right|
        \\
        &=
        \frac{1}{d}
        \left|
        \sum_b 
        \left(\bra{b}UPU^\dagger\ket{b}+\bra{b}VPV^\dagger\ket{b}\right)
        \left(\bra{b}UPU^\dagger\ket{b}-\bra{b}VPV^\dagger\ket{b}\right)
        \right|
        \\
        &\leq
        \frac{1}{d}
        \left(
            \sum_b
            \left(
                \bra{b}UPU^\dagger\ket{b}
                +
                \bra{b}VPV^\dagger\ket{b}
            \right)^2
        \right)^{1/2}
        \left(
            \sum_b
            \left(
                \bra{b}UPU^\dagger\ket{b}
                -
                \bra{b}VPV^\dagger\ket{b}
            \right)^2
        \right)^{1/2}.
        \label{app:eq:m_P_lipschitz_cs}
    \end{align}
    The first factor satisfies
    \begin{align}
        &
        \left(
            \sum_b
            \left(
                \bra{b}UPU^\dagger\ket{b}
                +
                \bra{b}VPV^\dagger\ket{b}
            \right)^2
        \right)^{1/2}
        \nonumber\\
        &\leq
        \left(
            \sum_b
                \bra{b}UPU^\dagger\ket{b}^2
        \right)^{1/2}
        +
        \left(
            \sum_b
                \bra{b}VPV^\dagger\ket{b}^2
        \right)^{1/2}
        \nonumber\\
        &\leq
        \|UPU^\dagger\|_2+\|VPV^\dagger\|_2
        =
        2\sqrt d.
    \end{align}
    For the second factor,
    \begin{equation}
        \left(
            \sum_b
                \bra{b}
                \left(
                    UPU^\dagger-VPV^\dagger
                \right)
                \ket{b}^2
        \right)^{1/2}
        \leq
        \|UPU^\dagger-VPV^\dagger\|_2.
    \end{equation}
    Substituting these bounds into Eq.~\eqref{app:eq:m_P_lipschitz_cs} gives
    \begin{align}
        |m_P(U)-m_P(V)|
        &\leq
        \frac{2}{\sqrt d}
        \|UPU^\dagger-VPV^\dagger\|_2
        \nonumber\\
        &\leq
        \frac{2}{\sqrt d}
        \left(
            \|(U-V)PU^\dagger\|_2
            +
            \|VP(U^\dagger-V^\dagger)\|_2
        \right)
        \nonumber\\
        &\leq
        \frac{2}{\sqrt d}
        \left(
            \|U-V\|_\infty\|PU^\dagger\|_2
            +
            \|VP\|_2\|U^\dagger-V^\dagger\|_\infty
        \right)
        \nonumber\\
        &=
        4\|U-V\|_\infty.
    \end{align}
    Here, we used
    $\|AB\|_2\leq\|A\|_\infty\|B\|_2$,
    $\|P\|_2=\sqrt d$, and
    $\|U^\dagger-V^\dagger\|_\infty=\|U-V\|_\infty$.
\end{proof}

\subsection{Approximate unitary $t$-designs}

For a unitary ensemble $\mathcal E$, define its $t$-th moment channel by
\begin{equation}
    \Phi_{\mathcal E}^{(t)}(X)
    =
    \mathbb E_{U\sim\mathcal E}
    \left[U^{\otimes t}X(U^\dagger)^{\otimes t}\right].
\end{equation}
The ensemble is an exact unitary $t$-design if $\Phi_{\mathcal E}^{(t)}=\Phi_{\mathrm H}^{(t)}$, where $\Phi_{\mathrm H}^{(t)}$ is the Haar moment channel. Approximate designs relax this equality. We use the relative error definition~\cite{schuster2025random}: an ensemble $\mathcal E$ is an $\epsilon_{\mathrm{des}}$-approximate unitary $t$-design if
\begin{equation}\label{app:eq:approx_design}
    (1-\epsilon_{\mathrm{des}})
    \Phi_{\mathrm H}^{(t)}
    \preceq
    \Phi_{\mathcal E}^{(t)}
    \preceq
    (1+\epsilon_{\mathrm{des}})
    \Phi_{\mathrm H}^{(t)},
\end{equation}
where $\Phi_1\preceq\Phi_2$ means that $\Phi_2-\Phi_1$ is completely positive. Relative error approximate unitary $t$-designs can be implemented in a 1D ancilla-free architecture with circuit depth~\cite{schuster2025random}
\begin{equation}
    \mathrm{depth}
    =
    O\!\left(
        t\operatorname{poly}\log t\cdot
        \log\frac{n}{\epsilon_{\mathrm{des}}}
    \right).
\end{equation}
For fixed $t$ and $\epsilon_{\mathrm{des}}$, this is logarithmic in the system size.

For classical shadow estimation, Ref.~\cite{schuster2025random} uses the Haar inverse shadow channel $\mathcal M_{\mathrm H}^{-1}$ to derive guarantees from the approximate design condition. For measurement data generated using $U\sim\mathcal E$, this gives $\hat{\rho}_{\mathrm H} =\mathcal M_{\mathrm H}^{-1} (U^\dagger\ket b\!\bra bU)$ and $\hat{o}_{\mathrm H} =\operatorname{Tr}(O\hat{\rho}_{\mathrm H})$. Since $\mathcal M_{\mathrm H}^{-1}$ need not invert $\mathcal M_{\mathcal E}$, this estimator is generally biased. Although Ref.~\cite{schuster2025random} assumes $O\succeq0$, decomposing $O_0$ into its positive and negative parts gives the following bounds for arbitrary Hermitian $O$ when $\mathcal E$ is an $\epsilon_{\mathrm{des}}$-approximate unitary $3$-design:
\begin{align}\label{app:eq:approx_design_shadow}
    \left|
        \mathbb{E}[\hat{o}_{\mathrm H}]
        -
        \operatorname{Tr}(O\rho)
    \right|
    &=
    O\!\left(
        \epsilon_{\mathrm{des}}\|O_0\|_1
    \right),
    \nonumber\\
    \operatorname{Var}[\hat{o}_{\mathrm H}]
    &=
    O\!\left(
        \|O_0\|_2^2
        +
        \epsilon_{\mathrm{des}}\|O_0\|_1^2
    \right).
\end{align}
These bounds hold uniformly over the input state $\rho$. For observables with bounded trace norm, the additional terms in Eq.~\eqref{app:eq:approx_design_shadow} remain controlled. For example, if $O=\ket{\psi}\!\bra{\psi}$, then $\|O_0\|_2^2=1-1/d$ and $\|O_0\|_1=2(1-1/d)$. The bias is therefore $O(\epsilon_{\mathrm{des}})$ and the variance remains $O(1)$. Choosing $\epsilon_{\mathrm{des}}=O(\epsilon)$ gives the same sample complexity as global Clifford measurements for overlap estimation, with measurement circuits of depth $O(\log(n/\epsilon))$.

\subsection{Matrix Bernstein inequality with spectral clipping}

\begin{theorem}[Matrix Bernstein inequality~\cite{tropp2015introduction}]
    Let $Y_1,\ldots,Y_T$ be independent, centered, $d\times d$
    Hermitian random matrices. Suppose that, almost surely,
    $\|Y_i\|_\infty\leq R$, and that
    $\|\mathbb{E}[Y_i^2]\|_\infty\leq\sigma^2$. We refer to $R$ as the
    radius and to $\sigma^2$ as the variance parameter. Then, for every
    $t>0$,
    \begin{equation}\label{app:eq:matrix_bernstein}
        \Pr\left[
            \left\|
                \frac{1}{T}\sum_{i=1}^T Y_i
            \right\|_\infty
            \geq t
        \right]
        \leq
        2d\exp\left(
            -\frac{Tt^2}{2\sigma^2+2Rt/3}
        \right).
    \end{equation}
    Equivalently, with probability at least $1-\delta$,
    \begin{equation}\label{app:eq:matrix_bernstein_high_Probability}
        \left\|
            \frac{1}{T}\sum_{i=1}^T Y_i
        \right\|_\infty
        =
        O \left(
        \sqrt{\frac{\sigma^2\log(d/\delta)}{T}}
        +
        \frac{R\log(d/\delta)}{T}
        \right)
        .
    \end{equation}
\end{theorem}

Even when an estimator has a bounded second moment, its operator norm may not admit a useful uniform upper bound, making it difficult to control the radius $R$ in the matrix Bernstein inequality. We therefore apply spectral clipping to bound the radius and control the resulting bias using the same second moment. For $K>0$, define
\begin{equation}\label{app:eq:clip_definition}
    \operatorname{clip}_K(x)
    =
    \begin{cases}
        -K, & x<-K,\\
        x,  & -K\leq x\leq K,\\
        K,  & x>K.
    \end{cases}
\end{equation}

If $A=\sum_j\lambda_j\ket{j}\!\bra{j}$ is Hermitian, then  $\operatorname{clip}_K(A) =\sum_j\operatorname{clip}_K(\lambda_j)\ket{j}\!\bra{j}$.
The clipped estimator has a spectral radius of at most $K$, but it is generally biased. The following lemma separately bounds the bias introduced by clipping and the statistical fluctuations of the empirical average.

\begin{lemma}[Matrix Bernstein inequality with spectral clipping]\label{appx:lem:clip}
    Let $\hat X_1,\ldots,\hat X_T$ be independent copies of a Hermitian
    random matrix $\hat X$ satisfying $\mathbb{E}[\hat X]=X$ and
    $\|\mathbb{E}[\hat X^2]\|_\infty\leq\sigma^2$.
    Define
    \begin{equation}
        \overline X^{(K)}_T
        =
        \frac{1}{T}\sum_{i=1}^T
        \operatorname{clip}_K(\hat X_i).
    \end{equation}
    Then, for every $K>0$ and $0<\delta<1$, with probability at least
    $1-\delta$,
    \begin{equation}
        \left\|\overline X^{(K)}_T-X\right\|_\infty
        =
        O \left(
        \frac{\sigma^2}{K}
        +
        \sqrt{\frac{\sigma^2\log(d/\delta)}{T}}
        +
        \frac{K\log(d/\delta)}{T}
        \right)\label{app:eq:clipped_matrix_concentration}.
    \end{equation}
\end{lemma}

\begin{proof}
    The scalar inequalities $|x-\operatorname{clip}_K(x)|\leq x^2/K$ and $\operatorname{clip}_K(x)^2\leq x^2$ imply
    \begin{equation}
        \left\|
            \mathbb{E}[\operatorname{clip}_K(\hat X)]-X
        \right\|_\infty
        \leq
        \frac{\sigma^2}{K},
        \qquad
        \mathbb{E}[(\operatorname{clip}_K(\hat X)-\mathbb{E}[\operatorname{clip}_K(\hat X)])^2]
        \preceq
        \mathbb{E}[\operatorname{clip}_K(\hat X)^2]
        \preceq
        \mathbb{E}[\hat X^2].
    \end{equation}
    Therefore, the estimation error can be bounded as follows:
    \begin{align}
        \left\|\overline X^{(K)}_T-X\right\|_\infty
        &\leq
        \left\|
            \overline X^{(K)}_T
            -
            \mathbb{E}[\operatorname{clip}_K(\hat X)]
        \right\|_\infty
        +
        \left\|
            \mathbb{E}[\operatorname{clip}_K(\hat X)]-X
        \right\|_\infty \\
        &\leq
        \left\|
            \overline X^{(K)}_T
            -
            \mathbb{E}[\operatorname{clip}_K(\hat X)]
        \right\|_\infty
        +
        \frac{\sigma^2}{K} \\
        &=
        O\left(
            \frac{\sigma^2}{K}
            +
            \sqrt{\frac{\sigma^2\log(d/\delta)}{T}}
            +
            \frac{K\log(d/\delta)}{T}
        \right).
    \end{align}
    To obtain the last line, apply the matrix Bernstein inequality (Eq.~\eqref{app:eq:matrix_bernstein_high_Probability}) to $Y_i=\operatorname{clip}_K(\hat X_i) -\mathbb{E}[\operatorname{clip}_K(\hat X)]$. These matrices satisfy $\mathbb{E}[Y_i]=0$, $\|Y_i\|_\infty\leq \|\operatorname{clip}_K(\hat X_i)\|_\infty+\|\mathbb{E}[\operatorname{clip}_K(\hat X)]\|_\infty \leq2K$, and
    \begin{align}
        \mathbb{E}[Y_i^2]
        &=
        \mathbb{E}[\operatorname{clip}_K(\hat X)^2]
        -
        \mathbb{E}[\operatorname{clip}_K(\hat X)]^2 \\
        &\preceq
        \mathbb{E}[\operatorname{clip}_K(\hat X)^2]
        \preceq
        \mathbb{E}[\hat X^2],
    \end{align}
    and hence
    $\|\mathbb{E}[Y_i^2]\|_\infty\leq\sigma^2$.
\end{proof}

\subsection{A variance lower bound for arbitrary estimators}
\label{app:subsec:chapman_robbins}

Consider an arbitrary real-valued estimator $\hat o=\hat o(z)$ computed from a measurement outcome $z$ whose distribution is either $p$ or $q$. Suppose that $p$ and $q$ are close (small chi-square divergence), while the expectations of $\hat o$ under these distributions are well separated. If $\hat o$ had small variance under $q$, then its values would be concentrated around $\mathbb E_{z\sim q}[\hat o(z)]$. Since $p$ and $q$ assign similar weights to the same outcomes, changing the distribution from $q$ to $p$ could not substantially change the expectation of $\hat o$.
Therefore, a large separation between the two expectations despite a small chi-square divergence requires a large variance. The following theorem quantifies this tradeoff without assuming that the estimator is unbiased or linear.

\begin{theorem}[Hammersley--Chapman--Robbins inequality~\cite{nishiyama2019new}]
\label{app:thm:var_lower}
    Let $p$ and $q$ be probability distributions on a common sample
    space $\mathcal Z$, with $q(z)>0$ whenever $p(z)>0$. Let
    $\hat o=\hat o(z)$ be any real-valued estimator with finite second
    moment under $q$. Define
    \begin{equation}
        \chi^2\!\left(
            p
            \middle\Vert
            q
        \right)
        :=
        \mathbb E_{z\sim q}
        \left[
            \left(
                \frac{p(z)}{q(z)}-1
            \right)^2
        \right].
    \end{equation}
    If this quantity is finite and nonzero, then
    \begin{equation}
        \operatorname{Var}_{z\sim q}
        \!\left[
            \hat o(z)
        \right]
        \geq
        \frac{
            \left|
                \mathbb E_{z\sim p}[\hat o(z)]
                -
                \mathbb E_{z\sim q}[\hat o(z)]
            \right|^2
        }{
            \chi^2\!\left(
                p
                \middle\Vert
                q
            \right)
        }.
    \end{equation}
\end{theorem}

\subsection{Entrywise inequalities for nonnegative matrices}

We collect several elementary properties of the entrywise order for nonnegative matrices and vectors that will be used below. For two matrices of the same size, we write 
\begin{equation}
    A\leq_{\mathrm{ew}}B
\end{equation}
if $A_{ij}\leq B_{ij}$ for every $i,j$, and use the same notation for vectors. If $A\leq_{\mathrm{ew}}\widetilde A$ and $x\leq_{\mathrm{ew}}\widetilde x$ are compatible and nonnegative, then $Ax\leq_{\mathrm{ew}}\widetilde A\widetilde x$. Likewise, if $B\leq_{\mathrm{ew}}\widetilde B$, then $AB\leq_{\mathrm{ew}}\widetilde A\widetilde B$. Iterating this property gives, for nonnegative square matrices of the same size,
\begin{equation}
    A_i\leq_{\mathrm{ew}}\widetilde A_i
    \ \text{for all }i
    \quad\Longrightarrow\quad
    A_1\cdots A_m
    \leq_{\mathrm{ew}}
    \widetilde A_1\cdots\widetilde A_m
    \quad\Longrightarrow\quad
    \operatorname{Tr}(A_1\cdots A_m)
    \leq
    \operatorname{Tr}(\widetilde A_1\cdots\widetilde A_m).
    \label{app:eq:elementwise_product_bound}
\end{equation}

\begin{definition}[Spectrum and spectral radius]
    For a square matrix $M$, let $\operatorname{spec}(M)$ denote its spectrum, the set of its eigenvalues. Its spectral radius is $\rho(M):=\max_{\mu\in\operatorname{spec}(M)}|\mu|$.
\end{definition}

We also use the following consequence of the Collatz--Wielandt inequality.

\begin{fact}[Collatz--Wielandt bound%
~\cite{maccluer2000many}]
\label{app:fact:cw_bound}
    Let $A$ be a nonnegative $r\times r$ matrix and let
    $g=(g_1,\ldots,g_r)^T$ satisfy $g_i>0$ for every $i$. If
    $Ag\leq_{\mathrm{ew}}\lambda g$, then
    \begin{equation}
        \rho(A)\leq\lambda,
        \qquad
        \operatorname{Tr}(A^m)
        \leq
        r\lambda^m
    \end{equation}
    for every positive integer $m$.
\end{fact}

We refer to such a positive vector $g$ as a Collatz--Wielandt test vector. A valid value of $\lambda$ can be obtained by bounding the componentwise ratios $\max_i(Ag)_i/g_i$.

\section{Tensor-network formulation of two-layer Clifford measurements}

The technical step in our analysis is a tensor-network description of the two-layer block Clifford ensemble. Rather than comparing the ensemble with an approximate unitary design, we directly exploit the staggered geometry of the circuit. We divide the system into half-blocks of $k/2$ qubits. Each $k$-qubit block consists of two adjacent half-blocks, and the blocks in the second layer are shifted by one half-block relative to those in the first layer. This structure allows us to assign a local transfer matrix to each $k$-qubit block. Multiplying these matrices along the circuit gives matrix product state~(MPS) representations of the visibility $m_P$ and the correlated visibility $\tau(P,Q)$.
Labeling every Pauli operator on a $k$-qubit block would naively require $4^k$ labels. The transitivity of the Clifford group reduces them to two types for a single Pauli: identity and nonidentity. For a Pauli pair, six types remain: $(I,I)$, $(P,I)$, $(I,P)$, equal nonidentity Paulis $(P,P)$, distinct commuting Paulis $(P,Q)$ with $PQ=QP$, and anticommuting Paulis $(P,Q)$ with $PQ=-QP$.
The Clifford average is the same for all Pauli operators or Pauli pairs within each type. This reduces the local transfer matrix dimensions to sizes independent of the block Hilbert space dimension $2^k$. We use these tensor network representations for two purposes. First, in the proof of Theorem~\ref{main:thm:1}, we combine the MPS representations of $m_P$ and $\tau(P,Q)$ and bound the associated transfer matrix product. Second, using the factorization structure of $m_P$, we construct an exact MPS for $m_P^{-1}$ with bond dimension $\chi=n/k$. This permits exact implementation of the inverse shadow channel $\mathcal M^{-1}$, without using the Haar inverse or a numerical approximation.

\subsection{MPS representation of $m_P$}
\begin{figure*}[t]
    \centering
    \includegraphics[width=0.9\linewidth, trim=2.5cm 22.2cm 4cm 1cm, clip]{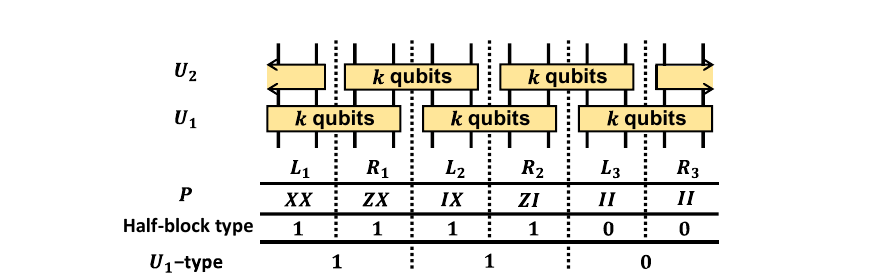}
    \caption{\textbf{Half-block decomposition of the visibility $m_P$.} The first layer acts on $(L_i,R_i)$, while the second layer acts on $(R_i,L_{i+1})$. A half block has type $0$ if the Pauli operator restricted to it is the identity and type $1$ otherwise. The type of a first layer block is obtained by combining the types of its two half blocks. In the example shown here, the half block type word is $(1,1,1,1,0,0)$ and the block support word is $(1,1,0)$. 
    }
    \label{app:fig:m_P}
\end{figure*}

The MPS representation of the visibility $m_P$ of the two-layer Clifford ensemble has been studied in Ref.~\cite{cho2025sample}. We give a different derivation here because the same construction will be used for the correlated visibility $\tau(P,Q)$. Assume that $k$ is even and divides $n$, and let $m=n/k$. We divide the system into $2m$ half-blocks, each containing $k/2$ qubits. The blocks in the first layer are labeled by $(L_i,R_i)$ for $i=1,\ldots,m$, while the blocks in the second layer are labeled by $(R_i,L_{i+1})$, with the cyclic convention $L_{m+1}=L_1$.
Accordingly, we write $U=U_2U_1$, where $U_1$ and $U_2$ are products of independent $k$-qubit Clifford unitaries acting on the blocks in the first and second layers, respectively.

For Pauli operators $P$ and $Q$, define $\operatorname{type}(P)=0$ if $P=I$ and $\operatorname{type}(P)=1$ otherwise. 
For $a,b\in\{0,1\}$, define $a\star b$ by the following table:
\begin{equation}\label{app:eq:single_pauli_type_composition}
    \begin{array}{cc|c}
        a & b & a\star b\\
        \hline
        0 & 0 & 0\\
        0 & 1 & 1\\
        1 & 0 & 1\\
        1 & 1 & 1
    \end{array}.
\end{equation}
Thus,
$\operatorname{type}(P\otimes Q)=\operatorname{type}(P)\star\operatorname{type}(Q)$. Let $P_i=P_{L_i}\otimes P_{R_i}$ be the restriction of $P$ to the $i$th block in the first layer, and define $r_i:=\operatorname{type}(P_i)$. The binary vector $\boldsymbol{r}(P)=(r_1,\ldots,r_m)$ will be called the block support vector of $P$. In the example shown in Fig.~\ref{app:fig:m_P}, the half-block types are $(1,1,1,1,0,0)$ and $\boldsymbol{r}(P)=(1,1,0)$.

The visibility $m_P$ depends on $P$ only through this block support vector. To see this, consider two Pauli operators $P$ and $Q$ with the same block support vector. The transitivity of the Clifford group implies that there is a block Clifford $V_1\in\mathrm{Cl}(k)^{\otimes m}$ such that $Q=V_1PV_1^\dagger$, up to phase. Since $U_1V_1$ has the same distribution as $U_1$, we have $m_Q=m_P$. Thus, the detailed Pauli operator within each active block is not needed when computing the visibility $m_P$.

For a $k$-qubit Pauli operator $P$, let $r\in\{0,1\}$ denote its support type, with $r=0$ for the identity and $r=1$ otherwise. Define
\begin{align}
    p_r(a,b)
    &:=
    \Pr_{V\sim\mathrm{Cl}(k)}
    \left[
    \mathrm{type}((VPV^\dagger)_L)=a ,
    \mathrm{type}((VPV^\dagger)_R)=b
    \right].
    \label{app}
\end{align}
If $r=0$, the Pauli operator on the block is the identity and hence $(a,b)=(0,0)$. If $r=1$, the transitivity property of the Clifford group makes the output $VPV^{\dagger}$ uniform over all $4^k-1$ non-identity Pauli operators. A half-block contains $4^{k/2}-1=2^k-1$ non-identity Pauli operators. Therefore,
\begin{equation}\label{app:eq:single_pauli_half_type_distribution}
    \begin{array}{c|cccc}
        (a,b)
        &(0,0)&(0,1)&(1,0)&(1,1)\\
        \hline
        p_0(a,b)
        &1&0&0&0\\[2pt]
        p_1(a,b)
        &0
        &\dfrac{1}{2^k+1}
        &\dfrac{1}{2^k+1}
        &\dfrac{2^k-1}{2^k+1}
    \end{array}.
\end{equation}

For a uniformly random $k$-qubit Clifford, let $s(0)$ and $s(1)$ denote the probabilities that a Pauli of type $0$ and type $1$, respectively, is mapped into $\pm\mathcal Z$. Clifford conjugation fixes the identity, so $s(0):=1$. A nonidentity Pauli is mapped uniformly to one of the $4^k-1$ nonidentity Paulis, of which $2^k-1$ lie in $\pm\mathcal Z$. Hence, $s(1):=(2^k-1)/(4^k-1)=(2^k+1)^{-1}$.

After applying $U_1$, define $a_i:=\operatorname{type}((U_1PU_1^\dagger)_{L_i})$ and $b_i:=\operatorname{type}((U_1PU_1^\dagger)_{R_i})$. The input to the $i$th block of $U_2$ acts on $(R_i,L_{i+1})$ and has type $b_i\star a_{i+1}$, where $a_{m+1}=a_1$. Thus, its probability of being mapped into $\pm\mathcal Z$ by the corresponding Clifford is $s(b_i\star a_{i+1})$. Let $w_{k,2}(P)$ denote the number of second-layer blocks on which $P$ is nonidentity. Independence of the Clifford blocks in both layers then gives

\begin{align}
    m_P
    &=
    \mathbb{E}_{U}
    \left[
        \mathbf{1}
        \left\{
            UPU^\dagger\in\pm\mathcal{Z}
        \right\}
    \right]
    =
    \mathbb{E}_{U_1}\mathbb{E}_{U_2}
    \left[
        \mathbf{1}
        \left\{
            U_2U_1PU_1^\dagger U_2^\dagger
            \in\pm\mathcal{Z}
        \right\}
    \right]
    \\
    &=
    \mathbb{E}_{U_1}
    \left[
        (2^k+1)^{-w_{k,2}(U_1PU_1^\dagger)}
    \right]
    =
    \mathbb{E}_{U_1}
    \left[
        \prod_{i=1}^{m}
        s(b_i\star a_{i+1})
    \right]
    \\
    &=
    \sum_{\substack{
        a_1,\ldots,a_m\in\{0,1\}\\
        b_1,\ldots,b_m\in\{0,1\}
    }}
    \prod_{i=1}^{m}
    p_{r_i}(a_i,b_i)\,
    s(b_i\star a_{i+1})\\
    &=
    \sum_{a_1,\ldots,a_m\in\{0,1\}}
    \prod_{i=1}^{m}
    \left[
        \sum_{b_i\in\{0,1\}}
        p_{r_i}(a_i,b_i)\,
        s(b_i\star a_{i+1})
    \right]
    \\
    &=
    \sum_{a_1,\ldots,a_m\in\{0,1\}}
    \prod_{i=1}^{m}
    (\mathbb{K}_{r_i})_{a_i,a_{i+1}}
    \\
    &=
    \operatorname{Tr}
    \left(
        \mathbb{K}_{r_1}\cdots\mathbb{K}_{r_m}
    \right),
    \label{app:eq:visibility_mps}
\end{align}
Here, for $r\in\{0,1\}$, the transfer matrix $\mathbb{K}_r$ is defined by
\begin{equation}\label{app:eq:visibility_transfer_definition}
    (\mathbb{K}_r)_{aa'}
    :=
    \sum_{b\in\{0,1\}}
    p_r(a,b)\,
    s(b\star a'),
    \qquad
    a,a'\in\{0,1\}.
\end{equation}
The row index $a$ records the type on $L_i$, the column index $a'$ records the type on $L_{i+1}$, and the type $b$ on $R_i$ is summed over. Using Eq.~\eqref{app:eq:single_pauli_half_type_distribution}, the two transfer matrices are
\begin{equation}\label{app:eq:visibility_transfer_matrices}
    \mathbb{K}_0
    =
    \begin{pmatrix}
        1 & \dfrac{1}{2^k+1}\\
        0 & 0
    \end{pmatrix},
    \qquad
    \mathbb{K}_1
    =
    \frac{1}{(2^k+1)^2}
    \begin{pmatrix}
        1 & 1\\
        2^{k+1} & 2^k
    \end{pmatrix},
\end{equation}
as in Ref.~\cite{cho2025sample}. Thus, $P\mapsto m_P$ has an exact MPS representation with bond dimension two. The physical index of this MPS is the block type $r_i\in\{0,1\}$.

\subsection{MPS representation of $m_P^{-1}$}\label{app:subsec:eff_mpo}

The MPS representation of $m_P$ obtained above does not immediately give an MPS representation of $m_P^{-1}$. In particular, taking the inverse of each transfer matrix or taking the reciprocal of each matrix entry does not give the entrywise reciprocal of the full contraction. For the two-layer ensemble, an exact representation can instead be obtained from the  fact that $\operatorname{rank}(\mathbb{K}_0)=1$.

Let $m=n/k$, and let $\boldsymbol{r}(P)=(r_1,\ldots,r_m)\in\{0,1\}^m$ be the block support vector of $P$ defined in the previous subsection. Thus,
\begin{equation}
    m_P
    =
    \operatorname{Tr}
    \left(
        \mathbb{K}_{r_1}\cdots\mathbb{K}_{r_m}
    \right).
\end{equation}
The matrix $\mathbb{K}_0$ can be written as
\begin{equation}\label{app:eq:visibility_rank_one_reset}
    \mathbb{K}_0
    =
    st^{T},
    \qquad
    s = (1, 0)^T,
    \qquad
    t = (1,  \frac{1}{2^k+1})^T.
\end{equation}
For a nonnegative integer $\ell$, define $\alpha_\ell:=t^{T}\mathbb{K}_1^\ell s$.

\begin{fact}[Appendix G in \cite{cho2025sample}]\label{app:fact:alpha}
    The eigenvalues of $\mathbb{K}_1$ are $\lambda_{\pm}=\frac{2^k+1\pm\sqrt{4^k+6\cdot2^k+1}}{2(2^k+1)^2}$.
    For every $\ell\geq0$,
    \begin{equation}
        \alpha_{\ell}
        =
        \alpha_{+}\lambda_{+}^{\ell}
        +
        \alpha_{-}\lambda_{-}^{\ell},
    \end{equation}
    where $\alpha_{\pm}=\frac{\alpha_1-\lambda_{\mp}}{\lambda_{\pm}-\lambda_{\mp}}, \alpha_1=\frac{3\cdot2^k+1}{(2^k+1)^3}.$
\end{fact}

For $\ell \ge 0$, Eq.~\eqref{app:eq:visibility_rank_one_reset} gives
\begin{equation}\label{app:eq:visibility_reset_identity}
    \mathbb{K}_0
    \mathbb{K}_1^\ell
    \mathbb{K}_0
    =
    st^{T} \mathbb{K}_1^{\ell} st^{T}
    =
    s (t^{T} \mathbb{K}_1^{\ell} s)t^{T}
    =
    \alpha_\ell\mathbb{K}_0.
\end{equation}
Suppose that $\boldsymbol{r}(P)$ contains at least one zero, equivalently $\boldsymbol{r}(P)\neq1^m$. Since the trace is invariant under cyclic permutations, we can write the support vector as $0\,1^{\ell_1}0\,1^{\ell_2}\cdots0\,1^{\ell_z}$, where $\sum_{i=1}^{z} \ell_i=|\boldsymbol{r}(P)|$ is the number of ones in the vector. Adjacent zeros are allowed and correspond to $\ell_j=0$. Repeated application of Eq.~\eqref{app:eq:visibility_reset_identity} and the fact that $\mathbb{K}_0^2=\mathbb{K}_0$ give 
\begin{equation}\label{app:eq:mP_decomp}
    m_P=\prod_{j=1}^{z}\alpha_{\ell_j}.
\end{equation}
For example,
\begin{align}
    m=6,\quad 
    &\boldsymbol{r}(P)=101101 \sim 011011 = 01^201^2 \longmapsto\\ &\quad m_P=\operatorname{Tr}(\mathbb{K}_{0}\mathbb{K}_{1}^{2}\mathbb{K}_{0}\mathbb{K}_{1}^{2})=\alpha_{2}\operatorname{Tr}(\mathbb{K}_{0}\mathbb{K}_{1}^{2})=\alpha_2^2.
\end{align}
For the remaining support vector $\boldsymbol{r}(P)=1^m$, we have
$m_P=\operatorname{Tr}(\mathbb K_1^m)$. Motivated by the factorization
in Eq.~\eqref{app:eq:mP_decomp}, define the following matrices:
\begin{equation}\label{app:eq:inverse_visibility_mps_matrices}
    \mathbb{L}_0
    =
    \sum_{\ell=0}^{m-1}
    \alpha_\ell^{-1}|0\rangle\langle\ell|,
    \qquad
    \mathbb{L}_1
    =
    \sum_{\ell=0}^{m-2}
    |\ell+1\rangle\langle\ell|
    +
    \frac{
        [\operatorname{Tr}(\mathbb{K}_1^m)]^{-1}
    }{m}
    |0\rangle\langle m-1|.
\end{equation}
The matrices in Eq.~\eqref{app:eq:inverse_visibility_mps_matrices} are not unique. The essential properties are $\operatorname{rank}(\mathbb L_0)=1$ and $\mathbb L_0\mathbb L_1^\ell\mathbb L_0 =\alpha_\ell^{-1}\mathbb L_0$. The case $\boldsymbol{r}(P)=1^m$ requires only a separate adjustment.
\begin{lemma}[Exact MPS representation of the inverse visibility $m_P^{-1}$]
    \label{app:lem:exact_inverse_visibility_mps}
    For every Pauli operator $P$ with block support vector
    $\boldsymbol{r}(P)=(r_1,\ldots,r_m)$,
    \begin{equation}\label{app:eq:exact_inverse_visibility_mps}
        m_P^{-1}
        =
        \operatorname{Tr}
        \left(
            \mathbb{L}_{r_1}\cdots\mathbb{L}_{r_m}
        \right).
    \end{equation}
\end{lemma}

\begin{proof}
    We first establish the counterpart of Eq.~\eqref{app:eq:visibility_reset_identity}. For $0\leq\ell\leq m-1$, $\mathbb{L}_1^\ell|0\rangle=|\ell\rangle$. Therefore,
    \begin{align}
        \mathbb{L}_0\mathbb{L}_1^\ell\mathbb{L}_0
        &=
        \mathbb{L}_0\mathbb{L}_1^\ell
        \sum_{r=0}^{m-1}
        \alpha_{r}^{-1}|0\rangle\langle r|
        =
        \sum_{r=0}^{m-1}
        \alpha_{r}^{-1}
        \left(
            \mathbb{L}_0\mathbb{L}_1^\ell|0\rangle
        \right)
        \langle r|
        \nonumber
        \\
        &=
        \sum_{r=0}^{m-1}
        \alpha_{r}^{-1}
        \left(
            \mathbb{L}_0|\ell\rangle
        \right)
        \langle r|
        =
        \sum_{r=0}^{m-1}
        \alpha_{r}^{-1}\alpha_\ell^{-1}
        |0\rangle\langle r|
        =
        \alpha_\ell^{-1}\mathbb{L}_0.
        \label{app:eq:inverse_visibility_reset_identity}
    \end{align}
    In particular, setting $\ell=0$ and using $\alpha_0=1$ gives $\mathbb{L}_0^2=\mathbb{L}_0$. Suppose first that $\boldsymbol{r}(P)$ contains at least one zero ($\boldsymbol{r}(P)\neq 1^m$). By cyclicity of the trace, we may write it as $0\,1^{\ell_1}0\,1^{\ell_2}\cdots0\,1^{\ell_z}$. Repeated application of Eq.~\eqref{app:eq:inverse_visibility_reset_identity} gives $\operatorname{Tr}( \mathbb{L}_{r_1}\cdots\mathbb{L}_{r_m}) =\prod_{j=1}^{z}\alpha_{\ell_j}^{-1}$.
    Since $m_P=\prod_{j=1}^{z}\alpha_{\ell_j}$, this contraction is equal to $m_P^{-1}$. It remains to consider $\boldsymbol{r}(P)=1^m$ where $m_P=\operatorname{Tr}(\mathbb{K}_1^m)$. By construction, $\mathbb{L}_1$ satisfies $\mathbb{L}_1^m=\frac{[\operatorname{Tr}(\mathbb{K}_1^m)]^{-1}}{m}I_m$, where $I_m$ is the $m\times m$ identity matrix. Taking the trace gives $\operatorname{Tr}(\mathbb{L}_1^m) =[\operatorname{Tr}(\mathbb{K}_1^m)]^{-1}=m_P^{-1}$.
\end{proof}

The matrices $\mathbb{L}_0$ and $\mathbb{L}_1$ 
\begin{equation}
    \mathbb{L}_0
    =
    \begin{pmatrix}
        \alpha_0^{-1}
        & \alpha_1^{-1}
        & \cdots
        & \alpha_{m-2}^{-1}
        & \alpha_{m-1}^{-1}
        \\
        0 & 0 & \cdots & 0 & 0
        \\
        \vdots & \vdots & & \vdots & \vdots
        \\
        0 & 0 & \cdots & 0 & 0
    \end{pmatrix},
    \qquad
    \mathbb{L}_1
    =
    \begin{pmatrix}
        0 & 0 & \cdots & 0
        & \frac{[\operatorname{Tr}(\mathbb{K}_1^m)]^{-1}}{m}
        \\
        1 & 0 & \cdots & 0 & 0
        \\
        0 & 1 & \ddots & \vdots & \vdots
        \\
        \vdots & \ddots & \ddots & 0 & 0
        \\
        0 & \cdots & 0 & 1 & 0
    \end{pmatrix}.
\end{equation}
have dimension $m=n/k$. Hence, $P\mapsto m_P^{-1}$ has an exact MPS representation with bond dimension $\chi=n/k$. To obtain the matrix product operator~(MPO) of the inverse shadow channel $\mathcal{M}^{-1}$ in the Pauli transfer matrix representation, we extend this MPS diagonally in the local Pauli basis. On each block, $\mathbb{L}_0$ is placed on the diagonal entry corresponding to the identity Pauli, while $\mathbb{L}_1$ is placed on the diagonal entries corresponding to all non-identity Paulis. Contracting the virtual indices then assigns the diagonal coefficient $m_P^{-1}$ to every Pauli operator $P$. Since this diagonal extension does not change the virtual indices, the MPO has the same bond dimension $\chi=n/k$. For $k=O(\log n)$, the operator $U^\dagger\ket b\!\bra bU$ has polynomial MPO bond dimension, and so does the snapshot $\mathcal M^{-1}(U^\dagger\ket b\!\bra bU)$. Consequently, $\hat o=\operatorname{Tr}(O\hat\rho)$ can be evaluated in polynomial time whenever $O$ also has polynomial MPO bond dimension.

\subsection{MPS representation of $\tau(P,Q)$}
\begin{figure*}[t]
    \centering
    \includegraphics[width=0.9\linewidth, trim=2.5cm 21.8cm 4cm 1cm, clip]{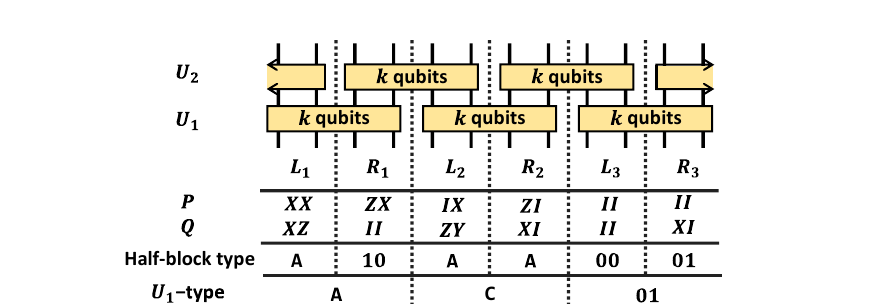}
    \caption{
    \textbf{Half-block decomposition used to evaluate the correlated visibility $\tau(P,Q)$.}
    The restrictions of $P$ and $Q$ to each half-block $L_i$ or $R_i$ determine its half-block type. The two half-block types within each first-layer block define the corresponding $U_1$ type according to Eq.~\eqref{app:eq:pauli_pair_type_composition}. The staggered second layer $U_2$ couples $R_i$ and $L_{i+1}$, with periodic boundary conditions indicated by the arrows.
    }
    \label{app:fig:tau_P}
\end{figure*}
We now apply the transfer matrix construction used for $m_P$ to the correlated visibility $\tau(P,Q)$. The main difference is that each half-block now carries a pair of Pauli operators. We therefore need to record the relation between the two Pauli operators, rather than only whether a single Pauli operator is the identity. Let
\begin{equation}
    \mathcal{S}
    :=
    \{
        \texttt{00},
        \texttt{10},
        \texttt{01},
        \texttt{E},
        \texttt{C},
        \texttt{A}
    \}.
\end{equation}
For Pauli operators $P$ and $Q$ acting on the same half-block or full block,
define $\operatorname{type}(P,Q)\in\mathcal{S}$ according to
\begin{equation}\label{app:eq:pauli_pair_types}
    \begin{array}{c|c}
        \operatorname{type}(P,Q) & \text{condition}\\
        \hline
        \texttt{00} & P=I,\quad Q=I\\
        \texttt{10} & P\neq I,\quad Q=I\\
        \texttt{01} & P=I,\quad Q\neq I\\
        \texttt{E}  & P=Q\neq I\\
        \texttt{C}  & P, Q\neq I,\quad P\neq Q,\quad PQ=QP\\
        \texttt{A}  & PQ=-QP
    \end{array}.
\end{equation}
The types $\texttt{10}$ and $\texttt{01}$ are kept separate because the pair $(P,Q)$ is ordered. The types $\texttt{C}$ and $\texttt{A}$ also need to be distinguished, even though both Pauli operators are non-identity in either case.

For $a,b\in\mathcal{S}$, define $a\ast b$ as the type obtained by taking
the tensor product of pair types $a$ and $b$. Thus, if
$\operatorname{type}(P_1,Q_1)=a$ and
$\operatorname{type}(P_2,Q_2)=b$, then
\begin{equation}
    \operatorname{type}
    (P_1\otimes P_2,Q_1\otimes Q_2)
    =
    a\ast b.
\end{equation}
The complete composition table of $a\ast b$ is
\begin{equation}\label{app:eq:pauli_pair_type_composition}
    \begin{array}{c|cccccc}
        a\backslash b
        &\texttt{00}&\texttt{10}&\texttt{01}
        &\texttt{E}&\texttt{C}&\texttt{A}\\
        \hline
        \texttt{00}
        &\texttt{00}&\texttt{10}&\texttt{01}
        &\texttt{E}&\texttt{C}&\texttt{A}\\
        \texttt{10}
        &\texttt{10}&\texttt{10}&\texttt{C}
        &\texttt{C}&\texttt{C}&\texttt{A}\\
        \texttt{01}
        &\texttt{01}&\texttt{C}&\texttt{01}
        &\texttt{C}&\texttt{C}&\texttt{A}\\
        \texttt{E}
        &\texttt{E}&\texttt{C}&\texttt{C}
        &\texttt{E}&\texttt{C}&\texttt{A}\\
        \texttt{C}
        &\texttt{C}&\texttt{C}&\texttt{C}
        &\texttt{C}&\texttt{C}&\texttt{A}\\
        \texttt{A}
        &\texttt{A}&\texttt{A}&\texttt{A}
        &\texttt{A}&\texttt{A}&\texttt{C}
    \end{array}.
\end{equation}
In particular, $\texttt{A}\ast\texttt{A}=\texttt{C}$ because two anticommutation signs cancel under the tensor product. Let $P_i=P_{L_i}\otimes P_{R_i}$ and $Q_i=Q_{L_i}\otimes Q_{R_i}$ be the restrictions of $P$ and $Q$ to the $i$th block in the first layer. Define $t_i:=\operatorname{type}(P_i,Q_i)$ and $\boldsymbol{t}(P,Q)=(t_1,\ldots,t_m)$. We refer to $\boldsymbol{t}(P,Q)$ as the block pair type vector of $(P,Q)$. An example of the block pair type vector is shown in Fig.~\ref{app:fig:tau_P}. 
In this example, the half-block pair types are $(\texttt{A},\texttt{10},\texttt{A},\texttt{A}, \texttt{00},\texttt{01})$. 
Composing the two half-block types within each first-layer block gives
\begin{equation}
    \boldsymbol{t}(P,Q)
    =
    (\texttt{A},\texttt{C},\texttt{01}).
\end{equation}
The correlated visibility depends on $(P,Q)$ only through this vector. Indeed, the six types in Eq.~\eqref{app:eq:pauli_pair_types} are preserved under simultaneous Clifford conjugation  $\mathrm{type}(P_i,Q_i)=\mathrm{type}(CP_iC^{\dagger},CQ_iC^{\dagger})$ where $C\in \mathrm{Cl}(k)$. Moreover, two ordered Pauli pairs of the same type can be mapped to each other by a Clifford unitary, up to signs. Therefore, if $(P,Q)$ and $(P',Q')$ have the same block pair type vector, there is a block Clifford $V_1\in\mathrm{Cl}(k)^{\otimes m}$ such that $P'=V_1PV_1^\dagger$ and $Q'=V_1QV_1^\dagger$, up to signs. Since $U_1V_1$ has the same distribution as $U_1$, it follows that $\tau(P',Q')=\tau(P,Q)$.

We next describe the output type distribution of one first-layer block. For compactness, let $D=2^k$. A half-block contains $4^{k/2}=D$ Pauli operators, while a full $k$-qubit block contains $4^k=D^2$ such operators. Let $n_t$ denote the number of ordered Pauli pairs of type $t$ on a half-block, and let $N_t$ denote the corresponding number on a full block. Their values are
\begin{equation}\label{app:eq:pauli_pair_orbit_sizes}
    \begin{array}{c|cc}
        t & n_t & N_t\\
        \hline
        \texttt{00}
        &1&1\\[2pt]
        \texttt{10}
        &D-1&D^2-1\\[2pt]
        \texttt{01}
        &D-1&D^2-1\\[2pt]
        \texttt{E}
        &D-1&D^2-1\\[2pt]
        \texttt{C}
        &\dfrac{(D-1)(D-4)}{2}
        &\dfrac{(D^2-1)(D^2-4)}{2}\\[8pt]
        \texttt{A}
        &\dfrac{D(D-1)}{2}
        &\dfrac{D^2(D^2-1)}{2}
    \end{array}.
\end{equation}
For each $t\in\mathcal{S}$, choose any Pauli pair $(P,Q)$ of type $t$ on a $k$-qubit block. Let $\pi_t(a,b)$ denote the probability that the output pair types on its left and right half-blocks are $a$ and $b$, respectively. More explicitly,
\begin{align}
    \pi_t(a,b)
    &:=
    \Pr_{V\sim\mathrm{Cl}(k)}
    \left[
        \begin{array}{l}
        \operatorname{type}
        \bigl((VPV^\dagger)_L,(VQV^\dagger)_L\bigr)=a,\\
        \operatorname{type}
        \bigl((VPV^\dagger)_R,(VQV^\dagger)_R\bigr)=b
        \end{array}
    \right]
    \nonumber\\
    &=
    \frac{n_an_b}{N_t}\,
    \mathbf{1}\{a\ast b=t\}.
    \label{app:eq:pauli_pair_output_distribution}
\end{align}
This probability does not depend on the choice of $(P,Q)$. The second equality follows because Clifford conjugation is uniform within each pair type. The factor $n_an_b$ counts the pairs whose left and right types are $a$ and $b$, while $a\ast b=t$ ensures that their full block type is $t$. 

For a Pauli pair $(P,Q)$ of type $t\in\mathcal S$, define $\eta(t):=\Pr_{C\sim\mathrm{Cl}(k)} [CPC^\dagger,\ CQC^\dagger\in\pm\{I,Z\}^{\otimes k}]$, where $C$ is sampled uniformly. By Eq.~\eqref{app:eq:global_cliff_Ord_correlated_visibility} with $n=k$, its values are
\begin{equation}\label{app:eq:local_pair_visibility}
    \begin{array}{c|cccccc}
        t
        &\texttt{00}&\texttt{10}&\texttt{01}
        &\texttt{E}&\texttt{C}&\texttt{A}\\
        \hline
        \eta(t)
        &1
        &\dfrac{1}{D+1}
        &\dfrac{1}{D+1}
        &\dfrac{1}{D+1}
        &\dfrac{2}{(D+1)(D+2)}
        &0
    \end{array}.
\end{equation}

For fixed $U_1$, set $\widetilde P=U_1PU_1^\dagger$ and $\widetilde Q=U_1QU_1^\dagger$, ignoring their signs, and define
\begin{equation}
    a_i
    :=
    \operatorname{type}
    (\widetilde P_{L_i},\widetilde Q_{L_i}),
    \qquad
    b_i
    :=
    \operatorname{type}
    (\widetilde P_{R_i},\widetilde Q_{R_i}).
\end{equation}
The first layer preserves the pair type on each full block, so $a_i\ast b_i=t_i$, and the probability of obtaining $(a_i,b_i)$ is $\pi_{t_i}(a_i,b_i)$.

The $i$th block in the second layer acts on $(R_i,L_{i+1})$. The input pair type of this block is therefore $b_i\ast a_{i+1}$, and its contribution to the correlated visibility is $\eta(b_i\ast a_{i+1})$. Using the cyclic convention $a_{m+1}=a_1$ and the independence of the Clifford blocks, we obtain
\begin{align}
    \tau(P,Q)
    &=
    \mathbb{E}_{U_1}\mathbb{E}_{U_2}
    \left[
        \mathbf{1}
        \left\{
            U_2U_1PU_1^\dagger U_2^\dagger
            \in\pm\mathcal Z
        \right\}
        \mathbf{1}
        \left\{
            U_2U_1QU_1^\dagger U_2^\dagger
            \in\pm\mathcal Z
        \right\}
    \right]
    \nonumber\\
    &=
    \mathbb{E}_{U_1}
    \left[
        \prod_{i=1}^{m}
        \eta(b_i\ast a_{i+1})
    \right]
    \nonumber\\
    &=
    \sum_{\substack{
        a_1,\ldots,a_m\in\mathcal{S}\\
        b_1,\ldots,b_m\in\mathcal{S}
    }}
    \prod_{i=1}^{m}
    \pi_{t_i}(a_i,b_i)\,
    \eta(b_i\ast a_{i+1})
    \label{app:eq:correlated_visibility_type_sum}\\
    &=
    \sum_{a_1,\ldots,a_m\in\mathcal{S}}
    \prod_{i=1}^{m}
    \left[
        \sum_{b_i\in\mathcal{S}}
        \pi_{t_i}(a_i,b_i)\,
        \eta(b_i\ast a_{i+1})
    \right]
    \nonumber\\
    &=
    \sum_{a_1,\ldots,a_m\in\mathcal{S}}
    \prod_{i=1}^{m}
    \left(
        \mathbb{T}_{t_i}
    \right)_{a_i,a_{i+1}}
    \nonumber\\
    &=
    \operatorname{Tr}
    \left(
        \mathbb{T}_{t_1}
        \cdots
        \mathbb{T}_{t_m}
    \right).
    \label{app:eq:correlated_visibility_mps}
\end{align}
Here, for $t,a,a'\in\mathcal{S}$, the pair transfer matrix is defined by
\begin{align}
    \left(
        \mathbb{T}_t
    \right)_{a,a'}
    :=
    \sum_{b\in\mathcal{S}}
    \pi_t(a,b)\,
    \eta(b\ast a')=
    \frac{n_a}{N_t}
    \sum_{\substack{b\in\mathcal{S}\\a\ast b=t}}
    n_b\,\eta(b\ast a').
    \label{app:eq:correlated_visibility_transfer}
\end{align}
The row index $a$ records the pair type on $L_i$, the column index $a'$ records the pair type on $L_{i+1}$, and the pair type $b$ on $R_i$ is summed over. Equations~\eqref{app:eq:pauli_pair_type_composition}, \eqref{app:eq:pauli_pair_orbit_sizes}, and \eqref{app:eq:local_pair_visibility} specify every entry of the six transfer matrices. Therefore, $(P,Q)\mapsto\tau(P,Q)$ has an exact MPS representation with bond dimension six. Its physical index is the first layer block pair type $t_i\in\mathcal{S}$, and the bond dimension is 6.

\section{Proof of Theorem~\ref{main:thm:1}}
In this section, we state a formal version of Theorem~\ref{main:thm:1} in the main text and prove it in two steps. We first rewrite the variance $\operatorname{Var}[\hat o]$ and show that the desired Frobenius norm bound follows from controlling a spectral norm. We then use the transfer matrix representations developed in the previous section to establish the required norm bound. 

\begin{theorem}\label{appx:thm:1}
    For the two-layer Clifford ensemble with even block size $k\ge8$ dividing $n$ and satisfying $k2^{k/2}\ge Cn$ for some constant $C>0$,and for every state $\rho$ and Hermitian observable $O$, the unbiased estimator $\hat{o}=\operatorname{Tr}(O\hat{\rho})$ satisfies
    \begin{equation}
        \operatorname{Var}[\hat{o}]
        \leq
        C_F\lVert O_0\rVert_2^2,
        \qquad
        C_F=8e^{45/C},
    \end{equation}
    where $O_0=O-\frac{\operatorname{Tr}(O)}{d}I$.
\end{theorem}

\subsection{Reduction to a transfer matrix bound}
We first separate the dependence on the input state $\rho$ from the remaining calculation. Since every exact shadow snapshot $\hat\rho = \mathcal{M}^{-1}(U^{\dagger}\ket{b}\!\bra{b}U)$ has unit trace, the identity part of $O$ does not contribute to the variance $\operatorname{Var}[\operatorname{Tr}(O\hat{\rho})] = \operatorname{Var}[\operatorname{Tr}(O_0\hat{\rho})]$. We therefore work with $O_0$ throughout this subsection. It is useful to collect the second moment into the following operator:
\begin{align}
    \Gamma(O_0)
    &:=
    \mathbb{E}_{U}
    \sum_{b}
    U^\dagger\ket{b}\!\bra{b}U
        \operatorname{Tr}
        \left(
            O_0\,
            \mathcal{M}^{-1}
            \left(
                U^\dagger\ket{b}\!\bra{b}U
            \right)
        \right)^2 \label{app:eq:prediction_operator}.
\end{align}
$\Gamma(O_0)$ is positive semidefinite because it is the sum of projectors with positive coefficients. Born's rule gives $\mathbb{E}_{U,b} [\operatorname{Tr}(O_0\hat{\rho})^2] = \operatorname{Tr} (\rho\Gamma(O_0))$. Consequently,
\begin{equation}\label{app:eq:variance_prediction_bound}
    \operatorname{Var}[\hat{o}]
    \leq
    \operatorname{Tr}(\rho \Gamma(O_0))
    \le
    \left\|
        \Gamma(O_0)
    \right\|_\infty .
\end{equation}

Thus, bounding the variance $\operatorname{Var}[\hat o]$ reduces to bounding the spectral norm of $\Gamma(O_0)$.
For each $\boldsymbol u\in\{0,1\}^m$, let 
\begin{equation}
    O_{\boldsymbol u}:=\sum_{P\in\mathcal{P}_{\boldsymbol u}}o_P P,
\end{equation}
where $\mathcal{P}_{\boldsymbol u}:=\{P\in\mathcal{P}_n:\mathbf{r}(P)=\boldsymbol u\}$ and $o_P=\frac{\operatorname{Tr}(O_0P)}{d}$. Then
\begin{equation}\label{app:eq:observable_support_decomposition}
    O_0=\sum_{\boldsymbol u\in\{0,1\}^m}O_{\boldsymbol u},
    \qquad
    \lVert O_0\rVert_2^2
    =
    \sum_{\boldsymbol u\in\{0,1\}^m}\lVert O_{\boldsymbol u}\rVert_2^2. 
\end{equation}
The second equality follows from $\operatorname{Tr}(O_{\boldsymbol u} O_{\boldsymbol v})=0$ if $\boldsymbol{u}\neq \boldsymbol{v}$. 
Substituting the support decomposition of $O_0$ into Eq.~\eqref{app:eq:prediction_operator} gives
\begin{align}
    \Gamma(O_0) 
    &= 
    \mathbb{E}_{U}\sum_{b}
    U^\dagger\ket{b}\!\bra{b}U
    [\operatorname{Tr}(O_0\mathcal{M}^{-1}(U^\dagger\ket{b}\!\bra{b}U))]^2\\
    &=
    \sum_{\boldsymbol u,\boldsymbol v}\mathbb{E}_{U}\sum_{b}
    U^\dagger\ket{b}\!\bra{b}U
    \operatorname{Tr}(O_{\boldsymbol u}\mathcal{M}^{-1}(U^\dagger\ket{b}\!\bra{b}U))
    \operatorname{Tr}(O_{\boldsymbol v}\mathcal{M}^{-1}(U^\dagger\ket{b}\!\bra{b}U))\\
    &=
    \sum_{\boldsymbol u,\boldsymbol v}\mathbb{E}_{U}\sum_{b}
    U^\dagger\ket{b}\!\bra{b}U
    \operatorname{Tr}(\mathcal{M}^{-1}(O_{\boldsymbol u})U^\dagger\ket{b}\!\bra{b}U)
    \operatorname{Tr}(\mathcal{M}^{-1}(O_{\boldsymbol v})U^\dagger\ket{b}\!\bra{b}U)\\
    &=
    \sum_{\boldsymbol u,\boldsymbol v}\frac{1}{m(\boldsymbol u) m(\boldsymbol v)}\mathbb{E}_{U}\sum_{b}
    U^\dagger\ket{b}\!\bra{b}U
    \operatorname{Tr}(O_{\boldsymbol u}U^\dagger\ket{b}\!\bra{b}U)
    \operatorname{Tr}(O_{\boldsymbol v}U^\dagger\ket{b}\!\bra{b}U)\\
    &=
    \sum_{\boldsymbol u,\boldsymbol v}\frac{1}{m(\boldsymbol u) m(\boldsymbol v)} \sum_{P \in \mathcal{P}_{\boldsymbol u}}\sum_{Q \in \mathcal{P}_v}o_P o_Q \tau(P,Q) PQ\\
    &=
    \sum_{\boldsymbol u,\boldsymbol v}\frac{G_{\boldsymbol u \boldsymbol v}(O_0)}{m(\boldsymbol u) m(\boldsymbol v)},
\end{align}
Here, Fact~\ref{app:fact:dual} is used to move $\mathcal M^{-1}$ from the measurement snapshot to $O_{\boldsymbol u}$ and $O_{\boldsymbol v}$. We also used $\mathcal M^{-1}(O_{\boldsymbol u})=m(\boldsymbol u)^{-1}O_{\boldsymbol u}$, which follows from $m_P=m(\boldsymbol u)$ for every $P\in\mathcal P_u$. We define
\begin{align}
    G_{\boldsymbol u \boldsymbol v}(O_0)
    &:=
    \sum_{P \in \mathcal{P}_{\boldsymbol u}}\sum_{Q \in \mathcal{P}_v}
    o_P o_Q \tau(P,Q) PQ\\
    &=
    \sum_{P \in \mathcal{P}_{\boldsymbol u}}\sum_{Q \in \mathcal{P}_v}
    o_P o_Q \operatorname{Tr}(\mathbb{T}_{t_1}\mathbb{T}_{t_2}\dots\mathbb{T}_{t_m}) PQ\\
    &=
    \sum_{a_1,...a_m \in \mathcal{S}}
    \sum_{P \in \mathcal{P}_{\boldsymbol u}}\sum_{Q \in \mathcal{P}_v}
    o_P o_Q
    (\mathbb{T}_{t_1})_{a_1, a_2}(\mathbb{T}_{t_2})_{a_2, a_3}\dots (\mathbb{T}_{t_m})_{a_m, a_1} PQ\\
    &=
    \sum_{a_1,...a_m \in \mathcal{S}}
    \sum_{P \in \mathcal{P}_{\boldsymbol u}}\sum_{Q \in \mathcal{P}_v}
    o_P o_Q
    \bigotimes_{i=1}^{m} (\mathbb{T}_{t_i})_{a_i, a_{i+1}}P_i Q_i \label{app:eq:G}.
\end{align}
Recall that $t_i=t_i(P,Q)=\operatorname{type}(P_i,Q_i) \in\mathcal S$ denotes the correlation type of the local Pauli pair $(P_i,Q_i)$ on the $i$th $k$-qubit block of the first layer.
The visibility type vectors $\boldsymbol{u}$ and $\boldsymbol{v}$ determine which of the local Pauli operators $P_i$ and $Q_i$ are identities. On blocks where $u_i=v_i=1$, however, they do not determine whether the local correlation type is $\texttt{E}$, $\texttt{C}$, or $\texttt{A}$, and therefore do not determine the transfer matrix $\mathbb T_{t_i(P,Q)}$. The following lemma separates this remaining dependence on the local Pauli pair.
\begin{lemma}[Local decomposition identity] \label{app:lem:local_ordered_product_decomposition}
    Let $P$ and $Q$ be non-identity Pauli operators acting on the same block, and let $t=t(P,Q)\in\{\texttt{E},\texttt{C},\texttt{A}\}$ denote their possible correlation type. Then, for $a,a'\in \mathcal{S}$,
    \begin{align}
        (\mathbb T_t)_{aa'}PQ
        =
        (\mathbb T_{\texttt{E}}-\mathbb T_{\texttt{C}})_{aa'}
        \mathbf 1\{P=Q\}PQ
        +
        \left(
            \frac{\mathbb T_{\texttt{C}}+\mathbb T_{\texttt{A}}}{2}
        \right)_{aa'}PQ
        +
        \left(
            \frac{\mathbb T_{\texttt{C}}-\mathbb T_{\texttt{A}}}{2}
        \right)_{aa'}QP.
        \label{app:eq:local_decomp}
    \end{align}
\end{lemma}

\begin{proof}
    It suffices to verify the identity for the three possible values of $t\in\{\texttt{E},\texttt{C},\texttt{A}\}$.
    
    If $t=\texttt{E}$, then $P=Q$. 
    \begin{equation}
        (\mathrm{RHS}) = (\mathbb T_{\texttt{E}}-\mathbb T_{\texttt{C}})_{aa'}I 
        + \left(\frac{\mathbb T_{\texttt{C}}+\mathbb T_{\texttt{A}}}{2}\right)_{aa'}I
        + \left(\frac{\mathbb T_{\texttt{C}}-\mathbb T_{\texttt{A}}}{2}\right)_{aa'}I
        =(\mathbb T_{\texttt{E}})_{aa'}I.
    \end{equation}
    If $t=\texttt{C}$, then $P\neq Q$ and $QP=PQ$. 
    \begin{equation}
        (\mathrm{RHS}) = 0
        + \left(\frac{\mathbb T_{\texttt{C}}+\mathbb T_{\texttt{A}}}{2}\right)_{aa'}PQ
        + \left(\frac{\mathbb T_{\texttt{C}}-\mathbb T_{\texttt{A}}}{2}\right)_{aa'}PQ
        =(\mathbb T_{\texttt{C}})_{aa'}PQ.
    \end{equation}
    If $t=\texttt{A}$, then $P\neq Q$ and $QP=-PQ$. 
    \begin{equation}
        (\mathrm{RHS}) = 0
        + \left(\frac{\mathbb T_{\texttt{C}}+\mathbb T_{\texttt{A}}}{2}\right)_{aa'}PQ
        - \left(\frac{\mathbb T_{\texttt{C}}-\mathbb T_{\texttt{A}}}{2}\right)_{aa'}PQ
        =(\mathbb T_{\texttt{A}})_{aa'}PQ.
    \end{equation}
    Thus the decomposition holds for every
    $t\in\{\texttt{E},\texttt{C},\texttt{A}\}$.
\end{proof}

We now apply the \textit{local decomposition identity} of Lemma~\ref{app:lem:local_ordered_product_decomposition} simultaneously to all first-layer blocks. Fix two visibility type vectors $\boldsymbol u,\boldsymbol v \in \{0,1\}^m$, and define $J_{ab}:=\{i\in[m]:(u_i,v_i)=(a,b)\}$. These sets give the disjoint partition
\begin{equation}
    [m]=J_{00}\sqcup J_{01}\sqcup J_{10}\sqcup J_{11}.
\end{equation}

For a fixed Pauli pair
$P\in\mathcal P_u$ and $Q\in\mathcal P_v$, let
\begin{equation}
    \boldsymbol t(P,Q)
    =
    \bigl(t_1(P,Q),\ldots,t_m(P,Q)\bigr)
\end{equation}
denote its correlation type vector. The visibility type vectors $u$ and
$v$ determine the correlation type on every block outside $J_{11}$:
\begin{equation}
    t_i(P,Q)
    =
    \begin{cases}
        \texttt{00},&i\in J_{00},\\
        \texttt{01},&i\in J_{01},\\
        \texttt{10},&i\in J_{10}.
    \end{cases}
\end{equation}
On a block $i\in J_{11}$, however,
\begin{equation}
    t_i(P,Q)\in\{\texttt{E},\texttt{C},\texttt{A}\}
\end{equation}
On a block $i\in J_{11}$, the correlation type $t_i(P,Q)$ depends on the detailed local Pauli pair $(P_i,Q_i)$. Thus, the visibility type vectors $\boldsymbol{u}$ and $\boldsymbol{v}$ do not determine the transfer matrix $\mathbb T_{t_i(P,Q)}$ on these blocks. Eq.~\eqref{app:eq:local_decomp} expresses each active-active contribution $u_i=v_i=1$ as a sum of three terms. Expanding the product $\bigotimes_{i=1}^{m}(\mathbb T_{t_i})_{a_i,a_{i+1}}P_iQ_i$ gives a sum over the three local terms at each block in $J_{11}$. We encode the selected terms by an auxiliary vector
\begin{equation}
    \boldsymbol{\kappa}=(\kappa_1,\ldots,\kappa_m).
\end{equation}

For fixed visibility type vectors $\boldsymbol{u}$ and $\boldsymbol{v}$, define the local set at block $i$ by
\begin{equation}
    \mathcal K_{\boldsymbol u,\boldsymbol v}^{(i)}
    :=
    \begin{cases}
        \{0\},
        &i\notin J_{11},\\[2pt]
        \{\mathord{=},\mathtt{fwd},\mathtt{bwd}\},
        &i\in J_{11}.
    \end{cases}
\end{equation}
Here $0$ is a dummy label on blocks outside $J_{11}$. The set of labels is then given by the Cartesian product
\begin{equation}
    \mathcal K_{\boldsymbol u,\boldsymbol v}:=\prod_{i=1}^m\mathcal K_{\boldsymbol u,\boldsymbol v}^{(i)}.
\end{equation}
Since every block in $J_{11}$ has three labels, $|\mathcal K_{\boldsymbol u,\boldsymbol v}|=3^{|J_{11}|}$. For $\xi\in\{\mathord{=},\mathtt{fwd},\mathtt{bwd}\}$, define
\begin{equation}
    J_{11, \boldsymbol{\kappa}}^{(\xi)}
    :=
    \{i\in J_{11}:\kappa_i=\xi\}.
    \label{app:eq:kappa_active_sets}
\end{equation}
Therefore, for given $\boldsymbol{u}$ and $\boldsymbol{v}$, each $\boldsymbol\kappa \in \mathcal{K}_{\boldsymbol u,\boldsymbol v}$ induces the disjoint partition of $[m]$
\begin{align}
    [m]
    ={}&
    J_{00}
    \sqcup J_{01}
    \sqcup J_{10}
    \nonumber\\
    &
    \sqcup J_{11, \boldsymbol{\kappa}}^{(=)}
    \sqcup
    J_{11, \boldsymbol{\kappa}}^{(\mathtt{fwd})}
    \sqcup
    J_{11, \boldsymbol{\kappa}}^{(\mathtt{bwd})}.
    \label{app:eq:kappa_partition}
\end{align}
For every $\boldsymbol\kappa\in\mathcal K_{\boldsymbol u,\boldsymbol v}$, define the local auxiliary transfer matrix at block $i$ by
\begin{equation}
    \mathbb R_{u_iv_i}^{(\kappa_i)}
    :=
    \begin{cases}
        \mathbb T_{\texttt{00}},
        &(u_i,v_i,\kappa_i)=(0,0,0),\\
        \mathbb T_{\texttt{01}},
        &(u_i,v_i,\kappa_i)=(0,1,0),\\
        \mathbb T_{\texttt{10}},
        &(u_i,v_i,\kappa_i)=(1,0,0),\\
        \mathbb T_{\texttt{E}}-\mathbb T_{\texttt{C}},
        &(u_i,v_i,\kappa_i)=(1,1,{=}),\\
        \dfrac{\mathbb T_{\texttt{C}}+\mathbb T_{\texttt{A}}}{2},
        &(u_i,v_i,\kappa_i)=(1,1,\mathtt{fwd}),\\
        \dfrac{\mathbb T_{\texttt{C}}-\mathbb T_{\texttt{A}}}{2},
        &(u_i,v_i,\kappa_i)=(1,1,\mathtt{bwd}).
    \end{cases}
    \label{app:eq:local_packet_transfer}
\end{equation}
\begin{remark}
    The auxiliary label $\kappa_i$ is introduced by applying Lemma~\ref{app:lem:local_ordered_product_decomposition} on the blocks in $J_{11}$. It is not the correlation type $t_i(P,Q)$, but labels one of the three terms in the local decomposition in Eq.~\eqref{app:eq:local_decomp}. Thus, after introducing $\boldsymbol{u}$, $\boldsymbol{v}$, and $\boldsymbol{\kappa}$, the transfer matrix $\mathbb T_{t_i(P,Q)}$ at each block is replaced by $\mathbb R_{u_iv_i}^{(\kappa_i)}$.
\end{remark}
With this notation, the local identity of Lemma~\ref{app:lem:local_ordered_product_decomposition} on a block $i\in J_{11}$ becomes
\begin{align}
    (\mathbb T_{t_i(P,Q)})_{aa'}P_iQ_i
    =
    (\mathbb R_{11}^{(=)})_{aa'}
    \mathbf 1\{P_i=Q_i\}P_iQ_i
    +
    (\mathbb R_{11}^{(\mathtt{fwd})})_{aa'}P_iQ_i
    +
    (\mathbb R_{11}^{(\mathtt{bwd})})_{aa'}Q_iP_i.
    \label{app:eq:local_packet_expansion}
\end{align}
Certain local terms reverse the Pauli multiplication order from $P_iQ_i$ to $Q_iP_i$, while the equality term imposes the constraint $P_i=Q_i$. This information must be retained when the local decomposition is extended to the full Pauli pair $(P,Q)$. We encode these equality constraints and multiplication-order choices in the following ordered Pauli product with respect to $\boldsymbol{\kappa}$.

\begin{definition}[Ordered Pauli product]
    For Pauli operators $P$ and $Q$ with visibility type vectors $\boldsymbol{u}$ and $\boldsymbol{v}$, define
    \begin{align}
        \operatorname{Ord}_{\boldsymbol{\kappa}}(P,Q)
        &:=
        \left(
            \prod_{i\in J_{11, \kappa}^{(=)}}
            \mathbf{1}\{P_i=Q_i\}
        \right)
        \bigotimes_{i=1}^{m}
        \begin{cases}
            P_iQ_i,
            & i\notin J_{11, \kappa}^{(\texttt{bwd})},\\
            Q_iP_i,
            & i\in J_{11, \kappa}^{(\texttt{bwd})}.
        \end{cases}\\
        &=
        \left(
            \prod_{i\in J_{11, \kappa}^{(=)}}
            \mathbf{1}\{P_i=Q_i\}
        \right)
        (P^{T_{S}}Q^{T_{S}})^{T_{S}},
    \end{align}
    where $S=J_{11, \boldsymbol{\kappa}}^{(\mathtt{bwd})}$, and $O^{T_S}$ denotes the partial transpose of $O$ over the blocks in $S$. We can extend this definition bilinearly. For $A=\sum_P a_P P$ and $B=\sum_Q b_Q Q$, define
    \begin{equation}\label{app:eq:ordered_product_bilinear_extension}
        \operatorname{Ord}_{\boldsymbol{\kappa}}(A,B)
        =
        \sum_{P,Q}
        a_Pb_Q
        \operatorname{Ord}_{\boldsymbol{\kappa}}(P,Q).
    \end{equation}
\end{definition}

\begin{corollary}\label{app:cor:kappa_ord}
    Applying Eq.~\eqref{app:eq:local_packet_expansion} at every block in $J_{11}$ gives the following identity.
    For $P \in \mathcal{P}_{\boldsymbol u}$ and $Q \in \mathcal{P}_v$,
    \begin{equation}
        \bigotimes_{i=1}^{m}
        \left[
            (\mathbb T_{t_i(P,Q)})_{a_i,a_{i+1}}P_iQ_i
        \right]
        =
        \sum_{\boldsymbol\kappa\in\mathcal K_{\boldsymbol u,\boldsymbol v}}
        \left[
            \prod_{i=1}^{m}
            (\mathbb R_{u_iv_i}^{(\kappa_i)})_{a_i,a_{i+1}}
        \right]
        \operatorname{Ord}_{\boldsymbol\kappa}(P,Q).
    \end{equation}
\end{corollary}

\begin{proof}
    For $i\notin J_{11}$, $\kappa_i=0$ and
    $\mathbb R_{u_iv_i}^{(0)}=\mathbb T_{t_i(P,Q)}$.  For $i\in J_{11}$,
    Eq.~\eqref{app:eq:local_packet_expansion} expands the local factor into
    the three terms labeled by
    $\kappa_i\in\{\mathord{=},\mathtt{fwd},\mathtt{bwd}\}$.
    By multilinearity of the tensor product, each sum over $\kappa_i$ can be
    moved outside the tensor product, giving
    $\sum_{\boldsymbol\kappa\in\mathcal K_{\boldsymbol u,\boldsymbol v}}$.  For each
    $\boldsymbol\kappa$, the transfer matrix factors form the product of
    $\mathbb R_{u_iv_i}^{(\kappa_i)}$, while the Pauli multiplication order
    and equality constraints are encoded in
    $\operatorname{Ord}_{\boldsymbol\kappa}(P,Q)$.
\end{proof}

Substituting Corollary~\ref{app:cor:kappa_ord} into
Eq.~\eqref{app:eq:G} gives
\begin{align}
    G_{\boldsymbol u \boldsymbol v}(O_0)
    &=
    \sum_{a_1,...a_m \in \mathcal{S}}
    \sum_{P \in \mathcal{P}_{\boldsymbol u}}\sum_{Q \in \mathcal{P}_v}
    o_P o_Q
    \bigotimes_{i=1}^{m} (\mathbb{T}_{t_i})_{a_i, a_{i+1}}P_i Q_i\\
    &=
    \sum_{a_1,...a_m \in \mathcal{S}}
    \sum_{P \in \mathcal{P}_{\boldsymbol u}}\sum_{Q \in \mathcal{P}_v}
    o_P o_Q
    \sum_{\kappa \in \mathcal{K}_{\boldsymbol u,\boldsymbol v}}
    (\mathbb{R}_{u_1, v_1}^{(\kappa_1)})_{a_1, a_{2}}
    \dots
    (\mathbb{R}_{u_m, v_m}^{(\kappa_m)})_{a_m, a_1}
    \operatorname{Ord}_{\kappa}(P,Q)\\
    &=
    \sum_{\kappa \in \mathcal{K}_{\boldsymbol u,\boldsymbol v}}
    \operatorname{Tr}(\mathbb{R}_{u_1, v_1}^{(\kappa_1)} \dots \mathbb{R}_{u_m, v_m}^{(\kappa_m)})
    \sum_{P \in \mathcal{P}_{\boldsymbol u}}\sum_{Q \in \mathcal{P}_v}
    o_P o_Q \operatorname{Ord}_{\kappa}(P,Q)\\
    &=
    \sum_{\kappa \in \mathcal{K}_{\boldsymbol u,\boldsymbol v}}
    \operatorname{Tr}(\mathbb{R}_{u_1, v_1}^{(\kappa_1)} \dots \mathbb{R}_{u_m, v_m}^{(\kappa_m)})
    \operatorname{Ord}_{\kappa}(O_{\boldsymbol u},O_{\boldsymbol v}).
\end{align}

\begin{lemma}\label{app:lem:mixed_order_product}
    Let $\mathcal H_1$, $\mathcal H_2$, and $\mathcal H_3$ be finite-dimensional Hilbert spaces, and let $T_S$ denote the partial transpose over an arbitrary tensor factor $S$ of $\mathcal H_2$. Then, for arbitrary $A_{12}\in\mathcal L(\mathcal H_1\otimes\mathcal H_2)$ and $B_{23}\in\mathcal L(\mathcal H_2\otimes\mathcal H_3)$,
    \begin{equation}\label{app:eq:mixed_order_product}
        \left\|
        \left(
        (A_{12}\otimes I_3)^{T_S}
        (I_1\otimes B_{23})^{T_S}
        \right)^{T_S}
        \right\|_\infty
        \leq
        \lVert A_{12}\rVert_2
        \lVert B_{23}\rVert_2 .
    \end{equation}
\end{lemma}

\begin{proof}
    Since the norm inequality $\|\cdot\|_{\infty}\le\|\cdot\|_{2}$ and partial transpose preserves the Frobenius norm,
    \begin{align}
        \left\|
        \left(
        (A_{12}\otimes I_3)^{T_S}
        (I_1\otimes B_{23})^{T_S}
        \right)^{T_S}
        \right\|_\infty
        &\leq
        \left\|
        \left(
        (A_{12}\otimes I_3)^{T_S}
        (I_1\otimes B_{23})^{T_S}
        \right)^{T_S}
        \right\|_2\\
        &=
        \left\|
        (A_{12}\otimes I_3)^{T_S}
        (I_1\otimes B_{23})^{T_S}
        \right\|_2\\
        &=
        \left\|
        (\tilde A_{12}\otimes I_3)
        (I_1\otimes \tilde B_{23})
        \right\|_2\\
        &=
        \left|
        \operatorname{Tr}(
        (\tilde A_{12}^{\dagger}\tilde A_{12}\otimes I_3)
        (I_1 \otimes \tilde B_{23}\tilde B_{23}^{\dagger})
        )
        \right|^{1/2}\\
        &=
        \left|\operatorname{Tr}(
        \operatorname{Tr}_1(\tilde A_{12}^{\dagger}\tilde A_{12})
        \operatorname{Tr}_3(\tilde B_{23}\tilde B_{23}^{\dagger})
        )
        \right|^{1/2}\\
        &\leq
        \left|\operatorname{Tr}(\tilde A_{12}^{\dagger}\tilde A_{12})
        \right|^{1/2}
        \left|\operatorname{Tr}(\tilde B_{23}\tilde B_{23}^{\dagger})
        \right|^{1/2}\\
        &=
        \|\tilde A_{12}\|_2 \|\tilde B_{23}\|_2
        =
        \|A_{12}^{T_S}\|_2 \|B_{23}^{T_S}\|_2\\
        &=
        \|A_{12}\|_2 \|B_{23}\|_2,
    \end{align}
    where $\tilde A_{12}=A_{12}^{T_S}, \tilde B_{23}=B_{23}^{T_S}$. In the third line from the end, we use $\operatorname{Tr}(AB)\leq\operatorname{Tr}(A)\!\operatorname{Tr}(B)$ for positive semidefinite operators $A$ and $B$.
\end{proof}

\begin{corollary}[Spectral norm bound]
\label{app:cor:ordered_product_bound}
    For $\boldsymbol u, \boldsymbol v\in\{0,1\}^m$ and
    $\boldsymbol\kappa\in\mathcal K_{\boldsymbol u,\boldsymbol v}$,
    \begin{equation}
        \left\|
            \operatorname{Ord}_{\boldsymbol\kappa}(O_{\boldsymbol u},O_{\boldsymbol v})
        \right\|_\infty
        \leq
        F(\boldsymbol u,\boldsymbol v,\boldsymbol\kappa)
        \lVert O_{\boldsymbol u}\rVert_2
        \lVert O_{\boldsymbol v}\rVert_2,
        \label{app:eq:ordered_product_bound}
    \end{equation}
    where
    \begin{align}
        F(\boldsymbol u,\boldsymbol v,\boldsymbol\kappa)
        &:=
        2^{-k\left(
            |J_{00}|
            +\frac{|J_{10}|+|J_{01}|}{2}
            +|J_{11, \boldsymbol{\kappa}}^{(=)}|
        \right)}
        \label{app:eq:ordered_product_dimension_factor}\\
        &=
        \prod_{i=1}^m
        f(u_i,v_i,\kappa_i).
    \end{align}
    The local factor $f(u_i, v_i, \kappa_i)$ is defined by
    \begin{equation}
        f(u_i, v_i, \kappa_i)
        :=
        \begin{cases}
            2^{-k},
            &(u_i, v_i, \kappa_i)=(0,0,0),\\[2pt]
            2^{-k/2},
            &(u_i, v_i, \kappa_i)=(0,1,0),\\[2pt]
            2^{-k/2},
            &(u_i, v_i, \kappa_i)=(1,0,0),\\[2pt]
            2^{-k},
            &(u_i, v_i, \kappa_i)=(1,1,{=}),\\[2pt]
            1,
            &(u_i, v_i, \kappa_i)=(1,1,\mathtt{fwd}),\\[2pt]
            1,
            &(u_i, v_i, \kappa_i)=(1,1,\mathtt{bwd}).
        \end{cases}
        \label{app:eq:local_dimension_factor}
    \end{equation}
\end{corollary}

\begin{proof}
    For fixed visibility type vector $\boldsymbol u$, $\boldsymbol v$, and $\boldsymbol\kappa$, let $I_{ab}$ denote
    the identity operator on the subsystem formed by the $k$-qubit
    blocks in $J_{ab}$. In particular,
    $\lVert I_{ab}\rVert_2^2=2^{k|J_{ab}|}$. We expand
    \begin{align}
        O_{\boldsymbol u}
        &=
        I_{00}\otimes I_{01}\otimes
        \sum_{
            P\in \{I,X,Y,Z\}^{k|J_{11, \boldsymbol{\kappa}}^{(=)}|}
        }
        P\otimes A_P,
        \\
        O_{\boldsymbol v}
        &=
        I_{00}\otimes I_{10}\otimes
        \sum_{
            Q\in\{I,X,Y,Z\}^{k|J_{11, \boldsymbol{\kappa}}^{(=)}|}
        }
        Q\otimes B_Q,
    \end{align}
    where $P$ and $Q$ act on the subsystem formed by the blocks in
    $J_{11, \boldsymbol{\kappa}}^{(=)}$, and $A_P$ and $B_Q$ act on the
    remaining subsystem. Orthogonality of the Pauli basis gives
    \begin{align}
        \lVert O_{\boldsymbol u}\rVert_2^2
        &=
        2^{
            k\left(
                |J_{00}|+|J_{01}|
                +|J_{11, \boldsymbol{\kappa}}^{(=)}|
            \right)
        }
        \sum_P\lVert A_P\rVert_2^2,
        \\
        \lVert O_{\boldsymbol v}\rVert_2^2
        &=
        2^{
            k\left(
                |J_{00}|+|J_{10}|
                +|J_{11, \boldsymbol{\kappa}}^{(=)}|
            \right)
        }
        \sum_Q\lVert B_Q\rVert_2^2.
    \end{align}
    Then
    \begin{align}
        \operatorname{Ord}_{\kappa}(O_{\boldsymbol u}, O_{\boldsymbol v})
        &=
        I_{00} \otimes \sum_{P,Q} PQ \cdot\mathbf{1}\{P=Q\} \otimes ((I_{01}\otimes A_P )^{T_S}(I_{10}\otimes B_Q)^{T_S})^{T_S}\\
        &=
        I_{00} \otimes I_{J_{11, \boldsymbol{\kappa}}^{(=)}}\otimes \sum_{P}((I_{01}\otimes A_P )^{T_S}(I_{10}\otimes B_P)^{T_S})^{T_S},
    \end{align}
    where $S=J_{11}^{(\texttt{bwd})}(\boldsymbol\kappa)$. Thus, operator norm is bounded as follows:
    \begin{align}
        \left\|
        \operatorname{Ord}_{\kappa}(O_{\boldsymbol u}, O_{\boldsymbol v})
        \right\|_{\infty}
        &=
        \left\|
        \sum_{P}((I_{01}\otimes A_P )^{T_S}(I_{10}\otimes B_P)^{T_S})^{T_S}
        \right\|_{\infty}\\
        &\le
        \sum_{P}
        \left\|
        ((I_{01}\otimes A_P )^{T_S}(I_{10}\otimes B_P)^{T_S})^{T_S}
        \right\|_{\infty}\\
        &\le
        \sum_{P}\|A_P\|_2 \|B_P\|_2\\
        &\le
        \left(\sum_P \|A_P\|_2^2\right)^{1/2} \left(\sum_P \|B_P\|_2^2\right)^{1/2}\\
        &=
        2^{-k(
            |J_{00}|
            +\frac{|J_{10}|+|J_{01}|}{2}
            +|J_{11, \boldsymbol{\kappa}}^{(=)}|)}\|O_{\boldsymbol u}\|_2\|O_{\boldsymbol v}\|_2,
    \end{align}
    as desired.
    
\end{proof}
Combining Corollaries~\ref{app:cor:kappa_ord} and
\ref{app:cor:ordered_product_bound}, we obtain
\begin{align}
    \lVert G_{\boldsymbol u \boldsymbol v}(O_0)\rVert_\infty
    &\leq
    \sum_{\boldsymbol\kappa\in\mathcal K_{\boldsymbol u,\boldsymbol v}}
    \left|
        \operatorname{Tr}\left(
            \mathbb R_{u_1v_1}^{(\kappa_1)}
            \cdots
            \mathbb R_{u_mv_m}^{(\kappa_m)}
        \right)
    \right|
    \left\|
        \operatorname{Ord}_{\boldsymbol\kappa}(O_{\boldsymbol u},O_{\boldsymbol v})
    \right\|_\infty\\
    &\leq
    \lVert O_{\boldsymbol u}\rVert_2\lVert O_{\boldsymbol v}\rVert_2
    \sum_{\boldsymbol\kappa\in\mathcal K_{\boldsymbol u,\boldsymbol v}}
    F(\boldsymbol u,\boldsymbol v,\boldsymbol\kappa)
    \left|
        \operatorname{Tr}\left(
            \mathbb R_{u_1v_1}^{(\kappa_1)}
            \cdots
            \mathbb R_{u_mv_m}^{(\kappa_m)}
        \right)
    \right|.
    \label{app:eq:packet_trace_reduction}
\end{align}

In the ordered basis $\mathcal S = (\texttt{00},\texttt{10},\texttt{01}, \texttt{E},\texttt{C},\texttt{A})$, let $H$ act as the identity on the first four bond states and define
\begin{equation}
    H\big|_{\{\texttt{C},\texttt{A}\}}
    =
    \frac12
    \begin{pmatrix}
        1&1\\
        1&-1
    \end{pmatrix},
    \qquad
    H^{-1}\big|_{\{\texttt{C},\texttt{A}\}}
    =
    \begin{pmatrix}
        1&1\\
        1&-1
    \end{pmatrix}.
    \label{app:eq:H}
\end{equation}
Trace invariance under the common similarity transformation gives
\begin{align}
    \sum_{\boldsymbol\kappa\in\mathcal K_{\boldsymbol u,\boldsymbol v}}
    F(\boldsymbol u,\boldsymbol v,\boldsymbol\kappa)
    \left|
        \operatorname{Tr}\left(
            \mathbb R_{u_1v_1}^{(\kappa_1)}
            \cdots
            \mathbb R_{u_mv_m}^{(\kappa_m)}
        \right)
    \right|
    &=
    \sum_{\boldsymbol\kappa\in\mathcal K_{\boldsymbol u,\boldsymbol v}}
    \left|
        \operatorname{Tr}\left(
            \prod_{i=1}^m
            f(u_i,v_i,\kappa_i)
            H
            \mathbb R_{u_iv_i}^{(\kappa_i)}
            H^{-1}
        \right)
    \right|\\
    &\leq
    \sum_{\boldsymbol\kappa\in\mathcal K_{\boldsymbol u,\boldsymbol v}}
    \operatorname{Tr}\left(
        \prod_{i=1}^m
        f(u_i,v_i,\kappa_i)
        \left|
            H
            \mathbb R_{u_iv_i}^{(\kappa_i)}
            H^{-1}
        \right|_{\mathrm{ew}}
    \right)\\
    &=
    \operatorname{Tr}\left(
        \prod_{i=1}^m \sum_{\kappa_i \in \mathcal{K}_{\boldsymbol u,\boldsymbol v}^{(i)}}
        f(u_i,v_i,\kappa_i)
        \left|
            H\mathbb R_{u_iv_i}^{(\kappa_i)}H^{-1}
        \right|_{\mathrm{ew}}
    \right)\\
    &=
    \operatorname{Tr}\left(
        \prod_{i=1}^m \widetilde{\mathbb B}_{u_i,v_i}
    \right)
    ,
    \label{app:eq:weighted_packet_transfer_bound}
\end{align}
where $|\cdot|_{\mathrm{ew}}$ denotes entrywise absolute value $(|A|_{\mathrm{ew}})_{ij}=|(A)_{ij}|$ and
\begin{equation}\label{appx:eq:t_B}
    \widetilde{\mathbb B}_{u_iv_i}
    := 
    \sum_{\kappa_i \in \mathcal{K}_{\boldsymbol u,\boldsymbol v}^{(i)}}
    f(u_i, v_i, \kappa_i)
    \left|
        H\mathbb R_{u_iv_i}^{(\kappa_i)}H^{-1}
    \right|_{\mathrm{ew}}.
\end{equation}
Applying the $H$ conjugation before taking the entrywise absolute value removes an exponentially growing factor in our bound (see Appendix~\ref{app:sec:role_H} for details).
Substituting Eq.~\eqref{app:eq:weighted_packet_transfer_bound} into Eq.~\eqref{app:eq:packet_trace_reduction} gives
\begin{align}
    \|\Gamma(O_0)\|_{\infty} 
    &=
    \left\|
    \sum_{\boldsymbol u,\boldsymbol v}\frac{G_{\boldsymbol u \boldsymbol v}(O_0)}{m(\boldsymbol u) m(\boldsymbol v)}
    \right\|_{\infty}\le
    \sum_{\boldsymbol u,\boldsymbol v}\frac{\|G_{\boldsymbol u \boldsymbol v}(O_0)\|_{\infty}}{m(\boldsymbol u) m(\boldsymbol v)}\\
    &\le
    \sum_{\boldsymbol u,\boldsymbol v}\frac{\|O_{\boldsymbol u}\|_{2}\|O_{\boldsymbol v}\|_{2}}{m(\boldsymbol u) m(\boldsymbol v)}
    \operatorname{Tr}
    \left(
    \widetilde{\mathbb B}_{u_1, v_1}
    \dots
    \widetilde{\mathbb B}_{u_m, v_m}
    \right)\\
    &=
    \sum_{\boldsymbol u,\boldsymbol v} (\widetilde K)_{\boldsymbol u \boldsymbol v} \|O_{\boldsymbol u}\|_{2}\|O_{\boldsymbol v}\|_{2},\label{app:eq:gamma_quadratic}
\end{align}
where $\widetilde K$ is defined by the $2^m \times 2^m$ matrix with
\begin{equation}
    (\widetilde K)_{\boldsymbol u \boldsymbol v}=\frac{\operatorname{Tr}
    \left(
    \widetilde{\mathbb B}_{u_1, v_1}
    \dots
    \widetilde{\mathbb B}_{u_m, v_m}
    \right)}{m(\boldsymbol u)m(\boldsymbol v)}.
\end{equation}

\subsection{Entrywise bounds for the matrices $\widetilde{\mathbb B}_{ab}$}

The exact matrices $\widetilde{\mathbb B}_{ab}$ are nonnegative $6\times6$ matrices whose entries are rational functions of $D$ and $\sqrt D$, where $D=2^{k}$. Directly retaining these expressions in a product of $m$ matrices would make the subsequent estimates unnecessarily complicated. We instead keep an upper bound on the leading coefficient and the leading power of $D$ in each entry. In the ordered basis
\begin{equation}
    \mathcal S
    =
    (\texttt{00},\texttt{10},\texttt{01},
     \texttt{E},\texttt{C},\texttt{A}),
\end{equation}
define
\begin{equation}
    \mathbb{B}_{00}
    :=
    \begin{pmatrix}
        D^{-1} & D^{-2} & D^{-2} & D^{-2} & 2D^{-3} & 2D^{-3}\\
        0&0&0&0&0&0\\
        0&0&0&0&0&0\\
        0&0&0&0&0&0\\
        0&0&0&0&0&0\\
        0&0&0&0&0&0
    \end{pmatrix}.
    \label{app:eq:B00_entrywise_envelope}
\end{equation}
For the two mismatch matrices, define
\begin{align}
    \mathbb{B}_{10}
    &:=
    D^{-1/2}
    \begin{pmatrix}
        D^{-2} & D^{-2} & 2D^{-3} & 2D^{-3} & 2D^{-3} & 2D^{-3}\\
        2D^{-1} & D^{-1} & 3D^{-2} & 3D^{-2} & 2D^{-2} & 2D^{-2}\\
        0&0&0&0&0&0\\
        0&0&0&0&0&0\\
        0&0&0&0&0&0\\
        0&0&0&0&0&0
    \end{pmatrix},
    \label{app:eq:B10_entrywise_envelope}\\
    \mathbb{B}_{01}
    &:=
    D^{-1/2}
    \begin{pmatrix}
        D^{-2} & 2D^{-3} & D^{-2} & 2D^{-3} & 2D^{-3} & 2D^{-3}\\
        0&0&0&0&0&0\\
        2D^{-1} & 3D^{-2} & D^{-1} & 3D^{-2} & 2D^{-2} & 2D^{-2}\\
        0&0&0&0&0&0\\
        0&0&0&0&0&0\\
        0&0&0&0&0&0
    \end{pmatrix}.
    \label{app:eq:B01_entrywise_envelope}
\end{align}
These matrices retain the structures needed below:
\[
    \operatorname{rank}
    \bigl(\mathbb{B}_{00}\bigr)=1,
    \qquad
    \operatorname{rank}
    \bigl(\mathbb{B}_{01}\bigr)
    =
    \operatorname{rank}
    \bigl(\mathbb{B}_{10}\bigr)
    =2.
\]
For the active--active matrix, define
\begin{equation}
    \mathbb{B}_{11}
    :=
    \begin{pmatrix}
        D^{-3} & 4D^{-4} & 4D^{-4} & D^{-3} & 4D^{-4} & 4D^{-4}\\
        6D^{-3} & 2D^{-3} & 4D^{-3} & 4D^{-3} & 2D^{-3} & 2D^{-3}\\
        6D^{-3} & 4D^{-3} & 2D^{-3} & 4D^{-3} & 2D^{-3} & 2D^{-3}\\
        2D^{-2} & 7D^{-3} & 7D^{-3} & D^{-2} & 4D^{-3} & 4D^{-3}\\
        \frac52D^{-2} & D^{-2} & D^{-2} & D^{-2} & D^{-2} & 3D^{-3}\\
        \frac52D^{-2} & D^{-2} & D^{-2} & D^{-2} & 3D^{-3} & D^{-2}
    \end{pmatrix}.
    \label{app:eq:B11_entrywise_envelope}
\end{equation}

\begin{lemma}[Entrywise bounds for $\widetilde{\mathbb B}_{ab}$]
\label{app:lem:B_entrywise_bounds}
For $D=2^k\geq256$,
\begin{align}
    &\widetilde{\mathbb B}_{00}
    \leq_{\mathrm{ew}}
    \mathbb{B}_{00},
    \qquad
    \widetilde{\mathbb B}_{01}
    \leq_{\mathrm{ew}}
    \mathbb{B}_{01},
    \nonumber\\
    &\widetilde{\mathbb B}_{10}
    \leq_{\mathrm{ew}}
    \mathbb{B}_{10},
    \qquad
    \widetilde{\mathbb B}_{11}
    \leq_{\mathrm{ew}}
    \left(1+\frac{6}{D}\right)
    \mathbb{B}_{11}.
    \label{app:eq:B_entrywise_bounds}
\end{align}
\end{lemma}

\begin{proof}
    The bounds for $\widetilde{\mathbb B}_{00}$, $\widetilde{\mathbb B}_{01}$, and
    $\widetilde{\mathbb B}_{10}$ follow directly from their explicit entries by using
    \begin{equation}
        D-1\leq D,
        \qquad
        D+1\geq D,
        \qquad
        D+2\geq D.
    \end{equation}
    For each entry of $\widetilde{\mathbb B}_{11}$, we can write
    \begin{equation}
        \left(\widetilde{\mathbb B}_{11}\right)_{rs}
        =
        \frac{
            A_0D^\alpha+\sum_{i\geq1}A_iD^{\alpha-i}
        }{
            B_0D^\beta+\sum_{j\geq1}B_jD^{\beta-j}
        },
        \qquad
        \left(\mathbb B_{11}\right)_{rs}
        =
        \frac{A_0}{B_0}D^{\alpha-\beta},
    \end{equation}
    where $A_0,B_0>0$. For $D\geq256$, define
    \begin{align}
        A_{\mathrm{rest}}^{(r,s)}
        :=
        \sum_{i\geq1}
        \frac{A_i}{A_0}\,256^{1-i}\mathbf{1}\{A_i>0\},
        \quad
        B_{\mathrm{rest}}^{(r,s)}
        :=
        \sum_{j\geq1}
        \frac{|B_j|}{B_0}\,256^{1-j}\mathbf{1}\{B_j<0\}.
    \end{align}
    Keeping only the positive numerator corrections and
    negative denominator corrections gives
    \begin{align}
        \frac{
            \left(\widetilde{\mathbb B}_{11}\right)_{rs}
        }{
            \left(\mathbb B_{11}\right)_{rs}
        }
        \leq
        \frac{
            1+A_{\mathrm{rest}}^{(r,s)}/D
        }{
            1-B_{\mathrm{rest}}^{(r,s)}/D
        }
        \leq
        1+\frac{6}{D}
    \end{align}
    with the maximum attained at $(r,s)=(4,1)$. Therefore,
    \begin{equation}
        \widetilde{\mathbb B}_{11}
        \leq_{\mathrm{ew}}
        \left(1+\frac{6}{D}\right)\mathbb B_{11},
        \qquad D\geq256.
    \end{equation}
\end{proof}

Lemma~\ref{app:lem:B_entrywise_bounds} and
Eq.~\eqref{app:eq:elementwise_product_bound} imply
\begin{align}
    \operatorname{Tr}\!\left(
        \widetilde{\mathbb B}_{u_1v_1}\cdots
        \widetilde{\mathbb B}_{u_mv_m}
    \right)
    &\leq
    \left(1+\frac{6}{D}\right)^{|J_{11}|}
    \operatorname{Tr}\!\left(
        \mathbb{B}_{u_1v_1}\cdots
        \mathbb{B}_{u_mv_m}
    \right)
    \nonumber\\
    &\leq
    e^{6m/D}
    \operatorname{Tr}\!\left(
        \mathbb{B}_{u_1v_1}\cdots
        \mathbb{B}_{u_mv_m}
    \right).
    \label{app:eq:B_product_entrywise_bound}
\end{align}
In the last line, we use $|J_{11}|\le m$ and $(1+6/D)\le e^{6/D}$. The factor $e^{6m/D}$ will be included in the final bound on $(\widetilde K)_{\boldsymbol u \boldsymbol v}$.

\subsection{Lower bounds for the visibility $m(u)$}

\begin{lemma}[Lower bound for a nonempty run]\label{app:lem:visibility_run_lower_bound}
    For every $\ell\geq1$,
    \begin{equation}
        \alpha_\ell\geq\alpha_1\lambda_+^{\ell-1}.
        \label{app:eq:visibility_run_lower_bound}
    \end{equation}
\end{lemma}

\begin{proof}
    Recall from Fact~\ref{app:fact:alpha} that
    \begin{equation}
        \alpha_\ell
        =
        \alpha_+\lambda_+^\ell
        +
        \alpha_-\lambda_-^\ell,
        \label{app:eq:alpha_spectral_representation}
    \end{equation}
    where
    $\alpha_+=(\alpha_1-\lambda_-)/(\lambda_+-\lambda_-)$,
    $\alpha_-=(\lambda_+-\alpha_1)/(\lambda_+-\lambda_-)$, and
    \begin{equation}
        \lambda_\pm
        =
        \frac{
            D+1\pm\sqrt{D^2+6D+1}
        }{
            2(D+1)^2
        }.
        \label{app:eq:visibility_eigenvalues}
    \end{equation}
    Since $\lambda_-<0<\lambda_+$ and
    $0<\alpha_1=(3D+1)/(D+1)^3\leq 1/(D+1)\leq\lambda_+$, both 
    \begin{equation}
        \alpha_+>0, \quad \alpha_->0.
    \end{equation}
    Moreover,
    $\lambda_++\lambda_-=1/(D+1)>0$, and therefore
    $\lambda_+>|\lambda_-|$.  Hence
    $0<|\lambda_-|/\lambda_+<1$. We rewrite
    Eq.~\eqref{app:eq:alpha_spectral_representation} as
    \begin{equation}
        \alpha_\ell
        =
        \alpha_{+}\lambda_{+}^{\ell}+(-1)^{\ell}\alpha_-|\lambda_-|^{\ell}
        =
        \lambda_+^\ell
        \left[
            \alpha_+
            +
            (-1)^\ell\alpha_-
            \left(
                \frac{|\lambda_-|}{\lambda_+}
            \right)^\ell
        \right].
        \label{app:eq:alpha_parity_representation}
    \end{equation}
    Taking $\ell=1$ gives
    \begin{equation}
        \frac{\alpha_1}{\lambda_+}
        =
        \alpha_+
        -
        \alpha_-
        \frac{|\lambda_-|}{\lambda_+}.
    \end{equation}
    Suppose first that $\ell$ is odd.  Since
    $0<\frac{|\lambda_{-}|}{\lambda_{+}}<1$,
    the nonnegativity of $\alpha_-$ then gives
    \begin{equation}
        \frac{\alpha_{\ell}}{\lambda_{+}^{\ell}}
        =
        \alpha_+
        -
        \alpha_-
        \left(
            \frac{|\lambda_-|}{\lambda_+}
        \right)^\ell
        \geq
        \alpha_+
        -
        \alpha_-
        \frac{|\lambda_-|}{\lambda_+}
        =
        \frac{\alpha_1}{\lambda_+}.
    \end{equation}
    Suppose next that $\ell$ is even.  Since $\alpha_-\geq0$,
    \begin{equation}
        \frac{\alpha_{\ell}}{\lambda_{+}^{\ell}}
        =
        \alpha_+
        +
        \alpha_-
        \left(
            \frac{|\lambda_-|}{\lambda_+}
        \right)^\ell
        \geq
        \alpha_+
        \geq
        \alpha_+
        -
        \alpha_-
        \frac{|\lambda_-|}{\lambda_+}
        =
        \frac{\alpha_1}{\lambda_+}.
    \end{equation}
    Thus, for $\ell\ge1$, $\alpha_\ell\geq\alpha_1\lambda_+^{\ell-1}$.
\end{proof}


\begin{lemma}[Visibility estimates]\label{app:lem:visibility_estimates}
    Suppose that $D\geq256$.  Then
        \begin{align}
            D^2\alpha_1
            &\geq
            \frac{5}{2},
            \label{app:eq:alpha_one_lower_bound}\\
            \frac{1}{D+1}
            &\leq
            \lambda_+
            \leq
            \frac{1}{D},
            \label{app:eq:mu_plus_bounds}\\
            m(1^m)
            &\geq
            \left(1-\frac{1}{D}\right)\lambda_+^m.
            \label{app:eq:all_one_visibility_lower_bound}
        \end{align}
\end{lemma}

\begin{proof}
    Using $\alpha_1=(3D+1)/(D+1)^3$ and $D\ge 256$, 
    we can verify Eq.~\eqref{app:eq:alpha_one_lower_bound}. Next, the inequalities
    \begin{equation}
        D+1
        \leq
        \sqrt{D^2+6D+1}
        \leq
        D+3
    \end{equation}
    give
    \begin{equation}
        \frac{1}{D+1}
        \leq
        \lambda_+
        \leq
        \frac{D+2}{(D+1)^2}
        \leq
        \frac{1}{D}
        .
    \end{equation}
    This proves Eq.~\eqref{app:eq:mu_plus_bounds}. For the support vector $u=1^m$, $m(1^m)=\operatorname{Tr}(\mathbb K_1^m)=\lambda_+^m+\lambda_-^m$.
    Moreover,
    \begin{equation}
        \frac{|\lambda_-|}{\lambda_+}
        =
        \frac{
            4D
        }{
            \left(
                \sqrt{D^2+6D+1}+D+1
            \right)^2
        }
        <
        \frac1D.
    \end{equation}
    Therefore,
    \begin{align}
        m(1^m)
        &\geq
        \lambda_+^m-|\lambda_-|^m\\
        &=
        \lambda_+^m
        \left[
            1-
            \left(
                \frac{|\lambda_-|}{\lambda_+}
            \right)^m
        \right]\\
        &\geq
        \left(1-\frac1D\right)\lambda_+^m,
    \end{align}
    which proves Eq.~\eqref{app:eq:all_one_visibility_lower_bound}.
\end{proof}

\subsection{Bounding the entries $(K)_{uv}$}
We now bound each entry 
\begin{equation}
    (K)_{\boldsymbol u \boldsymbol v}:=\frac{\operatorname{Tr}
    \left(
    {\mathbb B}_{u_1v_1}
    \dots
    {\mathbb B}_{u_mv_m}
    \right)}{m(\boldsymbol u)m(\boldsymbol v)}
\end{equation}
by estimating the transfer matrix contraction in the numerator together with the two visibility factors in the
denominator. At every position $i\in J_{00}$, we have
\begin{equation}
    \mathbb B_{u_iv_i}
    =
    \mathbb B_{00},
    \qquad
    \mathbb K_{u_i}
    =
    \mathbb K_{v_i}
    =
    \mathbb K_0.
\end{equation}
The explicit matrix in
Eq.~\eqref{app:eq:B00_entrywise_envelope} has only one nonzero row, so
$\mathbb B_{00}$ has rank one.
Equation~\eqref{app:eq:visibility_rank_one_reset} gives the rank-one
form of $\mathbb K_0$. These rank-one structures factorize the numerator
and the two visibility contractions at the same positions. We therefore
divide the analysis according to whether $J_{00}$ is empty.

Suppose first that $J_{00}\neq\emptyset$ and that the full matrix product is not $\mathbb{B}_{00}^m$. After a cyclic
permutation, it can be written as
\begin{equation}
    \operatorname{Tr} 
    \left(
    \prod_{i=1}^m \mathbb B_{u_iv_i}
    \right)
    =
    \operatorname{Tr} 
    \left(
    \mathbb{B}_{00}^{t_1}W_1
    \mathbb{B}_{00}^{t_2}W_2
    \cdots
    \mathbb{B}_{00}^{t_{n_W}}W_{n_W}
    \right),
    \qquad
    t_a\geq1,
    \label{app:eq:common_zero_component_decomposition}
\end{equation}
where each $W_a$ is a nonempty matrix product containing no
$\mathbb B_{00}$. Here, $n_W$ is the number of maximal nonempty subproducts separated by the $\mathbb B_{00}$ runs, and $W_a$ denotes the $a$th such subproduct.
Consequently, every factor in $W_a$ belongs to
$\{\mathbb B_{11},\mathbb B_{10},\mathbb B_{01}\}$.
We refer to $\mathbb B_{10}$ and $\mathbb B_{01}$ as mismatch matrices, so each $W_a$ either consists only of $\mathbb B_{11}$ or contains at least one mismatch matrix.

Every $W$ containing a mismatch has a unique maximal-run decomposition of the form
\begin{equation}
    W
    =
    \mathbb{B}_{11}^{r_0}
    \mathbb{B}_{\sigma_1}^{\ell_1}
    \mathbb{B}_{11}^{r_1}
    \cdots
    \mathbb{B}_{\sigma_z}^{\ell_z}
    \mathbb{B}_{11}^{r_z}.
    \label{app:eq:mismatch_component_decomposition}
\end{equation}
Here, $z\geq1$, $\sigma_j\in\{10,01\}$ and $\ell_j\geq1$ for
$1\leq j\leq z$, while $r_j\geq0$ for $0\leq j\leq z$.
The exponents $r_0$ and $r_z$ count the $\mathbb B_{11}$ factors at the
two boundaries of $W$ and may vanish without any restriction on the
adjacent mismatch orientation. For every internal index $1\leq j\leq z-1$, the condition $r_j=0$ requires $\sigma_j\neq\sigma_{j+1}$; otherwise, the two adjacent powers can be combined, contradicting the maximality of the decomposition.

We write
\begin{equation}
    r:=\sum_{j=0}^z r_j,
    \qquad
    h:=\sum_{j=1}^z\ell_j,
    \qquad
    L:=r+h.
    \label{app:eq:component_counts}
\end{equation}
Each mismatch matrix $\mathbb B_{10}$ and $\mathbb B_{01}$ has exactly
two nonzero rows and therefore admits the rank-two factorization
\begin{equation}
    \mathbb B_\sigma
    =
    D^{-1/2}E_\sigma R_\sigma,
    \label{app:eq:mismatch_rank_two_factorization}
\end{equation}
where 
\begin{equation}
    E_{10}
    =
    \begin{pmatrix}
        1&0\\
        0&1\\
        0&0\\
        0&0\\
        0&0\\
        0&0
    \end{pmatrix},
    \qquad
    E_{01}
    =
    \begin{pmatrix}
        1&0\\
        0&0\\
        0&1\\
        0&0\\
        0&0\\
        0&0
    \end{pmatrix}.
    \label{app:eq:mismatch_embeddings}
\end{equation}
For $\sigma\in\{10,01\}$, let $R_\sigma$ be the $2\times6$ matrix consisting of the two nonzero rows of $D^{1/2}\mathbb B_\sigma$ [Eqs.~\eqref{app:eq:B10_entrywise_envelope} and~\eqref{app:eq:B01_entrywise_envelope}]. Define the internal mismatch transfer by
\begin{equation}
    J_\sigma
    :=
    R_\sigma E_\sigma.
    \label{app:eq:J_sigma_definition}
\end{equation}
Both orientations give the same matrix:
\begin{equation}
    J
    =
    J_{10}
    =
    J_{01}
    =
    \begin{pmatrix}
        D^{-2}&D^{-2}\\
        2D^{-1}&D^{-1}
    \end{pmatrix}.
    \label{app:eq:J_sigma_explicit}
\end{equation}
Consequently, for every $\ell\geq1$,
\begin{equation}
    \mathbb B_\sigma^\ell
    =
    D^{-\ell/2}
    E_\sigma
    J_\sigma^{\ell-1}
    R_\sigma.
    \label{app:eq:mismatch_run_factorization}
\end{equation}
Applying this factorization to  Eq.~\eqref{app:eq:mismatch_component_decomposition} gives
\begin{align}
    W
    &=
    \mathbb B_{11}^{r_0} \mathbb B_{\sigma_1}^{\ell_1}  \mathbb B_{11}^{r_1}\dots\mathbb B_{\sigma_z}^{\ell_z}\mathbb B_{11}^{r_z}\\
    &=
    D^{-h/2}
    \mathbb B_{11}^{r_0}
    (E_{\sigma_1}
    J^{\ell_1-1}
    R_{\sigma_1})
    \mathbb B_{11}^{r_1}
    \cdots
    (E_{\sigma_z}
    J^{\ell_z-1}
    R_{\sigma_z})
    \mathbb B_{11}^{r_z}.
    \label{app:eq:w_decomp}
\end{align}
This representation allows both mismatch orientations to be analyzed with the same propagation estimate. The powers within each mismatch run are governed by the two-dimensional matrix $J$, and the orientation $\sigma\in\{10,01\}$ appears only in $E_\sigma$ and $R_\sigma$. Equation~\eqref{app:eq:mismatch_component_decomposition} covers every component containing at least one mismatch matrix. The two remaining forms are $W=\mathbb B_{11}^{L}$, when no mismatch matrix is present, and the full product $\mathbb B_{00}^{m}$, which contains no component $W_a$. These forms are treated separately.

We divide the proof according to whether $J_{00}$ is empty. In Case 1, $J_{00}\neq\emptyset$, and the rank-one matrices at the common-zero positions reduce the numerator and both visibility factors to open components $W_a$. In Case 2, $J_{00}=\emptyset$, and the numerator remains a single cyclic trace. We first establish the decompositions and estimates used in the case analysis.
Write the rank-one matrix $\mathbb B_{00}$ as
\begin{equation}
    \mathbb B_{00}
    =
    e_0b^T,
    \qquad
    e_0^T=(1,0,0,0,0,0),
    \qquad
    b^T=
    \left(
        D^{-1},
        D^{-2},
        D^{-2},
        D^{-2},
        2D^{-3},
        2D^{-3}
    \right).
\end{equation}
Since $b^Te_0=D^{-1}$, we have $\mathbb B_{00}^t=D^{-t+1}\mathbb B_{00}$ for $t\geq1$ and $\mathbb B_{00}X\mathbb B_{00}=(b^TXe_0)\mathbb B_{00}$. Applying these identities to Eq.~\eqref{app:eq:common_zero_component_decomposition}, and using $|J_{00}|=\sum_{a=1}^{n_W}t_a$, gives
\begin{equation}
    \operatorname{Tr}\!\left(
        \mathbb B_{u_1v_1}\cdots
        \mathbb B_{u_mv_m}
    \right)
    =
    D^{-|J_{00}|+n_W}
    \prod_{a=1}^{n_W}b^TW_ae_0.
    \label{app:eq:numerator_component_factorization}
\end{equation}

The same decomposition applies to the two visibility factors. For a component $W=W_a$ lying between two consecutive $\mathbb B_{00}$ runs in Eq.~\eqref{app:eq:common_zero_component_decomposition}, relabel its $L$ positions and write $W=\mathbb B_{u_1v_1}\cdots\mathbb B_{u_Lv_L}$. Define
\begin{equation}
    \nu_{\boldsymbol u}(W)
    :=
    \operatorname{Tr}\!\left(
        \mathbb K_0
        \mathbb K_{u_1}\cdots\mathbb K_{u_L}
    \right),
    \qquad
    \nu_{\boldsymbol v}(W)
    :=
    \operatorname{Tr}\!\left(
        \mathbb K_0
        \mathbb K_{v_1}\cdots\mathbb K_{v_L}
    \right).
    \label{app:eq:component_visibility_factors}
\end{equation}
The rank-one form of $\mathbb K_0$ implies $\mathbb K_0X\mathbb K_0 =\operatorname{Tr}(\mathbb K_0X)\mathbb K_0$. Together with $\mathbb K_0^{t_a}=\mathbb K_0$, this gives
\begin{align}
    m(\boldsymbol u)
    =
    \prod_{a=1}^{n_W}\nu_u(W_a),
    \quad
    m(\boldsymbol v)
    =
    \prod_{a=1}^{n_W}\nu_v(W_a).
    \label{app:eq:visibility_component_factorization}
\end{align}
Combining Eqs.~\eqref{app:eq:numerator_component_factorization} and~\eqref{app:eq:visibility_component_factorization}, we obtain
\begin{equation}
    (K)_{\boldsymbol u \boldsymbol v}
    =
    D^{-|J_{00}|+n_W}
    \prod_{a=1}^{n_W}
    \frac{
        b^TW_ae_0
    }{
        \nu_u(W_a)\nu_v(W_a)
    }.
    \label{app:eq:normalized_component_factorization}
\end{equation}
For example, $W=\mathbb B_{01}\mathbb B_{10}\mathbb B_{11}$ induces the support sequences $011$ and $101$, so
\begin{equation}
    W=\mathbb B_{01}\mathbb B_{10}\mathbb B_{11} \longmapsto \nu_{\boldsymbol u}(W)=\alpha_2,\quad \nu_{\boldsymbol v}(W)=\alpha_1^2.
\end{equation}

For a mismatch component with the decomposition in Eq.~\eqref{app:eq:mismatch_component_decomposition}, set
\begin{equation}
    I_{\partial}:=\mathbf 1_{\{r_0>0\}}+\mathbf 1_{\{r_z>0\}}.
\end{equation}

\begin{lemma}[Visibility of an open mismatch component]
\label{app:lem:open_mismatch_component_visibility}
    Every open mismatch component $W$ satisfies
    \begin{equation}
        \nu_{\boldsymbol u}(W)\nu_{\boldsymbol v}(W)
        \geq
        \left(
            \frac{\alpha_1}{\lambda_+}
        \right)^{z+I_{\partial}}
        \lambda_+^{2r+h}
        =
        \alpha_1^{z+I_{\partial}}
        \lambda_+^{2r+h-z-I_{\partial}}.
        \label{app:eq:open_mismatch_component_visibility}
    \end{equation}
\end{lemma}

\begin{proof}
    In the $u$- and $v$-sequences defining $\nu_{\boldsymbol u}(W)$ and $\nu_{\boldsymbol v}(W)$, each $\mathbb B_{11}$ contributes one $1$ to both sequences, and each mismatch matrix contributes one $1$ to exactly one sequence. Hence these sequences contain $2r+h$ ones in total. The $z$ maximal mismatch runs give $z$ nonempty runs of ones, while the two boundary active runs contribute $\mathbf 1_{\{r_0>0\}}+\mathbf 1_{\{r_z>0\}}=I_{\partial}$ additional runs. Thus, $\nu_{\boldsymbol u}(W)\nu_{\boldsymbol v}(W)$ contains $z+I_{\partial}$ factors $\alpha_\ell$, whose indices sum to $2r+h$. Applying Eq.~\eqref{app:eq:visibility_run_lower_bound} to each factor proves Eq.~\eqref{app:eq:open_mismatch_component_visibility}.
\end{proof}

We next collect the estimates used to bound the numerator factors $b^TWe_0$. Equation~\eqref{app:eq:w_decomp} gives
\begin{align}
    b^TWe_0
    =&
    D^{-h/2}
    \left(
        b^T\mathbb B_{11}^{r_0}E_{\sigma_1}
    \right)
    J^{\ell_1-1}
    \left(
        R_{\sigma_1}
        \mathbb B_{11}^{r_1}
        E_{\sigma_2}
    \right)
    J^{\ell_2-1}
    \cdots
    \nonumber\\
    &\quad\cdots
    \left(
        R_{\sigma_{z-1}}
        \mathbb B_{11}^{r_{z-1}}
        E_{\sigma_z}
    \right)
    J^{\ell_z-1}
    \left(
        R_{\sigma_z}
        \mathbb B_{11}^{r_z}e_0
    \right),
    \label{app:eq:bwe}
\end{align}
where the middle chain is absent when $z=1$. If an internal $r_j$ vanishes, the corresponding factor becomes $R_{\sigma_j}E_{\sigma_{j+1}}$, with $\sigma_j\neq\sigma_{j+1}$. All factors in Eq.~\eqref{app:eq:bwe} are entrywise nonnegative, so the estimates below may be applied successively from right to left. Define
\begin{equation}
    g
    :=
    \left(
        \frac{3}{D^{3/2}},
        \frac{4}{D},
        \frac{4}{D},
        \frac{3}{D^{1/2}},
        1,
        1
    \right)^T,
    \qquad
    p:=(1,D)^T,
    \label{app:eq:dimensionful_test_vectors}
\end{equation}
and
\begin{equation}
    \Lambda_D:=1+\frac{4}{\sqrt D},
    \qquad
    \kappa_D:=1+\frac{2}{D},
    \qquad
    \gamma_D:=\frac{\Lambda_D}{(D\lambda_+)^2}.
    \label{app:eq:propagation_constants}
\end{equation}

\begin{lemma}[Elementary transfer bounds]
\label{app:lem:elementary_transfer_bounds}
    Suppose that $D=2^k\geq256$. For
    $\sigma,\tau\in\{10,01\}$, the following entrywise bounds hold:
    \begin{align}
        \mathbb B_{11}g
        &\leq_{\mathrm{ew}}
        D^{-2}\Lambda_Dg,
        \label{app:eq:active_propagation_bound}\\
        Jp
        &\leq_{\mathrm{ew}}
        D^{-1}\kappa_Dp,
        \label{app:eq:mismatch_propagation_bound}\\
        \mathbb B_{11}E_\tau p
        &\leq_{\mathrm{ew}}
        2D^{-1}g,
        \label{app:eq:mismatch_to_active_bound}\\
        R_\sigma g
        &\leq_{\mathrm{ew}}
        10D^{-3}p,
        \label{app:eq:active_to_mismatch_bound}\\
        R_\sigma E_\tau p
        &\leq_{\mathrm{ew}}
        5D^{-2}p,
        \qquad \sigma\neq\tau,
        \label{app:eq:orientation_change_bound}\\
        b^Tg
        &\leq
        7D^{-5/2},
        \label{app:eq:left_active_boundary_bound}\\
        b^TE_\sigma p
        &=
        2D^{-1},
        \label{app:eq:left_mismatch_boundary_bound}\\
        R_\sigma e_0
        &\leq_{\mathrm{ew}}
        2D^{-2}p,
        \label{app:eq:right_mismatch_boundary_bound}\\
        \mathbb B_{11}e_0
        &\leq_{\mathrm{ew}}
        \frac{2}{3}D^{-3/2}g,
        \label{app:eq:right_active_boundary_bound}\\
        b^Te_0
        &=
        D^{-1}.
        \label{app:eq:reset_bound}
    \end{align}
\end{lemma}

\begin{proof}
    Direct multiplication gives
    \begin{align}
        D^2\operatorname{diag}(g)^{-1}\mathbb B_{11}g
        &=
        \begin{pmatrix}
            1+\frac{8}{3}D^{-1/2}+D^{-1}+\frac{32}{3}D^{-3/2}\\
            1+3D^{-1/2}+6D^{-1}+\frac{9}{2}D^{-3/2}\\
            1+3D^{-1/2}+6D^{-1}+\frac{9}{2}D^{-3/2}\\
            1+\frac{8}{3}D^{-1/2}+2D^{-1}+\frac{56}{3}D^{-3/2}\\
            1+3D^{-1/2}+11D^{-1}+\frac{15}{2}D^{-3/2}\\
            1+3D^{-1/2}+11D^{-1}+\frac{15}{2}D^{-3/2}
        \end{pmatrix}\\
        &\leq_{\mathrm{ew}}
        (1+4D^{-1/2})\mathbf 1_6,
    \end{align}
    This proves Eq.~\eqref{app:eq:active_propagation_bound}. The matrix $\mathbb B_{01}$ is obtained from $\mathbb B_{10}$ by exchanging the second and third rows and columns. It is therefore enough to compute
    \begin{align}
        Jp
        &=
        \begin{pmatrix}
            D^{-2} & D^{-2}\\
            2D^{-1} & D^{-1}
        \end{pmatrix}
        \begin{pmatrix}
            1\\
            D
        \end{pmatrix}
        =
        \begin{pmatrix}
            D^{-1}+D^{-2}\\
            1+2D^{-1}
        \end{pmatrix}
        \leq_{\mathrm{ew}}
        D^{-1}\kappa_Dp,
        \label{app:eq:elementary_Jp}\\
        \mathbb B_{11}E_{10}p
        &=
        \mathbb B_{11}
        \begin{pmatrix}
            1\\
            D\\
            0\\
            0\\
            0\\
            0
        \end{pmatrix}
        =
        \begin{pmatrix}
            5D^{-3}\\
            2D^{-2}+6D^{-3}\\
            4D^{-2}+6D^{-3}\\
            9D^{-2}\\
            D^{-1}+\frac52D^{-2}\\
            D^{-1}+\frac52D^{-2}
        \end{pmatrix}
        \leq_{\mathrm{ew}}
        2D^{-1}g,
        \label{app:eq:elementary_BEp}\\
        R_{10}g
        &=
        \begin{pmatrix}
            D^{-2} & D^{-2} & 2D^{-3} & 2D^{-3} & 2D^{-3} & 2D^{-3}\\
            2D^{-1} & D^{-1} & 3D^{-2} & 3D^{-2} & 2D^{-2} & 2D^{-2}
        \end{pmatrix}
        \begin{pmatrix}
            3D^{-3/2}\\
            4D^{-1}\\
            4D^{-1}\\
            3D^{-1/2}\\
            1\\
            1
        \end{pmatrix}
        \\
        &=
        \begin{pmatrix}
            8D^{-3}+9D^{-7/2}+8D^{-4}\\
            8D^{-2}+15D^{-5/2}+12D^{-3}
        \end{pmatrix}
        \leq_{\mathrm{ew}}
        10D^{-3}p,
        \label{app:eq:elementary_Rg}\\
        R_{10}E_{01}p
        &=
        \begin{pmatrix}
            D^{-2} & 2D^{-3}\\
            2D^{-1} & 3D^{-2}
        \end{pmatrix}
        \begin{pmatrix}
            1\\
            D
        \end{pmatrix}
        =
        \begin{pmatrix}
            3D^{-2}\\
            5D^{-1}
        \end{pmatrix}
        \leq_{\mathrm{ew}}
        5D^{-2}p.
        \label{app:eq:elementary_REp}
    \end{align}
    These inequalities follow directly from $D\geq256$. The boundary terms satisfy
    \begin{align}
        b^Tg
        &=
        6D^{-5/2}+12D^{-3}
        \leq
        7D^{-5/2},
        \label{app:eq:elementary_bg}\\
        b^TE_{10}p
        &=
        2D^{-1},
        \label{app:eq:elementary_bEp}\\
        R_{10}e_0
        &=
        \begin{pmatrix}
            D^{-2}\\
            2D^{-1}
        \end{pmatrix}
        \leq_{\mathrm{ew}}
        2D^{-2}p,
        \label{app:eq:elementary_Re}\\
        \mathbb B_{11}e_0
        &=
        \begin{pmatrix}
            D^{-3}\\
            6D^{-3}\\
            6D^{-3}\\
            2D^{-2}\\
            \frac52D^{-2}\\
            \frac52D^{-2}
               \end{pmatrix}
        \leq_{\mathrm{ew}}
        \frac23D^{-3/2}g,
        \label{app:eq:elementary_Be}\\
        b^Te_0
        &=
        D^{-1}.
        \label{app:eq:elementary_be}
    \end{align}
    The corresponding statements for $\mathbb B_{01}$ can be proved in the same way.
\end{proof}

The remaining estimates are used only at boundaries and at transitions between the six-dimensional transfer space and the two-dimensional mismatch space. For these estimates, only the powers of $D$ are relevant, as the fixed coefficients do not affect the asymptotic scaling.

\begin{corollary}[Connector and boundary bounds]
\label{app:cor:connector_and_boundary_bounds}
    For $r\geq0$, set $I_r:=\mathbf 1_{\{r>0\}}$. Then, for
    $\sigma,\tau\in\{10,01\}$,
    \begin{align}
        R_\sigma\mathbb B_{11}^re_0
        &\leq_{\mathrm{ew}}
        \frac{20}{3}
        D^{-2r-2-I_r/2}
        \Lambda_D^r p,
        \label{app:eq:right_endpoint_bound}\\
        R_\sigma\mathbb B_{11}^rE_\tau p
        &\leq_{\mathrm{ew}}
        20D^{-2r-2}\Lambda_D^r p,
        \label{app:eq:internal_connector_bound}\\
        b^T\mathbb B_{11}^rE_\sigma p
        &\leq
        14D^{-2r-1-I_r/2}\Lambda_D^r.
        \label{app:eq:left_endpoint_bound}
    \end{align}
    In Eq.~\eqref{app:eq:internal_connector_bound}, $r=0$ requires $\sigma\neq\tau$.
\end{corollary}

\begin{proof}
    Suppose first that $r=0$. Then $I_r=0$, and
    Eqs.~\eqref{app:eq:right_mismatch_boundary_bound},
    \eqref{app:eq:orientation_change_bound}, and
    \eqref{app:eq:left_mismatch_boundary_bound} give
    \begin{align}
        R_\sigma e_0
        &\leq_{\mathrm{ew}}
        2D^{-2}p
        \leq_{\mathrm{ew}}
        \frac{20}{3}D^{-2}p,\\
        R_\sigma E_\tau p
        &\leq_{\mathrm{ew}}
        5D^{-2}p
        \leq_{\mathrm{ew}}
        20D^{-2}p,
        \qquad \sigma\neq\tau,\\
        b^TE_\sigma p
        &=
        2D^{-1}
        \leq
        14D^{-1}.
    \end{align}

    Let $r\geq1$. Applying the corresponding boundary estimate once
    and propagating the remaining $r-1$ factors with
    Eq.~\eqref{app:eq:active_propagation_bound} gives
    \begin{align}
        R_\sigma\mathbb B_{11}^re_0
        &=
        R_\sigma\mathbb B_{11}^{r-1}
        \bigl(\mathbb B_{11}e_0\bigr)
        \nonumber\\
        &\leq_{\mathrm{ew}}
        \frac{20}{3}
        D^{-2r-5/2}
        \Lambda_D^{r-1}p
        \leq_{\mathrm{ew}}
        \frac{20}{3}
        D^{-2r-5/2}
        \Lambda_D^rp,\\
        R_\sigma\mathbb B_{11}^rE_\tau p
        &=
        R_\sigma\mathbb B_{11}^{r-1}
        \bigl(\mathbb B_{11}E_\tau p\bigr)
        \nonumber\\
        &\leq_{\mathrm{ew}}
        20D^{-2r-2}
        \Lambda_D^{r-1}p
        \leq_{\mathrm{ew}}
        20D^{-2r-2}
        \Lambda_D^rp,\\
        b^T\mathbb B_{11}^rE_\sigma p
        &=
        b^T\mathbb B_{11}^{r-1}
        \bigl(\mathbb B_{11}E_\sigma p\bigr)
        \nonumber\\
        &\leq
        14D^{-2r-3/2}
        \Lambda_D^{r-1}
        \leq
        14D^{-2r-3/2}
        \Lambda_D^r.
    \end{align}
    Here we used $\Lambda_D\geq1$. Since $I_r=1$ for $r\geq1$,
    these are the stated bounds.
\end{proof}

We now distinguish the possible forms of the matrix product. If $J_{00}\neq\emptyset$ and the full product is not $\mathbb B_{00}^m$, each component $W_a$ either contains a mismatch matrix or has the form $\mathbb B_{11}^{L_a}$; the product $\mathbb B_{00}^m$ is treated separately. If $J_{00}=\emptyset$, the product remains cyclic and is classified as pure active, non-pure mismatch, or pure mismatch. We begin with an open component containing at least one mismatch matrix.

\subsubsection{Case 1-1: Open mismatch components}

Let $W$ be a component in
Eq.~\eqref{app:eq:common_zero_component_decomposition} containing at
least one mismatch matrix.

\begin{corollary}[Mismatch-component numerator bound]
\label{app:cor:mismatch_component_numerator_bound}
    For $D\geq256$, every such component satisfies
    \begin{equation}
        b^TWe_0
        \leq
        5\cdot20^z
        \Lambda_D^r
        \kappa_D^{h-z}
        D^{-2r-3h/2-z-1-I_{\partial}/2}.
        \label{app:eq:mismatch_component_numerator_bound}
    \end{equation}
\end{corollary}

\begin{proof}
    Equation~\eqref{app:eq:w_decomp} contributes the factor
    $D^{-h/2}$. Iterating
    Eq.~\eqref{app:eq:mismatch_propagation_bound} through the mismatch
    runs contributes $D^{-(h-z)}\kappa_D^{h-z}$, since
    $\sum_{j=1}^z(\ell_j-1)=h-z$.

    Reading Eq.~\eqref{app:eq:bwe} from right to left, we apply
    Eq.~\eqref{app:eq:right_endpoint_bound} once,
    Eq.~\eqref{app:eq:internal_connector_bound} $z-1$ times, and
    Eq.~\eqref{app:eq:left_endpoint_bound} once. If an internal $r_j$
    vanishes, maximality gives $\sigma_j\neq\sigma_{j+1}$, so the
    condition in Eq.~\eqref{app:eq:internal_connector_bound} is
    satisfied. Therefore,
    \begin{align}
        b^TWe_0
        &\leq
        D^{-h/2}
        \left(
            14D^{-2r_0-1-I_{r_0}/2}\Lambda_D^{r_0}
        \right)
        \prod_{j=1}^{z}
        \left(
            D^{-(\ell_j-1)}\kappa_D^{\ell_j-1}
        \right)
        \nonumber\\
        &\quad\times
        \prod_{j=1}^{z-1}
        \left(
            20D^{-2r_j-2}\Lambda_D^{r_j}
        \right)
        \left(
            \frac{20}{3}
            D^{-2r_z-2-I_{r_z}/2}\Lambda_D^{r_z}
        \right)
        \nonumber\\
        &=
        14\cdot\frac{20}{3}\cdot20^{z-1}
        \Lambda_D^{\sum_{j=0}^{z}r_j}
        \kappa_D^{\sum_{j=1}^{z}(\ell_j-1)}
        \nonumber\\
        &\quad\times
        D^{-h/2-\sum_{j=1}^{z}(\ell_j-1)
        -2\sum_{j=0}^{z}r_j
        -1-2(z-1)-2-(I_{r_0}+I_{r_z})/2}
        \nonumber\\
        &=
        14\cdot\frac{20}{3}\cdot20^{z-1}
        \Lambda_D^r
        \kappa_D^{h-z}
        D^{-h/2-(h-z)-2r-2z-1-I_{\partial}/2}
        \nonumber\\
        &\leq
        5\cdot20^z
        \Lambda_D^r
        \kappa_D^{h-z}
        D^{-2r-3h/2-z-1-I_{\partial}/2}.
    \end{align}
\end{proof}

\begin{corollary}[Normalized mismatch-component bound]
\label{app:cor:normalized_mismatch_component_bound}
    For $D\geq256$, every open component containing at least one
    mismatch matrix satisfies
    \begin{equation}
        \frac{b^TWe_0}{\nu_{\boldsymbol u}(W)\nu_{\boldsymbol v}(W)}
        \leq
        \gamma_D^{L+1}q_D^h,
        \label{app:eq:normalized_mismatch_component_bound}
    \end{equation}
    where $q_D:=40/\sqrt D$.
\end{corollary}

\begin{proof}
    Dividing
    Eq.~\eqref{app:eq:mismatch_component_numerator_bound} by
    Eq.~\eqref{app:eq:open_mismatch_component_visibility} gives
    \begin{align}
        \frac{b^TWe_0}{\nu_{\boldsymbol u}(W)\nu_{\boldsymbol v}(W)}
        \leq{}&
        5\cdot20^z
        \Lambda_D^r
        \kappa_D^{h-z}
        D^{-2r-3h/2-z-1-I_{\partial}/2}
        \alpha_1^{-z-I_{\partial}}
        \lambda_+^{-2r-h+z+I_{\partial}}
        \nonumber\\
        \leq{}&
        5\cdot8^z
        \gamma_D^r
        \left(
            \frac{\kappa_D}{D\lambda_+}
        \right)^{h-z}
        D^{-h/2}
        \left(\frac{2}{5}\right)^{I_{\partial}}
        (D\lambda_+)^{I_{\partial}}
        D^{-1+I_{\partial}/2},
    \end{align}
    where the second inequality uses
    Eq.~\eqref{app:eq:alpha_one_lower_bound}. Since $I_{\partial}\in\{0,1,2\}$ and $D\lambda_+<1$, we have
    \begin{equation}
        \left(\frac{2}{5}\right)^{I_{\partial}}
        (D\lambda_+)^{I_{\partial}}
        D^{-1+I_{\partial}/2}
        \leq 1.
    \end{equation}
    Moreover, Eq.~\eqref{app:eq:mu_plus_bounds} gives $\kappa_D/(D\lambda_+)\leq(1+1/D)^3<8$. Since $1\leq z\leq h$, we have $5\cdot8^z(\kappa_D/(D\lambda_+))^{h-z}\leq40^h$. Hence
    \begin{equation}
        \frac{b^TWe_0}{\nu_{\boldsymbol u}(W)\nu_{\boldsymbol v}(W)}
        \leq
        \gamma_D^rq_D^h
        \leq
        \gamma_D^{L+1}q_D^h,
    \end{equation}
    where the last inequality follows from $\gamma_D>1$ and $L=r+h$.
\end{proof}

\subsubsection{Case 1-2: Open active components}

An open component containing no mismatch matrix has the form $W=\mathbb B_{11}^{L}$ with $L\geq1$. Since both local support sequences are $1^L$, we have $\nu_{\boldsymbol u}(W)=\nu_{\boldsymbol v}(W)=\alpha_L$.

\begin{lemma}[Normalized active-component bound]
\label{app:lem:normalized_active_component_bound}
    Suppose that $D\geq256$. Then
    \begin{equation}
        \frac{b^TWe_0}{\nu_{\boldsymbol u}(W)\nu_{\boldsymbol v}(W)}
        =
        \frac{b^T\mathbb B_{11}^{L}e_0}{\alpha_L^2}
        \leq
        \gamma_D^{L+1}.
        \label{app:eq:normalized_active_component_bound}
    \end{equation}
\end{lemma}

\begin{proof}
    Applying Eqs.~\eqref{app:eq:right_active_boundary_bound},
    \eqref{app:eq:active_propagation_bound}, and
    \eqref{app:eq:left_active_boundary_bound} gives
    \begin{align}
        b^T\mathbb B_{11}^{L}e_0
        &=
        b^T\mathbb B_{11}^{L-1}
        \bigl(\mathbb B_{11}e_0\bigr)
        \nonumber\\
        &\leq
        \frac{2}{3}D^{-3/2}
        \left(D^{-2}\Lambda_D\right)^{L-1}
        b^Tg
        \nonumber\\
        &\leq
        \frac{14}{3}
        D^{-2L-2}\Lambda_D^{L-1}.
        \label{app:eq:active_component_numerator_bound}
    \end{align}
    Equation~\eqref{app:eq:visibility_run_lower_bound} gives
    $\alpha_L\geq\alpha_1\lambda_+^{L-1}$. Therefore,
    \begin{align}
        \frac{b^T\mathbb B_{11}^{L}e_0}{\alpha_L^2}
        &\leq
        \frac{14}{3}
        \frac{
            D^{-2L-2}\Lambda_D^{L-1}
        }{
            \alpha_1^2\lambda_+^{2L-2}
        }
        \nonumber\\
        &\leq
        \frac{14}{3}
        \left(\frac{2}{5}\right)^2
        \left(
            \frac{\Lambda_D}{(D\lambda_+)^2}
        \right)^{L-1}
        \nonumber\\
        &=
        \frac{56}{75}\gamma_D^{L-1}
        \leq
        \gamma_D^{L+1}.
    \end{align}
    Here the second inequality uses Eq.~\eqref{app:eq:alpha_one_lower_bound}, and the last uses $56/75<1$ and $\gamma_D>1$.
\end{proof}

This agrees with
Eq.~\eqref{app:eq:normalized_mismatch_component_bound} when $h=0$.

\subsubsection{Case 1-3: The product $\mathbb B_{00}^{m}$}

When every factor is $\mathbb B_{00}$, there is no open component $W_a$, and the normalized entry can be evaluated directly.

\begin{lemma}[Pure-reset product bound]
\label{app:lem:pure_reset_product_bound}
    Suppose that $D\geq256$. Then
    \begin{equation}
        \frac{\operatorname{Tr}(\mathbb B_{00}^{m})}{m(0^m)^2}
        =
        D^{-m}
        \leq
        \gamma_D^m.
        \label{app:eq:pure_reset_product_bound}
    \end{equation}
\end{lemma}

\begin{proof}
    Since $\mathbb B_{00}=e_0b^T$ and $b^Te_0=D^{-1}$, we have $\operatorname{Tr}(\mathbb B_{00}^m)=(b^Te_0)^m=D^{-m}$. Moreover, $\mathbb K_0^m=\mathbb K_0$ and $\operatorname{Tr}(\mathbb K_0)=1$, so $m(0^m)=1$. The final inequality follows from $D^{-1}<1<\gamma_D$.
\end{proof}

\subsubsection{Combined bound for Case 1}

Cases 1-1 and 1-2 may occur in the same matrix product, while Case 1-3 covers the product with no open component. We now combine these estimates.

\begin{lemma}[Common-zero product bound]
\label{app:lem:common_zero_case_bound}
    Suppose that $D\geq256$ and $J_{00}\neq\emptyset$. Then
    \begin{equation}
        \frac{
            \operatorname{Tr}\!\left(
                \mathbb B_{u_1v_1}\cdots
                \mathbb B_{u_mv_m}
            \right)
        }{
            m(\boldsymbol u)m(\boldsymbol v)
        }
        \leq
        \gamma_D^m q_D^{|\boldsymbol u \oplus \boldsymbol v|}.
        \label{app:eq:common_zero_case_bound}
    \end{equation}
\end{lemma}

\begin{proof}
    If the full matrix product is $\mathbb B_{00}^m$, then $\boldsymbol u=\boldsymbol v=0^m$, and the claim follows from Lemma~\ref{app:lem:pure_reset_product_bound}.

    Suppose now that the full product is not $\mathbb B_{00}^m$. For each component $W_a$, let $L_a$ and $h_a$ denote its length and number of mismatch matrices, respectively. Then
    \begin{equation}
        \sum_{a=1}^{n_W}L_a
        =
        m-|J_{00}|,
        \qquad
        \sum_{a=1}^{n_W}h_a
        =
        |\boldsymbol u \oplus \boldsymbol v|.
        \label{app:eq:common_zero_component_counts}
    \end{equation}
    Corollary~\ref{app:cor:normalized_mismatch_component_bound} and Lemma~\ref{app:lem:normalized_active_component_bound} give, for every component $W_a$,
    \begin{equation}
        \frac{
            b^TW_ae_0
        }{
            \nu_u(W_a)\nu_v(W_a)
        }
        \leq
        \gamma_D^{L_a+1}q_D^{h_a}.
        \label{app:eq:uniform_open_component_bound}
    \end{equation}
    Substituting this estimate into Eq.~\eqref{app:eq:normalized_component_factorization} gives
    \begin{align}
        (K)_{\boldsymbol u \boldsymbol v}
        &\leq
        D^{-|J_{00}|+n_W}
        \prod_{a=1}^{n_W}
        \gamma_D^{L_a+1}q_D^{h_a}
        \nonumber\\
        &=
        D^{-(|J_{00}|-n_W)}
        \gamma_D^{m-|J_{00}|+n_W}
        q_D^{|\boldsymbol u \oplus \boldsymbol v|}
        \nonumber\\
        &=
        \gamma_D^m
        (D\gamma_D)^{-(|J_{00}|-n_W)}
        q_D^{|\boldsymbol u \oplus \boldsymbol v|}
        \nonumber\\
        &\leq
        \gamma_D^m q_D^{|\boldsymbol u \oplus \boldsymbol v|}.
    \end{align}
    Here $|J_{00}|\geq n_W$ because every $\mathbb B_{00}$ run is nonempty, and $D\gamma_D\geq1$.
\end{proof}

This completes the case $J_{00}\neq\emptyset$. We now assume $J_{00}=\emptyset$, so the numerator remains a single cyclic trace.

\subsubsection{Case 2-1: The pure active cycle}

If $J_{00}=\emptyset$ and no mismatch matrix is present, then $\boldsymbol u=\boldsymbol v=1^m$ and the full matrix product is $\mathbb B_{11}^m$.

\begin{lemma}[Pure active cycle bound]
\label{app:lem:pure_active_cycle_bound}
    Suppose that $D\geq256$. Then
    \begin{equation}
        \frac{\operatorname{Tr}(\mathbb B_{11}^{m})}{m(1^m)^2}
        \leq
        8\gamma_D^m.
        \label{app:eq:pure_active_cycle_bound}
    \end{equation}
\end{lemma}

\begin{proof}
    Iterating Eq.~\eqref{app:eq:active_propagation_bound} and applying Fact~\ref{app:fact:cw_bound} in the six-dimensional transfer space give
    \begin{align}
        \mathbb B_{11}^m g
        &\leq_{\mathrm{ew}}
        \left(D^{-2}\Lambda_D\right)^m g,
        \nonumber\\
        \operatorname{Tr}(\mathbb B_{11}^m)
        &\leq
        6\left(D^{-2}\Lambda_D\right)^m.
    \end{align}
    Equation~\eqref{app:eq:all_one_visibility_lower_bound} gives
    $m(1^m)\geq(1-D^{-1})\lambda_+^m$. Hence
    \begin{align}
        \frac{\operatorname{Tr}(\mathbb B_{11}^{m})}{m(1^m)^2}
        &\leq
        \frac{6}{(1-D^{-1})^2}
        \left(
            \frac{\Lambda_D}{(D\lambda_+)^2}
        \right)^m
        \nonumber\\
        &=
        \frac{6}{(1-D^{-1})^2}\gamma_D^m
        \leq
        8\gamma_D^m,
    \end{align}
    where the last inequality uses $D\geq256$.
\end{proof}

\subsubsection{Case 2-2: Non-pure cyclic mismatch products}

Assume that $J_{00}=\emptyset$ and that the full product contains a mismatch matrix but is neither $\mathbb B_{10}^m$ nor $\mathbb B_{01}^m$. After applying the same cyclic permutation to $\boldsymbol{u}$ and $\boldsymbol{v}$, we may write
\begin{align}
    \mathbb B_{u_1v_1}\cdots\mathbb B_{u_mv_m}
    &=
    \mathbb B_{\sigma_1}^{\ell_1}
    \mathbb B_{11}^{r_1}
    \cdots
    \mathbb B_{\sigma_z}^{\ell_z}
    \mathbb B_{11}^{r_z},
    \label{app:eq:cyclic_mismatch_run_decomposition}
\end{align}
where 
\begin{equation}
    h
    =
    \sum_{j=1}^z\ell_j,
    \qquad
    r=
    \sum_{j=1}^zr_j,
    \qquad m=r+h.
\end{equation}
Here $z\geq1$, $\sigma_j\in\{10,01\}$, $\ell_j\geq1$, and $r_j\geq0$, with $\sigma_{z+1}:=\sigma_1$. If $r_j=0$, maximality gives $\sigma_j\neq\sigma_{j+1}$. The excluded pure mismatch cycles are precisely the case $z=1$ and $r_1=0$. Moreover, $r=|J_{11}|$ and $h=|J_{10}|+|J_{01}|=|\boldsymbol u \oplus \boldsymbol v|$.

\begin{lemma}[Cyclic mismatch numerator bound]
\label{app:lem:nonpure_cyclic_mismatch_numerator}
    Suppose that $D\geq256$. Then
    \begin{equation}
        \operatorname{Tr}\!\left(
            \mathbb B_{u_1v_1}\cdots\mathbb B_{u_mv_m}
        \right)
        \leq
        2\cdot20^z
        \Lambda_D^r
        \kappa_D^{h-z}
        D^{-2r-3h/2-z}.
        \label{app:eq:nonpure_cyclic_mismatch_numerator}
    \end{equation}
\end{lemma}

\begin{proof}
    Applying Eq.~\eqref{app:eq:mismatch_run_factorization} to every mismatch run and using cyclicity of the trace gives
    \begin{align}
        &Q
        :=
        J^{\ell_1-1}
        \left(
            R_{\sigma_1}
            \mathbb B_{11}^{r_1}
            E_{\sigma_2}
        \right)
        \cdots
        J^{\ell_z-1}
        \left(
            R_{\sigma_z}
            \mathbb B_{11}^{r_z}
            E_{\sigma_1}
        \right),
        \nonumber\\
        &\operatorname{Tr}\!\left(
            \mathbb B_{u_1v_1}\cdots\mathbb B_{u_mv_m}
        \right)
        =
        D^{-h/2}\operatorname{Tr}(Q).
        \label{app:eq:cyclic_mismatch_reduced_trace}
    \end{align}
    The matrix $Q$ is a nonnegative $2\times2$ matrix. Each power $J^{\ell_j-1}$ contributes $D^{-(\ell_j-1)}\kappa_D^{\ell_j-1}$, while each factor between adjacent mismatch runs contributes $20D^{-2r_j-2}\Lambda_D^{r_j}$. Therefore,
    \begin{align}
        Qp
        &\leq_{\mathrm{ew}}
        20^z
        D^{-2r-2z}
        \Lambda_D^r
        D^{-(h-z)}
        \kappa_D^{h-z}p
        \nonumber\\
        &=
        20^z
        \Lambda_D^r
        \kappa_D^{h-z}
        D^{-2r-h-z}p.
        \label{app:eq:cyclic_mismatch_product_propagation}
    \end{align}
    The condition in Eq.~\eqref{app:eq:internal_connector_bound} holds for every $j$ by the decomposition above. Since $p$ has strictly positive entries, Fact~\ref{app:fact:cw_bound} gives
    \begin{equation}
        \operatorname{Tr}(Q)
        \leq
        2\cdot20^z
        \Lambda_D^r
        \kappa_D^{h-z}
        D^{-2r-h-z}.
    \end{equation}
    Multiplying this bound by $D^{-h/2}$ proves the claim.
\end{proof}

\begin{lemma}[Visibility of a non-pure cyclic mismatch product]
\label{app:lem:nonpure_cyclic_mismatch_visibility}
    Suppose that $D\geq256$. Then
    \begin{equation}
        m(\boldsymbol u)m(\boldsymbol v)
        \geq
        \left(1-\frac1D\right)
        \alpha_1^z
        \lambda_+^{2r+h-z}.
        \label{app:eq:nonpure_cyclic_mismatch_visibility}
    \end{equation}
\end{lemma}

\begin{proof}
    Suppose first that $\boldsymbol u\neq1^m$ and $\boldsymbol v\neq1^m$. Since no $\mathbb B_{00}$ is present, there is no boundary from $\mathbb B_{00}$. Therefore, the counting argument of Lemma~\ref{app:lem:open_mismatch_component_visibility} applies with $I_{\partial}=0$ and gives
    \begin{equation}
        m(\boldsymbol u)m(\boldsymbol v)
        \geq
        \alpha_1^z\lambda_+^{2r+h-z}
        \geq
        \left(1-\frac1D\right)
        \alpha_1^z\lambda_+^{2r+h-z}.
    \end{equation}

    It remains to consider the case in which one of the two sequences is $1^m$. Suppose that $u=1^m$. Then no $\mathbb B_{01}$ occurs. Since the pure mismatch cycle is excluded, the $z$ mismatch runs are separated by nonempty active runs. Thus, the sequence $v$ has $z$ nonempty runs of ones with total length $r$. Therefore,
    \begin{align}
        m(\boldsymbol u)
        &\geq
        \left(1-\frac1D\right)\lambda_+^{r+h},
        \nonumber\\
        m(\boldsymbol v)
        &\geq
        \alpha_1^z\lambda_+^{r-z}.
    \end{align}
    Multiplying these inequalities proves
    Eq.~\eqref{app:eq:nonpure_cyclic_mismatch_visibility}.
    The case $v=1^m$ follows by exchanging $\boldsymbol{u}$ and $\boldsymbol{v}$.
\end{proof}

\begin{corollary}[Normalized cyclic mismatch bound]
\label{app:cor:normalized_cyclic_mismatch_bound}
    Suppose that $D\geq256$. Then
    \begin{equation}
        \frac{
            \operatorname{Tr}\!\left(
                \mathbb B_{u_1v_1}\cdots\mathbb B_{u_mv_m}
            \right)
        }{
            m(\boldsymbol u)m(\boldsymbol v)
        }
        \leq
        \gamma_D^m q_D^h.
        \label{app:eq:normalized_cyclic_mismatch_bound}
    \end{equation}
\end{corollary}

\begin{proof}
    Dividing
    Eq.~\eqref{app:eq:nonpure_cyclic_mismatch_numerator} by
    Eq.~\eqref{app:eq:nonpure_cyclic_mismatch_visibility} gives
    \begin{align}
        \frac{
            \operatorname{Tr}\!\left(
                \mathbb B_{u_1v_1}\cdots\mathbb B_{u_mv_m}
            \right)
        }{
            m(\boldsymbol u)m(\boldsymbol v)
        }
        &\leq
        \frac{2}{1-D^{-1}}
        20^z
        \Lambda_D^r
        \kappa_D^{h-z}
        D^{-2r-3h/2-z}
        \alpha_1^{-z}
        \lambda_+^{-2r-h+z}
        \nonumber\\
        &\leq
        \frac{2}{1-D^{-1}}
        8^z
        \Lambda_D^r
        \kappa_D^{h-z}
        D^{-2r-3h/2+z}
        \lambda_+^{-2r-h+z}
        \nonumber\\
        &=
        \frac{2}{1-D^{-1}}
        8^z
        \gamma_D^m
        D^{-h/2}
        \frac{\kappa_D^{h-z}}{\Lambda_D^h}
        (D\lambda_+)^{h+z}.
        \label{app:eq:normalized_cyclic_mismatch_intermediate}
    \end{align}
    Here the second inequality uses
    Eq.~\eqref{app:eq:alpha_one_lower_bound}, and the last equality
    uses $m=r+h$. Since $\kappa_D\leq\Lambda_D$,
    $D\lambda_+<1$, and $\Lambda_D>1$,
    \begin{equation}
        \frac{\kappa_D^{h-z}}{\Lambda_D^h}
        (D\lambda_+)^{h+z}
        \leq
        \Lambda_D^{-z}(D\lambda_+)^{h+z}
        \leq1.
    \end{equation}
    Moreover, $D\geq256$ and $1\leq z\leq h$ give
    \begin{equation}
        \frac{2}{1-D^{-1}}8^z
        <
        3\cdot8^z
        \leq
        40^h.
    \end{equation}
    Substituting these estimates into
    Eq.~\eqref{app:eq:normalized_cyclic_mismatch_intermediate} and
    using $q_D=40/\sqrt D$ proves the claim.
\end{proof}

\subsubsection{Case 2-3: Pure mismatch cycles}

It remains to consider a cycle consisting entirely of one mismatch orientation. For $\sigma\in\{10,01\}$, the full matrix product is $\mathbb B_\sigma^m$, the visibility factors are $m(0^m)$ and $m(1^m)$, and $|\boldsymbol u \oplus \boldsymbol v|=m$. In Case 1, this form is included by taking $z=1$ and $r_0=r_z=0$, since the two boundary runs impose no condition on $\sigma_1$. 

\begin{lemma}[Pure mismatch cycle bound]
\label{app:lem:pure_mismatch_cycle_bound}
    Suppose that $D\geq256$. For every
    $\sigma\in\{10,01\}$,
    \begin{equation}
        \frac{
            \operatorname{Tr}(\mathbb B_\sigma^m)
        }{
            m(0^m)m(1^m)
        }
        \leq
        \left(\frac{2}{\sqrt D}\right)^m
        \leq
        \gamma_D^m q_D^m.
        \label{app:eq:pure_mismatch_cycle_bound}
    \end{equation}
\end{lemma}

\begin{proof}
    Equation~\eqref{app:eq:mismatch_rank_two_factorization},
    $R_\sigma E_\sigma=J$, and cyclicity of the trace give
    \begin{align}
        \operatorname{Tr}(\mathbb B_\sigma^m)
        &=
        D^{-m/2}
        \operatorname{Tr}\!\left(
            E_\sigma J^{m-1}R_\sigma
        \right)
        \nonumber\\
        &=
        D^{-m/2}\operatorname{Tr}(J^m).
        \label{app:eq:pure_mismatch_trace_reduction}
    \end{align}
    Comparing Eqs.~\eqref{app:eq:J_sigma_explicit} and
    \eqref{app:eq:visibility_transfer_matrices} gives
    $J=(1+1/D)^2\mathbb K_1$. Since
    $\operatorname{Tr}(\mathbb K_1^m)=m(1^m)$ and $m(0^m)=1$,
    \begin{align}
        \frac{
            \operatorname{Tr}(\mathbb B_\sigma^m)
        }{
            m(0^m)m(1^m)
        }
        &=
        D^{-m/2}
        \left(1+\frac1D\right)^{2m}
        \nonumber\\
        &\leq
        \left(\frac{2}{\sqrt D}\right)^m
        \leq
        \gamma_D^m q_D^m,
    \end{align}
    where the last line uses $(1+1/D)^2\leq2$,
    $\gamma_D>1$, and $q_D=40/\sqrt D$.
\end{proof}

Cases 2-1--2-3 complete the analysis for $J_{00}=\emptyset$. Together with the common-zero bound, these estimates give the uniform entrywise bound used below.

\subsection{From entrywise bounds to a quadratic-form bound}
\label{app:subsubsec:entrywise_to_quadratic_form}

The estimates established in the preceding cases imply that, for every $\boldsymbol u, \boldsymbol v\in\{0,1\}^m$,
\begin{equation}
    \frac{
        \operatorname{Tr}\!\left(
            \mathbb B_{u_1v_1}\cdots
            \mathbb B_{u_mv_m}
        \right)
    }{
        m(\boldsymbol u)m(\boldsymbol v)
    }
    \leq
    8\gamma_D^m q_D^{|\boldsymbol u \oplus \boldsymbol v|}.
    \label{app:eq:uniform_envelope_entry_bound}
\end{equation}
Here the four cases are covered by
Lemma~\ref{app:lem:common_zero_case_bound},
Lemma~\ref{app:lem:pure_active_cycle_bound},
Corollary~\ref{app:cor:normalized_cyclic_mismatch_bound}, and
Lemma~\ref{app:lem:pure_mismatch_cycle_bound}.

Returning to the exact matrices and applying
Eq.~\eqref{app:eq:B_product_entrywise_bound}, we obtain
\begin{align}
    (\widetilde K)_{\boldsymbol u \boldsymbol v}
    &=
    \frac{
        \operatorname{Tr}\!\left(
            \widetilde{\mathbb B}_{u_1v_1}\cdots
            \widetilde{\mathbb B}_{u_mv_m}
        \right)
    }{
        m(\boldsymbol u)m(\boldsymbol v)
    }\\
    &\leq
    8e^{6m/D}\gamma_D^m
    q_D^{|\boldsymbol u \oplus \boldsymbol v|}.
    \label{app:eq:uniform_exact_K_entry_bound}
\end{align}

For $\boldsymbol u, \boldsymbol v\in\{0,1\}^m$, define the $2^m\times2^m$ matrix $A$ by
\begin{equation}
    (A)_{\boldsymbol u \boldsymbol v}
    :=
    q_D^{|\boldsymbol u \oplus \boldsymbol v|}.
    \label{app:eq:hamming_matrix_definition}
\end{equation}
The additivity of the Hamming distance $|\boldsymbol u \oplus \boldsymbol v|=\sum_i|u_i \oplus v_i|$ gives
\begin{align}
    (A)_{\boldsymbol u \boldsymbol v}
    =
    q_D^{|\boldsymbol u \oplus \boldsymbol v|}
    =
    \prod_i q_D^{|u_i\oplus v_i|}
    &\longmapsto 
    A=
    \begin{pmatrix}
        1&q_D\\
        q_D&1
    \end{pmatrix}^{\otimes m}\\
    &\longmapsto 
    \lVert A\rVert_{\infty}
    =
    (1+q_D)^m.
    \label{app:eq:hamming_matrix_factorization}
\end{align}

Eq.~\eqref{app:eq:gamma_quadratic} 
and Eq.~\eqref{app:eq:uniform_exact_K_entry_bound} give
\begin{align}
    \lVert\Gamma(O_0)\rVert_\infty
    &\leq
    \sum_{\boldsymbol u,\boldsymbol v}(\widetilde K)_{\boldsymbol u \boldsymbol v}\|O_{\boldsymbol u}\|_{2}\|O_{\boldsymbol v}\|_{2}
    \\
    &\leq
    8e^{6m/D}\gamma_D^m
    \sum_{\boldsymbol u,\boldsymbol v}(A)_{\boldsymbol u \boldsymbol v}\|O_{\boldsymbol u}\|_{2}\|O_{\boldsymbol v}\|_{2}
    \\
    &\leq
    8e^{6m/D}\gamma_D^m
    \|A\|_{\infty} \sum_u\|O_{\boldsymbol u}\|_{2}^2\\
    &=
    8e^{6m/D}\gamma_D^m (1+q_D)^m \|O_0\|_2^2\\
    &\leq
    8e^{6m/D}\gamma_D^m
    \left(
        1+\frac{40}{\sqrt D}
    \right)^m
    \lVert O_0\rVert_2^2.
    \label{app:eq:entrywise_to_quadratic_form}
\end{align}

Using $\gamma_D\leq(1+4/\sqrt D)(1+1/D)^2$, $q_D=40/\sqrt D$, and $1+x\leq e^x$, we obtain
\begin{align}
    e^{6/D}\gamma_D(1+q_D)
    &\leq
    e^{6/D}
    \left(
        1+\frac4{\sqrt D}
    \right)
    \left(
        1+\frac1D
    \right)^2
    \left(
        1+\frac{40}{\sqrt D}
    \right)
    \\
    &\leq
    \exp\!\left(
        \frac{44}{\sqrt D}+\frac8D
    \right)
    \leq
    \exp\!\left(
        \frac{45}{\sqrt D}
    \right),
    \label{app:eq:dimension_factor_bound}
\end{align}
where the last inequality uses $D\geq256$.
Applying this estimate to Eq.~\eqref{app:eq:entrywise_to_quadratic_form} under the assumption $k2^{k/2}\ge Cn$ gives
\begin{align}
    \lVert\Gamma(O_0)\rVert_\infty
    &\leq
    8\exp\!\left(
        \frac{45m}{\sqrt D}
    \right)
    \lVert O_0\rVert_2^2\\
    &=
    8\exp\left(\frac{45n}{k2^{k/2}}\right)\lVert O_0\rVert_2^2\\
    &\leq
    8\exp\left(\frac{45}{C}\right)\lVert O_0\rVert_2^2.
    \label{app:eq:final_prediction_operator_bound}
\end{align}
This completes the proof of Theorem~\ref{appx:thm:1}.

\subsection{The role of the $H$ conjugation}\label{app:sec:role_H}

In Eq.~\eqref{appx:eq:t_B}, we conjugate each local transfer matrix $\mathbb R_{u_iv_i}^{(\kappa_i)}$ by $H$ before taking its entrywise absolute value. Applying the same conjugation at every position leaves each cyclic product trace unchanged:
\begin{equation}
    \operatorname{Tr}\!\left(
    \mathbb R_{u_1v_1}^{(\kappa_1)}
    \cdots
    \mathbb R_{u_mv_m}^{(\kappa_m)}
    \right)
    =
    \operatorname{Tr}\!\left(
    H\mathbb R_{u_1v_1}^{(\kappa_1)}H^{-1}
    \cdots
    H\mathbb R_{u_mv_m}^{(\kappa_m)}H^{-1}
    \right).
\end{equation}
The factors $H^{-1}H$ cancel between adjacent matrices, and cyclicity of the trace removes the remaining $H$ and $H^{-1}$. The conjugation changes the individual matrix entries and hence their entrywise absolute values. It therefore changes the upper bound obtained after taking these absolute values, while leaving the exact MPS trace unchanged.

For $u_i=v_i=1$, the sum in Eq.~\eqref{appx:eq:t_B} contains the $\mathord{=}$, $\texttt{fwd}$, and $\texttt{bwd}$ terms. With $H=I$, Eq.~\eqref{appx:eq:t_B} becomes
\begin{equation}
    \widetilde{\mathbb B}_{11}^{(H=I)}
    =
    \frac{1}{D}
    \left|
    \mathbb R_{11}^{(=)}
    \right|_{\mathrm{ew}}
    +
    \left|
    \mathbb R_{11}^{(\texttt{fwd})}
    \right|_{\mathrm{ew}}
    +
    \left|
    \mathbb R_{11}^{(\texttt{bwd})}
    \right|_{\mathrm{ew}}.
    \label{appx:eq:h_free_B11_definition}
    \end{equation}
    In the ordered basis
    $\mathcal S=(\texttt{00},\texttt{10},\texttt{01},
    \texttt E,\texttt C,\texttt A)$, the leading term of each entry is
    \begin{equation}
    \widetilde{\mathbb B}_{11}^{(H=I)}
    \simeq
    \begin{pmatrix}
    D^{-3}&4D^{-4}&4D^{-4}&D^{-3}&4D^{-4}&2D^{-4}\\
    6D^{-3}&2D^{-3}&4D^{-3}&4D^{-3}&2D^{-3}&2D^{-3}\\
    6D^{-3}&4D^{-3}&2D^{-3}&4D^{-3}&2D^{-3}&2D^{-3}\\
    2D^{-2}&7D^{-3}&7D^{-3}&D^{-2}&4D^{-3}&2D^{-3}\\
    5D^{-2}&2D^{-2}&2D^{-2}&2D^{-2}&D^{-2}&D^{-2}\\
    5D^{-2}&2D^{-2}&2D^{-2}&2D^{-2}&D^{-2}&D^{-2}
    \end{pmatrix}.
    \label{appx:eq:h_free_B11}
\end{equation}
Here, $\simeq$ retains the leading term in each entry; the omitted terms are smaller by a relative factor $O(D^{-1})$.

All entries of $\widetilde{\mathbb B}_{11}^{(H=I)}$ are nonnegative. Expanding the trace of its $m$th power gives a sum of products of nonnegative entries. Keeping only the terms whose intermediate indices belong to $\{\texttt C,\texttt A\}$ can therefore only decrease the trace.

Before taking entrywise absolute values, the $\texttt{fwd}$ and $\texttt{bwd}$ matrices on these two indices have the leading forms
\begin{equation}
    \left.
    \mathbb R_{11}^{(\texttt{fwd})}
    \right|_{\{\texttt C,\texttt A\}}
    =
    \frac{D^{-2}}{2}
    \begin{pmatrix}
    1&1\\
    1&1
    \end{pmatrix}
    +
    O(D^{-3}),
    \qquad
    \left.
    \mathbb R_{11}^{(\texttt{bwd})}
    \right|_{\{\texttt C,\texttt A\}}
    =
    \frac{D^{-2}}{2}
    \begin{pmatrix}
    1&-1\\
    -1&1
    \end{pmatrix}
    +
    O(D^{-3}).
    \label{appx:eq:h_free_signed_terms}
\end{equation}
Apart from the common factor $D^{-2}$, their leading terms are orthogonal rank-one projectors. When $H=I$, the entrywise absolute value removes their opposite off-diagonal signs. The three terms in Eq.~\eqref{appx:eq:h_free_B11_definition} then give
\begin{align}
    \left.
    \widetilde{\mathbb B}_{11}^{(H=I)}
    \right|_{\{\texttt C,\texttt A\}}
    &=
    \left.
    D^{-1}
    \left|\mathbb R_{11}^{(=)}\right|_{\mathrm{ew}}
    \right|_{\{\texttt C,\texttt A\}}
    +
    \left.
    \left|\mathbb R_{11}^{(\texttt{fwd})}\right|_{\mathrm{ew}}
    \right|_{\{\texttt C,\texttt A\}}
    +
    \left.
    \left|\mathbb R_{11}^{(\texttt{bwd})}\right|_{\mathrm{ew}}
    \right|_{\{\texttt C,\texttt A\}}
    \nonumber\\
    &=
    D^{-3}
    \begin{pmatrix}
    1&0\\
    0&1
    \end{pmatrix}
    +
    \frac{D^{-2}}{2}
    \begin{pmatrix}
    1&1\\
    1&1
    \end{pmatrix}
    +
    \frac{D^{-2}}{2}
    \begin{pmatrix}
    1&1\\
    1&1
    \end{pmatrix}
    +
    O(D^{-3})
    \nonumber\\
    &=
    D^{-2}
    \begin{pmatrix}
    1&1\\
    1&1
    \end{pmatrix}
    +
    O(D^{-3}),
\label{appx:eq:h_free_active_submatrix}
\end{align}
corresponding to the submatrix indexed by $\{\texttt{C}, \texttt{A}\}$. The two eigenvalues of this matrix are $2D^{-2}(1+O(D^{-1}))$ and $O(D^{-3})$. 
Consequently, whenever $k2^{k/2}=\Omega(n)$,
\begin{align}
    \frac{
    \operatorname{Tr}\!\left(
    \left(
    \widetilde{\mathbb B}_{11}^{(H=I)}
    \right)^m
    \right)
    }{
    m(1^m)^2
    }
    &\geq
    \frac{
    \operatorname{Tr}\!\left(
    \left(
    \left.
    \widetilde{\mathbb B}_{11}^{(H=I)}
    \right|_{\{\texttt C,\texttt A\}}
    \right)^m
    \right)
    }{
    m(1^m)^2
    }
    \nonumber\\
    &=
    \Omega(2^m)
    =
    \Omega\!\left(2^{n/k}\right).
    \label{appx:eq:h_free_exponential_factor}
\end{align}
It shows that the entrywise bound constructed from Eq.~\eqref{appx:eq:t_B} with $H=I$ contains the factor $2^{n/k}$. The choice of $H$ in Eq.~\eqref{app:eq:H} removes this factor by sending the two rank-one projectors to different diagonal entries before their entrywise absolute values are taken:
\begin{align}
    \left.
    \widetilde{\mathbb B}_{11}
    \right|_{\{\texttt C,\texttt A\}}
    &=
    \left.
    D^{-1}
    \left|
    H\mathbb R_{11}^{(=)}H^{-1}
    \right|_{\mathrm{ew}}
    \right|_{\{\texttt C,\texttt A\}}
    +
    \left.
    \left|
    H\mathbb R_{11}^{(\texttt{fwd})}H^{-1}
    \right|_{\mathrm{ew}}
    \right|_{\{\texttt C,\texttt A\}}
    +
    \left.
    \left|
    H\mathbb R_{11}^{(\texttt{bwd})}H^{-1}
    \right|_{\mathrm{ew}}
    \right|_{\{\texttt C,\texttt A\}}
    \nonumber\\
    &=
    D^{-3}
    \begin{pmatrix}
    1&0\\
    0&1
    \end{pmatrix}
    +
    D^{-2}
    \begin{pmatrix}
    1&0\\
    0&0
    \end{pmatrix}
    +
    D^{-2}
    \begin{pmatrix}
    0&0\\
    0&1
    \end{pmatrix}
    +
    O(D^{-3})
    \nonumber\\
    &=
    D^{-2}
    \begin{pmatrix}
    1&0\\
    0&1
    \end{pmatrix}
    +
    O(D^{-3}).
    \label{appx:eq:h_active_submatrix}
\end{align}
Its largest eigenvalue is $D^{-2}(1+O(D^{-1}))$, so its $m$th power does not contain the additional factor $2^m$.

\section{Approximate designs are not sufficient}\label{app:subsec:approximate_designs_not_sufficient}
Throughout this section, $\Phi_{\mathcal E}$ denotes the third-moment channel $\Phi_{\mathcal E}^{(3)}$.

\begin{fact}[\cite{collins2006integration}]
\label{app:fact:low_order_haar_choi_spectra}
    For an integer $t\in\{1,2,3\}$, define
    \begin{equation}
        J_s^{(t)}
        :=
        \mathbb E_{U\sim\operatorname{Haar}(\mathrm U(s))}
        \left[
            \left|U^{\otimes t}\right\rangle\!\right\rangle
            \left\langle\!\left\langle U^{\otimes t}\right|
        \right],
        \label{app:eq:haar_q_moment_choi}
    \end{equation}
    where $|A\rangle\!\rangle$ is the vectorization of matrix $A$.
    For $s\geq3$, the distinct positive eigenvalues of
    $J_s^{(1)}$, $J_s^{(2)}$, and $J_s^{(3)}$ are
    \begin{align}
        J_s^{(1)}:\quad&
        \frac1s,
        \label{app:eq:haar_choi_spectrum_t1}\\
        J_s^{(2)}:\quad&
        \frac{2}{s(s+1)},
        \qquad
        \frac{2}{s(s-1)},
        \label{app:eq:haar_choi_spectrum_t2}\\
        J_s^{(3)}:\quad&
        \frac{6}{s(s+1)(s+2)},
        \qquad
        \frac{6}{s(s^2-1)},
        \qquad
        \frac{6}{s(s-1)(s-2)}.
        \label{app:eq:haar_choi_spectrum_t3}
    \end{align}
\end{fact}

\begin{fact}[Haar decomposition of the Pauli commutant]
\label{app:fact:haar_decomposition_pauli_commutant}
    Let $P\in\mathcal P_n\setminus\{I\}$, and let $X$ and $Y$ be its positive and negative eigenspaces, respectively. Define $\mathrm H_P:=\{h\in\mathrm U(d):[h,P]=0\}$.
    Every $h\in\mathrm H_P$ can be written uniquely as $h=u\oplus v$, where $u\in\mathrm U(X)$ and $v\in\mathrm U(Y)$. Moreover, if $h\sim\operatorname{Haar}(\mathrm H_P)$, then $\boldsymbol{u}$ and $\boldsymbol{v}$ are independent Haar-random unitaries on $X$ and $Y$, respectively.
\end{fact}

\begin{proof}
    For every $|x\rangle\in X$, the relation $[h,P]=0$ gives $P h|x\rangle=hP|x\rangle=h|x\rangle$. Hence, $h$ preserves $X$. The same argument shows that $h$ preserves $Y$. Since $h$ is unitary, its restrictions to $X$ and $Y$ are unitaries, and therefore
    \begin{equation}
        h=u\oplus v
    \end{equation}
    for unique $u\in\mathrm U(X)$ and $v\in\mathrm U(Y)$. Conversely, every such $u\oplus v$ commutes with $P$. 

    Now sample $\boldsymbol{u}$ and $\boldsymbol{v}$ independently from the Haar measures on
    $\mathrm U(X)$ and $\mathrm U(Y)$. For every fixed
    $u_0\oplus v_0\in\mathrm H_P$,
    \begin{equation}
        (u_0\oplus v_0)(u\oplus v)
        =
        (u_0u)\oplus(v_0v)
    \end{equation}
    has the same distribution as $u\oplus v$ by the invariance of the two Haar measures. The resulting probability measure on $\mathrm H_P$ is therefore left invariant. By the uniqueness of the Haar probability measure on the compact group $\mathrm H_P$~\cite{mele2024introduction},
    it is precisely $\operatorname{Haar}(\mathrm H_P)$. Thus, under Haar sampling from $\mathrm H_P$, the two blocks $\boldsymbol{u}$ and $\boldsymbol{v}$ are independent and Haar distributed on their respective spaces.
\end{proof}

\begin{lemma}[Block-Haar third-moment domination]
\label{app:lem:block_haar_third_moment_domination}
    Let $d=2r\geq16$, and fix a Pauli operator $P\in \mathcal{P}_n \setminus\{I\}$. Define $\mathrm{H}_P:=\left\{U\in\mathrm U(d):[U,P]=0\right\}$. Then
    \begin{equation}
        \Phi_{\mathrm{H}_P}
        \preceq_{\mathrm{CP}}
        15\Phi_{\mathrm H}.
        \label{app:eq:block_haar_third_moment_domination}
    \end{equation}
    \end{lemma}
    
    \begin{proof}
    By Fact~\ref{app:fact:haar_decomposition_pauli_commutant}, a Haar-random $h\in\mathrm H_P$ can be written as $h=u\oplus v$, where $\boldsymbol{u}$ and $\boldsymbol{v}$ are independent Haar-random unitaries on $X$ and $Y$, respectively. The vectors $|u\rangle\!\rangle$ and $|v\rangle\!\rangle$ lie in the orthogonal spaces $X\otimes X^*$ and $Y\otimes Y^*$, respectively. Moreover,
    \begin{equation}
        |h\rangle\!\rangle
        =
        |u\rangle\!\rangle+|v\rangle\!\rangle.
    \end{equation}
    We therefore write $|h^{\otimes3}\rangle\!\rangle=|\Sigma_0\rangle\!\rangle+|\Sigma_1\rangle\!\rangle+|\Sigma_2\rangle\!\rangle+|\Sigma_3\rangle\!\rangle$, where
    \begin{align}
        |\Sigma_0\rangle\!\rangle
        &:=
        |u\rangle\!\rangle\otimes|u\rangle\!\rangle\otimes|u\rangle\!\rangle,\\
        |\Sigma_1\rangle\!\rangle
        &:=
        |u\rangle\!\rangle\otimes|u\rangle\!\rangle\otimes|v\rangle\!\rangle
        +
        |u\rangle\!\rangle\otimes|v\rangle\!\rangle\otimes|u\rangle\!\rangle
        +
        |v\rangle\!\rangle\otimes|u\rangle\!\rangle\otimes|u\rangle\!\rangle,\\
        |\Sigma_2\rangle\!\rangle
        &:=
        |u\rangle\!\rangle\otimes|v\rangle\!\rangle\otimes|v\rangle\!\rangle
        +
        |v\rangle\!\rangle\otimes|u\rangle\!\rangle\otimes|v\rangle\!\rangle
        +
        |v\rangle\!\rangle\otimes|v\rangle\!\rangle\otimes|u\rangle\!\rangle,\\
        |\Sigma_3\rangle\!\rangle
        &:=
        |v\rangle\!\rangle\otimes|v\rangle\!\rangle\otimes|v\rangle\!\rangle.
    \end{align}
    The Haar distribution of $u$ is invariant under $u\mapsto e^{i\theta}u$, which induces $|u\rangle\!\rangle\mapsto e^{i\theta}|u\rangle\!\rangle$. Under this transformation, $|\Sigma_p\rangle\!\rangle$ acquires the phase $e^{i(3-p)\theta}$. Hence, 
    \begin{equation}
        \mathbb E\!\left[|\Sigma_p\rangle\!\rangle\langle\! \langle \Sigma_q|\right]
        =
        e^{i(q-p)\theta}
        \mathbb E\!\left[|\Sigma_p\rangle\!\rangle\langle\! \langle \Sigma_q|\right]
    \end{equation}
    for every $\theta$, which implies $\mathbb E\!\left[|\Sigma_p\rangle\!\rangle\langle\! \langle \Sigma_q|\right]=\delta_{pq}\mathbb E\!\left[|\Sigma_p\rangle\!\rangle\langle\! \langle \Sigma_p|\right]$.
    It follows that
    \begin{align}
        J_{\mathrm H_P}
        &=
        \mathbb E_{h\sim\mathrm H_P}
        \left[\left(
            |h\rangle\!\rangle\langle\!\langle h|
        \right)^{\otimes3}\right]\\
        &=
        J_0\oplus J_1\oplus J_2\oplus J_3,
        \label{app:eq:block_haar_choi_decomposition}
    \end{align}
    where $J_p:=\mathbb E\!\left[|\Sigma_p\rangle\!\rangle\langle\! \langle \Sigma_p|\right]$.
    The two unmixed blocks satisfy $J_0=J_3=J_r^{(3)}$. Fact~\ref{app:fact:low_order_haar_choi_spectra} gives
    \begin{equation}
        \|J_0\|_\infty
        =
        \|J_3\|_\infty
        =
        \frac{6}{r(r-1)(r-2)}.
        \label{app:eq:unmixed_block_norm}
    \end{equation}
    
    To bound $J_1$, applying the inequality
    \begin{equation}
        (|a_1\rangle+|a_2\rangle+|a_3\rangle)(\langle a_1|+\langle a_2|+\langle a_3|)
        \preceq
        3\sum_{j=1}^3|a_j\rangle\langle a_j|.
        \label{app:eq:three_vector_bound}
    \end{equation}
    to $|\Sigma_1\rangle\!\rangle \langle\! \langle \Sigma_1|$ gives
    \begin{align}
        |\Sigma_1\rangle\!\rangle \langle\! \langle \Sigma_1| 
        \preceq 
        3(
        \nonumber&|u\rangle\!\rangle \langle\! \langle u|\otimes|u\rangle\!\rangle \langle\! \langle u|\otimes|v\rangle\!\rangle \langle\! \langle v|\\
        \quad\nonumber&+
        |u\rangle\!\rangle \langle\! \langle u|\otimes|v\rangle\!\rangle \langle\! \langle v|\otimes|u\rangle\!\rangle \langle\! \langle u|\\
        \quad&+
        |v\rangle\!\rangle \langle\! \langle v|\otimes|u\rangle\!\rangle \langle\! \langle u|\otimes|u\rangle\!\rangle \langle\! \langle u|).
    \end{align}
    Since $\boldsymbol{u}$ and $\boldsymbol{v}$ are independent, each term in $\mathbb{E}[|\Sigma_1\rangle\!\rangle \langle\! \langle \Sigma_1|]$ is a permutation of $3J_r^{(2)}\otimes J_r^{(1)}$. The three terms are mutually orthogonal due to the fact that $\langle\! \langle u|v\rangle\!\rangle=0$. Therefore,
    \begin{align}
        \|J_1\|_\infty
        \leq
        3
        \left\|J_r^{(2)}\right\|_\infty
        \left\|J_r^{(1)}\right\|_\infty
        =
        \frac{6}{r^2(r-1)}.
        \label{app:eq:first_mixed_block_norm}
    \end{align}
    The same bound holds in $\|J_2\|_\infty$. Since
    \begin{equation}
        \|J_1\|_\infty=\|J_2\|_\infty
        \leq
        \frac{6}{r^2(r-1)}
        \leq
        \frac{6}{r(r-1)(r-2)}=\|J_0\|_\infty=\|J_3\|_\infty,
    \end{equation}
    Eq.~\eqref{app:eq:block_haar_choi_decomposition} yields
    \begin{equation}
        \|J_{\mathrm{H}_P}\|_\infty
        =
        \frac{6}{r(r-1)(r-2)}.
        \label{app:eq:block_haar_choi_norm}
    \end{equation}
    
    Every element of $\mathrm{H}_P$ is a unitary on $\mathbb C^d$. Hence
    \begin{align}
        \operatorname{Im}(J_{\mathrm{H}_P})
        &=
        \operatorname{span}
        \left\{
            |h^{\otimes3}\rangle\!\rangle:
            h\in \mathrm{H}_P
        \right\}\\
        &\subseteq
        \operatorname{span}
        \left\{
            |U^{\otimes3}\rangle\!\rangle:
            U\in\mathrm U(d)
        \right\}\\
        &=
        \operatorname{Im}(J_{\mathrm{H}}).
        \label{app:eq:block_haar_image_inclusion}
    \end{align}
    Let $\Pi_{\mathrm H}$ and $\Pi_{\mathrm H_P}$ be the orthogonal projectors onto $\operatorname{Im}(J_{\mathrm H})$ and $\operatorname{Im}(J_{\mathrm H_P})$, respectively. Equations \eqref{app:eq:haar_choi_spectrum_t3} and \eqref{app:eq:block_haar_choi_norm} give
    \begin{align}
        J_{\mathrm{H}_P}
        &\preceq
        \|J_{\mathrm{H}_P}\|_\infty\Pi_{\mathrm{H}_P}
        \preceq
        \|J_{\mathrm{H}_P}\|_\infty\Pi_{\mathrm{H}}
        \\
        &\preceq
        \frac{
            \|J_{\mathrm{H}_P}\|_\infty
        }{
            \lambda_{\min}^{+}(J_{\mathrm{H}})
        }
        J_{\mathrm{H}}
        =
        \frac{
            d(d+1)(d+2)
        }{
            r(r-1)(r-2)
        }
        J_{\mathrm{H}}.
    \end{align}
    For $d=2r\geq16$,
    \begin{equation}
        \frac{
            d(d+1)(d+2)
        }{
            r(r-1)(r-2)
        }
        \leq15.
    \end{equation}
    Thus $J_{\mathrm{H}_P}\preceq15J_{\mathrm{H}}$. The Choi correspondence~\cite{choi1975completely} gives
    \begin{equation}
        \Phi_{\mathrm{H}_P}
        \preceq_{\mathrm{CP}}
        15\Phi_{\mathrm H},
    \end{equation}
    which proves the claim.
\end{proof}

For two unitary ensembles $\mathcal U$ and $\mathcal V$, let $\mathcal U\mathcal V$ denote the distribution of $UV$, where $U\sim\mathcal U$ and $V\sim\mathcal V$ are sampled independently. Their moment channels satisfy
\begin{equation}
    \Phi_{\mathcal U\mathcal V}
    =
    \Phi_{\mathcal U}\circ\Phi_{\mathcal V}.
    \label{app:eq:product_ensemble_moment_channel}
\end{equation}
Define the Clifford ensembles
\begin{align}
    \mathcal E_P
    &:=
    \left\{
        C\in\operatorname{Cl}(n):
        CPC^\dagger\in\pm\mathcal Z
    \right\},\\
    \mathcal E_P^{(c)}
    &:=
    \left\{
        C\in\operatorname{Cl}(n):
        CPC^\dagger\notin\pm\mathcal Z
    \right\},
\end{align}
each equipped with the uniform distribution. Uniform Clifford
conjugation maps $P$ uniformly to a signed non-identity Pauli. Since $\Pr_{C\sim \mathrm{Cl}(n)} (CPC^{\dagger} \in \pm \mathcal{Z})=1/(d+1)$,
\begin{equation}
    \Phi_{\operatorname{Cl}(n)}
    =
    \frac{1}{d+1}\Phi_{\mathcal E_P}
    +
    \frac{d}{d+1}\Phi_{\mathcal E_P^{(c)}}.
    \label{app:eq:cliff_Ord_detection_decomposition}
\end{equation}

\begin{lemma}[Approximate design with large variance]
\label{app:lem:approximate_design_large_variance}
    Let $d=2^n\geq16$ and fix
    $P\in\mathcal P_n\setminus\{I\}$. For $0<\eta<1$, define
    \begin{equation}
        \mathcal E_\eta
        :=
        (1-\eta)\mathcal E_P^{(c)}\mathrm H_P
        +
        \eta\operatorname{Cl}(n).
        \label{app:eq:M_eta}
    \end{equation}
    Then $\mathcal E_\eta$ has an invertible shadow channel and satisfies
    \begin{equation}
        \left(
            1-\frac{14(1-\eta)}{d}
        \right)
        \Phi_{\mathrm H}
        \preceq_{\mathrm{CP}}
        \Phi_{\mathcal E_\eta}
        \preceq_{\mathrm{CP}}
        \left(
            1+\frac{1-\eta}{d}
        \right)
        \Phi_{\mathrm H}.
        \label{app:eq:eta_design}
    \end{equation}
    Moreover, the associated shadow channel satisfies
    \begin{equation}
        \mathcal M_{\mathcal E_\eta}(P)=\frac{\eta}{d+1}P.
        \label{app:eq:eta_m_P}
    \end{equation}
\end{lemma}

\begin{proof}
    Composing Eq.~\eqref{app:eq:cliff_Ord_detection_decomposition} on the
    right with $\Phi_{\mathrm H_P}$ and using that
    $\operatorname{Cl}(n)$ is an exact unitary $3$-design gives
    \begin{equation}
        \Phi_{\mathrm H}
        =
        \frac{1}{d+1}\Phi_{\mathcal E_P\mathrm H_P}
        +
        \frac{d}{d+1}\Phi_{\mathcal E_P^{(c)}\mathrm H_P}.
        \label{app:eq:exact_design_conditional_decomposition}
    \end{equation}
    Here we used the right invariance of Haar measure $\Phi_{\operatorname{Cl}(n)\mathrm H_P}=\Phi_{\mathrm H}$. Lemma~\ref{app:lem:block_haar_third_moment_domination} and
    Eq.~\eqref{app:eq:product_ensemble_moment_channel} give
    \begin{align}
        \Phi_{\mathcal E_P\mathrm H_P}
        &=
        \Phi_{\mathcal E_P}\circ\Phi_{\mathrm H_P}
        \nonumber\\
        &\preceq_{\mathrm{CP}}
        15\Phi_{\mathcal E_P}\circ\Phi_{\mathrm H}=15\Phi_{\mathrm H},
        \label{app:eq:detected_ensemble_domination}
    \end{align}
    where the last equality follows from the left invariance of Haar
    measure. Rearranging
    Eq.~\eqref{app:eq:exact_design_conditional_decomposition} gives
    $\Phi_{\mathcal E_P^{(c)}\mathrm H_P}
    =(1+d^{-1})\Phi_{\mathrm H}
    -d^{-1}\Phi_{\mathcal E_P\mathrm H_P}$. Therefore,
    \begin{equation}
        \left(1-\frac{14}{d}\right)\Phi_{\mathrm H}
        \preceq_{\mathrm{CP}}
        \Phi_{\mathcal E_P^{(c)}\mathrm H_P}
        \preceq_{\mathrm{CP}}
        \left(1+\frac{1}{d}\right)\Phi_{\mathrm H}.
        \label{app:eq:no_detection_approximate_design}
    \end{equation}
    Since $\Phi_{\operatorname{Cl}(n)}=\Phi_{\mathrm H}$,
    Eq.~\eqref{app:eq:no_detection_approximate_design} gives
    \begin{equation}
        \left(1-\frac{14(1-\eta)}{d}\right)\Phi_{\mathrm H}
        \preceq_{\mathrm{CP}}
        \underbrace{
            (1-\eta)\Phi_{\mathcal E_P^{(c)}\mathrm H_P}
            +\eta\Phi_{\operatorname{Cl}(n)}
        }_{=\Phi_{\mathcal E_\eta}}
        \preceq_{\mathrm{CP}}
        \left(1+\frac{1-\eta}{d}\right)\Phi_{\mathrm H}.
    \end{equation}
    which proves Eq.~\eqref{app:eq:eta_design}. For $U=Ch\sim\mathcal E_P^{(c)}\mathrm H_P$, since $[h,P]=0$ and $CPC^\dagger\notin\pm\mathcal Z$, we have 
    \begin{align}
        \langle b|UPU^\dagger|b\rangle &= 
        \langle b|ChPh^{\dagger}C^\dagger|b\rangle\\
        &=\langle b|CPC^\dagger|b\rangle \label{app:eq:CPC}
        =0.
    \end{align}
    Hence
    $\mathcal M_{\mathcal E_P^{(c)}\mathrm H_P}(P)=0$, and
    \begin{equation}
        \mathcal M_{\mathcal E_\eta}(P)
        =
        (1-\eta)\mathcal M_{\mathcal E_P^{(c)}\mathrm H_P}(P)+
        \eta\mathcal M_{\operatorname{Cl}(n)}(P)
        =
        \frac{\eta}{d+1}P.
    \end{equation}
    Since $\eta>0$, mixing with the global Clifford ensemble ensures that $\mathcal M_{\mathcal E_\eta}$ is invertible.
\end{proof}
Although the unitary ensemble $\mathcal E_\eta$ in Eq.~\eqref{app:eq:M_eta} satisfies the approximate unitary $3$-design condition in Eq.~\eqref{main:eq:cp}, the visibility of $P$ is only $\eta/(d+1)$, as shown in Eq.~\eqref{app:eq:eta_m_P}. This visibility can be made arbitrarily small by decreasing $\eta$. This gives the following no-go theorem.

\begin{theorem}[Approximate designs are not sufficient]
\label{app:thm:approximate_design_arbitrary_estimator}
    Let $d=2^n\geq16$. For any $\epsilon_{\mathrm{des}}>14/2^n$, $\alpha>0$, and $0\leq\beta<1$, there exists an $\epsilon_{\mathrm{des}}$-approximate unitary $3$-design $\mathcal E$ satisfying Eq.~\eqref{main:eq:cp} such that any real-valued estimator $\hat f_O(U,b)$ satisfying
    \begin{equation}
        \left|
            \mathbb E_{U,b}[\hat f_O(U,b)]
            -
            \operatorname{Tr}(\rho O)
        \right|
        \leq
        \beta
        \left|
            \operatorname{Tr}(\rho O)
        \right|
        \label{app:eq:uniform_relative_bias}
    \end{equation}
    for every state $\rho$ and Hermitian observable $O$ has some
    state $\rho$ and Hermitian observable $O$ with $O_0\neq0$ for which
    \begin{equation}
        \operatorname{Var}_{U,b}[\hat f_O(U,b)]
        \geq
        \alpha\lVert O_0\rVert_2^2.
        \label{app:eq:arbitrary_estimator_variance_lower_bound}
    \end{equation}
    Here the expectation and variance are over the randomized
    measurement outcome $(U,b)$ generated from $\rho$ using
    $\mathcal E$.
\end{theorem}

Since $\alpha$ can be arbitrarily large, the approximate unitary design condition alone does not imply the variance guarantee of global Clifford measurements. The relative bias condition in Eq.~\eqref{app:eq:uniform_relative_bias} includes every unbiased estimator, corresponding to $\beta=0$. For example, the shadow estimator $\hat f_O(U,b)=\operatorname{Tr}\!\left[ O\mathcal M_{\mathcal E}^{-1} (U^\dagger\ket b\!\bra b U) \right]$ is unbiased, where $\mathcal M_{\mathcal E}^{-1}$ is the exact inverse of the shadow channel associated with $\mathcal E$. The condition also allows biased estimators with relative bias at most $\beta$. An accuracy condition is necessary: without one, an estimator that always outputs a constant has zero variance for every measurement distribution.

\begin{proof}[Proof of Theorem~\ref{app:thm:approximate_design_arbitrary_estimator}]
    Fix $P\in\mathcal P_n\setminus\{I\}$. Choose $0<\eta<1$ sufficiently small that $(1-\beta)^2(d+1)/(\eta d)\geq\alpha$, and let $\mathcal E=\mathcal E_\eta$ be the ensemble in Eq.~\eqref{app:eq:M_eta}. By Lemma~\ref{app:lem:approximate_design_large_variance}, $\mathcal E$ satisfies Eq.~\eqref{main:eq:cp} with error at most $\epsilon_{\mathrm{des}}$. Consider
    \begin{equation}
        \rho_+=\frac{I+P}{d},
        \qquad
        \rho_0=\frac{I}{d},
        \qquad
        O=\frac{P}{\sqrt d}.
    \label{app:eq:hard_state_observable_pair}
    \end{equation}
    These are valid states, and $\operatorname{Tr}(\rho_+ O)=1/\sqrt d$ and $\operatorname{Tr}(\rho_0O)=0$, while $O_0=O$ and $\lVert O_0\rVert_2^2=1$.
    
    Let $p_+$ and $p_0$ denote the distributions of $z=(U,b)$ generated from $\rho_+$ and $\rho_0$, respectively. For fixed $U$, their conditional probabilities and chi-square divergence are 
    \begin{equation}
        p_+(b\!\mid\!U)=\frac{1+\langle b|UPU^\dagger|b\rangle}{d}, \qquad p_0(b\!\mid\!U)=\frac 1d, \qquad
        \chi^2
        \left(
        p_+(b\!\mid\!U)\| p_0(b\!\mid\!U)
        \right)
        =\frac{1}{d}\sum_b\langle b|UPU^\dagger|b\rangle^2. \label{app:eq:cond_prob}
    \end{equation}
    For $U\sim\mathcal E_P^{(c)}\mathrm H_P$, Eq.~\eqref{app:eq:CPC} gives $\langle b|UPU^\dagger|b\rangle=0$. In the Clifford part, this matrix element is nonzero with probability $1/(d+1)$. Therefore,
    \begin{align}
        \chi^2(p_+\| p_0)
        &=
        \mathbb{E}_{U\sim \mathcal E_{\eta}} \chi^2(p_+(b\!\mid\!U)\| p_0(b\!\mid\!U))\\
        &=
        (1-\eta)\,\mathbb{E}_{U\sim\mathcal E_P^{(c)}\mathrm H_P} \chi^2(p_+(b\!\mid\!U)\| p_0(b\!\mid\!U)) 
        +
        \eta\,\mathbb{E}_{U\sim \mathrm{Cl}(n)} \chi^2(p_+(b\!\mid\!U)\| p_0(b\!\mid\!U))
        \\
        &=
        \frac{\eta}{d}\,
        \mathbb E_{U\sim\operatorname{Cl}(n)}
        \sum_b
        \langle b|UPU^\dagger|b\rangle^2
        =
        \frac{\eta}{d+1}.
        \label{app:eq:hard_pair_chi_square}
    \end{align}
    
    Applying Eq.~\eqref{app:eq:uniform_relative_bias} to $\rho_+$ and $\rho_0$ gives
    \begin{equation}
        \left|
        \mathbb E_{z\sim p_+}[\hat f_O(z)]
        -
        \mathbb E_{z\sim p_0}[\hat f_O(z)]
        \right|
        \geq
        \frac{1-\beta}{\sqrt d}.
        \label{app:eq:hard_pair_mean_separation}
    \end{equation}
    Here, $\mathbb E_{z\sim p_0}[\hat f_O(z)]=0$ because $\operatorname{Tr}(\rho_0O)=0$. Combining Theorem~\ref{app:thm:var_lower} with Eqs.~\eqref{app:eq:hard_pair_chi_square} and \eqref{app:eq:hard_pair_mean_separation} gives
    \begin{align}
        \operatorname{Var}_{z\sim p_0}[\hat f_O(z)]
        &\geq
        \frac{(1-\beta)^2/d}{
        \eta/(d+1)
        }
        \nonumber\\
        &=
        \frac{(1-\beta)^2(d+1)}{\eta d}
        \geq
        \alpha\lVert O_0\rVert_2^2.
        \label{app:eq:hard_pair_variance_lower_bound}
    \end{align}
    Thus Eq.~\eqref{app:eq:arbitrary_estimator_variance_lower_bound}
    holds for $\rho=\rho_0$ and $O=P/\sqrt d$.
\end{proof}

\begin{corollary}[Sample lower bounds for state learning]
\label{app:cor:approximate_design_state_learning_no_go}
    Let $\mathcal E_\eta$ be the unitary ensemble defined in Eq.~\eqref{app:eq:M_eta}. Suppose that an algorithm uses $N$ i.i.d.\ measurement outcomes generated using $\mathcal E_\eta$ and solves, for every state $\rho$, either quantum state tomography or stabilizer group learning with
    \begin{equation}
        \Pr\!\left[
            \lVert\hat\rho-\rho\rVert_1<\frac12
        \right]
        \geq
        \frac{2}{3},\qquad
        \Pr\!\left[
            \widehat S\supseteq\operatorname{Weyl}(\rho),
            \quad
            \frac{1}{|\widehat S|}
            \sum_{a\in\widehat S}
            \operatorname{Tr}(\rho W_a)^2
            \geq \frac{3}{4}
        \right]
        \geq \frac{2}{3},
        \label{app:eq:both}
    \end{equation}
    where $\hat\rho$ and $\widehat S$ are the outputs of the tomography and stabilizer-learning algorithms, respectively. Then
    \begin{equation}
        N
        =
        \Omega\!\left(\frac{d}{\eta}\right).
    \end{equation}
    Since $\eta$ can be arbitrarily small, the approximate unitary $3$-design condition alone does not provide a uniform sample-complexity bound for either task.
\end{corollary}

\begin{proof}
    Fix $P\in\mathcal P_n\setminus\{I\}$ and consider the ensemble $\mathcal E_\eta$ in Eq.~\eqref{app:eq:M_eta}. By Lemma~\ref{app:lem:approximate_design_large_variance}, $\mathcal E_\eta$ satisfies the $\epsilon_{\mathrm{des}}$-approximate unitary $3$-design condition and has an invertible shadow channel for every $0<\eta<1$. Consider the two states
    \begin{equation}
        \rho_0=\frac{I}{d},
        \qquad
        \rho_+=\frac{I+P}{d}.
    \end{equation}
    Let $p_0$ and $p_+$ denote the distributions of $(U,b)$ generated by  $\rho_0$ and $\rho_+$, respectively. By Eq.~\eqref{app:eq:cond_prob},
    \begin{align}
        d_{\mathrm{TV}}(p_+,p_0)
        &=
        \frac{1}{2d}
        \mathbb E_{U\sim\mathcal E_\eta}
        \sum_b
        \left|
            \langle b|UPU^\dagger|b\rangle
        \right|=
        \frac{\eta}{2(d+1)}.
        \label{app:eq:state_learning_single_tv}
    \end{align}
    A test with success probability at least $2/3$ must satisfy, by Le Cam's two-point method~\cite{yu1997assouad},
    \begin{equation}
        \frac{1}{3}
        \leq
        d_{\mathrm{TV}}
        \left(
            p_+^{\otimes N},
            p_0^{\otimes N}
        \right)
        \leq
        N d_{\mathrm{TV}}(p_+,p_0)
        =
        \frac{N\eta}{2(d+1)}.
    \end{equation}
    Therefore,
    \begin{equation}
        N
        \geq
        \frac{2(d+1)}{3\eta}
        =
        \Omega\!\left(\frac{d}{\eta}\right).
        \label{app:eq:state_learning_sample_lower_bound}
    \end{equation}
    Any algorithm satisfying the corresponding guarantee in Eq.~\eqref{app:eq:both} would distinguish $\rho_0$ from $\rho_+$ with probability at least 2/3. Therefore, Eq.~\eqref{app:eq:state_learning_sample_lower_bound} applies to both tasks.
\end{proof}

\section{Approximate designs are not necessary}
\label{app:subsec:approximate_designs_not_necessary}

We prove Proposition~\ref{main:prop:not_necessary} using two disconnected copies of the periodic two-layer Clifford circuit.  Divide the $n$-qubit system into two registers $A$ and $B$, each containing $n/2$ qubits, and independently sample
\begin{equation}
    U_A\sim\mathcal E_A,
    \qquad
    U_B\sim\mathcal E_B,
\end{equation}
where $\mathcal E_A$ and $\mathcal E_B$ are periodic two-layer Clifford ensembles with the same block size $k$.  The ensemble on the full system is
\begin{equation}
    \mathcal E_{AB}
    :=
    \mathcal E_A\otimes\mathcal E_B.
    \label{app:eq:tensor_product_ensemble}
\end{equation}
We first show that both sides of the relative approximate-design condition in Eq.~\eqref{main:eq:cp} fail.  We then prove that the variance bound of Theorem~\ref{main:thm:1} remains valid.

Each register $A, B$ has Hilbert-space dimension $\sqrt d$. For $X\in\{A,B\}$ and $\pi\in S_3$, let $V_\pi^X$ permute the three copies of $\mathcal H_X$ according to $\pi$, and define
\begin{equation}
    \Pi_X
    :=
    \frac{1}{6}\sum_{\pi\in S_3}V_\pi^X
    \qquad
    \Pi_{AB}
    :=
    \frac{1}{6}\sum_{\pi\in S_3}
    V_\pi^A\otimes V_\pi^B.
    \label{app:eq:three_copy_projectors}
\end{equation}
\begin{lemma}[Permutation-projector inclusion]
\label{app:lem:permutation_projector_inclusion}
    The projectors in Eq.~\eqref{app:eq:three_copy_projectors} satisfy
    \begin{equation}
        \Pi_A\otimes\Pi_B
        \preceq
        \Pi_{AB}.
        \label{app:eq:permutation_projector_inclusion}
    \end{equation}
    Consequently, $\Pi_{AB}-\Pi_A\otimes\Pi_B$ is positive semidefinite.
\end{lemma}

\begin{proof}
    Suppose that $\ket{\psi}$ satisfies
    \begin{equation}
        (V_\pi^A\otimes V_\sigma^B)\ket{\psi}
        =
        \ket{\psi}
        \qquad
        \text{for all }\pi,\sigma\in S_3.
    \end{equation}
    Setting $\pi=\sigma=\tau$ gives
    \begin{equation}
        V_\tau^{AB}\ket{\psi}
        =
        (V_\tau^A\otimes V_\tau^B)\ket{\psi}
        =
        \ket{\psi}
        \qquad
        \text{for all }\tau\in S_3.
    \end{equation}
    Thus every vector fixed by $\Pi_A\otimes\Pi_B$ is also fixed by $\Pi_{AB}$, and hence $\Pi_A\otimes\Pi_B\preceq\Pi_{AB}$.
\end{proof}

Define the positive operators
\begin{equation}
    O_+
    :=
    \Pi_A\otimes\Pi_B,
    \qquad
    O_-
    :=
    \Pi_{AB}-\Pi_A\otimes\Pi_B,
    \label{app:eq:product_design_witnesses}
\end{equation}
together with the three-copy input
\begin{equation}
    \rho_0
    :=
    \left(
        \ket{0}_A\!\bra{0}
        \otimes
        \ket{0}_B\!\bra{0}
    \right)^{\otimes3}.
    \label{app:eq:product_design_test_state}
\end{equation}

\begin{lemma}[Failure of the lower relative bound]
\label{app:lem:product_lower_relative_bound}
    For every $\epsilon_{\mathrm{des}}<1$,
    \begin{equation}
        (1-\epsilon_{\mathrm{des}})
        \Phi_{\mathrm H}^{(3)}
        \npreceq
        \Phi_{\mathcal E_{AB}}^{(3)}.
        \label{app:eq:product_lower_relative_bound}
    \end{equation}
\end{lemma}

\begin{proof}
    For every $U_A\otimes U_B$ in $\mathcal E_{AB}$, the state $(U_A\otimes U_B)^{\otimes3}\rho_0 (U_A^\dagger\otimes U_B^\dagger)^{\otimes3}$ is fixed by $\Pi_A\otimes\Pi_B$.  Hence
    \begin{equation}
        \operatorname{Tr}\!\left[
            O_-\Phi_{\mathcal E_{AB}}^{(3)}(\rho_0)
        \right]
        =0.
        \label{app:eq:product_lower_witness_zero}
    \end{equation}
    The Haar third moment is
    \begin{equation}
        \Phi_{\mathrm H}^{(3)}(\rho_0)
        =
        \frac{\Pi_{AB}}{\operatorname{Tr}(\Pi_{AB})},
    \end{equation}
    and therefore
    \begin{equation}
        \operatorname{Tr}\!\left[
            O_-\Phi_{\mathrm H}^{(3)}(\rho_0)
        \right]
        =
        1-
        \frac{
            \operatorname{Tr}(\Pi_A)\operatorname{Tr}(\Pi_B)
        }{
            \operatorname{Tr}(\Pi_{AB})
        }
        >0.
        \label{app:eq:product_lower_witness_haar}
    \end{equation}
    Applying the proposed CP-order inequality to the positive input $\rho_0$ and testing it against $O_-\succeq0$ would give
    \begin{equation}
        0
        \geq
        (1-\epsilon_{\mathrm{des}})
        \operatorname{Tr}\!\left[
            O_-\Phi_{\mathrm H}^{(3)}(\rho_0)
        \right]
        >0,
    \end{equation}
    which is impossible.
\end{proof}

\begin{lemma}[Failure of the upper relative bound]
\label{app:lem:product_upper_relative_bound}
    Suppose that $d\geq16$.  For every
    $\epsilon_{\mathrm{des}}<1$,
    \begin{equation}
        \Phi_{\mathcal E_{AB}}^{(3)}
        \npreceq
        (1+\epsilon_{\mathrm{des}})
        \Phi_{\mathrm H}^{(3)}.
        \label{app:eq:product_upper_relative_bound}
    \end{equation}
\end{lemma}

\begin{proof}
    The same product structure gives
    \begin{equation}
        \operatorname{Tr}\!\left[
            O_+\Phi_{\mathcal E_{AB}}^{(3)}(\rho_0)
        \right]
        =1.
        \label{app:eq:product_upper_witness_one}
    \end{equation}
    On the Haar side,
    \begin{align}
        \operatorname{Tr}\!\left[
            O_+\Phi_{\mathrm H}^{(3)}(\rho_0)
        \right]
        &=
        \frac{
            \operatorname{Tr}(\Pi_A)\operatorname{Tr}(\Pi_B)
        }{
            \operatorname{Tr}(\Pi_{AB})
        }
        \nonumber\\
        &=
        \frac{
            \binom{\sqrt{d}+2}{3}^{\,2}
        }{
            \binom{d+2}{3}
        }
        =
        \frac{
            (\sqrt{d}+1)^2(\sqrt{d}+2)^2
        }{
            6(d+1)(d+2)
        }
        \leq
        \frac{1}{2},
        \label{app:eq:product_upper_witness_haar}
    \end{align}
    where the last inequality holds for $d\geq16$. The proposed upper CP-order inequality would imply
    \begin{equation}
        1
        \leq
        (1+\epsilon_{\mathrm{des}})
        \operatorname{Tr}\!\left[
            O_+\Phi_{\mathrm H}^{(3)}(\rho_0)
        \right]
        <1,
    \end{equation}
    which is a contradiction.
\end{proof}

Lemmas~\ref{app:lem:product_lower_relative_bound} and~\ref{app:lem:product_upper_relative_bound} show that neither side of Eq.~\eqref{main:eq:cp} holds for $\mathcal E_{AB}$ with any $\epsilon_{\mathrm{des}}<1$.  We next establish the variance bound directly from the tensor-product structure of the two periodic components.

\begin{lemma}[Variance bound for two periodic components]
\label{app:lem:two_periodic_component_variance}
    Suppose that $k\geq8$ is even, $k$ divides $n/2$, and $k2^{k/2}\geq Cn$ for some constant $C>0$.  Let $\hat\rho_{AB}$ be the unbiased shadow estimator associated with $\mathcal E_{AB}$. Then, for every $n$-qubit state $\rho$ and Hermitian observable $O$,
    \begin{equation}
        \operatorname{Var}\!\left[
            \operatorname{Tr}(O\hat\rho_{AB})
        \right]
        \leq
        64e^{45/C}\lVert O_0\rVert_2^2,
        \qquad
        O_0
        :=
        O-\frac{\operatorname{Tr}(O)}{2^n}I.
        \label{app:eq:two_periodic_component_variance}
    \end{equation}
\end{lemma}

\begin{proof}
    Write
    \begin{equation}
        m_A=m_B=\frac{n}{2k},
        \qquad
        m=m_A+m_B=\frac{n}{k},
        \qquad
        D=2^k.
    \end{equation}
    The shadow channel and its exact inverse factorize as
    \begin{equation}
        \mathcal M_{AB}
        =
        \mathcal M_A\otimes\mathcal M_B,
        \qquad
        \mathcal M_{AB}^{-1}
        =
        \mathcal M_A^{-1}\otimes\mathcal M_B^{-1}.
        \label{app:eq:product_shadow_channel_factorization}
    \end{equation}

    Let $K_A$ and $K_B$ be the matrices entering
    Eq.~\eqref{app:eq:gamma_quadratic} for the two periodic components.
    For support strings
    \[
        u=(u^A,u^B),
        \qquad
        v=(v^A,v^B),
    \]
    both the exact visibility and the numerator contraction factorize. Consequently, the corresponding matrix for the product ensemble satisfies
    \begin{equation}
        (K_{AB})_{uv}
        =
        (K_A)_{u^Av^A}
        (K_B)_{u^Bv^B}.
        \label{app:eq:product_K_factorization}
    \end{equation}
    Applying Eq.~\eqref{app:eq:uniform_exact_K_entry_bound} separately to the two components gives
    \begin{align}
        (K_{AB})_{uv}
        &\leq
        64
        e^{6(m_A+m_B)/D}
        \gamma_D^{m_A+m_B}
        q_D^{|u^A\oplus v^A|+|u^B\oplus v^B|}
        \nonumber\\
        &=
        64
        e^{6m/D}
        \gamma_D^m
        q_D^{|\boldsymbol u \oplus \boldsymbol v|}.
        \label{app:eq:product_K_entry_bound}
    \end{align}

    Using the support decomposition of $O_0$ and repeating the quadratic-form estimate in Eqs.~\eqref{app:eq:entrywise_to_quadratic_form} and~\eqref{app:eq:dimension_factor_bound}, we obtain
    \begin{align}
        \left\|
            \Gamma_{AB}(O_0)
        \right\|_\infty
        &\leq
        64e^{6m/D}\gamma_D^m
        (1+q_D)^m
        \lVert O_0\rVert_2^2\\
        &\leq
        64\exp\!\left(
            \frac{45m}{\sqrt D}
        \right)
        \lVert O_0\rVert_2^2
        \nonumber\\
        &=
        64\exp\!\left(
            \frac{45n}{k2^{k/2}}
        \right)
        \lVert O_0\rVert_2^2\\
        &\leq
        64e^{45/C}\lVert O_0\rVert_2^2.
        \label{app:eq:product_prediction_operator_bound}
    \end{align}
    Combining this with Eq.~\eqref{app:eq:variance_prediction_bound}
    proves Eq.~\eqref{app:eq:two_periodic_component_variance}.
\end{proof}

Lemmas~\ref{app:lem:product_lower_relative_bound} and~\ref{app:lem:product_upper_relative_bound} show that the tensor-product ensemble fails both sides of the relative approximate-design condition. Lemma~\ref{app:lem:two_periodic_component_variance} shows that the same ensemble satisfies
\begin{equation}
    \operatorname{Var}\!\left[
        \operatorname{Tr}(O\hat\rho_{AB})
    \right]
    =
    O\!\left(
        \lVert O_0\rVert_2^2
    \right)
\end{equation}
under the block size condition of Theorem~\ref{main:thm:1}.  This proves Proposition~\ref{main:prop:not_necessary}.  The construction has two disconnected periodic one-dimensional components and already establishes that Eq.~\eqref{main:eq:cp} is not necessary for the variance guarantee of Theorem~\ref{main:thm:1}.

\section{Proof of Theorem~\ref{main:thm:2}}
\label{app:sec:proof_of_multishot_depth_lower_bound}

Theorem~\ref{main:thm:1} shows that, for single-shot observable estimation, global Clifford measurements can be replaced by measurements generated by the two-layer block unitary ensemble while preserving the same Frobenius-norm variance scaling. This replacement also applies to several randomized quantum-state learning algorithms because the statistical guarantee for these problems can be expressed as observable estimation problems that can be covered by Theorem~\ref{main:thm:1}. We now consider multi-shot observable estimation, in which each sampled unitary is reused for several measurement shots. Reusing a measurement unitary reduces the number of distinct circuits that need to be compiled, making this setting experimentally relevant \cite{brandao2020fast,zhou2023performance,helsen2023thrifty}.

Let $\{U_i\}_{i=1}^{N_U}$ be sampled independently from a unitary ensemble $\mathcal E$. For each $U_i$, we perform $N_S$ measurements and obtain outcomes $\{b_{i,j}\}_{j=1}^{N_S}$. The multi-shot shadow estimator is
\begin{equation}
    \hat{\rho}_{\mathrm{mul}}
    =
    \frac{1}{N_U N_S}
    \sum_{i=1}^{N_U}
    \sum_{j=1}^{N_S}
    \mathcal{M}^{-1}
    \left(
        U_i^\dagger
        \ket{b_{i,j}}\!\bra{b_{i,j}}
        U_i
    \right),
    \label{app:eq:multishot_shadow_estimator}
\end{equation}
where $\mathcal{M}^{-1}$ is the exact inverse shadow channel. The corresponding estimator of
$\operatorname{Tr}(\rho O)$ is $\hat{o}=\operatorname{Tr}(O\hat{\rho}_{\mathrm{mul}})$. Then its variance can be written as 
\begin{equation}
    \operatorname{Var}[\hat{o}]
    =
    \frac{1}{N_U}
    \left(
        \frac{V_{\mathrm{shot}}}{N_S}
        +
        \frac{N_S-1}{N_S}
        V_{\mathrm{floor}}
    \right),
    \label{app:eq:multishot_variance_decomposition}
\end{equation}
where 
\begin{align}
    &V_{\mathrm{shot}}=
    \operatorname{Var}_{U,b}
    \left[\operatorname{Tr}(O\mathcal{M}^{-1}(U^{\dagger}\ket{b}\!\bra{b}U))
    \right],\\
    &V_{\mathrm{floor}}=
    \operatorname{Var}_{U}
    \left[
    \sum_{b\in \{0,1\}^n} \operatorname{Tr}(U\rho U^{\dagger}\ket{b}\!\bra{b})\operatorname{Tr}(O\mathcal{M}^{-1}(U^{\dagger}\ket{b}\!\bra{b}U)) 
    \right]\label{app:eq:variance_floor}
\end{align}

The $\operatorname{Var}[\hat o]$ depends on fourth-order unitary moments. Thus a global exact unitary $4$-design reproduces the corresponding Haar
average exactly. In particular, the lemma below shows that
\begin{equation}
    V_{\mathrm{floor}}
    \leq
    \frac{4}{d}
    \lVert O_0\rVert_2^2.
    \label{app:eq:exact_four_design_floor_preview}
\end{equation}
We may ask whether the two-layer block unitary ensemble can reproduce the same scaling in Eq.~\eqref{app:eq:exact_four_design_floor_preview}. We answer this question negatively by constructing a state $\rho$ and an observable $O$ for which the two-layer geometry leaves a much larger value of $V_{\mathrm{floor}}$. The separation remains even when every block unitary is independently Haar random and when the block size grows linearly with the system size.

At a high level, the proof exploits the limited propagation of the two-layer circuit. We choose a state and an observable for which the quantity defining $V_{\mathrm{floor}}$ separates into contributions from two spatial regions. The two contributions depend on disjoint sets of random block unitaries and are therefore independent. More precisely, the relevant random variable takes the form $Z=XY$, where $X$ and $Y$ are independent and exact unbiasedness gives
$\mathbb{E}[X]=\mathbb{E}[Y]=1$. Consequently,
\begin{align}
    \operatorname{Var}[Z]
    &=
    \operatorname{Var}[X]
    +
    \operatorname{Var}[Y]
    +
    \operatorname{Var}[X]\operatorname{Var}[Y].
    \label{app:eq:product_variance_overview}
\end{align}
Thus, a fluctuation confined to either region remains as a lower bound on the full variance floor.

We first derive the multi-shot variance bound for a global exact unitary $4$-design and then construct the state and observable that exhibit the two-layer variance separation.

\begin{lemma}[Multi-shot variance for exact unitary $4$-designs]
\label{app:lem:exact_four_design_multishot_variance}
    Let $\mathcal E$ be a global exact unitary $4$-design on a
    $d$-dimensional Hilbert space. Then, for every state $\rho$ and every
    Hermitian observable $O$,
    \begin{equation}
        V_{\mathrm{shot}}
        \leq
        3\lVert O_0\rVert_2^2,
        \qquad
        V_{\mathrm{floor}}
        \leq
        \frac{4}{d}\lVert O_0\rVert_2^2.
        \label{app:eq:exact_four_design_variance_terms}
    \end{equation}
    Consequently, the multi-shot estimator in
    Eq.~\eqref{app:eq:multishot_shadow_estimator} satisfies
    \begin{equation}
        \operatorname{Var}[\hat{o}]
        \leq
        \frac{1}{N_U}
        \left(
            \frac{3}{N_S}
            +
            \frac{4(N_S-1)}{dN_S}
        \right)
        \lVert O_0\rVert_2^2.
        \label{app:eq:exact_four_design_multishot_variance}
    \end{equation}
\end{lemma}

\begin{proof}
    The trace part of $O$ contributes only a deterministic constant to
    $\hat{o}$, so it suffices to consider the traceless observable $O_0$.
    Since an exact unitary $4$-design is also an exact unitary $2$-design,
    its shadow channel agrees with the Haar shadow channel and
    \begin{equation}
        \mathcal{M}_{\mathcal E}^{-1}(O_0)
        =
        (d+1)O_0.
        \label{app:eq:four_design_inverse_on_traceless_observable}
    \end{equation}

    Applying the Haar third-moment identity to a single measurement outcome~\cite{huang2020predicting} gives
    \begin{align}
        V_{\mathrm{shot}}
        &=
        \frac{d+1}{d+2}
        \left(
            \lVert O_0\rVert_2^2
            +
            2\operatorname{Tr}(\rho O_0^2)
        \right)
        -
        \operatorname{Tr}(\rho O_0)^2
        \leq
        3\lVert O_0\rVert_2^2.
        \label{app:eq:exact_four_design_single_shot_bound}
    \end{align}
    Here we used $\operatorname{Tr}(\rho O_0^2)\leq\lVert O_0\rVert_2^2$. For two outcomes obtained using the same unitary, the Haar fourth-moment identity~\cite{anshu2022distributed} gives
    \begin{equation}
        V_{\mathrm{floor}}
        \leq
        \frac{4d^2+9d+7}
        {d(d+2)(d+3)}
        \lVert O_0\rVert_2^2
        \leq
        \frac{4}{d}
        \lVert O_0\rVert_2^2.
        \label{app:eq:exact_four_design_floor_bound}
    \end{equation}
    The first inequality follows from $\operatorname{Tr}(\rho^2)\leq1$, $\operatorname{Tr}(\rho^2O_0^2)\leq\lVert O_0\rVert_2^2$, and $\operatorname{Tr}(\rho O_0)^2 \leq\operatorname{Tr}(\rho^2)\lVert O_0\rVert_2^2$. Substituting these bounds into Eq.~\eqref{app:eq:multishot_variance_decomposition} proves the claim.
\end{proof}

\begin{fact}[Variance of the purity of a Haar-random state~\cite{scott2003entangling}]\label{appx:fact:purity_var}
    Let $\mathcal{H}=\mathcal{H}_A\otimes\mathcal{H}_B$ with $d_A=d_B=\sqrt{d}$, and let $\rho=U\ket{0}\!\bra{0}U^\dagger$, where $U\sim\operatorname{Haar}(\mathrm{U}(d))$. Let $\rho_A=\operatorname{Tr}_B(\rho)$ and $p:=\operatorname{Tr}(\rho_A^2)$. Then
    \begin{align}
        \operatorname{Var}_{U}[p]
        &=
        \frac{2(d-1)^2}
        {(d+1)^2(d+2)(d+3)}.
    \end{align}
\end{fact}

\begin{lemma}[Haar second moment across two adjacent blocks]
\label{lem:two_block_haar_average}
Let $k$ be even and let $D=2^k$. Consider
$\mathcal H=\mathcal H_1\otimes\mathcal H_2
\otimes\mathcal H_3\otimes\mathcal H_4$, where
$\dim(\mathcal H_i)=\sqrt D$.
Let $A,B\in\mathcal L(\mathcal H_2\otimes\mathcal H_3)$, and let
\[
V_{12}\sim\operatorname{Haar}(\mathcal H_1\otimes\mathcal H_2),
\qquad
V_{34}\sim\operatorname{Haar}(\mathcal H_3\otimes\mathcal H_4)
\]
be independent Haar-random unitaries.
For fixed $b_{12},b_{34}\in\{0,1\}^k$, define
\begin{equation}
\rho_2
:=
\operatorname{Tr}_1\!\left(
V_{12}|b_{12}\rangle\langle b_{12}|V_{12}^{\dagger}
\right),
\qquad
\rho_3
:=
\operatorname{Tr}_4\!\left(
V_{34}|b_{34}\rangle\langle b_{34}|V_{34}^{\dagger}
\right).
\end{equation}
Then
\begin{align}
&\mathbb E_{V_{12},V_{34}}
\left[
\operatorname{Tr}\!\left(A(\rho_2\otimes\rho_3)\right)
\operatorname{Tr}\!\left(B(\rho_2\otimes\rho_3)\right)
\right]
\nonumber\\
&=
\frac{1}{D^2(D+1)^2}
\Big(
D^2\operatorname{Tr}(A)\operatorname{Tr}(B)
+D^{3/2}
\operatorname{Tr}\!\left(
\operatorname{Tr}_3(A)\operatorname{Tr}_3(B)
\right)
\nonumber\\
&\hspace{35mm}
+D^{3/2}
\operatorname{Tr}\!\left(
\operatorname{Tr}_2(A)\operatorname{Tr}_2(B)
\right)
+D\operatorname{Tr}(AB)
\Big).
\label{eq:two_block_haar_average}
\end{align}
\end{lemma}

\begin{proof}
By Haar invariance, $V_{12}|b_{12}\rangle$ and
$V_{34}|b_{34}\rangle$ are independent Haar-random unit vectors, and
their distributions do not depend on the fixed basis labels.
For a Haar-random unit vector on a $D$-dimensional space,
\begin{equation}
\mathbb E_V
\left[
\left(
V|b\rangle\langle b|V^\dagger
\right)^{\otimes 2}
\right]
=
\frac{I+\mathrm{SWAP}}{D(D+1)}.
\end{equation}
Let primed labels denote the second tensor copy. Taking the partial
trace over the two copies of the outer subsystem gives, for
$a\in\{2,3\}$,
\begin{equation}
\mathbb E\!\left[\rho_a\otimes\rho_{a'}\right]
=
\Phi_{aa'}
:=
\frac{
\sqrt D\,I_{aa'}+\mathrm{SWAP}_{aa'}
}{
\sqrt D(D+1)
}.
\label{eq:reduced_haar_second_moment}
\end{equation}
Indeed, the identity term gives $D I_{aa'}$, while the swap term gives
$\sqrt D\,\mathrm{SWAP}_{aa'}$ after tracing out the two outer copies.

Independence of $V_{12}$ and $V_{34}$ now yields
\begin{equation}
\mathbb E_{V_{12},V_{34}}
\left[
\operatorname{Tr}\!\left(A(\rho_2\otimes\rho_3)\right)
\operatorname{Tr}\!\left(B(\rho_2\otimes\rho_3)\right)
\right]
=
\operatorname{Tr}\!\left[
(A_{23}\otimes B_{2'3'})
(\Phi_{22'}\otimes\Phi_{33'})
\right].
\end{equation}
The product of the two second moments is
\begin{align}
\Phi_{22'}\otimes\Phi_{33'}
=
\frac{1}{D^2(D+1)^2}
\Big(
&D^2 I_{22'}\otimes I_{33'}
+D^{3/2}\mathrm{SWAP}_{22'}\otimes I_{33'}
\nonumber\\
&+D^{3/2}I_{22'}\otimes\mathrm{SWAP}_{33'}
+D\,\mathrm{SWAP}_{22'}\otimes\mathrm{SWAP}_{33'}
\Big).
\end{align}
Finally,
\begin{align}
    &\operatorname{Tr}\!\left[
    (A\otimes B)
    (\mathrm{SWAP}_{22'}\otimes I_{33'})
    \right]
    =
    \operatorname{Tr}\!\left(
    \operatorname{Tr}_3(A)\operatorname{Tr}_3(B)
    \right),\\
    &\operatorname{Tr}\!\left[
    (A\otimes B)
    (I_{22'}\otimes\mathrm{SWAP}_{33'})
    \right]
    =
    \operatorname{Tr}\!\left(
    \operatorname{Tr}_2(A)\operatorname{Tr}_2(B)
    \right),\\
    &\operatorname{Tr}\!\left[
    (A\otimes B)
    (\mathrm{SWAP}_{22'}\otimes\mathrm{SWAP}_{33'})
    \right]
    =
    \operatorname{Tr}(AB).
\end{align}
The identity term equals
$\operatorname{Tr}(A)\operatorname{Tr}(B)$, which proves
Eq.~\eqref{eq:two_block_haar_average}.
\end{proof}

\begin{theorem}[Multi-shot depth lower bound]
\label{app:thm:multishot_depth_lower_bound}
    Consider the two-layer block unitary ensemble on $n$ qubits, where each $k$-qubit block unitary is sampled independently from $\operatorname{Haar}(k)$, together with its exact inverse shadow channel. For $k\leq n/4$, there exist a state $\rho$ and an observable $O$ such that
    \begin{equation}
        \frac{
            V_{\mathrm{floor}}^{2\mathrm L}
        }{
            V_{\mathrm{floor}}^{\mathrm{Haar}}
        }
        =
        \Omega
        \left(
            2^{n-3k}
        \right),
        \label{app:eq:multishot_floor_separation}
    \end{equation}
    where $V_{\mathrm{floor}}^{2\mathrm L}$ is the variance floor of the two-layer ensemble and $V_{\mathrm{floor}}^{\mathrm{Haar}}$ is that of a global exact unitary $4$-design for the same $\rho$ and $O$. Consequently, reproducing the global exact unitary $4$-design variance scaling uniformly over all states and observables requires $k=\Omega(n)$.
\end{theorem}

\begin{proof}
    Set $D=2^k$ and $m=n/k$, so that $m\geq4$. We prove the claimed separation by taking every block gate independently Haar random. For a fixed measurement unitary $U$, define
    \begin{equation}
        g_U
        :=
        \sum_b
        \operatorname{Tr}
        \left(
            U\rho U^\dagger\ket b\!\bra b
        \right)
        \operatorname{Tr}
        \left(
            O\mathcal M^{-1}
            \left(
                U^\dagger\ket b\!\bra bU
            \right)
        \right)
        \label{app:eq:g_U}
    \end{equation}
    Equation~\eqref{app:eq:variance_floor} gives $V_{\mathrm{floor}}^{2\mathrm L}=\operatorname{Var}_U[g_U]$. Let $\Pi_A=\ket{0^k}\!\bra{0^k}$ on the first block and let $\Pi_B:=\bigotimes_{j=3}^{m-1}\ket{0^k}\!\bra{0^k}$ on blocks $3,\ldots,m-1$. Define
    \begin{align}
        \rho
        &=
        \Pi_A\otimes
        \frac{I_k}{D}\otimes
        \Pi_B\otimes
        \frac{I_k}{D},
        \nonumber\\
        O
        &=
        \Pi_A\otimes I_k\otimes
        \Pi_B\otimes I_k.
        \label{app:eq:multishot_counterexample_state_observable}
    \end{align}
    Write the $j$th first-layer block as $L_jR_j$. The two identity blocks ensure that the light cones of $\mathsf A=L_1R_1$ and $\mathsf B=L_3R_3\cdots L_{m-1}R_{m-1}$ remain disjoint under the action of the two-layer block unitary $U$. We denote the corresponding extended regions by
    \begin{equation}
        \mathsf{eA}
        :=
        R_m\,\mathsf A\,L_2,
        \qquad
        \mathsf{eB}
        :=
        R_2\,\mathsf B\,L_m.
        \label{app:eq:extended_region_decomposition}
    \end{equation}
    Thus, the relevant support pattern separates schematically as
    \begin{equation}
        \mathsf A\,(L_2R_2)\,\mathsf B\,(L_m R_m)
        \xrightarrow{U} 
        \mathsf{eA}\,\mathsf{eB}.
    \end{equation}
    
    We first show that $\mathcal{M}^{-1}(O)$ factorizes across $\mathsf{eA}$ and $\mathsf{eB}$.
    Write the Pauli expansions $\Pi_A=\sum_Pa_PP$ and $\Pi_B=\sum_Qb_QQ$, and let $\mathbb K(P)$ and $\mathbb K(Q)$ denote the products of transfer matrices determined by the visibility type vectors of $P$ and $Q$. The rank-one decomposition of $\mathbb{K}_0$ gives
    \begin{align}
        m_{P\otimes I_k\otimes Q\otimes I_k}
        &=
        \operatorname{Tr}
        \left(
            \mathbb K_0\mathbb K(P)
            \mathbb K_0\mathbb K(Q)
        \right)
        \nonumber\\
        &=
        \operatorname{Tr}
        \left(
            \mathbb K_0\mathbb K(P)
        \right)
        \operatorname{Tr}
        \left(
            \mathbb K_0\mathbb K(Q)
        \right)\\
        &=
        m_{\mathsf{eA}}(P)m_{\mathsf{eB}}(Q),
        \label{app:eq:counterexample_visibility_factorization}
    \end{align}
    where $m_{\mathsf{eA}}(P)=\operatorname{Tr}(\mathbb K_0\mathbb K(P))$ and $m_{\mathsf{eB}}(Q)=\operatorname{Tr}(\mathbb K_0\mathbb K(Q))$.
    Therefore,
    \begin{equation}
        \mathcal M^{-1}(O)
        =
        \widetilde{\Pi}_A\otimes I_k
        \otimes\widetilde{\Pi}_B\otimes I_k,
        \label{app:eq:counterexample_inverse_factorization}
    \end{equation}
    where
    \begin{equation}
        \widetilde{\Pi}_A
        :=
        \sum_P
        \frac{a_P}
        {m_{\mathsf{eA}}(P)}P,
        \qquad
        \widetilde{\Pi}_B
        :=
        \sum_Q
        \frac{b_Q}
        {m_{\mathsf{eB}}(Q)}Q.
    \end{equation}
    Both $\rho$ and $\mathcal M^{-1}(O)$ retain the separation of $\mathsf A$ and $\mathsf B$ by the intervening identity blocks. Hence, $g_U$ factorizes into a product of independent random variables:
    \begin{align}
        g_U 
        &= 
        \sum_b
        \operatorname{Tr}(U\rho U^\dagger\ket b\!\bra b)
        \operatorname{Tr}\left(\mathcal M^{-1}(O) U^\dagger\ket b\!\bra bU\right),\\
        &=\sum_b
        \raisebox{-1.5cm}{\includegraphics[height=3cm]{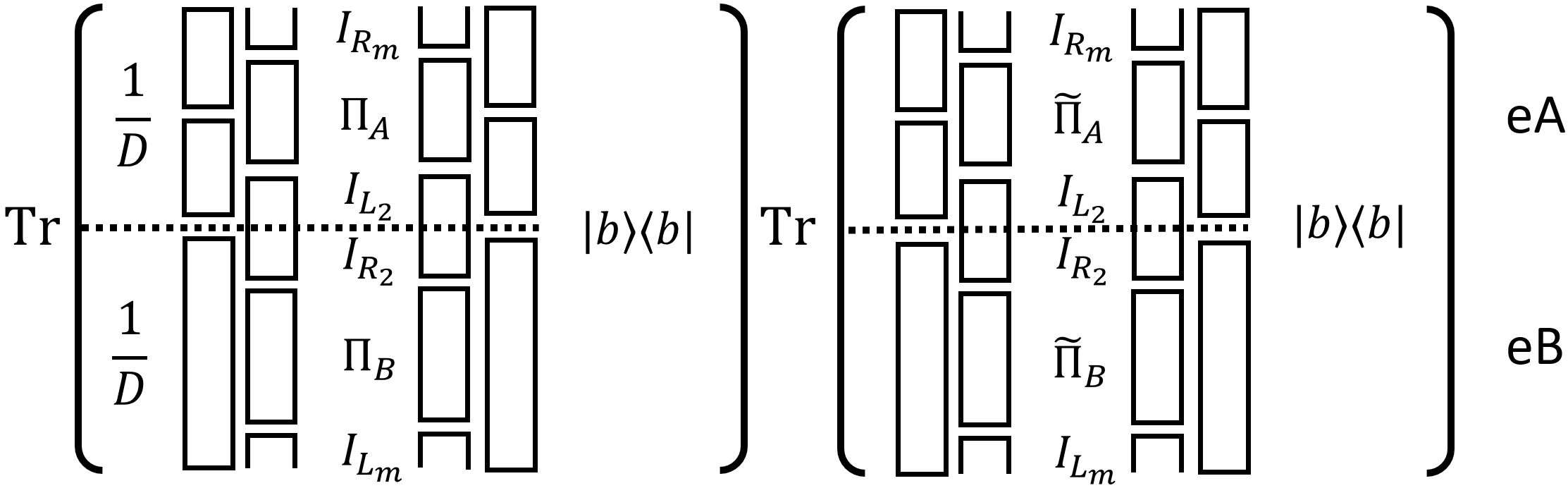}}\\
        &=\sum_{b_{\texttt{eA}}, b_{\texttt{eB}}}
        \raisebox{-1.5cm}{\includegraphics[height=3cm]{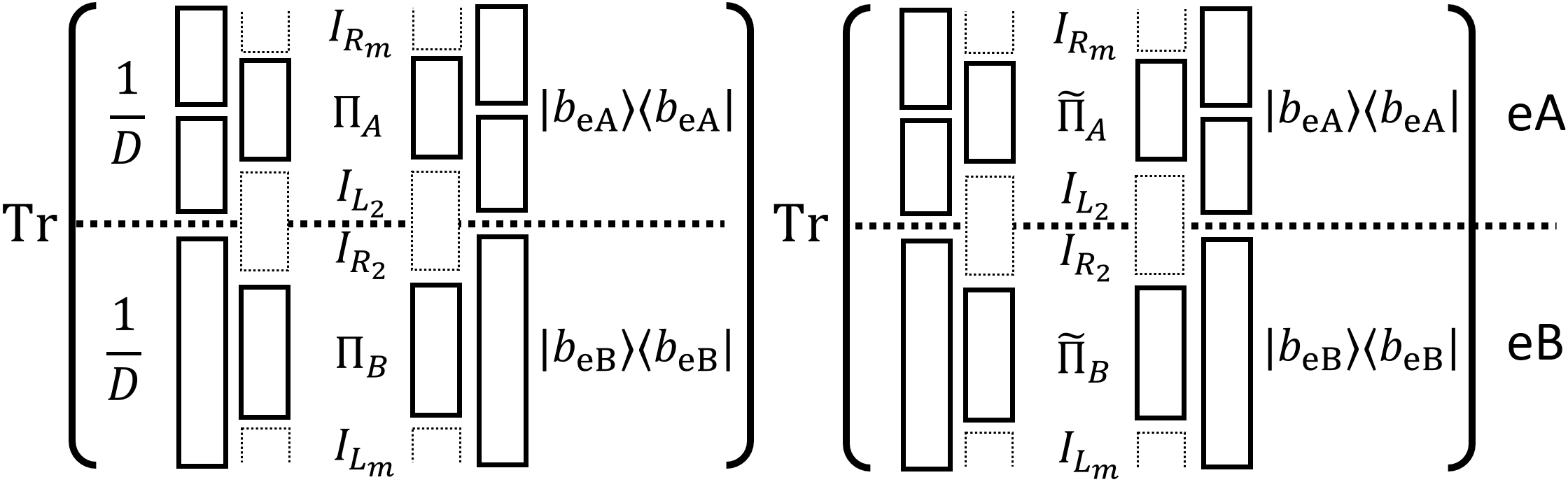}}\\
        &=\sum_{b_{\texttt{eA}}, b_{\texttt{eB}}}
        \raisebox{-1.5cm}{\includegraphics[height=3cm]{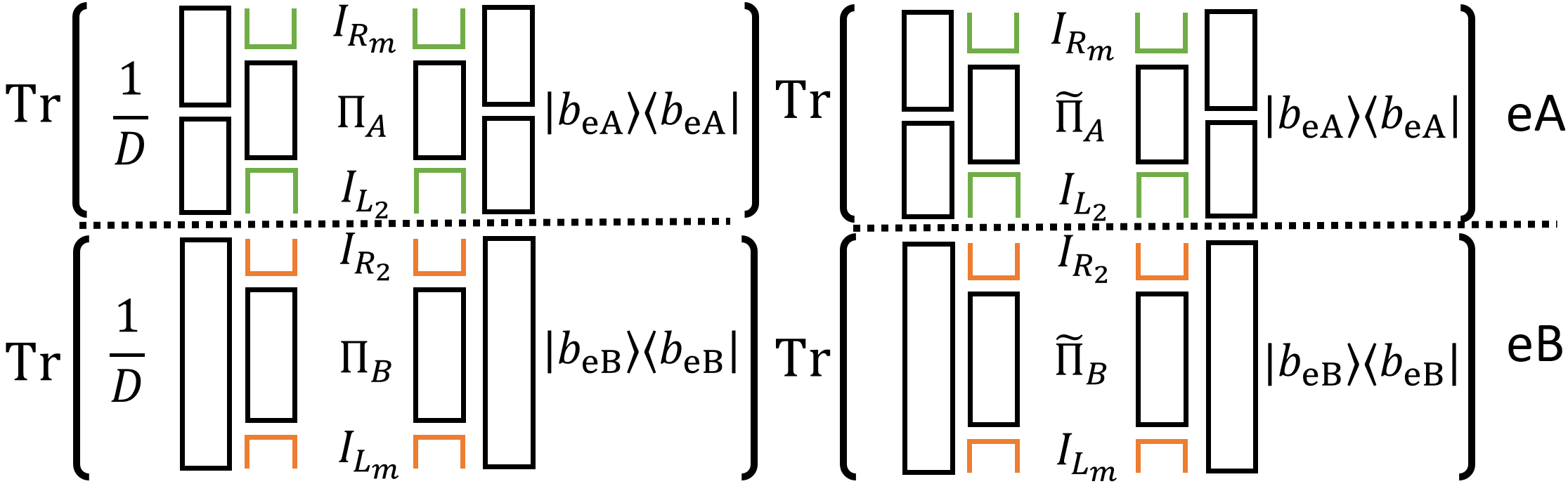}}\label{app:eq:orange_green}\\
        &=
        \sum_{b_{\texttt{eA}}}
        \operatorname{Tr}(U_{\texttt{eA}} \rho_{\texttt{eA}}U_{\texttt{eA}}^{\dagger} \ket{b_{\texttt{eA}}}\!\bra{b_{\texttt{eA}}})
        \operatorname{Tr}(\mathcal{M}_{\texttt{eA}}^{-1}(O_{\texttt{eA}})U_{\texttt{eA}}^{\dagger} \ket{b_{\texttt{eA}}}\!\bra{b_{\texttt{eA}}}U_{\texttt{eA}})\\
        &\nonumber \qquad \times
        \sum_{b_{\texttt{eB}}}
        \operatorname{Tr}(U_{\texttt{eB}} \rho_{\texttt{eB}}U_{\texttt{eB}}^{\dagger} \ket{b_{\texttt{eB}}}\!\bra{b_{\texttt{eB}}})
        \operatorname{Tr}(\mathcal{M}_{\texttt{eB}}^{-1}(O_{\texttt{eB}})U_{\texttt{eB}}^{\dagger} \ket{b_{\texttt{eB}}}\!\bra{b_{\texttt{eB}}}U_{\texttt{eB}})
        \\
        &=
        X_{\texttt{eA}} X_{\texttt{eB}},
    \end{align}
    where, for $R\in\{\mathsf{eA},\mathsf{eB}\}$, $\mathcal M_R$ is the shadow channel of the periodic two-layer unitary ensemble on $R$, and
    \begin{align}
        &X_{R}
        =
        \sum_{b_{R}}
        \operatorname{Tr}(U_{R} \rho_{R}U_{R}^{\dagger} \ket{b_{R}}\!\bra{b_{R}})
        \operatorname{Tr}(\mathcal{M}_{R}^{-1}(O_{R})U_{R}^{\dagger} \ket{b_{R}}\!\bra{b_{R}}U_{R}),\\
        &\rho_{\texttt{eA}}=\frac{1}{D}O_{\texttt{eA}} = \frac{1}{D}I_{R_m}\otimes \Pi_A \otimes I_{L_2},\\
        &\rho_{\texttt{eB}}=\frac{1}{D}O_{\texttt{eB}} = \frac{1}{D}I_{R_2}\otimes \Pi_B \otimes I_{L_m}.
    \end{align}
    The unitary blocks connecting $\mathsf{eA}$ and $\mathsf{eB}$ cancel because the adjacent state and observable factors are identities.
    Independent auxiliary random unitary blocks can therefore be inserted to impose periodic boundary conditions on the two regions, shown as the green and orange blocks in Eq.~\eqref{app:eq:orange_green}. Each part is then described by its local shadow channel, $\mathcal M_{\mathsf{eA}}$ or $\mathcal M_{\mathsf{eB}}$. Thus, $X_{\mathsf{eA}}$ and $X_{\mathsf{eB}}$ are the local unbiased shadow estimators on the two regions. 
    Although this auxiliary circuit differs from the original two-layer circuit, it preserves the estimator statistics for the chosen $\rho$ and $O$ and makes the two random variables independent. Their expectations are therefore
    \begin{align}
        &\mathbb E_{U_{\texttt{eA}}}[X_{\mathsf{eA}}]
        =\operatorname{Tr}(O_{\texttt{eA}}\rho_{\texttt{eA}})=1
        ,\\
        &\mathbb E_{U_{\texttt{eB}}}[X_{\mathsf{eB}}]
        =\operatorname{Tr}(O_{\texttt{eB}}\rho_{\texttt{eB}})=1.
    \end{align}
    Hence
    \begin{align}
        V_{\mathrm{floor}}^{2\mathrm L}
        &=
        \operatorname{Var}
        \left[
            X_{\mathsf{eA}}X_{\mathsf{eB}}
        \right]
        \nonumber\\
        &=
        \operatorname{Var}[X_{\mathsf{eA}}]
        +
        \operatorname{Var}[X_{\mathsf{eB}}]
        +
        \operatorname{Var}[X_{\mathsf{eA}}]
        \operatorname{Var}[X_{\mathsf{eB}}]
        \nonumber\\
        &\geq
        \operatorname{Var}[X_{\mathsf{eA}}].
        \label{app:eq:variance_floor_extended_region_reduction}
    \end{align}

    It remains to analyze the smaller region $\mathsf{eA}$. Let $U_{\texttt{eA}}=VW$ where $W$ denote the first-layer Haar gate on $\mathsf A$, and $V$ denote
    the two adjacent second-layer Haar gates. Define
    \begin{align}
        p_W
        &:=
        \operatorname{Tr}
        \left[
            \left(
                \operatorname{Tr}_{R_1}\!(W\Pi_{\texttt{A}}W^\dagger)
            \right)^2
        \right].
        \label{app:eq:local_reduced_purity}
    \end{align}
    Then Lemma~\ref{lem:two_block_haar_average} gives
    \begin{align}
        \mathbb E_V[ X_{\mathsf{eA}}\mid W]
        &=
        \mathbb{E}_V 
        \raisebox{-0.68cm}{\includegraphics[height=1.5cm]{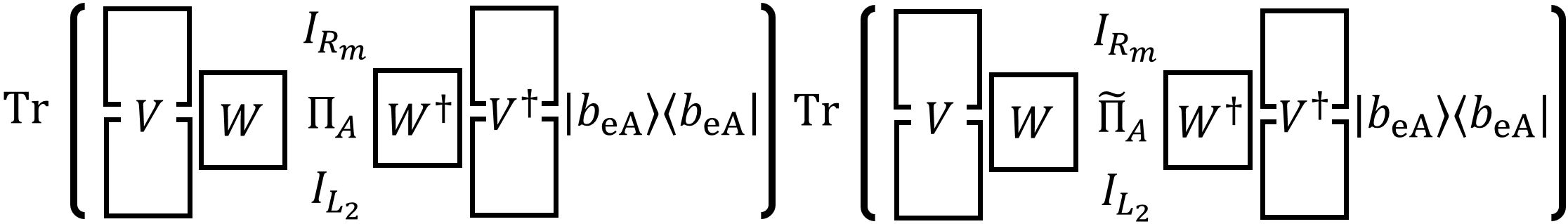}}
        \\
        &=
        \frac{2\sqrt D(D+1)}{3D+1}p_W
        +
        (\mathrm{const.}).
        \label{app:eq:conditional_extended_region_mean}
    \end{align}
    From Fact~\ref{appx:fact:purity_var}, 
    \begin{equation}
        \operatorname{Var}_W[p_W]=\frac{2(D-1)^2}{(D+1)^2(D+2)(D+3)}.
    \end{equation}
    
    Then the law of total variance and Eq.~\eqref{app:eq:conditional_extended_region_mean} imply
    \begin{align}
        V_{\mathrm{floor}}^{2\mathrm L}
        &\geq
        \operatorname{Var}_{V,W}[X_{\mathsf{eA}}]\\
        &=
        \operatorname{Var}_W[\mathbb{E}_{V}[X_{\mathsf{eA}} | W]]
        +
        \mathbb{E}_W[\operatorname{Var}_V[X_{\mathsf{eA}}|W]]\\
        &\geq
        \operatorname{Var}_W[\mathbb{E}_{V}[X_{\mathsf{eA}} | W]]\\
        &=
        \left(
            \frac{
                2\sqrt D(D+1)
            }{
                3D+1
            }
        \right)^2
        \operatorname{Var}_W[p_W]\\
        &=
        \left(
            \frac{
                2\sqrt D(D+1)
            }{
                3D+1
            }
        \right)^2
        \frac{
            2(D-1)^2
        }{
            (D+1)^2(D+2)(D+3)
        }\\
        &=
        \frac{
            8D(D-1)^2
        }{
            (3D+1)^2(D+2)(D+3)
        }
        >
        \frac{1}{7D}
        \label{app:eq:two_layer_floor_lower_bound}
    \end{align}
    where the last inequality holds for $D\geq4$. We finally compare this lower bound with the global exact unitary $4$-design value. Since $O$ is a rank-$D^2$ projector and the total Hilbert-space dimension is $D^m$, $\lVert O_0\rVert_2^2=D^2-D^{4-m}\leq D^2$. Lemma~\ref{app:lem:exact_four_design_multishot_variance} gives
    \begin{equation}
        V_{\mathrm{floor}}^{\mathrm{Haar}}
        \leq
        \frac{4}{D^m}
        \lVert O_0\rVert_2^2
        \leq
        4D^{2-m}.
        \label{app:eq:exact_design_counterexample_floor}
    \end{equation}
    Combining Eqs.~\eqref{app:eq:two_layer_floor_lower_bound} and~\eqref{app:eq:exact_design_counterexample_floor}, we obtain
    \begin{align}
        \frac{
            V_{\mathrm{floor}}^{2\mathrm L}
        }{
            V_{\mathrm{floor}}^{\mathrm{Haar}}
        }
        &>
        \frac{
            (7D)^{-1}
        }{
            4D^{2-m}
        }
        =
        \frac{D^{m-3}}{28}\\
        &=
        \frac{1}{28}
        2^{n-3k}
        =
        \Omega
        \left(
            2^{n-3k}
        \right).
    \end{align}
    Thus, matching the variance scaling of a global exact unitary $4$-design requires a block size $k=\Omega(n)$.
\end{proof}

\section{Proof of Theorem~\ref{main:thm:3}}
Let $\mathcal E_L^{\mathrm{all}}$ denote the ensemble of depth-$L$ all-to-all random two-qubit circuits on $n$ qubits. Each layer contains independent Haar-random two-qubit gates acting on an arbitrary set of disjoint qubit pairs. The interaction pattern can be fixed in advance and need not cover every qubit in each layer; we assume only that every qubit participates in at least one gate over the full circuit. We take the two-qubit gates in $\mathrm{SU}(4)$, since global phases do not affect the measurement statistics.

The counterexample in the previous section uses the two-layer structure to express $g_U$ in Eq.~\eqref{app:eq:g_U} as a product of two independent random variables. This separation cannot be assumed for all-to-all circuits, where the interaction light cone of a single qubit can cover the entire system within depth $O(\log n)$. We instead consider the event $E$ that every gate acting on a fixed qubit is close to the identity. Conditioned on the event $E$, the circuit is close to one in which that qubit is separated from the remaining system. We show that $\Pr(E)\geq \exp[-O(L\log L)]$, which gives a lower bound on the variance floor. First, we use the following fact about the Haar measure.

\begin{fact}[\cite{oszmaniec2021epsilon}]\label{app:fact:haar_small_ball}
    There exists a constant $a>0$ such that, for
    $U\sim\operatorname{Haar}(SU(4))$ and $0<R\leq 1$,
    \begin{equation}
        \Pr_U\!\left[
            \|U-I\|_\infty\leq R
        \right]
        \geq
        aR^{15}.
    \end{equation}
\end{fact}

\begin{theorem}[Multi-shot depth lower bound]
\label{app:thm:all_to_all_depth_lower_bound}
    Consider the ancilla-free all-to-all random two-qubit circuit ensemble $\mathcal E_L^{\mathrm{all}}$ on $n$ qubits with depth $L$.  There exists a state and a Pauli observable such that
    \begin{equation}
        \frac
        {V_{\mathrm{floor}}^{\mathrm{all},L}}
        {V_{\mathrm{floor}}^{\mathrm{Haar}}}
        \geq
        d\left(\frac{c_1}{L}\right)^{c_2L}
    \end{equation}
    for constants $c_1,c_2>0$. Consequently, matching the global exact unitary $4$-design variance floor within a constant factor uniformly over all states and observables requires $L\log L=\Omega(n)$, and hence $L=\Omega(n/\log n)$.
\end{theorem}

\begin{proof}
    Define
    \begin{equation}
        P=Z_1,
        \qquad
        \rho=\frac{I+P}{d}.
    \end{equation}
    Since every qubit is acted on by an independent Haar-random two-qubit gate, $\mathcal E_L^{\mathrm{all}}$ is Pauli invariant. By Fact~\ref{app:fact:m_P(U)},
    \begin{equation}
        \mathcal M_{\mathcal E_L^{\mathrm{all}}}(P)
        =
        m_P P,
        \qquad
        m_P
        =
        \mathbb E_U[m_P(U)],
        \qquad
        m_P(U)
        =
        \frac{1}{d}
        \sum_b
        \bra{b}UPU^\dagger\ket{b}^2.
    \end{equation}
    
    For each $U$ and $b$, define the reduced state $\sigma_{U,b}:=\operatorname{Tr}_{2,\ldots,n}\left(U^\dagger\ket{b}\!\bra{b}U\right)$.
    By the isotropy of the Haar-random gate acting on the first qubit, $m_{X_1}=m_{Y_1}=m_{Z_1}$. Therefore,
    \begin{align}
        m_P
        &=
        \frac{1}{d}
        \mathbb E_U
        \sum_b
        \operatorname{Tr}(\sigma_{U,b} Z)^2\\
        &=
        \frac{1}{3d}
        \mathbb E_U
        \sum_b
        \bigl(
        \operatorname{Tr}(\sigma_{U,b} X)^2+\operatorname{Tr}(\sigma_{U,b} Y)^2+\operatorname{Tr}(\sigma_{U,b} Z)^2
        \bigr)\\
        &=
        \frac{1}{3d}
        \mathbb E_U
        \sum_b
        \left(
        2\operatorname{Tr}(\sigma_{U,b}^2)-1
        \right)\leq
        \frac13 \label{app:eq:mean_visibility_upper_bound},
    \end{align}
    
        Let $E$ be the event that all two-qubit gates $\{U_i\}_i$ acting on the
    first qubit satisfy
    \begin{equation}
        \|U_i-I\|_\infty
        \leq
        R.
    \end{equation}
    There are at most $L$ such gates. Their independence and
    Fact~\ref{app:fact:haar_small_ball} imply
    \begin{equation}
        \Pr(E)
        \geq
        \left(aR^{15}\right)^L.
    \end{equation}
    
    Conditioned on the event $E$, replace every gate acting on the first qubit by the identity and denote the circuit obtained in this way by $U_0$. The circuits $U$ and $U_0$ differ at at most $L$ gates, and the operator-norm difference at each gate is at most $R$. No gate in $U_0$ acts on the first qubit, so $m_P(U_0)=1\{Z_1 \in \pm \mathcal Z\}=1$. Telescoping the circuit products gives
    \begin{equation}
        \|U-U_0\|_\infty
        \leq
        LR.
    \end{equation}
    
    By Lemma~\ref{app:lem:m_P_L_bound}, taking $R=1/(16L)$ gives,
    conditioned on the event $E$,
    \begin{align}
        m_P(U)
        &\geq
        m_P(U_0)-|m_P(U)-m_P(U_0)|
        \\
        &\geq
        m_P(U_0)-4\|U-U_0\|_\infty
        \\
        &\geq
        1-4LR
        =
        \frac{3}{4} \label{app:eq:conditional_visibility_lower_bound}.
    \end{align}
        For the chosen state and observable, Eq.~\eqref{app:eq:variance_floor}
    and $\sum_b\bra{b}UPU^\dagger\ket{b}=0$ give
    \begin{align}
        V_{\mathrm{floor}}^{\mathrm{all},L}
        &=
        \operatorname{Var}_{U}
        \left[
        \sum_{b} \operatorname{Tr}(U\rho U^{\dagger}\ket{b}\!\bra{b})\operatorname{Tr}(P\mathcal{M}^{-1}(U^{\dagger}\ket{b}\!\bra{b}U)) 
        \right]
        \\
        &=
        \operatorname{Var}_U
        \left[
            \frac{1}{dm_P}
            \sum_b
            \left(
                1+\bra{b}UPU^\dagger\ket{b}
            \right)
            \bra{b}UPU^\dagger\ket{b}
        \right]
        \\
        &=
        m_P^{-2}\operatorname{Var}_U[m_P(U)].
    \end{align}
    By Eq.~\eqref{app:eq:mean_visibility_upper_bound}, $m_P\leq1/3$, while Eq.~\eqref{app:eq:conditional_visibility_lower_bound} gives $m_P(U)\geq3/4$ on the event $E$. Therefore,
    \begin{align}
        V_{\mathrm{floor}}^{\mathrm{all},L}
        &=
        m_P^{-2}
        \mathbb E_U
        \left[
            \left(m_P(U)-m_P\right)^2
        \right]
        \\
        &\geq
        m_P^{-2}
        \Pr(E)\,
        \mathbb E_U\!
        \left[
            \left(m_P(U)-m_P\right)^2 \mid E \,
        \right]
        \\
        &\geq
        \Pr(E)m_P^{-2}
        \left(
            \frac{3}{4}-m_P
        \right)^2
        \\
        &\geq
        \frac{25}{16}\Pr(E)
        \geq
        \frac{25}{16}
        \left(
            \frac{a^{1/15}}{16L}
        \right)^{15L}.
    \end{align}
    For the same state and observable, the Haar fourth-moment
    identity~\cite{anshu2022distributed} gives
    \begin{equation}
        V_{\mathrm{floor}}^{\mathrm{Haar}}
        =
        \frac{2(d+2)}{d(d+3)}
        \leq
        \frac{2}{d}.
    \end{equation}
    Consequently,
    \begin{align}
        \frac{
            V_{\mathrm{floor}}^{\mathrm{all},L}
        }{
            V_{\mathrm{floor}}^{\mathrm{Haar}}
        }
        &\geq
        \frac{25d}{32}
        \left(
            \frac{a^{1/15}}{16L}
        \right)^{15L}
        \geq
        d\left(\frac{c_1}{L}\right)^{c_2L},
    \end{align}
    for a sufficiently small constant $c_1>0$ and $c_2=15$. Thus, the two variance floors can differ by at most a constant factor only if $L\log L=\Omega(n)$, which implies $L=\Omega(n/\log n)$.
\end{proof}

\section{Learning guarantee with approximate strong unitary designs}
\label{app:sec:strong_design_learning}

For a unitary ensemble $\mathcal E$, its $t$-th moment channel is
\begin{equation}
    \Phi_{\mathcal E}^{(t)}(X)
    :=
    \mathbb E_{U\sim\mathcal E}
    \left[
        U^{\otimes t}
        X
        U^{\dagger,\otimes t}
    \right].
\end{equation}
This standard moment channel fixes the same orientation of the unitary $U$ in all $t$ copies. To include moments with different orientations, let $p$ and $q$ be nonnegative integers satisfying $p+q=t$ and define
\begin{equation}
    \Phi_{\mathcal E}^{(p,q)}(X)
    :=
    \mathbb E_{U\sim\mathcal E}
    \left[
        \left(
            U^{\otimes p}\otimes U^{*,\otimes q}
        \right)
        X
        \left(
            U^{\dagger,\otimes p}\otimes U^{T,\otimes q}
        \right)
    \right].
    \label{app:eq:mixed_moment_channel}
\end{equation}
We call $\mathcal E$ an exact strong unitary $t$-design if $\Phi_{\mathcal E}^{(p,q)}=\Phi_{\mathrm H}^{(p,q)}$ for every $p+q=t$. The case $q=0$ reduces to the ordinary $t$-th moment channel.

There are several ways to define an approximate strong design~\cite{schuster2025strong}. In additive error, one requires
\begin{equation}
    \left\|
        \Phi_{\mathcal E}^{(p,q)}
        -
        \Phi_{\mathrm H}^{(p,q)}
    \right\|_{\diamond}
    \leq
    \epsilon_{\mathrm{des}}
\end{equation}
for every $p+q=t$. This gives an additive bound on the output of the mixed moment channel. Measurable error instead requires the averaged output of every experiment making at most $t$ queries chosen from $U$, $U^\dagger$, $U^T$, and $U^*$ to be within $\epsilon_{\mathrm{des}}$ in trace norm of the corresponding Haar output.

In this work, we use relative error. An ensemble $\mathcal E$ is an $\epsilon_{\mathrm{des}}$-approximate strong unitary $t$-design in relative error if, for every $p+q=t$,
\begin{equation}
    (1-\epsilon_{\mathrm{des}})
    \Phi_{\mathrm H}^{(p,q)}
    \preceq
    \Phi_{\mathcal E}^{(p,q)}
    \preceq
    (1+\epsilon_{\mathrm{des}})
    \Phi_{\mathrm H}^{(p,q)}.
    \label{app:eq:relative_strong_design}
\end{equation}
Here, $\Phi_1\preceq\Phi_2$ means that $\Phi_2-\Phi_1$ is completely positive. This condition gives multiplicative control whenever the mixed moment channel is contracted with positive inputs and positive measurement operators. 

The relative error approximate strong unitary design condition has a substantially larger depth cost. Ordinary relative error approximate unitary $t$-designs can be implemented without ancillary qubits at depth
\begin{equation}
    \mathrm{Approximate~designs:}\qquad
    \mathrm{depth}
    =
    O\!\left(
    t\,\operatorname{poly}\log t\,
    \log(n/\epsilon_{\mathrm{des}})
    \right).
\end{equation}
For approximate strong unitary $t$-designs in relative error, 1D local random circuits achieve
\begin{equation}
    \mathrm{Approximate~strong~designs:}\qquad
    \mathrm{depth}
    =
    O\!\left(
    \log^7(t)
    \left(
    nt+\log(1/\epsilon_{\mathrm{des}})
    \right)
    \right),
\end{equation}
for $\epsilon_{\mathrm{des}}\geq 2t^2/2^n$~\cite{schuster2025strong}. Conversely, for $\epsilon_{\mathrm{des}}<1/4$, every 1D circuit satisfying Eq.~\eqref{app:eq:relative_strong_design} has depth $\Omega(n+t/\log(nt))$, even with ancillary qubits. Thus, for fixed $t$ and constant $\epsilon_{\mathrm{des}}$, the required depth is $\Theta(n)$ in 1D~\cite{schuster2025strong}.

\subsection{Variance comparison for ensemble-independent post-processing}

We first assume that the post-processing map $\mathcal C$ is independent of the unitary ensemble. The same map is therefore used for both $\mathcal E$ and the Haar ensemble. For a measurement outcome $\mathbf{b}=\{b_i\}_{i=1}^m$, consider the real estimator
\begin{align}
    \hat o(U,\mathbf{b})
    &:=
    \operatorname{Tr}\!\left[
        O\,
        \mathcal C\!\left(
            U^{(p,q)}BU^{(p,q),\dagger}
        \right)
    \right]
    \nonumber\\
    &=
    \operatorname{Tr}\!\left[
        A\,
        U^{(p,q)}BU^{(p,q),\dagger}
    \right],
    \label{app:eq:ensemble_independent_estimator}
\end{align}
where $A:=\mathcal C^\dagger(O)$ and $U^{(p,q)}:=U^{\otimes p}\otimes U^{*,\otimes q}$. The matrix $B$ may depend on the measurement outcomes $\{b_i\}_{i=1}^m$.

We use the vectorization convention $\langle\!\langle X|Y\rangle\!\rangle=\operatorname{Tr}(X^\dagger Y)$ and define
\begin{equation}
    R(U)
    :=
    U^{(p,q)}
    \otimes
    U^{(p,q),*}.
    \label{app:eq:vectorized_unitary_action}
\end{equation}
Then, $|U^{(p,q)}BU^{(p,q),\dagger}\rangle\!\rangle=R(U)|B\rangle\!\rangle$ and $\hat o=\langle\!\langle A^\dagger|R(U)|B\rangle\!\rangle$. To apply the relative error comparison directly to the variance $\operatorname{Var}[\hat o]$, rather than only to the second moment $\mathbb E[\hat o^2]$, we introduce the normalized vectorized identity
\begin{equation}
    |e\rangle\!\rangle
    :=
    \frac{|I\rangle\!\rangle}
    {\sqrt{\operatorname{Tr}(I)}} ,
    \qquad
    R(U)|e\rangle\!\rangle
    =
    |e\rangle\!\rangle.
    \label{app:eq:invariant_identity_vector}
\end{equation}
Here, $I$ is the identity acting on the same space as $B$. Introduce a two-dimensional auxiliary system with orthonormal states $|0\rangle$ and $|1\rangle$, and set
\begin{align}
    |\widetilde B\rangle\!\rangle
    &:=
    |B\rangle\!\rangle\otimes|0\rangle
    +
    |e\rangle\!\rangle\otimes|1\rangle,
    \nonumber\\
    |\widetilde A_c\rangle\!\rangle
    &:=
    |A^\dagger\rangle\!\rangle\otimes|0\rangle
    -
    c|e\rangle\!\rangle\otimes|1\rangle.
    \label{app:eq:augmented_vectors}
\end{align}
This auxiliary label is used only in the proof and is not an ancillary qubit in the measurement circuit. Orthogonality of $|0\rangle$ and $|1\rangle$ gives
\begin{equation}
    \langle\!\langle\widetilde A_c|
    \left(R(U)\otimes I_2\right)
    |\widetilde B\rangle\!\rangle
    =
    \hat o(U,\mathbf{b})-c.
    \label{app:eq:centered_estimator_amplitude}
\end{equation}

For $\nu\in\{\mathcal E,\mathrm H\}$, define the squared loss
\begin{equation}
    L_\nu(c)
    :=
    \mathbb E_{U\sim\nu}
    \sum_{b_1 \dots b_m}
    \prod_{i=1}^m \langle b_i|U\rho U^\dagger|b_i\rangle
    \bigl(\hat o(U,\mathbf{b})-c\bigr)^2.
    \label{app:eq:squared_loss}
\end{equation}
Using Eq.~\eqref{app:eq:centered_estimator_amplitude}, each term in this average can be written as
\begin{align}
    &\prod_{i=1}^m \langle b_i|U\rho U^\dagger|b_i\rangle
    \left|
        \langle\!\langle\widetilde A_c|
        \left(R(U)\otimes I_2\right)
        |\widetilde B\rangle\!\rangle
    \right|^2
    \nonumber\\
    &=
    \operatorname{Tr}\!\Big[
        \left(
            \bigotimes_{i=1}^m |b_i\rangle\!\langle b_i|
            \otimes
            |\widetilde A_c\rangle\!\rangle
            \langle\!\langle\widetilde A_c|
        \right)
        \left(
            U^{\otimes m}\otimes R(U)\otimes I_2
        \right)
        \left(
            \rho^{\otimes m}
            \otimes
            |\widetilde B\rangle\!\rangle
            \langle\!\langle\widetilde B|
        \right)
        \left(
            U^{\dagger,\otimes m}\otimes R(U)^\dagger\otimes I_2
        \right)
    \Big].
    \label{app:eq:positive_squared_loss}
\end{align}
Both
\begin{equation}
    \rho^{\otimes m}\otimes
    |\widetilde B\rangle\!\rangle
    \langle\!\langle\widetilde B|
    \quad\text{and}\quad
    \bigotimes_{i=1}^m |b_i\rangle\!\langle b_i|\otimes
    |\widetilde A_c\rangle\!\rangle
    \langle\!\langle\widetilde A_c|
\end{equation}
are positive semidefinite. Up to a reordering of tensor factors, $R(U)$ contains $p+q$ copies of $U$ and $p+q$ copies of $U^*$. Including the copy of $U^{\otimes m}$ from the Born probability, Eq.~\eqref{app:eq:positive_squared_loss} is governed by the mixed moment channel $\Phi_\nu^{(p+q+m,p+q)}$. Thus, the relative CP-order bounds for $\Phi_{\mathcal E}^{(p+q+m,p+q)}$ give, for every real $c$,
\begin{equation}
    (1-\epsilon_{\mathrm{des}})L_{\mathrm H}(c)
    \leq
    L_{\mathcal E}(c)
    \leq
    (1+\epsilon_{\mathrm{des}})L_{\mathrm H}(c).
    \label{app:eq:relative_squared_loss}
\end{equation}

We next relate the squared loss to the variance.

\begin{lemma}[Variance as the minimum squared loss]\label{appx:lem:c}
    Let $Z$ be a real random variable with a finite second moment. Then
    \begin{equation}
        \operatorname{argmin}_{c\in\mathbb R}
        \mathbb E[(Z-c)^2]
        =
        \mathbb E[Z].
    \end{equation}
\end{lemma}
\begin{proof}
    Expanding around the mean gives
    \begin{align}
        \mathbb E[(Z-c)^2]
        &=
        \mathbb E\!\left[
            \bigl(
                Z-\mathbb E[Z]
                +
                \mathbb E[Z]-c
            \bigr)^2
        \right]
        \nonumber\\
        &=
        \operatorname{Var}(Z)
        +
        \bigl(c-\mathbb E[Z]\bigr)^2,
    \end{align}
    where its minimum occurs at $c=\mathbb E[Z]$.
\end{proof}

Define $c_{\mathcal E}:=\mathbb E_{U\sim\mathcal E, \mathbf{b}}[\hat o],c_{\mathrm H}:=\mathbb E_{U\sim\mathrm H, \mathbf{b}}[\hat o]$. Lemma~\ref{appx:lem:c} gives
\begin{equation}
    \operatorname{Var}_{\mathcal E}(\hat o)
    =
    L_{\mathcal E}(c_{\mathcal E}),
    \qquad
    \operatorname{Var}_{\mathrm H}(\hat o)
    =
    L_{\mathrm H}(c_{\mathrm H}).
\end{equation}
Since $c_{\mathcal E}$ minimizes $L_{\mathcal E}$, we may evaluate it at $c_{\mathrm H}$ and then apply Eq.~\eqref{app:eq:relative_squared_loss}:
\begin{align}
    \operatorname{Var}_{\mathcal E}(\hat o)
    &=
    L_{\mathcal E}(c_{\mathcal E})
    \nonumber\\
    &\leq
    L_{\mathcal E}(c_{\mathrm H})
    \nonumber\\
    &\leq
    (1+\epsilon_{\mathrm{des}})
    L_{\mathrm H}(c_{\mathrm H})
    \nonumber\\
    &=
    (1+\epsilon_{\mathrm{des}})
    \operatorname{Var}_{\mathrm H}(\hat o).
    \label{app:eq:relative_variance_upper}
\end{align}
To obtain the lower bound, we first apply Eq.~\eqref{app:eq:relative_squared_loss} at $c=c_{\mathcal E}$. Since $c_{\mathrm H}$ minimizes $L_{\mathrm H}$, we then have
\begin{align}
    \operatorname{Var}_{\mathcal E}(\hat o)
    &=
    L_{\mathcal E}(c_{\mathcal E})
    \nonumber\\
    &\geq
    (1-\epsilon_{\mathrm{des}})
    L_{\mathrm H}(c_{\mathcal E})
    \nonumber\\
    &\geq
    (1-\epsilon_{\mathrm{des}})
    L_{\mathrm H}(c_{\mathrm H})
    \nonumber\\
    &=
    (1-\epsilon_{\mathrm{des}})
    \operatorname{Var}_{\mathrm H}(\hat o).
    \label{app:eq:relative_variance_lower}
\end{align}
Combining the two bounds gives
\begin{equation}
    (1-\epsilon_{\mathrm{des}})
    \operatorname{Var}_{\mathrm H}(\hat o)
    \leq
    \operatorname{Var}_{\mathcal E}(\hat o)
    \leq
    (1+\epsilon_{\mathrm{des}})
    \operatorname{Var}_{\mathrm H}(\hat o).
    \label{app:eq:relative_variance_bound}
\end{equation}

\subsection{Ensemble-dependent post-processing}

We now allow the post-processing map $\mathcal C_{\mathcal E}$ to depend on the unitary ensemble. Once $\mathcal E$ is fixed, $\mathcal C_{\mathcal E}$ is independent of the individual unitary $U$ and outcome $b$. The upper bound in the relative variance comparison of Eq.~\eqref{app:eq:relative_variance_bound} gives
\begin{equation}
    \operatorname{Var}_{\mathcal E}(\hat o_{\mathcal E})
    \leq
    (1+\epsilon_{\mathrm{des}})
    \operatorname{Var}_{\mathrm H}(\hat o_{\mathcal E}),
    \label{app:eq:ensemble_dependent_variance_upper}
\end{equation}
where $\hat o_{\mathcal E}=\operatorname{Tr}(O\mathcal C_{\mathcal E}(U^{(p,q)}BU^{(p,q),\dagger}))$.
The Haar-averaged variance on the right-hand side uses $\mathcal C_{\mathcal E}$ rather than $\mathcal C_{\mathrm H}$. It therefore remains to compare $\operatorname{Var}_{\mathrm H}(\hat o_{\mathcal E})$ with $\operatorname{Var}_{\mathrm H}(\hat{o}_{\mathrm H})$. A common choice for $\mathcal C_{\mathcal{E}}$ in quantum state learning is the inverse shadow channel $\mathcal M_{\mathcal E}^{-1}$, which gives an unbiased estimator and satisfies the following bound.

\begin{lemma}[Inverse shadow post-processing]\label{app:lem:inverse_shadow_postprocessing}
    Let $\mathcal E$ be an $\epsilon_{\mathrm{des}}$-approximate strong unitary 2-design with $\epsilon_{\mathrm{des}}<1$ and assume that its distribution is invariant under both left and right multiplication by Pauli operators. Then, for every traceless Hermitian observable $O_0=\sum_{P\neq I}o_P P$,
    \begin{equation}
        \left\|
        \mathcal M_{\mathcal E}^{-1}(O_0)
        \right\|_2^2
        \leq
        \frac{1}{(1-\epsilon_{\mathrm{des}})^2}
        \left\|
        \mathcal M_{\mathrm H}^{-1}(O_0)
        \right\|_2^2.
    \end{equation}
\end{lemma}

\begin{proof}
    For $P\neq I$, by~\cite{bu2024classical} and~\cite{schuster2025strong},
    \begin{align}
        m_P
        &=
        \frac{1}{d}\sum_{b\in\{0,1\}^n}
        \mathbb E_{U\sim\mathcal E}
        \left[
            \operatorname{Tr}\!\left(
                PU^\dagger\ket b\!\bra b U
            \right)^2
        \right]\\
        &=
        \mathbb E_{U\sim\mathcal E}
        \left[
        \operatorname{Tr}\!\left(
        PU^\dagger\ket 0\!\bra 0 U
        \right)^2
        \right]\\
        &=
        \langle\!\langle 0|
        \Phi_{\mathcal E}^{(1,1)}
        \!\left(|P\rangle\!\rangle\langle\!\langle P|
        \right)
        |0\rangle\!\rangle\label{app:eq:m_P_strong}.
    \end{align}
    Here, $|0\rangle\!\rangle:=\ket 0\otimes\ket 0$ denotes the vectorization of $\ket 0\!\bra 0$. Applying the relative error condition to Eq.~\eqref{app:eq:m_P_strong} implies
    \begin{equation}
        (1-\epsilon_{\mathrm{des}})m_P^{\mathrm H}\leq m_P\leq (1+\epsilon_{\mathrm{des}})m_P^{\mathrm H},
    \end{equation}
    where $m_P^{\mathrm{H}}=(d+1)^{-1}$. Since $m_{P}^{\mathrm{H}}/m_P\le (1-\epsilon_{\mathrm {des}})^{-1}$,
    \begin{align}
        \left\|
        \mathcal M_{\mathcal E}^{-1}(O_0)
        \right\|_2^2
        &=
        \sum_{P\neq I} \frac{o_P^2}{m_P^2}2^n
        \le 
        \max_P (m_{P}^{\mathrm{H}}/m_P)^2\left\|
        \mathcal M_{\mathrm H}^{-1}(O_0)
        \right\|_2^2\\
        &\le
        \frac{1}{(1-\epsilon_{\mathrm{des}})^2}
        \left\|
        \mathcal M_{\mathrm H}^{-1}(O_0)
        \right\|_2^2.
    \end{align}
\end{proof}

Under the Pauli-invariance assumptions of Lemma~\ref{app:lem:inverse_shadow_postprocessing}, relative-error approximate strong unitary $3$- or $4$-designs suffice to recover the corresponding exact-design bounds for the applications below, up to the factors specified below.

\paragraph{Single-shot observable estimation.}
For the exact shadow estimator, self-adjointness of the inverse shadow channel gives
\begin{equation}
    \hat o_{\mathcal{E}}
    =
    \operatorname{Tr}\!\left(
    \mathcal M_{\mathcal E}^{-1}(O)
    U^\dagger\ket b\!\bra bU
    \right).
\end{equation}
The Haar third-moment identity~\cite{huang2020predicting} gives
\begin{align}
    \operatorname{Var}_{\mathrm H}(\hat o_{\mathcal{E}})
    \leq
    \frac{3}{(d+1)(d+2)}
    \left\|
    \mathcal M_{\mathcal E}^{-1}(O_0)
    \right\|_2^2.
    \label{app:eq:single_shot_general_frobenius_bound}
\end{align}
Combining Eq.~\eqref{app:eq:ensemble_dependent_variance_upper} with Lemma~\ref{app:lem:inverse_shadow_postprocessing} and $\mathcal M_{\mathrm H}^{-1}(O_0)=(d+1)O_0$ gives
\begin{align}
    \operatorname{Var}_{\mathcal E}(\hat o_{\mathcal{E}})
    &\leq
    (1+\epsilon_{\mathrm{des}})\operatorname{Var}_{\mathrm H}(\hat o_{\mathcal{E}})
    \\
    &\leq
    \frac{3(1+\epsilon_{\mathrm{des}})}{(d+1)(d+2)}
    \left\|
    \mathcal M_{\mathcal E}^{-1}(O_0)
    \right\|_2^2
    \nonumber\\
    &\leq
    \frac{3(1+\epsilon_{\mathrm{des}})}
    {(d+1)(d+2)}
    \frac{1}{(1-\epsilon_{\mathrm{des}})^2}
    \|\mathcal M_{\mathrm H}^{-1}(O_0)\|_2^2
    \nonumber\\
    &\leq
    \frac{3(1+\epsilon_{\mathrm{des}})}
    {(1-\epsilon_{\mathrm{des}})^2}
    \frac{d+1}{d+2}
    \|O_0\|_2^2
    \nonumber\\
    &\leq
    \frac{3(1+\epsilon_{\mathrm{des}})}
    {(1-\epsilon_{\mathrm{des}})^2}
    \|O_0\|_2^2.
    \label{app:eq:single_shot_ensemble_frobenius_bound}
\end{align}
Thus, this recovers the exact global unitary $3$-design bound up to the factor $(1+\epsilon_{\mathrm{des}})/(1-\epsilon_{\mathrm{des}})^2$.

\paragraph{Multi-shot observable estimation.}
When each unitary is reused for $N_S$ measurement shots, we use the multi-shot shadow estimator
\begin{equation}
    \hat o_{\mathcal E}
    =
    \operatorname{Tr}(O\hat{\rho}_{\mathrm{mul}})
    =
    \frac{1}{N_U N_S}
    \sum_{i=1}^{N_U}\sum_{j=1}^{N_S}
    \operatorname{Tr}\!\left[
        O\mathcal M_{\mathcal E}^{-1}\!\left(
            U_i^\dagger
            \ket{b_{i,j}}\!\bra{b_{i,j}}
            U_i
        \right)
    \right],
\end{equation}
where $\hat{\rho}_{\mathrm{mul}}$ is defined in Eq.~\eqref{app:eq:multishot_shadow_estimator} with $\mathcal M^{-1}=\mathcal M_{\mathcal E}^{-1}$. The Haar variance with this fixed inverse shadow channel takes the form
\begin{equation}
    \operatorname{Var}_{\mathrm H}(\hat o_{\mathcal{E}})
    =
    \frac{1}{N_U}
    \left(
    \frac{
    V_{\mathrm{shot},\mathrm H}
    \bigl(\mathcal M_{\mathcal E}^{-1}(O_0)\bigr)
    }{
    N_S
    }
    +
    \frac{N_S-1}{N_S}
    V_{\mathrm{floor},\mathrm H}
    \bigl(\mathcal M_{\mathcal E}^{-1}(O_0)\bigr)
    \right).
    \label{app:eq:multishot_general_postprocessing}
\end{equation}
The third- and fourth-order Haar identities give the following bounds for the shot and floor terms, with the corresponding variance estimates governed by $\Phi_{\mathcal E}^{(2,1)}(\cdot)$ and $\Phi_{\mathcal E}^{(2,2)}(\cdot)$, respectively:
\begin{align}
    V_{\mathrm{shot},\mathrm H}
    \bigl(\mathcal M_{\mathcal E}^{-1}(O_0)\bigr)
    &\leq
    \frac{3}{(d+1)(d+2)}
    \left\|
    \mathcal M_{\mathcal E}^{-1}(O_0)
    \right\|_2^2,
    \nonumber\\
    V_{\mathrm{floor},\mathrm H}
    \bigl(\mathcal M_{\mathcal E}^{-1}(O_0)\bigr)
    &\leq
    \frac{
    4d^2+9d+7
    }{
    d(d+1)^2(d+2)(d+3)
    }
    \left\|
    \mathcal M_{\mathcal E}^{-1}(O_0)
    \right\|_2^2.
    \label{app:eq:multishot_general_frobenius_bounds}
\end{align}
Combining these bounds with
Eq.~\eqref{app:eq:ensemble_dependent_variance_upper} and
Lemma~\ref{app:lem:inverse_shadow_postprocessing} gives
\begin{align}
    \operatorname{Var}_{\mathcal E}(\hat o_{\mathcal{E}})
    &\leq
    \frac{1+\epsilon_{\mathrm{des}}}
    {(1-\epsilon_{\mathrm{des}})^2N_U}
    \left(
    \frac{3}{N_S}
    +
    \frac{4(N_S-1)}{dN_S}
    \right)
    \|O_0\|_2^2.
    \label{app:eq:multishot_ensemble_frobenius_bound}
\end{align}
Thus, both terms retain the exact unitary $4$-design scaling up to the
factor $(1+\epsilon_{\mathrm{des}})/(1-\epsilon_{\mathrm{des}})^2$.

\paragraph{Fourth moment of a randomized matrix element.} The following fourth-moment estimate is the key ingredient used in Ref.~\cite{brandao2020fast} to lower-bound the expected $\ell_1$ distance between measurement outcome distributions via Berger's inequality. Let $X$ be traceless and Hermitian, and let $b$ be fixed. The fourth pure state moment gives
\begin{align}
    \mathbb E_{U\sim\mathrm H}
    \left[
        \operatorname{Tr}\!\left(
            XU^\dagger\ket b\!\bra bU
        \right)^4
    \right]
    &=
    \frac{
        3\operatorname{Tr}(X^2)^2
        +
        6\operatorname{Tr}(X^4)
    }{
        d(d+1)(d+2)(d+3)
    }
    \nonumber\\
    &\leq
    \frac{
        9
    }{
        d(d+1)(d+2)(d+3)
    }
    \|X\|_2^4.
    \label{app:eq:fourth_moment_frobenius_bound}
\end{align}
The last inequality follows from $\operatorname{Tr}(X^4)\leq\operatorname{Tr}(X^2)^2=\|X\|_2^4$. Here, no post-processing is applied, so the post-processing
map is the identity and is therefore ensemble independent. Consequently, a relative error approximate strong unitary $4$-design satisfies
\begin{align}
    \mathbb E_{U\sim\mathcal E}
    \left[
        \operatorname{Tr}\!\left(
            XU^\dagger\ket b\!\bra bU
        \right)^4
    \right]\leq
    (1+\epsilon_{\mathrm{des}})
    \frac{
        9
    }{
        d(d+1)(d+2)(d+3)
    }
    \|X\|_2^4.
    \label{app:eq:fourth_moment_ensemble_frobenius_bound}
\end{align}

\section{State-learning applications of the Frobenius-norm variance bound}
\label{app:sec:state_learning_applications}

We now apply Theorem~\ref{main:thm:1} to setting-efficient quantum state tomography, multiparameter quantum metrology, stabilizer structure learning, and unbiased purity estimation. This bound allows the corresponding global Clifford guarantees to be retained with shallow measurements. The observable and state used in each reduction depend on the learning task. We denote those to which Theorem~\ref{main:thm:1} is applied by $O_{\mathrm{thm}}$ and $\rho_{\mathrm{thm}}$, respectively. In some cases, $\rho_{\mathrm{thm}}$ differs from the state being learned.

\subsection{Setting-efficient quantum state tomography}
\label{app:sec:setting_efficient_qst}

Given an unknown state $\rho$ with $\operatorname{rank}(\rho)\leq r$, quantum state tomography aims to construct an estimator $\hat{\rho}$ satisfying $\|\hat{\rho}-\rho\|_1\leq\epsilon$. We use $N_U$ independently sampled measurement bases and perform $N_S$ shots in each basis, so that the total number of samples is $T=N_UN_S$. Any non-adaptive single-copy protocol requires $T=\Omega(dr^2/\epsilon^2)$ samples in the worst case~\cite{Lowe2025lower}. The number of measurement settings is a separate resource.
A rank-$r$ state has $2dr-r^2-1$ real parameters, while one orthonormal basis gives at most $d-1$ independent probabilities. Thus, even without statistical noise, dimension counting gives $N_U\geq(2dr-r^2-1)/(d-1)=\Omega(r)$.
Reducing $N_U$ is also useful experimentally, since each change of basis can require a new circuit compilation.

Ref.~\cite{guctua2020fast} achieved the near-optimal total sample complexity $\widetilde{O}(dr^2/\epsilon^2)$ using $N_U=d$ unitary $2$-design bases. Ref.~\cite{brandao2020fast} reduced the dependence of $N_U$ on $d$ to polylogarithmic scaling using global Clifford measurements or a global exact unitary $4$-design, but obtained total sample complexities $\widetilde{O}(dr^4/\epsilon^4)$ and $\widetilde{O}(dr^2/\epsilon^4)$, respectively.
Ref.~\cite{cho2025sample} achieved $\widetilde{O}(dr^2/\epsilon^2)$ samples using shallow two-layer Clifford measurements, but sampled a fresh basis for every copy, so that $N_U=T$ and $N_S=1$.
Here, we retain the near-optimal total sample complexity while making $N_U$ depend only polylogarithmically on $d$.

For a Hermitian matrix $A$, let $\mathrm{proj}_r(A)$ denote a best rank-$r$ approximation of $A$ in operator norm, equivalently obtained by retaining the $r$ eigenvalues of largest absolute value. Using the spectral clipping defined in Eq.~\eqref{app:eq:clip_definition}, we consider
\begin{equation}\label{app:eq:setting_efficient_qst_estimator}
    \hat{\rho}_{r,K}
    =
    \mathrm{proj}_{r}
    \left(
        \frac{1}{N_U}
        \sum_{i=1}^{N_U}
        \operatorname{clip}_K
        \left(
            \frac{1}{N_S}
            \sum_{j=1}^{N_S}
            \mathcal{M}^{-1}
            \left(
                U_i^\dagger
                \ket{b_{i,j}}\!\bra{b_{i,j}}
                U_i
            \right)
        \right)
    \right).
\end{equation}
The clipping is applied after averaging the $N_S$ shots from one basis. The $N_U$ matrices entering the outer average are therefore independent.

\begin{fact}[\cite{cho2025sample}]
\label{app:fact:first_term}
    For the two-layer Clifford ensemble under the block size condition
    of Theorem~\ref{main:thm:1}, and for every state $\rho$,
    \begin{align}
        \left\|
            \mathbb{E}_{U,b}
            \left[
                \mathcal{M}^{-1}
                \left(
                    U^\dagger\ket b\!\bra bU
                \right)^2
            \right]
        \right\|_\infty
        &=
        \left\|
            \frac{1}{d^2}\sum_{P,Q}
            \frac{\tau(P,Q)}{m_P m_Q} 
            \operatorname{Tr}(\rho PQ)PQ
        \right\|_\infty
        \\
        &=O(d).
        \label{app:eq:qst_single_shot_second_moment}
    \end{align}
\end{fact}

\begin{lemma}[Rank-$r$ spectral truncation~\cite{guctua2020fast}]\label{appx:lem:projr}
    Let $A,B$ be Hermitian matrices with $\operatorname{rank}(B)\leq r$.
    Let $\mathrm{proj}_{r}(A)$ denote the rank-$r$ projection of $A$, obtained by $\mathrm{proj}_{r}(A) = \mathrm{argmin}_{\mathrm{rank}(X)\leq r}\|A-X\|_{\infty}$. Then
    \begin{equation}
        \|\mathrm{proj}_{r}(A)-B\|_1
        \leq
        4r\|A-B\|_{\infty}.
    \end{equation}
\end{lemma}

\begin{proof}
    Since $\mathrm{proj}_{r}(A)$ is a best rank-$r$ approximation of $A$ in operator norm and $\operatorname{rank}(B)\leq r$, $\|A-\mathrm{proj}_{r}(A)\|_{\infty}\leq\|A-B\|_{\infty}$.
    Hence,
    \begin{equation}
        \|\mathrm{proj}_{r}(A)-B\|_{\infty}
        \leq
        \|\mathrm{proj}_{r}(A)-A\|_{\infty}+\|A-B\|_{\infty}
        \leq
        2\|A-B\|_{\infty}.
    \end{equation}
    Since $\operatorname{rank}(\mathrm{proj}_{r}(A)-B)\leq\operatorname{rank}(\mathrm{proj}_{r}(A))+\operatorname{rank}(B)\leq 2r$,
    \begin{align}
        \|\mathrm{proj}_{r}(A)-B\|_1
        &\leq
        \operatorname{rank}(\mathrm{proj}_{r}(A)-B)\|\mathrm{proj}_{r}(A)-B\|_{\infty}\\
        &\leq
        2r\|\mathrm{proj}_{r}(A)-B\|_{\infty}\\
        &\leq
        4r\|A-B\|_{\infty}.
    \end{align}
\end{proof}

\begin{fact}[Clifford ensemble identities~\cite{cho2025shallow}]
\label{app:fact:clifford_ensemble}
    Let $U$ be a Clifford unitary and define
    $(U,P,b):=\operatorname{Tr}(UPU^\dagger\ket{b}\!\bra{b})$.
    Then, for any Pauli operators $P$ and $\{P_i\}_i$,
    \begin{align}
        (U,P,b)
        &=
        (U,P,b)\,
        \mathbf{1}\!\left\{
            UPU^\dagger\in\pm\mathcal Z
        \right\},\\
        \sum_b\prod_i(U,P_i,b)
        &=
        \operatorname{Tr}\!\left(\prod_i P_i\right)
        \prod_i
        \mathbf{1}\!\left\{
            UP_iU^\dagger\in\pm\mathcal Z
        \right\}.
    \end{align}
\end{fact}

\begin{lemma}
\label{app:lem:qst_same_setting_second_moment}
    Let $b$ and $c$ be conditionally independent outcomes obtained from the same unitary $U$. Under the assumptions of
    Theorem~\ref{appx:thm:1},
    \begin{equation}
        \left\|
            \mathbb{E}_{U,b,c}
            \left[
                \mathcal{M}^{-1}
                \left(
                    U^\dagger\ket b\!\bra bU
                \right)
                \mathcal{M}^{-1}
                \left(
                    U^\dagger\ket c\!\bra cU
                \right)
            \right]
        \right\|_\infty
        =
        O(1).
        \label{app:eq:qst_same_setting_second_moment}
    \end{equation}
\end{lemma}

\begin{proof}
    By Fact~\ref{app:fact:clifford_ensemble} and Eq.~\eqref{app:eq:correlated_visibility},
    \begin{align}
        &\mathbb{E}_{U,b,c}
        \left[
            \mathcal{M}^{-1}
            \left(
                U^\dagger\ket b\!\bra bU
            \right)
            \mathcal{M}^{-1}
            \left(
                U^\dagger\ket c\!\bra cU
            \right)
        \right] \label{appx:eq:qst2}\\
        &=
        \mathbb{E}_{U}\sum_{b,c}
        \operatorname{Tr}(\rho U^{\dagger} \ket b \!\bra b U)
        \operatorname{Tr}(\rho U^{\dagger} \ket c \!\bra c U)
        \mathcal{M}^{-1}(U^\dagger\ket b\!\bra b U)
        \mathcal{M}^{-1}(U^\dagger\ket c\!\bra c U)
        \\
        &=
        \mathbb{E}_{U}\sum_{b,c}\sum_{P,Q,R,S}
        \frac{\operatorname{Tr}(\rho P)}{d}
        \frac{\operatorname{Tr}(\rho Q)}{d}
        \frac{m_R^{-1}}{d}
        \frac{m_S^{-1}}{d}
        (U,P,b)(U,R,b)
        (U,Q,c)(U,S,c)RS\\
        &=
        \sum_{P,Q}
        \frac{\operatorname{Tr}(\rho P)}{d}
        \frac{\operatorname{Tr}(\rho Q)}{d}
        m_P^{-1}m_Q^{-1}
        \tau(P,Q)
        PQ.
    \end{align}
    Let $\ket v$ be a normalized eigenvector associated with the largest eigenvalue of the positive semidefinite operator in Eq.~\eqref{appx:eq:qst2}. Then
    \begin{align}
        &\left\|
            \mathbb{E}_{U,b,c}
            \left[
                \mathcal{M}^{-1}
                \left(
                    U^\dagger\ket b\!\bra bU
                \right)
                \mathcal{M}^{-1}
                \left(
                    U^\dagger\ket c\!\bra cU
                \right)
            \right]
        \right\|_\infty\\
        &=
        \frac{1}{d^2}
        \sum_{P,Q}
        \operatorname{Tr}(\rho P)
        \operatorname{Tr}(\rho Q)
        m_P^{-1}m_Q^{-1}
        \tau(P,Q)
        \operatorname{Tr}
        \left(
            \ket v\!\bra vPQ
        \right).
        \label{appx:eq:qst1}
    \end{align}
    The right-hand side is the second moment for the observable $O_{\mathrm{thm}}=\rho$ and
    the input state $\rho_{\mathrm{thm}}=\ket v\!\bra v$. Hence Theorem~\ref{appx:thm:1} gives
    \begin{align}
        &\left\|
            \mathbb{E}_{U,b,c}
            \left[
                \mathcal{M}^{-1}
                \left(
                    U^\dagger\ket b\!\bra bU
                \right)
                \mathcal{M}^{-1}
                \left(
                    U^\dagger\ket c\!\bra cU
                \right)
            \right]
        \right\|_\infty\\
        &\leq
        C_F
        \left\|
            \rho-\frac{I}{d}
        \right\|_2^2
        +
        \operatorname{Tr}
        \left(
            \rho\ket v\!\bra v
        \right)^2
        =
        O(1),
    \end{align}
    where $C_F$ is the constant in Theorem~\ref{appx:thm:1}.
\end{proof}

\begin{corollary}[Same-basis second moment]\label{appx:cor:qst1}
    Let $U$ be sampled from the two-layer Clifford ensemble, and let $b_1,\ldots,b_{N_S}$ be independent outcomes obtained by measuring $U\rho U^\dagger$ in the computational basis. Define
    \begin{equation}
        \hat{\sigma}
        :=
        \frac{1}{N_S}
        \sum_{j=1}^{N_S}
        \mathcal{M}^{-1}
        \left(
            U^\dagger\ket{b_j}\!\bra{b_j}U
        \right).
    \end{equation}
    Under the assumptions of Theorem~\ref{appx:thm:1}, choosing $N_S=d$
    gives
    \begin{equation}
        \left\|
            \mathbb{E}\hat{\sigma}^2
        \right\|_\infty
        =
        O(1).
    \end{equation}
\end{corollary}
\begin{proof}
    By Fact~\ref{app:fact:first_term} and Lemma~\ref{app:lem:qst_same_setting_second_moment},
    \begin{align}
        \left\|\mathbb{E} \hat \sigma^2\right\|_{\infty}
        &=
        \left\|
        \frac{1}{N_S} 
        \mathbb{E}_{U,b}  
        \mathcal{M}^{-1}(U^\dagger\ket b\!\bra bU)^2 
        +
        \frac{N_S-1}{N_S}
        \mathbb{E}_{U,b,c}  
        \mathcal{M}^{-1}(U^\dagger\ket b\!\bra b U)
        \mathcal{M}^{-1}(U^\dagger\ket c\!\bra c U)
        \right\|_{\infty}\\
        &=
        O\left(
        \frac{d}{N_S}+1
        \right)=O(1).
    \end{align}
\end{proof}

\begin{theorem}[Setting-efficient quantum state tomography]
\label{app:thm:setting_efficient_qst}
    Consider the two-layer Clifford ensemble under the assumptions of Theorem~\ref{main:thm:1}. For every state $\rho$ with $\operatorname{rank}(\rho)\leq r$ and every $0<\epsilon,\delta<1$, the estimator in Eq.~\eqref{app:eq:setting_efficient_qst_estimator}, with $K=\Theta(r/\epsilon)$, satisfies $\|\hat{\rho}_{r,K}-\rho\|_1\leq\epsilon$ with probability at least $1-\delta$ using
    \begin{equation}
        N_U
        =
        O\left(
            \frac{r^2\log(d/\delta)}{\epsilon^2}
        \right),
        \qquad
        N_S=d,
        \qquad
        T
        =
        O\left(
            \frac{dr^2\log(d/\delta)}{\epsilon^2}
        \right).
        \label{app:eq:setting_efficient_qst_complexity}
    \end{equation}
\end{theorem}
\begin{proof}
    For each $i=1,\ldots,N_U$, sample a measurement basis $U_i$ and measure $U_i\rho U_i^\dagger$ independently $N_S$ times, obtaining outcomes $\{b_{i,j}\}_{j=1}^{N_S}$. 
    Define
    \begin{align}
    \hat{\rho}
        &:=
        \frac{1}{N_U}
        \sum_{i=1}^{N_U}
        \hat{\rho}_i,\\
        \hat{\rho}_i
        &:=
        \frac{1}{N_S}
        \sum_{j=1}^{N_S}
        \mathcal{M}^{-1}
        \left(
            U_i^\dagger
            \ket{b_{i,j}}\!\bra{b_{i,j}}
            U_i
        \right).
    \end{align}
    Corollary~\ref{appx:cor:qst1} gives $\|\mathbb{E}\hat{\rho}_i^2\|_\infty=O(1)$. 
    Therefore, Lemma~\ref{appx:lem:clip} implies that, with probability at least $1-\delta$,
    \begin{equation}
        \left\|
            \frac{1}{N_U}
            \sum_{i=1}^{N_U}
            \operatorname{clip}_K(\hat{\rho}_i)
            -
            \rho
        \right\|_\infty
        =
        O\left(
            \frac{1}{K}
            +
            \sqrt{
                \frac{\log(d/\delta)}{N_U}
            }
            +
            \frac{K\log(d/\delta)}{N_U}
        \right).
    \end{equation}
    Choose $K=\Theta(1/\epsilon_{\mathrm{op}})$ and $N_U=\Theta(\log(d/\delta)/\epsilon_{\mathrm{op}}^2)$. 
    Then
    \begin{equation}
        \left\|
            \frac{1}{N_U}
            \sum_{i=1}^{N_U}
            \operatorname{clip}_K(\hat{\rho}_i)
            -
            \rho
        \right\|_\infty
        \leq
        \epsilon_{\mathrm{op}}.
    \end{equation}
    Taking $\epsilon_{\mathrm{op}}=\epsilon/(4r)$ and applying Lemma~\ref{appx:lem:projr} gives
    \begin{equation}
        \Pr\left[
            \|\hat{\rho}_{r,K}-\rho\|_1
            \leq
            \epsilon
        \right]
        \geq
        1-\delta.
    \end{equation}
    These choices give $K=\Theta(r/\epsilon)$, $N_U=O(r^2\log(d/\delta)/\epsilon^2)$, $N_S=d$, and $T=O(dr^2\log(d/\delta)/\epsilon^2)$.
\end{proof}

\subsection{Multiparameter quantum metrology}
\label{app:sec:multiparameter_quantum_metrology}
Quantum metrology studies the estimation of unknown parameters encoded in a quantum state. For a single parameter, the quantum Cramér--Rao bound can be attained by a suitable measurement. In the multiparameter setting, however, measurements that are optimal for different parameters need not be compatible, so the quantum Fisher information cannot in general be attained simultaneously in all parameter directions. Finding an optimal measurement can also require detailed knowledge of the state. Randomized measurements provide a state-independent alternative to such optimized measurements. Ref.~\cite{zhou2026randomized} showed that randomized measurements based on unitary $3$-designs are near-optimal for pure states and well-conditioned approximately low-rank states, and weakly near-optimal for well-conditioned full-parameter mixed states. More recent works obtained shallow randomized measurements for pure-state metrology~\cite{du2026complexity,mao2026near}. Here, we show that the mixed-state guarantees of Ref.~\cite{zhou2026randomized} can also be obtained using the shallow two-layer Clifford ensemble. We first introduce the definitions and relations needed for the proof.

Let $\rho_{\theta}$ be a $d$-dimensional quantum state parameterized by $\theta=(\theta_1,\ldots,\theta_m)$, and let $\mathbb M=\{M_x\}_x$ be a single-copy POVM. 
We write $p_x=\operatorname{Tr}(\rho_\theta M_x)$ and $\partial_i=\partial/\partial\theta_i$. An estimator $\hat\theta$ is locally unbiased at $\theta=\theta^0$ if
\begin{align}
    &\sum_x\hat\theta(x)p_x|_{\theta=\theta^0}=\theta^0,\\
    &\partial_i\sum_x\hat\theta_j(x)p_x|_{\theta=\theta^0}=\delta_{ij}.
\end{align}
Its mean squared error matrix (MSEM) $V$, the classical Fisher information matrix (CFIM) $I$, and the quantum Fisher information matrix (QFIM) $J$ are defined by
\begin{align}
    V(\mathbb M,\hat\theta,\rho_\theta)_{ij}
    &:=
    \sum_x
    \bigl(\hat\theta_i(x)-\theta_i\bigr)
    \bigl(\hat\theta_j(x)-\theta_j\bigr)
    p_x,
    \nonumber\\
    I(\mathbb M,\rho_\theta)_{ij}
    &:=
    \sum_{x:p_x\neq0}
    \frac{
        (\partial_i p_x)(\partial_j p_x)
    }{
        p_x
    },
    \nonumber\\
    J(\rho_\theta)_{ij}
    &:=
    \frac{1}{2}
    \operatorname{Tr}\!\left(
        \rho_\theta\{L_i,L_j\}
    \right),
\end{align}
where the symmetric logarithmic derivative (SLD) $L_i$ is defined as follows:
\begin{equation}
    \partial_i\rho_\theta
    =
    \frac{1}{2}
    \left(
        L_i\rho_\theta+\rho_\theta L_i
    \right).
\end{equation}
We choose the canonical SLDs satisfying $\Pi_0L_i\Pi_0=0$, where $\Pi_0$ projects onto the zero-eigenvalue eigenspace of $\rho_\theta$. We suppress the arguments of $V$, $I$, and $J$ when they are clear.
Following Ref.~\cite{zhou2026randomized}, we assume throughout
this subsection that $I\succ0$ at the local estimation point,
which also implies $J\succ0$.
Every locally unbiased estimator satisfies
\begin{equation}
    V \succeq I^{-1} \succeq J^{-1}.
    \label{app:eq:metrology_crb}
\end{equation}

For multiparameter mixed-state estimation, the QCRB is generally not attainable even when collective measurements over multiple copies are allowed, because the optimal measurements for different parameter directions can be incompatible. 
The Holevo Cramér--Rao bound (HCRB) accounts for this incompatibility and gives the asymptotically attainable precision limit under collective measurements~\cite{yang2019attaining}. Since the randomized measurements considered here act individually on single copies, neither the QCRB nor the HCRB is in general the appropriate benchmark for their optimality. 
Existing studies~\cite{zhou2026randomized} have therefore adopted a different comparison, and we follow this approach by comparing them with the best precision achievable by arbitrary single-copy measurements.

For a positive semidefinite cost matrix $W$, the weighted mean squared error is $\operatorname{Tr}(WV)$. Define the Fisher Cramér--Rao bound as $C_{\mathrm F}(W):=\operatorname{Tr}(WJ^{-1})$ and denote the HCRB by $C_{\mathrm H}(W)$. The optimal weighted mean squared error over single-copy measurements is
\begin{equation}
    C(W)
    :=
    \min_{(\mathbb M,\hat\theta):\mathrm{locally\ unbiased}}
    \operatorname{Tr}(WV)
    =
    \min_{\mathbb M}
    \operatorname{Tr}(WI^{-1}).
\end{equation}
For every locally unbiased estimator based on a single-copy measurement,
\begin{equation}
    C_{\mathrm F}(W)
    \leq
    C_{\mathrm H}(W)
    \leq
    C(W)
    \leq
    \operatorname{Tr}(WV).
\end{equation}
Moreover, $C_{\mathrm H}(W)\leq2C_{\mathrm F}(W)$. In general, $C(W)$ can be strictly larger than $C_{\mathrm H}(W)$, whereas for pure states $C(W)=C_{\mathrm H}(W)$~\cite{zhou2026randomized}. 
Thus, $C(W)$ provides the relevant optimal benchmark when measurements are restricted to individual copies.

Following the definition in Ref.~\cite{zhou2026randomized}, a family of POVMs is \emph{near-optimal} if there exists a constant $c\geq1$, independent of the state and POVM, such that
\begin{equation}
    I
    \succeq
    \frac{1}{c}J.
    \label{app:eq:metrology_near_optimality}
\end{equation}
For general mixed states, a single-copy POVM satisfying this condition may not exist.
A family of POVMs is \emph{weakly near-optimal} if there exist positive constants $c_1$ and $c_2$, independent of the state and POVM, such that
\begin{align}
    \operatorname{Tr}(JI^{-1})
    &\leq
    c_1 C(J),
    \label{app:eq:metrology_weak_near_optimality_cost}\\
    I
    &\succeq
    c_2
    \frac{
        \operatorname{Tr}(J^{-1}I)
    }{
        m
    }
    J.
    \label{app:eq:metrology_weak_near_optimality_matrix}
\end{align}
The first condition compares the weighted error for $W=J$ with the optimal single-copy measurement. The second prevents the Fisher information from being concentrated in only a few parameter directions.

To prove these optimality guarantees, Ref.~\cite{zhou2026randomized} constructs an explicit locally unbiased estimator from classical shadows~\cite{huang2020predicting}.
Let $\hat{\rho}$ denote an unbiased shadow estimator satisfying $\mathbb{E}[\hat{\rho}]=\rho_{\theta}$. At a local point $\theta=\theta^0$, Hermitian operators $\{X_i\}_{i=1}^m$, called \emph{deviation observables}, are chosen to satisfy
\begin{equation}
    \left.
    \operatorname{Tr}(\rho_{\theta}X_i)
    \right|_{\theta=\theta^0}
    =0,
    \qquad
    \left.
    \operatorname{Tr}\!\left(
        (\partial_i\rho_{\theta})X_j
    \right)
    \right|_{\theta=\theta^0}
    =
    \delta_{ij}.
    \label{app:eq:deviation_observable}
\end{equation}
The corresponding \emph{local shadow estimator} is
\begin{equation}
    \hat{\theta}_i(\theta^0)
    =
    \theta_i^0
    +
    \operatorname{Tr}\!\left(
        \left.X_i\right|_{\theta=\theta^0}\hat{\rho}
    \right).
    \label{app:eq:local_shadow_estimator}
\end{equation}
Since $\mathbb{E}[\hat{\rho}]=\rho_{\theta}$, Eq.~\eqref{app:eq:deviation_observable} directly implies that $\hat{\theta}$ is locally unbiased at $\theta^0$. 

The deviation observables can be chosen differently depending on the family of states. For pure states and the full-parameter mixed-state settings considered below, Ref.~\cite{zhou2026randomized} uses
\begin{equation}
    X_i
    =
    \sum_j
    (J^{-1})_{ij}L_j.
    \label{app:eq:deviation_observable_sld}
\end{equation}
Indeed, $\operatorname{Tr}(\rho_{\theta}L_j)=0$ and $\operatorname{Tr}((\partial_i\rho_{\theta})L_j)=J_{ij}$, so Eq.~\eqref{app:eq:deviation_observable} follows immediately. For approximately low-rank states, a different choice of deviation observables is used to remove the contribution from the small eigenvalue subspace, as described below.

The role of the local shadow estimator is to upper bound the performance of the measurement without directly evaluating its CFIM. The MSEM $V$ depends on the choice of estimator $\hat \theta$, whereas $I$ depends only on the POVM $\mathbb{M}$. 
Nevertheless, since the local shadow estimator is locally unbiased, Eq.~\eqref{app:eq:metrology_crb} gives $I^{-1}\preceq V$.
Consequently, an upper bound on $V$ provides a lower bound on $I$. In particular, $V\preceq cJ^{-1}$ implies $I\succeq J/c$, while for the weak near-optimality condition,
\begin{equation}
    C(J)
    \leq
    \operatorname{Tr}(JI^{-1})
    \leq
    \operatorname{Tr}(JV).
    \label{app:eq:local_shadow_cost}
\end{equation}
Thus, for full-parameter mixed states, it is sufficient to show that $\operatorname{Tr}(JV)$ matches the optimal single-copy measurement cost $C(J)$ up to a constant factor. 

\paragraph{\textbf{Pure states.}} 
We first consider a parameterized pure state $\rho_\theta=\ket{\psi_\theta}\!\bra{\psi_\theta}$. For pure states, the local shadow estimator defined above admits a particularly simple structure, and the following fact holds 
\begin{fact}[\cite{zhou2026randomized}]
\label{app:fact:pure_sld}
    For a pure state $\rho_\theta$, the SLDs satisfy
    \begin{equation}
        \operatorname{Tr}(L_iL_j)=2J_{ij}.
    \end{equation}
\end{fact}

\begin{theorem}[Near-optimality of Pure states]
\label{app:thm:pure_state_metrology}
    For every parameterized pure state $\rho_\theta=\ket{\psi_\theta}\!\bra{\psi_\theta}$, the POVM $\mathbb{M}_{2\mathrm L}$ induced by the two-layer Clifford ensemble under the block size condition of Theorem~\ref{main:thm:1} satisfies, for some constant $c\geq1$,
    \begin{equation}
        I(\mathbb{M}_{2\mathrm L},\rho_\theta)
        \succeq
        \frac{1}{c}
        J(\rho_\theta).
    \end{equation}
    Hence, $\mathbb{M}_{2\mathrm L}$ is near-optimal for multiparameter pure-state estimation.
\end{theorem}
\begin{proof}
    Let $C_F>1$ be the constant in Theorem~\ref{appx:thm:1} such that
    \begin{equation}
        \operatorname{Var}\!\left[
            \operatorname{Tr}(O\hat{\rho})
        \right]
        \leq
        C_F\|O_0\|_2^2
    \end{equation}
    for every Hermitian observable $O$. For an arbitrary vector $\boldsymbol{a}\in\mathbb{R}^m$, define $X_{\boldsymbol{a}}:=\sum_i a_iX_i$. For a pure state, we can take $L_i=2\partial_i\rho_\theta$. Hence, $\operatorname{Tr}(L_i)=2\partial_i\operatorname{Tr}(\rho_\theta)=0$, and therefore $\operatorname{Tr}(X_i)=0$ for every $i$. Thus, $\operatorname{Tr}(X_{\boldsymbol{a}})=0$. Since the local shadow estimator is locally unbiased at $\theta^0$,
    \begin{align}
        \boldsymbol{a}^TV\boldsymbol{a}
        &=
        \operatorname{Var}\!\left[
            \sum_i a_i\hat{\theta}_i
        \right]
        =
        \operatorname{Var}\!\left[
            \operatorname{Tr}(X_{\boldsymbol{a}}\hat{\rho})
        \right]\\
        &\leq
        C_F\|X_{\boldsymbol{a}}\|_2^2,
    \end{align}
    where the last inequality follows from Theorem~\ref{appx:thm:1}.

    Using $X_i=\sum_j(J^{-1})_{ij}L_j$ and Fact~\ref{app:fact:pure_sld},
    \begin{align}
        \|X_{\boldsymbol{a}}\|_2^2
        &=
        \sum_{i,j,k,\ell}
        a_i a_j
        (J^{-1})_{ik}
        (J^{-1})_{j\ell}
        \operatorname{Tr}(L_kL_\ell)\\
        &=
        2\boldsymbol{a}^TJ^{-1}\boldsymbol{a}.
    \end{align}
    Therefore,
    \begin{equation}
        \boldsymbol{a}^TV\boldsymbol{a}
        \leq
        2C_F \boldsymbol{a}^TJ^{-1}\boldsymbol{a}
    \end{equation}
    for every $\boldsymbol{a}\in\mathbb{R}^m$, and hence
    $V\preceq2C_FJ^{-1}$.
    Since $I^{-1}\preceq V$,
    \begin{equation}
        I\succeq\frac{1}{2C_F}J.
    \end{equation}
    Thus, $\mathbb{M}_{2\mathrm L}$ is near-optimal with
    $c=2C_F$.
\end{proof}

\paragraph{\textbf{Approximately low-rank states.}}
The pure-state result extends to mixed states whose quantum Fisher information is largely preserved in a low-dimensional subspace with eigenvalues bounded away from zero. This property is formalized by approximately low-rank well-conditioned states \cite{zhou2026randomized}. Theorem~\ref{appx:thm:1} then gives a constant-factor comparison between the CFIM and QFIM for this class.

\begin{definition}[Approximately low-rank well-conditioned state~\cite{zhou2026randomized}]
    Let $\rho=\sum_a\lambda_a|\psi_a\rangle\langle\psi_a|$, and fix $\mu,c\in(0,1]$. Let $\Pi$ project onto the eigenvectors with $\lambda_a\ge\mu$, and set $\Pi^\perp=\mathbb{I}-\Pi$ and $p=\operatorname{Tr}(\rho\Pi^\perp)$. For $p>0$, let $J^\perp$ be the QFIM of $\Pi^\perp\rho\Pi^\perp/p$, and let $J^p$ be the Fisher information matrix of the Bernoulli distribution $\{p,1-p\}$, so $J^p_{ij}=(\partial_i p)(\partial_j p)/(p(1-p))$. The model is called $(\mu,c)$-approximately low-rank
    well-conditioned if 
    \begin{equation}
        \Pi\neq0,\qquad
        J-pJ^\perp-J^p \succeq cJ .
        \label{eq:approximately_low_rank_condition}
    \end{equation}
\end{definition}

\begin{fact}[Modified SLD bound~\cite{zhou2026randomized}]
Let $L_i$ be the SLDs of $\rho$, and define
\begin{equation}
    \widetilde L_i
    :=
    L_i-\Pi^\perp L_i\Pi^\perp
    +\frac{\partial_i p}{1-p}\Pi .
    \label{eq:modified_sld}
\end{equation}
Then $\operatorname{Tr}(\rho\widetilde L_i)=0$. The matrix
$\widetilde J$ is defined by
\begin{equation}
    \widetilde J_{ij}
    :=
    \operatorname{Tr}(\widetilde L_i\partial_j\rho)
    =
    J_{ij}-pJ^\perp_{ij}-J^p_{ij}.
    \label{eq:modified_qfim}
\end{equation}
Hence $\widetilde J\succeq cJ$. Moreover, defining
$\widetilde K_{ij}:=
\operatorname{Tr}(\widetilde L_i\widetilde L_j)$, one has
\begin{equation}
    \widetilde K\preceq\frac{5}{\mu}J .
    \label{eq:modified_sld_frobenius}
\end{equation}
\end{fact}

\begin{theorem}[Near-optimality of approximately low-rank states]\label{thm:approximately_low_rank_metrology}
    Let $\mathcal E$ be the measurement ensemble of Theorem~\ref{appx:thm:1} under the same circuit assumptions. For every $(\mu,c)$-approximately low-rank well-conditioned model with $J\succ0$, its CFIM satisfies
    \begin{equation}
        I\succeq\frac{\mu c^2}{5C_{\rm F}}J,
    \end{equation}
    where $C_F$ is the constant in Theorem~\ref{appx:thm:1}.
\end{theorem}

\begin{proof}
    Define
    \begin{equation}
        X_i
        :=
        \sum_j(\widetilde J^{-1})_{ij}\widetilde L_j .
        \label{eq:approximately_low_rank_deviation}
    \end{equation}
    Then $\operatorname{Tr}(\rho X_i)=0$ and $\operatorname{Tr}(X_i\partial_j\rho)=\delta_{ij}$. Thus the local shadow estimator $\widehat\theta_i=(\theta_0)_i+ \operatorname{Tr}(X_i\widehat\rho_{\mathcal E})$ is locally unbiased. For $\boldsymbol{a}\in\mathbb{R}^m$, let $X_{\boldsymbol{a}}=\sum_i a_iX_i$. Theorem~\ref{appx:thm:1} applied to $X_{\boldsymbol{a}}$ gives
    \begin{align}
        \boldsymbol{a}^{T}V\boldsymbol{a}
        &=
        \operatorname{Var}\!\left[
            \operatorname{Tr}(X_{\boldsymbol{a}}\widehat\rho_{\mathcal E})
        \right] \\
        &\le
        C_{\rm F}\lVert X_{\boldsymbol{a},0}\rVert_2^2
        \le
        C_{\rm F}\lVert X_{\boldsymbol{a}}\rVert_2^2 \\
        &=
        C_{\rm F}
        \boldsymbol{a}^{T}
        \widetilde J^{-1}\widetilde K\widetilde J^{-1}
        \boldsymbol{a},
    \end{align}
    where $X_{\boldsymbol{a},0}$ denotes the traceless part of $X_{\boldsymbol{a}}$. Since this holds for every $\boldsymbol{a}$,
    \begin{align}
        V
        &\preceq
        C_{\rm F}
        \widetilde J^{-1}\widetilde K\widetilde J^{-1} \\
        &\preceq
        \frac{5C_{\rm F}}{\mu}
        \widetilde J^{-1}J\widetilde J^{-1} \\
        &\preceq
        \frac{5C_{\rm F}}{\mu c}\widetilde J^{-1}
        \preceq
        \frac{5C_{\rm F}}{\mu c^2}J^{-1}.
    \end{align}
    Since $I^{-1}\preceq V$,
    \begin{equation}
        I\succeq\frac{\mu c^2}{5C_{\rm F}}J,
    \end{equation}
    as desired.
\end{proof}

\paragraph{\textbf{Full-parameter mixed states.}}

Let $\rho_\theta=\sum_{a=1}^r\lambda_a |\psi_a\rangle\langle\psi_a|$ have rank $r$. Let $\kappa$ be a uniform upper bound on $\lambda_{\max}/\lambda_{\min}^+$, where $\lambda_{\min}^+$ is the smallest positive eigenvalue. 
The value $\kappa$ measures how uneven the nonzero eigenvalues are. Since $\lambda_{\max}\geq1/r$, one has $\lambda_{\min}^+\geq1/(r\kappa)$. The bound below shows that the two-layer measurement loses at most a factor proportional to $r\kappa$. 
For models with the maximal number of unknown parameters, the Gill--Massar (GM) bound converts this estimate into weak near-optimality.

\begin{fact}[SLD bound~\cite{zhou2026randomized}]
\label{app:fact:full_parameter_sld}
    Let $L_i$ be the SLDs and define
    $K_{ij}:=\operatorname{Tr}(L_iL_j)$. For every rank-$r$ state
    with condition number at most $\kappa$,
    \begin{equation}
        K\preceq2r\kappa J.
        \label{app:eq:full_parameter_sld}
    \end{equation}
\end{fact}

\begin{theorem}[Rank-$r$ mixed-state bound]
\label{app:thm:full_parameter_mixed_metrology}
    Let $\rho_\theta$ be a smooth rank-$r$ model with $J\succ0$.
    Under the same circuit assumptions as Theorem~\ref{appx:thm:1},
    \begin{align}
        I
        &\succeq
        \frac{1}{2C_Fr\kappa}J,
        \label{app:eq:full_parameter_fisher}\\
        \operatorname{Tr}(JI^{-1})
        &\leq
        2C_Fr\kappa m,
        \label{app:eq:full_parameter_cost}
    \end{align}
    where $C_F$ is the constant in Theorem~\ref{appx:thm:1}.
\end{theorem}

\begin{proof}
    Use the deviation observables in Eq.~\eqref{app:eq:deviation_observable_sld}. For $\boldsymbol{a}\in\mathbb R^m$, set $X_{\boldsymbol{a}}:=\sum_i a_iX_i$. 
    Applying Theorem~\ref{main:thm:1} to every $X_{\boldsymbol{a}}$ gives
    \begin{align}
        \boldsymbol{a}^TV\boldsymbol{a}
        &\leq
        C_F\|(X_{\boldsymbol{a}})_0\|_2^2
        \leq
        C_F\|X_{\boldsymbol{a}}\|_2^2\\
        &=
        C_F\boldsymbol{a}^TJ^{-1}KJ^{-1}\boldsymbol{a}
        \leq
        2C_Fr\kappa\,\boldsymbol{a}^TJ^{-1}\boldsymbol{a},
    \end{align}
    where $(X_{\boldsymbol{a}})_0$ is the traceless part of $X_{\boldsymbol{a}}$. Since this holds for every $\boldsymbol{a}\in\mathbb{R}^m$, $V\preceq2C_Fr\kappa J^{-1}$. The Cramér--Rao inequality $I^{-1}\preceq V$ proves Eq.~\eqref{app:eq:full_parameter_fisher}. 
    Taking the trace after multiplication by $J$ proves Eq.~\eqref{app:eq:full_parameter_cost}.
\end{proof}

Combining Theorem~\ref{app:thm:full_parameter_mixed_metrology} with the Gill--Massar bounds proved for the three maximal-parameter mixed-state models in Ref.~\cite{zhou2026randomized} shows that both weak near-optimality conditions 
\eqref{app:eq:metrology_weak_near_optimality_cost} and
\eqref{app:eq:metrology_weak_near_optimality_matrix} are satisfied.
Thus, for bounded $\kappa$, $\mathbb M_{2\mathrm L}$ is weakly near-optimal for all three models, with $c_1=2C_F\kappa$ and $c_2=(2C_F\kappa)^{-1}$.

\subsection{Stabilizer structure learning}
\label{app:sec:stabilizer_structure_learning}

Stabilizer structure learning identifies Pauli operators that act as symmetries of a quantum state. This information is useful in quantum state tomography and in the study of quantum many-body systems \cite{grewal2025efficient,chia2024efficient}. States with $n-t$ independent Pauli symmetries include $t$-doped stabilizer states and ground states of stabilizer Hamiltonians \cite{leone2024learning,brell2011toric,kitaev2006anyons}.

\begin{definition}[Weyl representation]
    For $n$ qubits, let $V:=\mathbb F_2^{2n}$ and write
    $a=(a_x,a_z)\in V$. The Weyl operator associated with $a$ is
    \begin{equation}
        W_a
        :=
        i^{a_x\cdot a_z}
        \bigotimes_{j=1}^{n}
        X_j^{(a_x)_j}Z_j^{(a_z)_j},
        \label{app:eq:weyl_operator}
    \end{equation}
    where $a_x\cdot a_z$ in the phase is evaluated over the integers.
    Ignoring global phases, this identifies $V$ with $\{I,X,Y,Z\}^{\otimes n}$, and Pauli multiplication corresponds to addition in $V$. 
    A Clifford circuit $C$ induces a linear map $a\mapsto C(a)$ satisfying $CW_aC^\dagger=\pm W_{C(a)}$. The computational $Z$ subspace is $\mathcal Z:=\{(0^n,z):z\in\mathbb F_2^n\}$. For $a=(a_x,a_z)$ and $b=(b_x,b_z)$, define the symplectic inner product by 
    \begin{equation}
        [a,b]:=a_x\cdot b_z+a_z\cdot b_x\pmod 2.
    \end{equation}
    The operators $W_a$ and $W_b$ commute exactly when $[a,b]=0$. A subspace is \emph{isotropic} if this inner product vanishes for every pair of its elements. An isotropic subspace of dimension $n$ is \emph{Lagrangian}; in particular, $\mathcal Z$ is Lagrangian. For a state $\rho$, its Pauli symmetry subspace is
    \begin{equation}
        S:=\mathrm{Weyl}(\rho)
        :=
        \left\{
            a\in V:
            \operatorname{Tr}(W_a\rho)^2=1
        \right\}.
        \label{app:eq:weyl_symmetry_subspace}
    \end{equation}
    The subspace $S$ is isotropic, and $\dim S$ is called the stabilizer dimension of $\rho$.
\end{definition}

A single-copy learning protocol has two steps. A Clifford circuit first maps some elements of $S$ into $\mathcal Z$. Computational difference sampling~\cite{grewal2025efficient} then identifies these Pauli $Z$-strings, which are mapped back using the known Clifford circuit. For stabilizer states, $O(n)$ computational difference samples suffice to recover the visible subgroup $C(S) \cap \mathcal{Z}$. Thus, the main technical question is whether a new independent generator becomes visible with sufficiently large probability.

\begin{fact}[\cite{grewal2025efficient, chia2024efficient}]\label{app:fact:known_stabilizer_visibility}
    Let $\dim S=n-t$, and let $T\subset S$ be a codimension-one subspace. Global
    random Clifford measurements satisfy
    \begin{equation}
        \Pr_C\!\left[
            C(S\setminus T)\cap\mathcal Z\neq\emptyset
        \right]
        \ge
        \frac{1}{2^{t+1}+1}
        =
        \Omega(2^{-t}).
        \label{app:eq:stabilizer_visibility_target}
    \end{equation} 
\end{fact}

The following observation connects a visible symmetry to the measured bit string.

\begin{lemma}[Deterministic symmetry outcomes]
    \label{app:lem:deterministic_weyl_outcome}
    Let $a\in\mathrm{Weyl}(\rho)$. Then
    \begin{equation}
        W_a\rho
        =
        \operatorname{Tr}(W_a\rho)\rho.
        \label{app:eq:weyl_eigenstate_relation}
    \end{equation}
    Moreover, suppose that $C(a)\in\mathcal Z$ and that $b\in \{0,1\}^{n}$ is a measurement outcome obtained by measuring $C\rho C^\dagger$ in the computational basis. Then
    \begin{equation}
        \bra b CW_aC^\dagger\ket b
        =
        \operatorname{Tr}(W_a\rho).
        \label{app:eq:visible_weyl_outcome}
    \end{equation}
\end{lemma}

\begin{proof}
    The condition $a\in\mathrm{Weyl}(\rho)$ gives $\operatorname{Tr}(W_a\rho)=\pm1$. Since $W_a$ is a Hermitian unitary, $\rho$ is supported on the eigenspace with this eigenvalue, which proves Eq.~\eqref{app:eq:weyl_eigenstate_relation}. 
    If $C(a)\in\mathcal Z$ and $b \sim C\rho C^{\dagger}$, 
    \begin{align}
        \operatorname{Tr}(W_a\rho)\bra bC\rho C^{\dagger} \ket b
        &=\bra bC W_a\rho C^{\dagger} \ket b\\
        &=
        \bra b CW_aC^{\dagger} C\rho C^{\dagger} \ket b\\
        &=
        \bra b CW_aC^{\dagger}\ket b \bra bC\rho C^{\dagger} \ket b.
    \end{align}
    The first equality follows by conjugating Eq.~\eqref{app:eq:weyl_eigenstate_relation} by $C$ and taking the expectation value in $\ket b$. Because $b$ is sampled from $C\rho C^{\dagger}$, $\bra bC\rho C^{\dagger} \ket b \neq 0$. Thus,
    \begin{equation}
        \operatorname{Tr}(W_a\rho)=\bra b CW_aC^{\dagger}\ket b.
    \end{equation}
\end{proof}

\begin{theorem}[Worst-case stabilizer visibility]\label{app:thm:worst_case_stabilizer_visibility}
    Let $\rho$ be an $n$-qubit state with $S=\mathrm{Weyl}(\rho)$ and $\dim S=n-t$, and let $T\subset S$ have codimension one. Let $C$ be sampled from the two-layer Clifford ensemble under the assumptions of Theorem~\ref{appx:thm:1}. 
    Then
    \begin{equation}
        \Pr_C\!\left[
            C(S\setminus T)\cap\mathcal Z\neq\emptyset
        \right]
        \geq
        \frac{1}{1+2^{t+1}C_F}
        =
        \Omega(2^{-t}),
        \label{app:eq:two_layer_stabilizer_visibility}
    \end{equation}
    where $C_F$ is the constant in Theorem~\ref{appx:thm:1}.
\end{theorem}

\begin{proof}
    Define
    \begin{align}
        O_{S,T}
        &:=
        \frac{1}{d}
        \sum_{a\in S\setminus T}
        \operatorname{Tr}(W_a\rho)W_a,
        \label{app:eq:stabilizer_witness}\\
        Y
        &:=
        \operatorname{Tr}(O_{S,T}\hat\rho),
        \label{app:eq:stabilizer_random_variable}
    \end{align}
    where $\hat\rho=\mathcal{M}^{-1}(C^{\dagger} \ket b \! \bra b C)$ is the exact shadow snapshot associated with $C$ and its measurement outcome $b$.
    Since $|S\setminus T|=2^{n-t-1}$ and $\operatorname{Tr}(W_a\rho)^2=1$,
    \begin{equation}
        \operatorname{Tr}(O_{S,T}\rho)
        =
        \lVert O_{S,T}\rVert_2^2
        =
        \frac{|S\setminus T|}{d}
        =
        2^{-t-1}.
        \label{app:eq:stabilizer_witness_moments}
    \end{equation}
    The exact inverse shadow channel is diagonal in the Pauli basis. Lemma~\ref{app:lem:deterministic_weyl_outcome} therefore gives
    \begin{align}
        Y
        &=
        \frac{1}{d}
        \sum_{a\in S\setminus T}
        \operatorname{Tr}(W_a\rho)
        \operatorname{Tr}(W_a\mathcal{M}^{-1}(C^{\dagger} \ket b \! \bra b C))\\
        &=
        \frac{1}{d}
        \sum_{a\in S\setminus T}
        m_a^{-1}
        \operatorname{Tr}(W_a\rho)
        \operatorname{Tr}(W_aC^{\dagger} \ket b \! \bra b C)\\
        &=
        \frac{1}{d}
        \sum_{a\in S\setminus T}
        m_a^{-1}
        \bra b CW_aC^{\dagger} \ket b^2\\
        &=
        \frac{1}{d}
        \sum_{a\in S\setminus T}
        \underbrace{m_a^{-1}}_{>0}
        \underbrace{
            \mathbf 1\!\left\{C(a)\in\mathcal Z\right\}
        }_{\in\{0,1\}}
        \geq 0,
        \label{app:eq:stabilizer_positive_sum}
    \end{align}
    where $m_{a}$ is the visibility of $W_a$. Then
    \begin{equation}
        Y>0
        \Longleftrightarrow
        C(S\setminus T)\cap\mathcal Z\neq\emptyset.
        \label{app:eq:stabilizer_positive_event}
    \end{equation}
    Since $m_a^{-1}>0$, Eq.~\eqref{app:eq:stabilizer_positive_sum} is strictly positive if and only if $\mathbf 1\!\left\{C(a)\in\mathcal Z\right\}=1$ for some $a\in S\setminus T$.
    This is equivalent to $C(S\setminus T)\cap\mathcal Z\neq\emptyset$. Eq.~\eqref{app:eq:visibility} and Theorem~\ref{appx:thm:1} imply
    \begin{align}
        \mathbb E Y
        &=\frac{1}{d}
        \sum_{a\in S\setminus T}
        m_{a}^{-1}m_a
        =
        2^{-t-1},\\
        \operatorname{Var}(Y)
        &\leq
        C_F\lVert O_{S,T}\rVert_2^2
        =
        C_F2^{-t-1}.
        \label{app:eq:stabilizer_y_moments}
    \end{align}
    Since $Y\geq0$, the Paley--Zygmund inequality gives
    \begin{align}
        \Pr(Y>0)
        &\geq
        \frac{(\mathbb EY)^2}{\mathbb EY^2}\geq
        \frac{2^{-2t-2}}
        {C_F2^{-t-1}+2^{-2t-2}}\\
        &=
        \frac{1}{1+2^{t+1}C_F}=\Omega(2^{-t}).
        \label{app:eq:stabilizer_paley_zygmund}
    \end{align}
    Equation~\eqref{app:eq:stabilizer_positive_event} completes the proof. 
\end{proof}

\begin{definition}[Computational difference
sampling~\cite{grewal2025efficient,cho2026single}]
\label{app:def:computational_difference_sampling}
    Measure two independent copies of an $n$-qubit state
    $\rho$ in the computational basis, obtaining outcomes
    $b_1,b_2\in\mathbb F_2^n$.
    A computational difference sample is
    \begin{equation}
        x=b_1\oplus b_2,
    \end{equation}
    where $\oplus$ denotes bitwise addition modulo two.
    Each sample uses two copies of $\rho$, measured
    separately.
\end{definition}

\begin{fact}[{\cite[
    proof of Lemma~4.21(b)]{chen2025stabilizer}}]
\label{app:fact:conditional_success_bound}
    Let $X_1,\ldots,X_m$ be $\{0,1\}$-valued random
    variables such that, for some $p\in(0,1]$, $\Pr[X_i=1\mid X_1,\ldots,X_{i-1}]\geq p$ almost surely for every $i$.
    For any nonnegative integer $r$ and $\delta\in(0,1)$,
    if $m\geq\frac{2}{p}\left(r+\ln(1/\delta)\right)$, then
    \begin{equation}
        \Pr\!\left[\sum_{i=1}^{m}X_i\geq r\right]
        \geq1-\delta.
    \end{equation}
\end{fact}

\begin{corollary}[Worst-case coverage of the Weyl support]
\label{app:cor:worst_case_weyl_support_recovery}
    Let $S=\mathrm{Weyl}(\rho)$ with $\dim S=n-t$, and let
    $C_1,\ldots,C_m$ be sampled independently from the
    two-layer Clifford ensemble under the assumptions of
    Theorem~\ref{appx:thm:1}.
    For $\delta\in(0,1)$, suppose that $m\geq2\left(1+2^{t+1}C_F\right)(n-t+\ln(1/\delta))$, where $C_F$ is the constant in Theorem~\ref{appx:thm:1}. Then, with probability at least $1-\delta$,
    \begin{equation}
        \sum_{i=1}^{m}
        \left(
            S\cap C_i^\dagger(\mathcal Z)
        \right)
        =
        S.
        \label{app:eq:worst_case_weyl_support_recovery}
    \end{equation}
\end{corollary}
\begin{proof}
    Set
    \begin{equation}
        p:=\frac{1}{1+2^{t+1}C_F},
        \qquad
        T_0:=\{0\},
        \qquad
        T_i:=\sum_{j=1}^{i}
        \left(S\cap C_j^\dagger(\mathcal Z)\right).
    \end{equation}
    Condition on $C_1,\ldots,C_{i-1}$.
    If $T_{i-1}\neq S$, choose a codimension-one
    subspace $H\subset S$ containing $T_{i-1}$.
    By the independence of $C_i$ and
    Theorem~\ref{app:thm:worst_case_stabilizer_visibility},
    \begin{equation}
        \Pr\!\left[
            C_i(S\setminus H)\cap\mathcal Z
            \neq\emptyset
            \,\middle|\,
            C_1,\ldots,C_{i-1}
        \right]
        \geq p.
    \end{equation}
    This event adds an element outside $H$ to $T_i$, and hence $\dim T_i>\dim T_{i-1}$. Define
    \begin{equation}
        X_i:=
        \begin{cases}
            1, & T_{i-1}=S
                 \text{ or }\dim T_i>\dim T_{i-1},\\
            0, & \text{otherwise}.
        \end{cases}
    \end{equation}
    Thus,
    $\Pr[X_i=1\mid C_1,\ldots,C_{i-1}]\geq p$.
    Since $X_1,\ldots,X_{i-1}$ are determined by
    $C_1,\ldots,C_{i-1}$, it follows that
    \begin{equation}
        \Pr[X_i=1\mid X_1,\ldots,X_{i-1}]
        \geq p
    \end{equation}
    almost surely. If $T_m\neq S$, each $X_i=1$ corresponds to a strict increase in dimension, so
    \begin{equation}
        \sum_{i=1}^{m}X_i
        \leq\dim T_m<n-t.
    \end{equation}
    Applying Fact~\ref{app:fact:conditional_success_bound} with $r=n-t$ and the assumed lower bound on $m$ therefore gives
    \begin{equation}
        \Pr[T_m=S]
        \geq
        \Pr\!\left[\sum_{i=1}^{m}X_i\geq n-t\right]
        \geq1-\delta.
    \end{equation}
\end{proof}

\begin{fact}[Reconstruction in a fixed Clifford
basis~\cite{grewal2025efficient,cho2026single}]
\label{app:fact:cds_reconstruction}
    Let $\rho$ be an $n$-qubit state and $C$ a Clifford circuit. For $\epsilon,\delta\in(0,1)$, a reconstruction procedure using $N_{\mathrm{CDS}}= O\!\left((n+\log(1/\delta))/\epsilon\right)$ independent computational difference samples generated by measuring $C\rho C^\dagger$ in the computational basis returns a subspace $\widehat S_C\subseteq C^\dagger(\mathcal Z)$. With probability at least $1-\delta$, this subspace satisfies
    \begin{equation}
        \frac{1}{|\widehat S_C|}
        \sum_{a\in\widehat S_C}
        \operatorname{Tr}(\rho W_a)^2
        \geq1-\epsilon,
        \qquad
        \widehat S_C
        \supseteq
        \mathrm{Weyl}(\rho)\cap C^\dagger(\mathcal Z).
    \end{equation}
\end{fact}

\begin{fact}[Sum of heavy-weight subspaces~\cite{cho2026single}]
\label{app:fact:combining_stabilizer_subspaces}
    Let $\rho$ be an $n$-qubit state and let $\widehat S_1,\ldots,\widehat S_m$ be subspaces of $\mathbb F_2^{2n}$. Suppose that, for some $0<\epsilon<1/4$,
    $\frac{1}{|\widehat S_i|}
        \sum_{a\in\widehat S_i}
        \operatorname{Tr}(\rho W_a)^2
        \geq1-\frac{\epsilon}{2n}$ for every $i\in[m]$.
    Then $\widehat S:=\sum_{i=1}^{m}\widehat S_i$ is isotropic and satisfies
    \begin{equation}
        \frac{1}{|\widehat S|}
        \sum_{a\in\widehat S}
        \operatorname{Tr}(\rho W_a)^2
        \geq1-\epsilon.
    \end{equation}
\end{fact}

Combining these facts with Corollary~\ref{app:cor:worst_case_weyl_support_recovery} gives the following guarantee.

\begin{corollary}[Worst-case approximate stabilizer
group learning]
\label{app:cor:worst_case_stabilizer_learning}
    Let $\rho$ be an $n$-qubit state with $\dim\mathrm{Weyl}(\rho)=n-t$, and suppose that the two-layer Clifford ensemble satisfies the assumptions of Theorem~\ref{appx:thm:1}. For $0<\epsilon<1/4$ and any fixed $\delta\in(0,1)$, there is a single-copy protocol using $O\!\left(\frac{n^3 2^t}{\epsilon}\right)$ copies of $\rho$ that outputs, with probability at least $1-\delta$, an isotropic subspace $\widehat S$ satisfying
    \begin{equation}
        \widehat S\supseteq\mathrm{Weyl}(\rho),
        \qquad
        \frac{1}{|\widehat S|}
        \sum_{a\in\widehat S}
        \operatorname{Tr}(\rho W_a)^2
        \geq1-\epsilon.
    \end{equation}
\end{corollary}

\begin{proof}
    By Corollary~\ref{app:cor:worst_case_weyl_support_recovery}, choosing $m=O(n2^t)$ bases ensures
    \begin{equation}
        \sum_{i=1}^{m}
        \left(
            \mathrm{Weyl}(\rho)
            \cap C_i^\dagger(\mathcal Z)
        \right)
        =
        \mathrm{Weyl}(\rho)
    \end{equation}
    with probability at least $1-\delta/2$. In each basis, apply Fact~\ref{app:fact:cds_reconstruction} with error $\epsilon/(2n)$ and failure probability $\delta/(2m)$.
    This uses $O\!\left(\frac{n\bigl(n+\log(2m/\delta)\bigr)}{\epsilon}\right)$ computational difference samples per basis. By a union bound, all reconstructed subspaces $\widehat S_i$ satisfy both guarantees in Fact~\ref{app:fact:cds_reconstruction} simultaneously with probability at least $1-\delta/2$. By Fact~\ref{app:fact:combining_stabilizer_subspaces} and Corollary~\ref{app:cor:worst_case_weyl_support_recovery}, $\widehat S:=\sum_i\widehat S_i$ satisfies, with probability at least $1-\delta$,
    \begin{align}
        &\widehat S \text{ is isotropic},\\
        &\widehat S\supseteq\mathrm{Weyl}(\rho),\\
        &\frac{1}{|\widehat S|}
        \sum_{a\in\widehat S}
        \operatorname{Tr}(\rho W_a)^2
        \geq1-\epsilon.
    \end{align}
    For fixed $\delta$, the $m=O(n2^t)$ bases each require
    $O(n^2/\epsilon)$ copies, giving
    $O(n^3 2^t/\epsilon)$ copies in total.
\end{proof}

\subsection{Unbiased purity estimation}
The purity $\operatorname{Tr}(\rho^2)$ quantifies the mixedness of a quantum state: it equals one for a pure state and is smaller than one for a mixed state. It directly determines the second R\'enyi entropy, $R_2(\rho)=-\log\operatorname{Tr}(\rho^2)$. When $\rho_A$ is the reduced state of a subsystem $A$ of a global pure state, $R_2(\rho_A)$ measures entanglement. It is widely used as an experimentally accessible proxy for the von Neumann entanglement entropy, which is generally harder to measure \cite{satzinger2021realizing,brydges2019probing}.

Standard methods for estimating purity, such as the SWAP test~\cite{ekert2002direct} and Bell sampling~\cite{hangleiter2024bell}, require coherent access to two copies. 
Single-copy randomized measurements avoid this requirement. Throughout this discussion, we consider additive error $0<\epsilon<1$ and a fixed failure probability. An exact unitary $4$-design~\cite{anshu2022distributed} gives an unbiased estimator
with the optimal single-copy sample complexity
\begin{equation}
    T=O\!\left(
        \max\left\{
            \frac{1}{\epsilon^2},
            \frac{\sqrt d}{\epsilon}
        \right\}
    \right).
\end{equation}
Global Clifford measurements also give an unbiased estimator.
For additive error $\epsilon$,
choosing $N_S=\Theta(\sqrt d)$ and $N_U=\Theta(1/\epsilon^2)$ gives the sample complexity~\cite{zheng2025distributed}
\begin{equation}
    T=N_U N_S=O\left(\frac{\sqrt d}{\epsilon^2}\right).
\end{equation}
Thus, global Clifford measurements retain the $\sqrt d$ dependence at fixed accuracy, but have a worse dependence on $\epsilon$ than the exact $4$-design bound. 

An $\epsilon_{\mathrm{des}}$-approximate unitary $4$-design can be implemented in 1D circuit depth $O(\log(n/\epsilon_{\mathrm{des}}))$~\cite{schuster2025random}. The corresponding purity estimator satisfies
\begin{align}
    \left|
        \mathbb E\hat p_2-\operatorname{Tr}(\rho^2)
    \right|
    &=
    O(\epsilon_{\mathrm{des}}),\\
    \operatorname{Var}[\hat p_2]
    &=
    O\left(
        \frac{1}{N_U}
        \left(
            \frac{d}{N_S^2}
            + \frac{1}{N_S}
            + \frac{1}{d}
            + \epsilon_{\mathrm{des}}
        \right)
    \right).
\end{align}
See Ref.~\cite{zheng2025distributed} for the explicit construction of this generally biased estimator.
To control both the bias and variance, it suffices to
choose $\epsilon_{\mathrm{des}}=O(\epsilon^2)$ and take
\begin{equation}
    N_U
    =
    \Theta\left(
        \max\left\{1,\frac{1}{d\epsilon^2}\right\}
    \right),
    \qquad
    N_S
    =
    \Theta\left(
        \min\left\{d,\frac{\sqrt d}{\epsilon}\right\}
    \right).
\end{equation}
This recovers the optimal sample complexity
\begin{equation}
    T=N_U N_S
    =
    O\left(
        \max\left\{
            \frac{1}{\epsilon^2},
            \frac{\sqrt d}{\epsilon}
        \right\}
    \right)\label{app:eq:pur_4design},
\end{equation}
with circuit depth $O(\log (n/\epsilon))$.

Below, we show that two-layer Clifford measurements achieve the same $T=O(\sqrt d/\epsilon^2)$ guarantee as global Clifford measurements. Although this sample complexity has a worse dependence on $\epsilon$ than the bound in Eq.~\eqref{app:eq:pur_4design}, our estimator is unbiased and uses circuits of depth $O(\log n)$ independent of the target error. Higher precision can therefore be achieved by increasing the number of measurements without increasing the circuit depth.


\begin{fact}[Multi-shot purity estimation~\cite{cho2025shallow}]
\label{app:fact:multishot_purity}
    Let $\mathcal E\subseteq\mathrm{Cl}(n)$ be a Clifford ensemble whose shadow channel $\mathcal M_{\mathcal E}$ is invertible.
    Sample $U_1,\ldots,U_{N_U}$ independently from $\mathcal E$. For each
    $i=1,\ldots,N_U$, apply $U_i$ to $N_S\geq2$ independent copies of
    $\rho$ and measure them in the computational basis, obtaining outcomes
    $\{b_{i,j}\}_{j=1}^{N_S}$. Define
    \begin{align}
        &\hat p_2
        =
        \frac{1}{N_U}
        \sum_{i=1}^{N_U}\hat p_{2,i},\\
        &\hat p_{2,i}
        =
        \binom{N_S}{2}^{-1}
        \sum_{j>k}
        \operatorname{Tr}\!\left(
            U_i^\dagger\ket{b_{i,j}}\!\bra{b_{i,j}}U_i\,
            \mathcal M^{-1}
            \!\left(
                U_i^\dagger\ket{b_{i,k}}\!\bra{b_{i,k}}U_i
            \right)
        \right).
    \end{align}
    Then $\mathbb E[\hat p_2]=\operatorname{Tr}(\rho^2)$,
    and
    \begin{align}
        \operatorname{Var}(\hat p_2)
        &=
        \frac{1}{N_U}
        \bigg[
            \frac{(N_S-2)(N_S-3)}{N_S(N_S-1)}V_1(\rho)
            +
            \frac{4(N_S-2)}{N_S(N_S-1)}V_2(\rho)
            +
            \frac{2}{N_S(N_S-1)}V_3(\rho)
            -
            \operatorname{Tr}(\rho^2)^2
        \bigg]\\
        &\leq
        \frac{1}{N_U}
        \left(
            V_1(\rho)
            +
            \frac{4}{N_S}V_2(\rho)
            +
            \frac{2}{(N_S-1)^2}V_3(\rho)
        \right).
    \end{align}
    The three quantities $V_1$, $V_2$, and $V_3$ are defined by
    \begin{align}
        V_1(\rho)
        &:=
        \frac{1}{d^2}
        \sum_{P,Q}
        \operatorname{Tr}(\rho P)^2
        \operatorname{Tr}(\rho Q)^2
        \tau(P,Q)/(m_P m_Q),\\
        V_2(\rho)
        &:=
        \frac{1}{d^2}
        \sum_{P,Q}
        \operatorname{Tr}(\rho P)
        \operatorname{Tr}(\rho Q)
        \operatorname{Tr}(\rho PQ)
        \tau(P,Q)/(m_P m_Q),\\
        V_3(\rho)
        &:=
        \frac{1}{d^2}
        \sum_{P,Q}
        \operatorname{Tr}(\rho PQ)^2
        \tau(P,Q)/(m_P m_Q).
    \end{align}
    Moreover, the Cauchy--Schwarz inequality gives
    \begin{equation}
        V_1(\rho)
        \leq
        V_2(\rho)
        \leq
        V_3(\rho).
    \end{equation}
\end{fact}

\begin{theorem}[Unbiased multi-shot purity estimation]
\label{app:thm:unbiased_purity}
    Consider the two-layer Clifford ensemble $\mathcal E_{2\mathrm L}$ with block size $k$ satisfying $k2^{k/2}=\Omega(n)$. For every $0<\epsilon,\delta<1$, choosing
    \begin{equation}
        N_U
        =
        O\left(
            \frac{1}{\delta\epsilon^2}
        \right),
        \qquad
        N_S
        =
        \sqrt d,
        \qquad
        T
        =
        O\left(
            \frac{\sqrt d}{\delta\epsilon^2}
        \right)
        \label{app:eq:purity_sample_complexity}
    \end{equation}
    gives
    \begin{equation}
        \Pr\left[
            \left|
                \hat p_2-\operatorname{Tr}(\rho^2)
            \right|
            >
            \epsilon
        \right]
        \leq
        \delta.
    \end{equation}
\end{theorem}

\begin{proof}
    Applying Eq.~\eqref{app:eq:cliff_Ord_variance} together with Theorem~\ref{main:thm:1} for $O_{\mathrm{thm}}=\rho_{\mathrm{thm}}=\rho$ gives
    \begin{align}
        V_2(\rho)
        &=
        \operatorname{Var}
        \left[
            \operatorname{Tr}(\rho\hat\rho)
        \right]
        +
        \bigl(\operatorname{Tr}(\rho^2)\bigr)^2\\
        &=O(\|\rho-I/d\|_{2}^2+1)=O(1).
    \end{align}
    Then $V_1(\rho)\leq V_2(\rho)$ gives $V_1(\rho)=O(1)$. By Fact~\ref{app:fact:first_term},
    \begin{align}
        V_3(\rho) 
        &=
        \frac{1}{d^2}
        \sum_{P,Q}
        \operatorname{Tr}(\rho PQ)^2
        \tau(P,Q)/(m_P m_Q)\\
        &=
        \operatorname{Tr}\left(
            \rho\,
            \mathbb E_{U,b}
            \left[
                \hat\rho^2
            \right]
        \right)\le
        \|\rho\|_1
        \left
        \|\mathbb E_{U,b}[\hat\rho^2]
        \right\|_{\infty}\\
        &=O(d)
    \end{align}
    Substituting these bounds into
    Fact~\ref{app:fact:multishot_purity} yields
    \begin{align}
        \operatorname{Var}[\hat p_2]
        &=
        O\left(
            \frac{1}{N_U}
            \left(
                1
                +
                \frac{1}{N_S}
                +
                \frac{d}{(N_S-1)^2}
            \right)
        \right)
        \nonumber\\
        &=
        O\left(
            \frac{1}{N_U}
            \left(
                1+\frac{d}{N_S^2}
            \right)
        \right).
    \end{align}
    For $N_S=\sqrt d$, this becomes
    $\operatorname{Var}[\hat p_2]=O(1/N_U)$.
    Chebyshev's inequality then gives
    \begin{equation}
        \Pr\left[
            \left|
                \hat p_2-\operatorname{Tr}(\rho^2)
            \right|
            >
            \epsilon
        \right]
        =
        O\left(
            \frac{1}{N_U\epsilon^2}
        \right).
    \end{equation}
    Taking $N_U=O(1/(\delta\epsilon^2))$ proves the claim.
\end{proof}

\end{document}